\documentclass[11pt]{book}

\usepackage{enumitem, imakeidx, framed}
\PassOptionsToPackage{svgnames}{xcolor}
\usepackage{tikz, fancyhdr, stmaryrd, cmll}
\usepackage[a4paper]{geometry}
\usepackage[font=footnotesize]{caption}
\usepackage[british]{babel}
\usepackage{amsfonts}
\usepackage{amsmath}
\usepackage{amssymb}
\usepackage{array}
\usepackage[hidelinks]{hyperref}
\usepackage{amsthm}
\usepackage{mathtools}
\usepackage{mathrsfs}
\usepackage[svgnames]{xcolor}
\usepackage{tikz}
\usetikzlibrary{angles,quotes,shapes,shapes.misc,positioning}
\usepackage{stmaryrd}
\usepackage{longtable}
\usepackage{setspace}
\usepackage[utf8]{inputenc}
\usepackage[section]{placeins}
\usepackage{authblk}
\usepackage{booktabs}
\usepackage{csquotes}
\usepackage{comment}
\usepackage{tikz-cd}
\usepackage{float}
\usepackage{thm-restate}
\usepackage{placeins}
\usepackage[backend=biber,bibencoding=utf8,style=alphabetic]{biblatex}
\usepackage{chngcntr}
\usepackage{apptools}

\allowdisplaybreaks

\newcommand{\vc}[1]{
    \raisebox{-.5\height}{#1}
}

\let\Oldsection\section
\renewcommand{\section}{\FloatBarrier\Oldsection}

\let\Oldsubsection\subsection
\renewcommand{\subsection}{\FloatBarrier\Oldsubsection}

\usetikzlibrary{arrows, decorations.markings, shapes.geometric, decorations.pathmorphing, decorations.pathreplacing, intersections, patterns,calc, backgrounds}
\tikzcdset{arrow style=tikz, diagrams={>=latex}}

\pgfdeclarelayer{bg}
\pgfdeclarelayer{mid}
\pgfdeclarelayer{fg}
\pgfdeclarelayer{edgelayer}
\pgfdeclarelayer{nodelayer}
\pgfsetlayers{bg,edgelayer,mid,nodelayer,main,fg}

\tikzset{
	0c/.style={circle, draw, fill, inner sep=1.5pt},
	1c/.style={->, thick, shorten <=2pt, shorten >=2pt},
	1cboth/.style={<->, thick, shorten <=2pt, shorten >=2pt},
	1clong/.style={->, thick},
	1cthin/.style={->, shorten <=4pt, shorten >=4pt},
	1cdot/.style={->, dashed, thick, shorten <=2pt, shorten >=2pt},
	1cinc/.style={right hook->, thick, shorten <=2pt, shorten >=2pt},
	1cincl/.style={left hook->, thick, shorten <=2pt, shorten >=2pt},
	follow/.style={->, >=stealth, very thick, shorten <=3pt, shorten >=3pt, color=magenta},
	2c/.style={double, thick, shorten <=6pt, shorten >=8pt, decoration={markings,mark=at position -6pt with {\arrow[scale=1.75]{>}}}, preaction={decorate}},
	2cdot/.style={double, dashed, thick, shorten <=10pt, shorten >=10pt, decoration={markings,mark=at position -8pt with {\arrow[scale=1.75]{>}}}, preaction={decorate}},
	3c1/.style={thick, double, double distance=3pt, shorten <=9pt, shorten >=11pt},
    	3c2/.style={thick, shorten <=9pt, shorten >=10pt},
	3c3/.style={shorten <=9pt, shorten >=10pt, decoration={markings,mark=at position -8pt with {\arrow[scale=3]{>}}},preaction={decorate}},
	4c1/.style={thick, double, double distance=4pt, shorten <=1pt, shorten >=2.75pt},
	4c2/.style={thick, double, double distance=1pt, shorten <=1pt, shorten >=1.25pt, decoration={markings,mark=at position -.05pt with {\arrow[scale=3,ultra thin]{>}}},preaction={decorate}},
	edge/.style={line width=.8pt, color=black},
	edgedot/.style={densely dotted, line width=.8pt, color=black},
	edgethdot/.style={densely dotted, line width=.4pt, color=gray},
	edgeth/.style={line width=.4pt, color=gray!60},
	edgethin/.style={line width=.8pt, color=gray!60},
	edgedotdark/.style={densely dotted, line width=.8pt, color=gray!80},
	dot/.style={circle, draw=black, line width=.8pt, fill=white, inner sep=1.7pt},
	dotth/.style={circle, draw=gray!60, fill=gray!60, inner sep=1.5pt},
	dotwh/.style={circle, draw=gray!60, line width=.4pt, fill=white, inner sep=1.7pt},
	dotwhite/.style={circle, draw=black, line width=.8pt, fill=white, inner sep=1.8pt},
	dotdark/.style={circle, draw, fill=black, inner sep=1.5pt},
	dotgrey/.style={circle, draw=black, line width=.8pt, fill=gray!60, inner sep=1.8pt},
	trian/.style={regular polygon,regular polygon sides=3,shape border rotate=0,fill=white, line width=.8pt, draw=black, inner sep=1.8pt},
	trianh/.style={regular polygon,regular polygon sides=3,shape border rotate=0,fill=white, draw=gray!60, line width=.4pt, inner sep=1.8pt},
	trib/.style={regular polygon,regular polygon sides=3,shape border rotate=0,fill=black, draw, inner sep=1.5pt},
	tribh/.style={regular polygon,regular polygon sides=3,shape border rotate=0,fill=gray!60, draw=gray!60, inner sep=1.5pt},
	tribco/.style={regular polygon,regular polygon sides=3,shape border rotate=180,fill=black, draw, inner sep=1.5pt},
	cover/.style={circle, draw=gray!10, fill=gray!10, inner sep=3.5pt},
	coverc/.style={circle, draw=gray!80, line width=.4pt, inner sep=3pt},
	coverch/.style={circle, draw=gray!30, line width=.4pt, inner sep=3pt},
	coverb/.style={circle, draw=gray!80, line width=.4pt, fill=gray!10, inner sep=3.5pt},
	every node/.style={font={\scriptsize}},
	every matrix/.append style={nodes={font=\normalsize}}
}
\def\roundingAmount{0.5em}
\def\paddingAmount{1em}
\def\sepAmount{0.2em}
\def\minHeight{1em}
\def\minWidth{1em}
\def\tinyBoxSize{7pt}
\def\tinyCircleSize{7pt}
\definecolor{zx_red}{RGB}{232, 165, 165}
\definecolor{zx_green}{RGB}{216, 248, 216}
\definecolor{nice_green}{RGB}{116, 148, 116}
\def\hyellow{yellow!30}
\tikzstyle{spider}=[rectangle,rounded corners=\roundingAmount,fill=gray!1,draw=Black,
  line width=0.4 pt,inner sep=\sepAmount,minimum width=\paddingAmount,minimum height=\paddingAmount]
\tikzstyle{rect}=[rectangle,fill=gray!1,draw=Black,inner sep=\sepAmount,minimum width=\minWidth,minimum height=\minHeight,
  line width=0.4 pt]
\tikzstyle{box}=[rect,fill=gray!10, line width=0.4 pt]
\tikzstyle{trap}=[trapezium,trapezium angle=80,fill=gray!1,draw=Black,inner sep=\sepAmount,minimum width=\minWidth,minimum height=\minHeight,
  line width=0.4 pt]
\tikzstyle{roundedtrap}=[rectangle,fill=none,draw=none,
  line width=0.4 pt,inner sep=\sepAmount,minimum width=\paddingAmount,minimum height=\paddingAmount,
{path picture={
	\coordinate (A) at (path picture bounding box.north west);
  	\coordinate (B) at (path picture bounding box.north east);
	\coordinate (D) at (path picture bounding box.south west);
  	\coordinate (C) at (path picture bounding box.south east);
  	\coordinate (E) at (path picture bounding box.south);
	  \draw[fill=gray!1,draw=none,rounded corners=\roundingAmount] (A) -- (D) -- (C) -- (B);
	  \draw[rounded corners=\roundingAmount,line width=0.4pt,transform canvas={xshift = 0.2pt, yshift=0.2pt}] (A) -- (D) -- (E);
	  \draw[rounded corners=\roundingAmount,line width=0.4pt,transform canvas={xshift = -0.2pt, yshift=0.2pt}] (E) -- (C) -- (B);
	  \draw[black,line width=1.6pt] (A) -- (B);
}}]
\tikzstyle{roundedtrapT}=[roundedtrap,
{path picture={
	\coordinate (A) at (path picture bounding box.south west);
  	\coordinate (B) at (path picture bounding box.south east);
	\coordinate (D) at (path picture bounding box.north west);
  	\coordinate (C) at (path picture bounding box.north east);
  	\coordinate (E) at (path picture bounding box.north);
	  \draw[fill=gray!1,draw=none,rounded corners=\roundingAmount] (A) -- (D) -- (C) -- (B);
	  \draw[rounded corners=\roundingAmount,line width=0.4pt,transform canvas={xshift = 0.2pt, yshift=-0.2pt}] (A) -- (D) -- (E);
	  \draw[rounded corners=\roundingAmount,line width=0.4pt,transform canvas={xshift = -0.2pt, yshift=-0.2pt}] (E) -- (C) -- (B);
	  \draw[black,line width=1.6pt] (A) -- (B);
}}]
\tikzstyle{triangle}=[regular polygon,regular polygon sides=3,draw=Black,inner sep=\sepAmount,minimum width=\minWidth,minimum height=\minHeight,
  line width=0.4 pt]
\tikzstyle{trianglet}=[triangle,shape border rotate=180]
\tikzstyle{ZH}=[rect]
\tikzstyle{wide ZH}=[rect,minimum width=2.5em]
\tikzstyle{smallH}=[rect, minimum width=\tinyBoxSize, minimum height =\tinyBoxSize]
\tikzstyle{smallZH}=[smallH]
\tikzstyle{h2}=[smallH]
\tikzstyle{yh}=[rect, minimum width=\tinyBoxSize, minimum height =\tinyBoxSize,fill=\hyellow]
\tikzstyle{smallCircle}=[circle,draw=Black, minimum width=\tinyCircleSize, line width=0.4 pt, inner sep=0em]
\tikzstyle{smallZ}=[smallCircle,fill=gray!1]
\tikzstyle{Z}=[smallZ]
\tikzstyle{wide point}=[fill=white,draw,shape=isosceles triangle,shape border rotate=-90,isosceles triangle stretches=true,inner sep=0pt,minimum width=2cm,minimum height=6.12mm,yshift=-0.0mm]
\tikzstyle{bbindex}=[font={\color{blue}\footnotesize}]
\tikzstyle{medium gray box}=[rect, minimum width=4em,fill=gray!30]
\tikzstyle{semilarge box}=[rect, minimum width=4em]
\tikzstyle{wide box}=[rect, minimum width=6em]
\tikzstyle{small gray box}=[rect, fill=gray!30]
\tikzstyle{H}=[rect,fill=\hyellow]
\tikzstyle{h}=[smallH,fill=\hyellow]
\tikzstyle{gn}=[spider,fill=zx_green]
\tikzstyle{rn}=[spider,fill=zx_red]
\tikzstyle{white}=[spider,fill=gray!1]
\tikzstyle{smallWhite}=[smallZ]
\tikzstyle{smallwhite}=[smallZ]
\tikzstyle{grey}=[spider,fill=gray!30]
\tikzstyle{smallGrey}=[smallCircle,fill=gray!30]
\tikzstyle{smallgrey}=[smallGrey]
\tikzstyle{black}=[spider,fill=gray!70]
\tikzstyle{small black}=[smallCircle,fill=gray!70]
\tikzstyle{smallBlack}=[small black]
\tikzstyle{smallblack}=[small black]
\tikzstyle{b}=[black]
\tikzstyle{w}=[white]
\tikzstyle{qn}=[trap]
\tikzstyle{qnt}=[trap,shape border rotate=180]
\tikzstyle{string}=[circle,fill=gray,draw=gray,inner sep=1pt]
\tikzstyle{net}=[rectangle,draw=White,minimum width=1.5em,minimum height=1.5em,fill=white]
\tikzstyle{none}=[inner sep=0pt]
\tikzstyle{crossing}=[circle,minimum width=\tinyBoxSize,draw=black,{path picture={ }}]
\tikzstyle{bbox}=[rectangle,draw=blue!60!white,minimum width=1em,minimum height=1em]
\tikzstyle{polynomial}=[roundedtrap]
\tikzstyle{polynomialT}=[roundedtrapT]
\tikzstyle{poly}=[polynomial]
\tikzstyle{polyT}=[polynomialT]

\tikzstyle{arrow}=[->,draw=black]
\tikzstyle{hadamard edge}=[-, dashed, dash pattern=on 2pt off 1pt, thick, draw={rgb,255: red,68; green,136; blue,255}]
\tikzstyle{light-arrow}=[->,draw=gray]
\tikzstyle{blue}=[draw=blue!30!white]

\pgfdeclarearrow{
  name=reduce,
  parameters={\the\pgfarrowlength},  
  setup code={
   \pgfarrowssettipend{0pt}
   \pgfarrowssetlineend{-\pgfarrowlength}
   \pgfarrowlinewidth=\pgflinewidth
   \pgfarrowssavethe\pgfarrowlength
  },
  drawing code={
   \pgfpathmoveto{\pgfpoint{0\pgfarrowlength}{0\pgfarrowlength}}
   \pgfpathlineto{\pgfpoint{-1\pgfarrowlength}{0.5\pgfarrowlength}}
   \pgfpathlineto{\pgfpoint{-1\pgfarrowlength}{-0.5\pgfarrowlength}}
   \pgfpathlineto{\pgfpoint{0\pgfarrowlength}{0\pgfarrowlength}}
   \pgfusepathqstroke
  },
  defaults = { length = 3pt }
}

\def\bbZ{\mathbb{Z}}
\def\less{\setminus}
\def\bbD{\mathbb{D}}
\def\bbQ{\mathbb{Q}}
\def\bbE{\mathbb{E}}
\def\ZW{\text{ZW}}
\def\id{\text{id}}
\def\Gal{\text{Gal}}
\def\bbC{\mathbb{C}}
\def\bbN{\mathbb{N}}
\def\bbL{\mathbb{L}}
\def\bbR{\mathbb{R}}
\def\bbP{\mathbb{P}}

\def\Qubit{\Cat{Qubit}}
\newcommand\PRGC[1]{\mbox{\ensuremath{\Cat{Ring}\text{-}\Cat{GC} \left( \text{#1} \right) }}}
\newcommand\SigmaGC[1]{\mbox{\ensuremath{\text{$\Sigma$-}\Cat{GC} \left( #1 \right) }}}

\def\Hilbert{\mathbb{H}}
\def\F2{{\mathbb{F}_2}}
\newcommand\idot[1]{\ensuremath{\interpret{\cdot}_{#1}}}
\renewcommand{\cal}[1]{\mathcal{#1}}
\def\Rules{\ensuremath{\cal{R}}}
\def\ZW{\ensuremath{\text{ZW}}}
\def\ZH{\ensuremath{\text{ZH}}}
\def\ZHR{\ensuremath{\text{ZH}_R}}
\def\ZWR{\ensuremath{\text{ZW}_R}}
\def\ZHC{\text{ZH}_\bbC}

\def\ZQ{\ensuremath{\text{ZQ}}}

\def\ZX{\ensuremath{\text{ZX}}}

\def\entails{\vdash}
\def\syntactic{\vdash}
\def\semantic{\vDash}

\def\from{\leftarrow}
\newcommand{\set}[1]{\left\{#1\right\}}
\newcommand{\bb}[1]{\mathbb{#1}}
\newcommand{\at}{|}
\newcommand{\family}[1]{\set{\quad #1 \quad}}
\def\alphadelta{\alpha_1, \dots, \alpha_n, \delta_1, \dots, \delta_m}
\newcommand{\inv}{^{-1}}
\newcommand{\abs}[1]{|#1|}
\newcommand{\comp}{\circ}
\newcommand{\issound}{\text{\ is sound}}

\newcommand{\rto}{\redto}

\newcommand{\redto}{\redtosub{{}}}
\newcommand{\redtosub}[1]{\redarsubsup{{}}{#1}{->}}

\newcommand{\redarsubsup}[3]{\,\begin{tikzpicture}[baseline=-0.25em,scale=0.6, every node/.style={transform shape}]
\node [style=none] at (0.75, -0.2) {$_#2$};
\draw [#3,>=reduce] (0,0) to  node[above] {#1} (0.75,0);
\end{tikzpicture}\,}
\newcommand{\interpret}[1]{\left\llbracket \; #1 \; \right\rrbracket}
\def\tensor{\ensuremath\otimes}

\def\iso{\cong}

\newcommand\cat[1]{\ensuremath{\mathbf{#1}}}
\newcommand\Cat[1]{\cat{#1}}

\newcommand\rmod[1]{\ensuremath{#1\,\cat{Mod}}}

\newcommand\Model{\ensuremath{\cal{M}}}

\newcommand\braket[2]{\langle\, #1 \,|\, #2 \,\rangle}
\newcommand\bra[1]{\langle\, #1\,|}
\newcommand\ket[1]{|\, #1 \, \rangle}

\newcommand\restr[2]{{
  \left.\kern-\nulldelimiterspace 
  #1 
  \vphantom{\big|} 
  \right|_{#2}
  }}
\renewcommand*{\arraystretch}{1.2}

\newcommand\va[2][1]{
\scalebox{#1}{\raisebox{-.5\height}{
#2
}}
}

\newcommand\by[1]{\;
    \raisebox{-.5\height}{$
        \stackrel{=}{
            \mbox{\scriptsize$#1$}
        }$
    }\;
}
\newcommand\interpretedas{\;
    \raisebox{0em}{$
        \stackrel{
            \scalebox{0.5}{$\idot{}$}
        }{\longmapsto}$
    }\;
}

\def\st{\text{s.t.\ }}
\def\Hom{\text{Hom}}

\def\Product{\prod}
\def\infinity{\infty}
\def\w{\omega}
\def\between{\leftrightarrow}
\newcommand\ring{\text{\textsc{ring}}}

\def\rprop{\text{\textsc{ring-prop}}}
\def\Zero{\text{\textsc{zero}}}
\def\rpropint{\rprop^{\interpret{}}}
\def\monprop{\text{\textsc{monoid-prop}}}
\def\groupprop{\text{\textsc{v-group-prop}}}

\def\injects{\hookrightarrow}
\def\into{\injects}

\def\Mat{\text{Mat}}
\newcommand\maxdegree[3]{\text{deg}^{#1}_{#2}\left [ \;#3\; \right ]}
\def\ie{\text{i.e.}\quad}

\def\ZW{\text{ZW}}
\def\ZX{\text{ZX}}

\def\ZH{\text{ZH}}

\newcommand{\QRingZH}[1][]{
	\ensuremath{\ring_{ZH}}
}
\newcommand\half[1][1]{\frac{#1}{2}}
\newcommand\piby[1]{\pi/#1}

\newcommand\bit[1]{\ensuremath{#1\text{-bit}}}
\newcommand\bits[1]{\ensuremath{#1\text{-bits}}}
\def\SO3{\ensuremath{\text{SO3}(\bbR)}}
\def\SU2{\ensuremath{\text{SU2}(\bbC)}}
\def\UQuat{\ensuremath{\hat Q}}

\def\chapterbullet{$\triangleright$}
\def\brag{$\triangleright$}
\def\contradiction{\raisebox{-0.25\height}{
	\begin{tikzpicture}[scale=0.2]
	\begin{pgfonlayer}{nodelayer}
		\node [style=none] (0) at (-0.5, 1) {};
		\node [style=none] (1) at (-1, 0.5) {};
		\node [style=none] (2) at (-1, -0.5) {};
		\node [style=none] (3) at (-0.5, -1) {};
		\node [style=none] (4) at (0.5, -1) {};
		\node [style=none] (5) at (1, -0.5) {};
		\node [style=none] (6) at (1, 0.5) {};
		\node [style=none] (7) at (0.5, 1) {};
	\end{pgfonlayer}
	\begin{pgfonlayer}{edgelayer}
		\draw (0.center) to (5.center);
		\draw (1.center) to (4.center);
		\draw (2.center) to (7.center);
		\draw (6.center) to (3.center);
	\end{pgfonlayer}
\end{tikzpicture}}}

\def\rowgap{2em}

\DeclareUnicodeCharacter{27E9}{\rangle}
\DeclareUnicodeCharacter{3C0}{\pi}

\newcommand{\node}[2]{
	\raisebox{-0.5\height}{
		\begin{tikzpicture}[scale=0.3]
			\begin{pgfonlayer}{nodelayer}
				\node [style=#1] (1)   at (0, 0)   {$#2$};
				\node [style=none] (2)   at (0, 2)   {};
				\node [style=none] (3)   at (0, -2)   {};
			\end{pgfonlayer}
			\begin{pgfonlayer}{edgelayer}
				\draw [-] (1.center) to (2.center);
				\draw [-] (1.center) to (3.center);
			\end{pgfonlayer}
		\end{tikzpicture}
	}
}
\newcommand{\smallnode}[2]{
	\raisebox{-0.5\height}{
		\begin{tikzpicture}[scale=0.3]
			\begin{pgfonlayer}{nodelayer}
				\node [style=#1] (1)   at (0, 0)   {$#2$};
				\node [style=none] (2)   at (0, 1.25)   {};
				\node [style=none] (3)   at (0, -1.25)   {};
			\end{pgfonlayer}
			\begin{pgfonlayer}{edgelayer}
				\draw [-] (1.center) to (2.center);
				\draw [-] (1.center) to (3.center);
			\end{pgfonlayer}
		\end{tikzpicture}
	}
}

\newcommand{\smallstate}[2]{
	\raisebox{-0.25\height}{
		\begin{tikzpicture}[scale=0.3]
			\begin{pgfonlayer}{nodelayer}
				\node [style=#1] (1)   at (0, 0)   {$#2$};
				\node [style=none] (2)   at (0,1)   {};
			\end{pgfonlayer}
			\begin{pgfonlayer}{edgelayer}
				\draw [-] (1.center) to (2.center);
			\end{pgfonlayer}
		\end{tikzpicture}
	}
}

\newcommand{\state}[2]{
	\vc{
		\begin{tikzpicture}[scale=0.3]
			\begin{pgfonlayer}{nodelayer}
				\node [style=#1] (1)   at (0, 0)   {$#2$};
				\node [style=none] (2)   at (0,2)   {};
			\end{pgfonlayer}
			\begin{pgfonlayer}{edgelayer}
				\draw [-] (1.center) to (2.center);
			\end{pgfonlayer}
		\end{tikzpicture}
	}
}

\newcommand{\smalleffect}[2]{
	\vc{
		\begin{tikzpicture}[scale=0.3]
			\begin{pgfonlayer}{nodelayer}
				\node [style=#1] (1)   at (0, 0)   {$#2$};
				\node [style=none] (2)   at (0,-1)   {};
			\end{pgfonlayer}
			\begin{pgfonlayer}{edgelayer}
				\draw [-] (1.center) to (2.center);
			\end{pgfonlayer}
		\end{tikzpicture}
	}
}

\newcommand{\spider}[2]{
	\raisebox{-0.5\height}{
		\begin{tikzpicture}[scale=0.3]
			\begin{pgfonlayer}{nodelayer}
				\node [style=#1] (1)   at (0, 0)   {$#2$};
				\node [style=none] (2)   at (-1, 2)   {};
				\node [style=none] ()   at (0, 2)   {$\dots$};
				\node [style=none] (3)   at (1, 2)   {};
				\node [style=none] (4)   at (-1, -2)   {};
				\node [style=none] ()   at (0, -2)   {$\dots$};
				\node [style=none] (5)   at (1, -2)   {};
			\end{pgfonlayer}
			\begin{pgfonlayer}{edgelayer}
				\draw [-] (1.center) to (2.center);
				\draw [-] (1.center) to (3.center);
				\draw [-] (1.center) to (4.center);
				\draw [-] (1.center) to (5.center);
			\end{pgfonlayer}
		\end{tikzpicture}
	}
}

\newcommand{\scalar}[2]{
	\va{
		\begin{tikzpicture}[scale=0.3]
			\begin{pgfonlayer}{nodelayer}
				\node [style=#1] (1)   at (0, 0)   {$#2$};
			\end{pgfonlayer}
		\end{tikzpicture}
	}
}

\newcommand{\galpharpi}[1]{
	\vc{
		\begin{tikzpicture}[scale=0.3]
			\begin{pgfonlayer}{nodelayer}
				\node [style=gn] (1)   at (0, 1.2)   {$#1$};
				\node [style=gn] (2)   at (0, -1.2)   {$\pi$};
			\end{pgfonlayer}
			\begin{pgfonlayer}{edgelayer}
				\draw [hadamard edge] (1.center) to (2.center);
			\end{pgfonlayer}
		\end{tikzpicture}
	}
}

\newcommand{\rg}[2]{
	\vc{
		\begin{tikzpicture}[scale=0.3]
			\begin{pgfonlayer}{nodelayer}
				\node [style=gn] (1)   at (0, 1.2)   {$#1$};
				\node [style=gn] (2)   at (0, -1.2)   {$#2$};
			\end{pgfonlayer}
			\begin{pgfonlayer}{edgelayer}
				\draw [hadamard edge] (1.center) to (2.center);
			\end{pgfonlayer}
		\end{tikzpicture}
	}
}

\newcommand{\spidermn}[4]{
	\vc{
		\begin{tikzpicture}[scale=0.3]
			\begin{pgfonlayer}{nodelayer}
				\node [style=#1] (1)   at (0, 0)   {$#2$};
				\node [style=none] (2)   at (-1, 2)   {};
				\node [style=none] ()   at (0, 2.5)   {$\overset{#4}{\dots}$};
				\node [style=none] (3)   at (1, 2)   {};
				\node [style=none] (4)   at (-1, -2)   {};
				\node [style=none] ()   at (0, -2.5)   {$\underset{#3}{\dots}$};
				\node [style=none] (5)   at (1, -2)   {};
			\end{pgfonlayer}
			\begin{pgfonlayer}{edgelayer}
				\draw [-] (1.center) to (2.center);
				\draw [-] (1.center) to (3.center);
				\draw [-] (1.center) to (4.center);
				\draw [-] (1.center) to (5.center);
			\end{pgfonlayer}
		\end{tikzpicture}
	}
}

\newcommand{\wspider}[3]{
	\begin{tikzpicture}[scale=0.3]
		\begin{pgfonlayer}{nodelayer}
			\node [style=#1] (1)   at (0, 0)   {$#2$};
			\node [style=none] (2)   at (-2, 2)   {};
			\node [style=none] ()   at (0, 2)   {$\overset{#3}{\dots}$};
			\node [style=none] (3)   at (2, 2)   {};
		\end{pgfonlayer}
		\begin{pgfonlayer}{edgelayer}
			\draw [-] (1.center) to (2.center);
			\draw [-] (1.center) to (3.center);
		\end{pgfonlayer}
	\end{tikzpicture}
}

\newcommand{\ZHExample}[1]{
	\begin{tikzpicture}[scale=0.65]
		\node [style=ZH] (0) at (-3, 0) {};
		\node [style=ZH] (1) at (-2, -1) {};
		\node [style=none] (3) at (-2, -2) {};
		\node [style=ZH] (8) at (-3, 1) {X};
		\node [style=white] (9) at (-1, 0) {};
		\node [style=ZH] (10) at (-1, 1) {-X};
		\node [style=none] (15) at (-1.5, 1.5) {};
		\node [style=none] (16) at (-1.5, 0.5) {};
		\node [style=none] (17) at (-0.5, 0.5) {};
		\node [style=none] (18) at (-0.5, 1.5) {};
		\node [style=none] (19) at (0, -0.5) {=};
		\node [style=ZH] (20) at (1, 0) {};
		\node [style=ZH] (21) at (2, -1) {};
		\node [style=none] (22) at (2, -2) {};
		\node [style=ZH] (23) at (1, 1) {X};
		\node [style=white] (24) at (3, 0) {};
		\node [style=ZH] (25) at (3, 1) {-X};
		\node [style=none] (30) at (2.5, 1.5) {};
		\node [style=none] (31) at (2.5, 0.5) {};
		\node [style=none] (32) at (3.5, 0.5) {};
		\node [style=none] (33) at (3.5, 1.5) {};
		\draw (1) to (3.center);
		\draw (0.center) to (1.center);
		\draw (8.center) to (0.center);
		\draw (10.center) to (9.center);
		\draw (9.center) to (1.center);
		\draw [style=blue] (15.center) to (16.center);
		\draw [style=blue] (16.center) to (17.center);
		\draw [style=blue] (17.center) to (18.center);
		\draw [style=blue] (18.center) to (15.center);
		\draw (21.center) to (22.center);
		\draw (20.center) to (21.center);
		\draw (23.center) to (20.center);
		\draw (24.center) to (21.center);
		\draw [style=blue] (30.center) to (31.center);
		\draw [style=blue] (31.center) to (32.center);
		\draw [style=blue] (32.center) to (33.center);
		\draw [style=blue] (33.center) to (30.center);
		\node [style=ZH] (0) at (-3, 0) {};
		\node [style=ZH] (1) at (-2, -1) {};
		\node [style=none] (3) at (-2, -2) {};
		\node [style=ZH] (8) at (-3, 1) {X};
		\node [style=white] (9) at (-1, 0) {};
		\node [style=ZH] (10) at (-1, 1) {-X};
		\node [style=none] (15) at (-1.5, 1.5) {};
		\node [style=none] (16) at (-1.5, 0.5) {};
		\node [style=none] (17) at (-0.5, 0.5) {};
		\node [style=none] (18) at (-0.5, 1.5) {};
		\node [style=none] (19) at (0, -0.5) {=};
		\node [style=ZH] (20) at (1, 0) {};
		\node [style=ZH] (21) at (2, -1) {};
		\node [style=none] (22) at (2, -2) {};
		\node [style=ZH] (23) at (1, 1) {X};
		\node [style=white] (24) at (3, 0) {};
		\node [style=ZH] (25) at (3, 1) {-X};
		\node [style=none] (30) at (2.5, 1.5) {};
		\node [style=none] (31) at (2.5, 0.5) {};
		\node [style=none] (32) at (3.5, 0.5) {};
		\node [style=none] (33) at (3.5, 1.5) {};
	\end{tikzpicture}
}

\newcommand{\zxtwospider}[2]{
	\vc{
		\begin{tikzpicture}
			\begin{pgfonlayer}{nodelayer}
				\node [style=gn] (0) at (-1, 0.25) {$#1$};
				\node [style=gn] (1) at (-1, -0.25) {$#2$};
				\node [style=gn] (2) at (1, 0) {$#1 + #2$};
				\node [style=none] (3) at (0, 0) {=};
				\node [style=none] (4) at (-1.5, 0.75) {};
				\node [style=none] (5) at (-0.5, 0.75) {};
				\node [style=none] (6) at (-1.5, -0.75) {};
				\node [style=none] (7) at (-1, -0.75) {};
				\node [style=none] (8) at (-0.5, -0.75) {};
				\node [style=none] (9) at (0.5, -0.75) {};
				\node [style=none] (10) at (1, -0.75) {};
				\node [style=none] (11) at (1.5, -0.75) {};
				\node [style=none] (12) at (0.5, 0.75) {};
				\node [style=none] (13) at (1.5, 0.75) {};
			\end{pgfonlayer}
			\begin{pgfonlayer}{edgelayer}
				\draw (4.center) to (0.center);
				\draw (5.center) to (0.center);
				\draw (0.center) to (1.center);
				\draw (1.center) to (6.center);
				\draw (1.center) to (7.center);
				\draw (1.center) to (8.center);
				\draw (12.center) to (2.center);
				\draw (2.center) to (13.center);
				\draw (2.center) to (9.center);
				\draw (2.center) to (10.center);
				\draw (2.center) to (11.center);
			\end{pgfonlayer}
		\end{tikzpicture}
	}
}

\newcommand{\bap}[4]{
	\vc{
		\begin{tikzpicture}[scale=0.3]
			\begin{pgfonlayer}{nodelayer}
				\node [style=#1] (1)   at (0, 0)   {$#2$};
				\node [style=none] (2)   at (0, 2)   {};
				\node [style=white] (4)   at (-1, -2)   {$#3$};
				\node [style=white] (5)   at (1, -2)   {$#4$};
			\end{pgfonlayer}
			\begin{pgfonlayer}{edgelayer}
				\draw [-] (1.center) to (2.center);
				\draw [-] (1.center) to (4.center);
				\draw [-] (1.center) to (5.center);
			\end{pgfonlayer}
		\end{tikzpicture}
	}
}
\newcommand\bapwide[4]{
	\vc{
		\begin{tikzpicture}[scale=0.3]
			\begin{pgfonlayer}{nodelayer}
				\node [style=#1] (1)   at (0, 0)   {$#2$};
				\node [style=none] (2)   at (0, 2)   {};
				\node [style=white] (4)   at (-2, -2)   {$#3$};
				\node [style=white] (5)   at (2, -2)   {$#4$};
			\end{pgfonlayer}
			\begin{pgfonlayer}{edgelayer}
				\draw [-] (1.center) to (2.center);
				\draw [-] (1.center) to (4.center);
				\draw [-] (1.center) to (5.center);
			\end{pgfonlayer}
		\end{tikzpicture}
	}
}

\newcommand\bapwwide[4]{
	\vc{
		\begin{tikzpicture}[scale=0.3]
			\begin{pgfonlayer}{nodelayer}
				\node [style=#1] (1)   at (0, 0)   {$#2$};
				\node [style=none] (2)   at (0, 2)   {};
				\node [style=white] (4)   at (-3, -2)   {$#3$};
				\node [style=white] (5)   at (3, -2)   {$#4$};
			\end{pgfonlayer}
			\begin{pgfonlayer}{edgelayer}
				\draw [-] (1.center) to (2.center);
				\draw [-] (1.center) to (4.center);
				\draw [-] (1.center) to (5.center);
			\end{pgfonlayer}
		\end{tikzpicture}
	}
}

\def\tripleblobs{
	\vc{
		\begin{tikzpicture}[scale=0.3]
			\begin{pgfonlayer}{nodelayer}
				\node [style=gn] (1)   at (0, 1.2)   {};
				\node [style=rn] (2)   at (0, -1.2)   {};
			\end{pgfonlayer}
			\begin{pgfonlayer}{edgelayer}
				\draw [-] (1.center) to (2.center);
				\draw [-,bend left=35] (1.center) to (2.center);
				\draw [-,bend right=35] (1.center) to (2.center);
			\end{pgfonlayer}
		\end{tikzpicture}
	}
}

\def\halfblobs{\tripleblobs \tripleblobs}

\newcommand{\smallbinary}[2][ZH]{
	\vc{
		\begin{tikzpicture}[scale=0.3]
			\begin{pgfonlayer}{nodelayer}
				\node [style=#1] (1)   at (0, 0)   {$#2$};
				\node [style=none] (2)   at (0, 1)   {};
				\node [style=none] (4)   at (-0.5, -1)   {};
				\node [style=none] (5)   at (0.5, -1)   {};
			\end{pgfonlayer}
			\begin{pgfonlayer}{edgelayer}
				\draw [-] (1.center) to (2.center);
				\draw [-] (1.center) to (4.center);
				\draw [-] (1.center) to (5.center);
			\end{pgfonlayer}
		\end{tikzpicture}
	}
}

\newcommand{\binary}[2][ZH]{
	\vc{
		\begin{tikzpicture}[scale=0.3]
			\begin{pgfonlayer}{nodelayer}
				\node [style=#1] (1)   at (0, 0)   {$#2$};
				\node [style=none] (2)   at (0, 2)   {};
				\node [style=none] (4)   at (-1, -2)   {};
				\node [style=none] (5)   at (1, -2)   {};
			\end{pgfonlayer}
			\begin{pgfonlayer}{edgelayer}
				\draw [-] (1.center) to (2.center);
				\draw [-] (1.center) to (4.center);
				\draw [-] (1.center) to (5.center);
			\end{pgfonlayer}
		\end{tikzpicture}
	}
}

\newcommand{\smallunary}[2][ZH]{
	\vc{
		\begin{tikzpicture}[scale=0.3]
			\begin{pgfonlayer}{nodelayer}
				\node [style=#1] (1)   at (0, 0)   {$#2$};
				\node [style=none] (2)   at (0, 1)   {};
				\node [style=none] (5)   at (0, -1)   {};
			\end{pgfonlayer}
			\begin{pgfonlayer}{edgelayer}
				\draw [-] (1.center) to (2.center);
				\draw [-] (1.center) to (5.center);
			\end{pgfonlayer}
		\end{tikzpicture}
	}
}
\newcommand{\unary}[2][ZH]{
	\vc{
		\begin{tikzpicture}[scale=0.3]
			\begin{pgfonlayer}{nodelayer}
				\node [style=#1] (1)   at (0, 0)   {$#2$};
				\node [style=none] (2)   at (0, 2)   {};
				\node [style=none] (5)   at (0, -2)   {};
			\end{pgfonlayer}
			\begin{pgfonlayer}{edgelayer}
				\draw [-] (1.center) to (2.center);
				\draw [-] (1.center) to (5.center);
			\end{pgfonlayer}
		\end{tikzpicture}
	}
}

\newcommand{\nary}[2][ZH]{
	\vc{
		\begin{tikzpicture}[scale=0.3]
			\begin{pgfonlayer}{nodelayer}
				\node [style=#1] (1)   at (0, 0)   {$#2$};
				\node [style=none] (2)   at (0, 2)   {};
				\node [style=none] ()   at (0, -2)   {$\dots$};
				\node [style=none] (4)   at (-1, -2)   {};
				\node [style=none] (5)   at (1, -2)   {};
			\end{pgfonlayer}
			\begin{pgfonlayer}{edgelayer}
				\draw [-] (1.center) to (2.center);
				\draw [-] (1.center) to (4.center);
				\draw [-] (1.center) to (5.center);
			\end{pgfonlayer}
		\end{tikzpicture}
	}
}

\newcommand{\boundaryAddNone}[2]{
	\vc{
\begin{tikzpicture}
	\begin{pgfonlayer}{nodelayer}
		\node [style=box] (g1) at (0, 0.25) {$#1$};
		\node [style=gn] (0) at (0, -0.75) {$#2$};
		\node [style=none] (1) at (0, -0.25) {$\,\dots$};
		\node [style=none] (4) at (0, 0.75) {$\,\dots$};
		\node [style=none] (5) at (-0.5, 0.75) {};
		\node [style=none] (6) at (0.5, 0.75) {};
	\end{pgfonlayer}
	\begin{pgfonlayer}{edgelayer}
		\draw [bend right=45, looseness=1.25] (0.center) to (g1.center);
		\draw [bend left=45, looseness=1.25] (0.center) to (g1.center);
		\draw (5.center) to (g1);
		\draw (g1) to (6.center);
	\end{pgfonlayer}
\end{tikzpicture}
	}
}

\newcommand{\boundaryAddSome}[4]{
	\vc{
\begin{tikzpicture}
	\begin{pgfonlayer}{nodelayer}
		\node [style=box] (g1) at (0, 0.25) {$#1$};
		\node [style=gn] (0) at (0, -0.75) {$#2$};
		\node [style=none] (1) at (0, -0.25) {$\,\dots$};
		\node [style=none] (4) at (0, 0.75) {$\,\dots$};
		\node [style=none] (5) at (-0.5, 0.75) {};
		\node [style=none] (6) at (0.5, 0.75) {};
		\node [style=#3] (7) at (0, -1.5) {$#4$};
	\end{pgfonlayer}
	\begin{pgfonlayer}{edgelayer}
		\draw [bend right=45, looseness=1.25] (0.center) to (g1.center);
		\draw [bend left=45, looseness=1.25] (0.center) to (g1.center);
		\draw (5.center) to (g1);
		\draw (g1) to (6.center);
		\draw (0.center) to (7.center);
	\end{pgfonlayer}
\end{tikzpicture}
	}
}

\def\dcup{\raisebox{0.1em}{\vc{\begin{tikzpicture}
	\begin{pgfonlayer}{nodelayer}
		\node [style=none] (2) at (-0.5, 0.5) {};
		\node [style=none] (3) at (0.5, 0.5) {};
	\end{pgfonlayer}
	\begin{pgfonlayer}{edgelayer}
		\draw [bend right=90, looseness=1.50] (2.center) to (3.center);
	\end{pgfonlayer}
\end{tikzpicture}
}}}
\def\dcap{\raisebox{0.3em}{\vc{\begin{tikzpicture}
	\begin{pgfonlayer}{nodelayer}
		\node [style=none] (2) at (-0.5, 0.5) {};
		\node [style=none] (3) at (0.5, 0.5) {};
	\end{pgfonlayer}
	\begin{pgfonlayer}{edgelayer}
		\draw [bend left=90, looseness=1.50] (2.center) to (3.center);
	\end{pgfonlayer}
\end{tikzpicture}
}}}

\newcommand{\QRHadamard}{
	\raisebox{-0.5\height}{
		\begin{tikzpicture}[scale=0.5]
			\begin{pgfonlayer}{nodelayer}
				\node [style=white] (1)   at (0, 0)   {$-2$};
				\node [style=poly] (2)   at (0, 1)   {$1$};
				\node [style=polyT] (3)   at (0, -1)   {$1$};
				\node [style=none] (4)   at (0, 2)   {};
				\node [style=none] (5)   at (0, -2)   {};
			\end{pgfonlayer}
			\begin{pgfonlayer}{edgelayer}
				\draw [-] (4.center) to (5.center);
			\end{pgfonlayer}
		\end{tikzpicture}
	}
}

\newcommand{\ZHExampleThree}[1]{
	\vc{
	\begin{tikzpicture}[scale=0.75]
		\node [style=ZH] (0) at (-3, 0) {};
		\node [style=ZH] (1) at (-2, -1) {};
		\node [style=none] (3) at (-2, -2) {};
		\node [style=ZH] (8) at (-3, 1) {#1};
		\node [style=white] (9) at (-1, 0) {};
		\node [style=ZH] (10a) at (-1.75, 1) {-#1};
		\node [style=ZH] (10b) at (-1, 1) {-#1};
		\node [style=ZH] (10c) at (-0.25, 1) {-#1};
		\node [style=none] (19) at (0, -0.5) {=};
		\node [style=ZH] (20) at (1, 0) {};
		\node [style=ZH] (21) at (2, -1) {};
		\node [style=none] (22) at (2, -2) {};
		\node [style=ZH] (23) at (1, 1) {#1};
		\node [style=white] (24) at (3, 0) {};
		\node [style=ZH] (25) at (2.25, 1) {-#1};
		\node [style=ZH] (26) at (3, 1) {-#1};
		\node [style=ZH] (27) at (3.75, 1) {-#1};
		\node [style=none] (30) at (2.5, 1.5) {};
		\node [style=none] (31) at (2.5, 0.5) {};
		\node [style=none] (32) at (3.5, 0.5) {};
		\node [style=none] (33) at (3.5, 1.5) {};
		\draw (1) to (3.center);
		\draw (0) to (1);
		\draw (8) to (0);
		\draw (10a) to (9.center);
		\draw (10b) to (9.center);
		\draw (10c) to (9.center);
		\draw (9.center) to (1);
		\draw (21) to (22.center);
		\draw (20) to (21);
		\draw (23) to (20);
		\draw (24.center) to (21);
		\node [style=ZH] (0) at (-3, 0) {};
		\node [style=ZH] (1) at (-2, -1) {};
		\node [style=none] (3) at (-2, -2) {};
		\node [style=ZH] (8) at (-3, 1) {#1};
		\node [style=white] (9) at (-1, 0) {};
		\node [style=ZH] (10a) at (-1.75, 1) {-#1};
		\node [style=ZH] (10b) at (-1, 1) {-#1};
		\node [style=ZH] (10c) at (-0.25, 1) {-#1};
		\node [style=none] (19) at (0, -0.5) {=};
		\node [style=ZH] (20) at (1, 0) {};
		\node [style=ZH] (21) at (2, -1) {};
		\node [style=none] (22) at (2, -2) {};
		\node [style=ZH] (23) at (1, 1) {#1};
		\node [style=white] (24) at (3, 0) {};
		\node [style=ZH] (25) at (2.25, 1) {-#1};
		\node [style=ZH] (26) at (3, 1) {-#1};
		\node [style=ZH] (27) at (3.75, 1) {-#1};
		\node [style=none] (30) at (2.5, 1.5) {};
		\node [style=none] (31) at (2.5, 0.5) {};
		\node [style=none] (32) at (3.5, 0.5) {};
		\node [style=none] (33) at (3.5, 1.5) {};
	\end{tikzpicture}
	}
}

\newcommand{\ZHExampleTwo}[1]{
	\vc{
	\begin{tikzpicture}[scale=0.75]
		\node [style=ZH] (0) at (-3, 0) {};
		\node [style=ZH] (1) at (-2, -1) {};
		\node [style=none] (3) at (-2, -2) {};
		\node [style=ZH] (8) at (-3, 1) {#1};
		\node [style=white] (9) at (-1, 0) {};
		\node [style=ZH] (10a) at (-1.75, 1) {-#1};
		\node [style=ZH] (10b) at (-1, 1) {-#1};
		\node [style=none] (19) at (0, -0.5) {=};
		\node [style=ZH] (20) at (1, 0) {};
		\node [style=ZH] (21) at (2, -1) {};
		\node [style=none] (22) at (2, -2) {};
		\node [style=ZH] (23) at (1, 1) {#1};
		\node [style=white] (24) at (3, 0) {};
		\node [style=ZH] (25) at (2.25, 1) {-#1};
		\node [style=ZH] (26) at (3, 1) {-#1};
		\node [style=none] (30) at (2.5, 1.5) {};
		\node [style=none] (31) at (2.5, 0.5) {};
		\node [style=none] (32) at (3.5, 0.5) {};
		\node [style=none] (33) at (3.5, 1.5) {};
		\draw (1) to (3.center);
		\draw (0) to (1);
		\draw (8) to (0);
		\draw (10a) to (9.center);
		\draw (10b) to (9.center);
		\draw (9.center) to (1);
		\draw (21) to (22.center);
		\draw (20) to (21);
		\draw (23) to (20);
		\draw (24.center) to (21);
		\node [style=ZH] (0) at (-3, 0) {};
		\node [style=ZH] (1) at (-2, -1) {};
		\node [style=none] (3) at (-2, -2) {};
		\node [style=ZH] (8) at (-3, 1) {#1};
		\node [style=white] (9) at (-1, 0) {};
		\node [style=ZH] (10a) at (-1.75, 1) {-#1};
		\node [style=ZH] (10b) at (-1, 1) {-#1};
		\node [style=none] (15) at (-1.5, 1.5) {};
		\node [style=none] (16) at (-1.5, 0.5) {};
		\node [style=none] (17) at (-0.5, 0.5) {};
		\node [style=none] (18) at (-0.5, 1.5) {};
		\node [style=none] (19) at (0, -0.5) {=};
		\node [style=ZH] (20) at (1, 0) {};
		\node [style=ZH] (21) at (2, -1) {};
		\node [style=none] (22) at (2, -2) {};
		\node [style=ZH] (23) at (1, 1) {#1};
		\node [style=white] (24) at (3, 0) {};
		\node [style=ZH] (25) at (2.25, 1) {-#1};
		\node [style=ZH] (26) at (3, 1) {-#1};
		\node [style=none] (30) at (2.5, 1.5) {};
		\node [style=none] (31) at (2.5, 0.5) {};
		\node [style=none] (32) at (3.5, 0.5) {};
		\node [style=none] (33) at (3.5, 1.5) {};
	\end{tikzpicture}
	}
}

\newcommand{\ZHExampleOne}[1]{
	\vc{
	\begin{tikzpicture}[scale=0.75]
		\node [style=ZH] (0) at (-3, 0) {};
		\node [style=ZH] (1) at (-2, -1) {};
		\node [style=none] (3) at (-2, -2) {};
		\node [style=ZH] (8) at (-3, 1) {#1};
		\node [style=white] (9) at (-1, 0) {};
		\node [style=ZH] (10a) at (-1, 1) {-#1};
		\node [style=none] (19) at (0, -0.5) {=};
		\node [style=ZH] (20) at (1, 0) {};
		\node [style=ZH] (21) at (2, -1) {};
		\node [style=none] (22) at (2, -2) {};
		\node [style=ZH] (23) at (1, 1) {#1};
		\node [style=white] (24) at (3, 0) {};
		\node [style=ZH] (25) at (3, 1) {-#1};
		\node [style=none] (30) at (2.5, 1.5) {};
		\node [style=none] (31) at (2.5, 0.5) {};
		\node [style=none] (32) at (3.5, 0.5) {};
		\node [style=none] (33) at (3.5, 1.5) {};
		\draw (1) to (3.center);
		\draw (0) to (1);
		\draw (8) to (0);
		\draw (10a) to (9.center);
		\draw (9.center) to (1);
		\draw (21) to (22.center);
		\draw (20) to (21);
		\draw (23) to (20);
		\draw (24.center) to (21);
		\node [style=ZH] (0) at (-3, 0) {};
		\node [style=ZH] (1) at (-2, -1) {};
		\node [style=none] (3) at (-2, -2) {};
		\node [style=ZH] (8) at (-3, 1) {#1};
		\node [style=white] (9) at (-1, 0) {};
		\node [style=ZH] (10a) at (-1, 1) {-#1};
		\node [style=none] (15) at (-1.5, 1.5) {};
		\node [style=none] (16) at (-1.5, 0.5) {};
		\node [style=none] (17) at (-0.5, 0.5) {};
		\node [style=none] (18) at (-0.5, 1.5) {};
		\node [style=none] (19) at (0, -0.5) {=};
		\node [style=ZH] (20) at (1, 0) {};
		\node [style=ZH] (21) at (2, -1) {};
		\node [style=none] (22) at (2, -2) {};
		\node [style=ZH] (23) at (1, 1) {#1};
		\node [style=white] (24) at (3, 0) {};
		\node [style=none] (30) at (2.5, 1.5) {};
		\node [style=none] (31) at (2.5, 0.5) {};
		\node [style=none] (32) at (3.5, 0.5) {};
		\node [style=none] (33) at (3.5, 1.5) {};
	\end{tikzpicture}
	}
}
\newcommand{\ZHExampleZero}[1]{
	\vc{
	\begin{tikzpicture}[scale=0.75]
		\node [style=ZH] (0) at (-3, 0) {};
		\node [style=ZH] (1) at (-2, -1) {};
		\node [style=none] (3) at (-2, -2) {};
		\node [style=ZH] (8) at (-3, 1) {#1};
		\node [style=white] (9) at (-1, 0) {};
		\node [style=none] (19) at (0, -0.5) {=};
		\node [style=ZH] (20) at (1, 0) {};
		\node [style=ZH] (21) at (2, -1) {};
		\node [style=none] (22) at (2, -2) {};
		\node [style=ZH] (23) at (1, 1) {#1};
		\node [style=white] (24) at (3, 0) {};
		\node [style=none] (30) at (2.5, 1.5) {};
		\node [style=none] (31) at (2.5, 0.5) {};
		\node [style=none] (32) at (3.5, 0.5) {};
		\node [style=none] (33) at (3.5, 1.5) {};
		\draw (1) to (3.center);
		\draw (0) to (1);
		\draw (8) to (0);
		\draw (9.center) to (1);
		\draw (21) to (22.center);
		\draw (20) to (21);
		\draw (23) to (20);
		\draw (24.center) to (21);
		\node [style=ZH] (0) at (-3, 0) {};
		\node [style=ZH] (1) at (-2, -1) {};
		\node [style=none] (3) at (-2, -2) {};
		\node [style=ZH] (8) at (-3, 1) {#1};
		\node [style=white] (9) at (-1, 0) {};
		\node [style=none] (15) at (-1.5, 1.5) {};
		\node [style=none] (16) at (-1.5, 0.5) {};
		\node [style=none] (17) at (-0.5, 0.5) {};
		\node [style=none] (18) at (-0.5, 1.5) {};
		\node [style=none] (19) at (0, -0.5) {=};
		\node [style=ZH] (20) at (1, 0) {};
		\node [style=ZH] (21) at (2, -1) {};
		\node [style=none] (22) at (2, -2) {};
		\node [style=ZH] (23) at (1, 1) {#1};
		\node [style=white] (24) at (3, 0) {};
		\node [style=none] (30) at (2.5, 1.5) {};
		\node [style=none] (31) at (2.5, 0.5) {};
		\node [style=none] (32) at (3.5, 0.5) {};
		\node [style=none] (33) at (3.5, 1.5) {};
	\end{tikzpicture}
	}
}

\newcommand\exaGalois[3]{
	\vc{
		\begin{tikzpicture}
			\begin{pgfonlayer}{nodelayer}
				\node [style=none] (4) at (0, 0.5) {};
				\node [style=white] (5) at (0, -0.25) {};
				\node [style=polynomial] (6) at (0.5, -1) {};
				\node [style=white] (7) at (-0.75, -1) {$#1$};
				\node [style=white] (8) at (0, -1.75) {$#2$};
				\node [style=white] (9) at (1, -1.75) {$#3$};
			\end{pgfonlayer}
			\begin{pgfonlayer}{edgelayer}
				\draw (5.center) to (6.center);
				\draw (4.center) to (5.center);
				\draw (5.center) to (7.center);
				\draw (6.center) to (8.center);
				\draw (6.center) to (9.center);
			\end{pgfonlayer}
		\end{tikzpicture}
	}
}

\newtheorem{theorem}{Theorem}[chapter]
\newtheorem{corollary}[theorem]{Corollary}
\newtheorem{proposition}[theorem]{Proposition}
\newtheorem{lemma}[theorem]{Lemma}

\theoremstyle{definition}
\newtheorem{definition}[theorem]{Definition}
\theoremstyle{remark}
\newtheorem{remark}[theorem]{Remark}
\newtheorem{example}[theorem]{Example}

\AtAppendix{\counterwithin{theorem}{section}}
\AtAppendix{\counterwithin{lemma}{section}}
\AtAppendix{\counterwithin{proposition}{section}}
\AtAppendix{\counterwithin{corollary}{section}}
\AtAppendix{\counterwithin{remark}{section}}

\makeindex[title=Index]

\begin{document}

\pagestyle{empty}
\begin{titlepage}
\begin{center}
\mbox{}\\[6pt]
\Huge \textbf{Graphical Calculi and their Conjecture Synthesis} \\[100pt] 

\includegraphics[width=0.2\columnwidth]{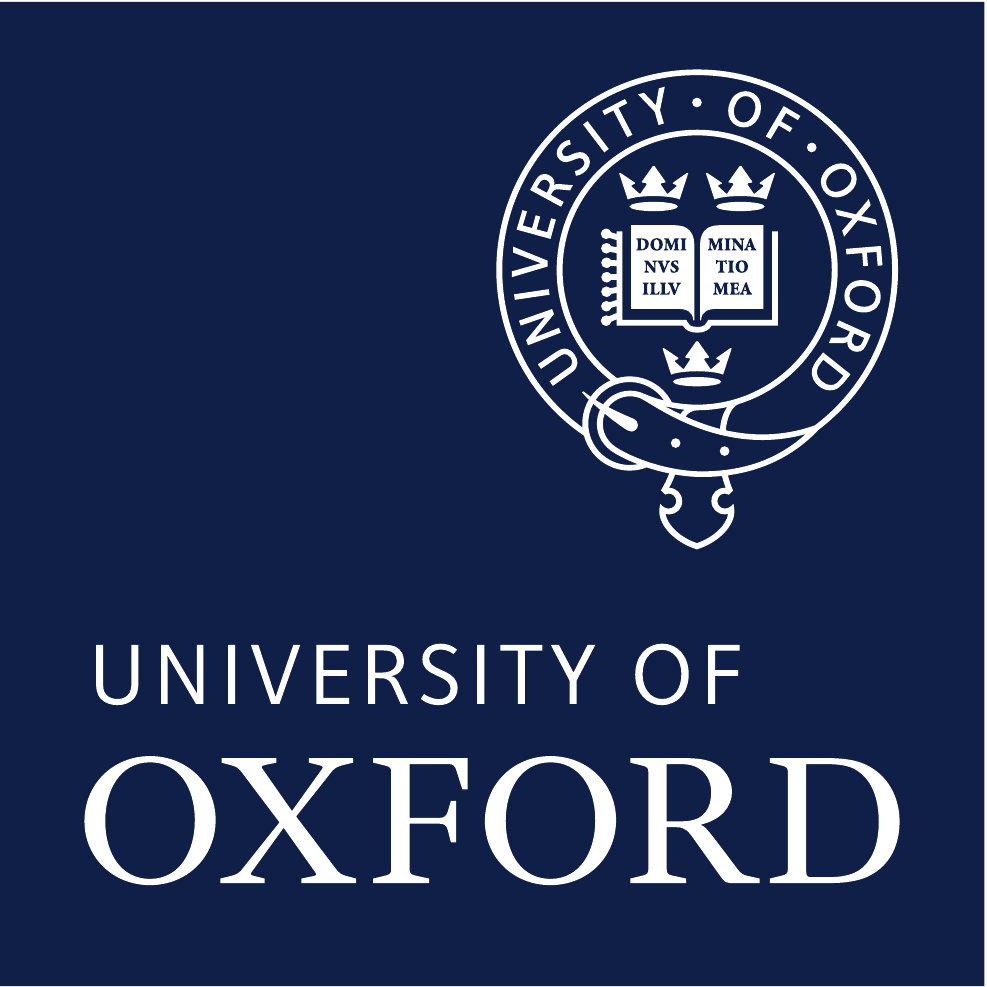} \\[80pt]

\LARGE Hector Miller-Bakewell \\[6pt]
\Large Wolfson College \\[3pt] 
University of Oxford \\[60pt]

A thesis submitted for the degree of \\[3pt]
\emph{Doctor of Philosophy} \\[12pt]
Hilary 2020
\end{center}
\end{titlepage}

\thispagestyle{empty} 

\cleardoublepage
\selectlanguage{british}
\begin{center}\textbf{Abstract}\end{center}

\vspace{10pt}

\begin{quotation}
\noindent \small 
Categorical Quantum Mechanics, and graphical calculi in particular,
has proven to be an intuitive and powerful
way to reason about quantum computing.
This work continues the exploration of graphical calculi,
inside and outside of the quantum computing setting,
by investigating the algebraic structures with which we label diagrams.
The initial aim for this was Conjecture Synthesis;
the algorithmic process of creating theorems.
To this process we introduce a generalisation step,
which itself requires the ability to infer
and then verify parameterised families of theorems.
This thesis introduces such inference
and verification frameworks,
in doing so forging novel links
between graphical calculi and
fields such as Algebraic Geometry and Galois Theory.
These frameworks inspired further research into the
design of graphical calculi,
and we introduce two important new calculi here.
First is the calculus RING,
which is initial among ring-based qubit graphical calculi,
and in turn inspired the introduction and classification of phase homomorphism pairs
also presented here.
The second is the calculus ZQ,
an edge-decorated calculus
which naturally expresses arbitrary qubit rotations,
eliminating the need for non-linear rules such as (EU) of ZX.
It is expected that these results will be of use 
to those creating optimisation schemes and intermediate representations for quantum computing,
to those creating new graphical calculi,
and for those performing conjecture synthesis.
\end{quotation}
\cleardoublepage

\frontmatter \pagestyle{plain} 
\thispagestyle{empty}
\renewcommand\thepage{}
\tableofcontents
\newpage \thispagestyle{empty}
\renewcommand\thepage{\arabic{page}}

\mainmatter \pagestyle{fancy}
\chapter{Introduction} \label{chapIntroduction}
\thispagestyle{plain}

Whether or not researchers have achieved quantum supremacy
is disputed \cite{QuantumSupremacy} \cite{IBMSupremacy},
but it is undeniable that the power and complexity of quantum computers
is increasing. In the words of Ref.~\cite{NISQ}:
`Now is a privileged time in the history of science and technology,
as we are witnessing the opening of the NISQ era.'
\footnote{NISQ devices are noisy, and have `a number of qubits ranging from 50 to a few hundred'.}
The ability to run wider and deeper quantum circuits
brings with it the desire to design more complicated quantum algorithms.
Diagrams allow researchers to express these algorithms more clearly,
and different types of diagram have different strengths.
For an example see Figure~\ref{figCCZFactory},
with the authors of that paper noting that `the ZX calculus graph is more
amenable to verification than the 3D diagram'.
Graphical calculi go a step beyond simply being diagrams in that they come with rules for manipulating the graphical data directly.

\begin{figure}[p]
	\center
	\includegraphics[width=0.4\textwidth,height=5cm,keepaspectratio]{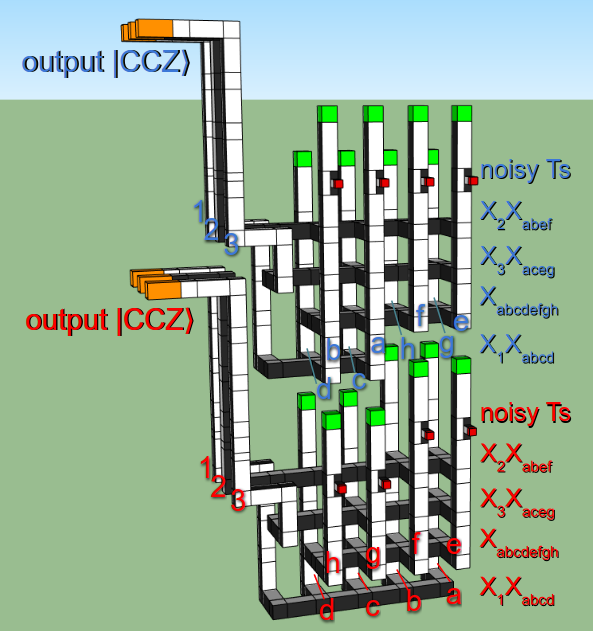}~\qquad~
	\includegraphics[width=0.4\textwidth,height=5cm,keepaspectratio]{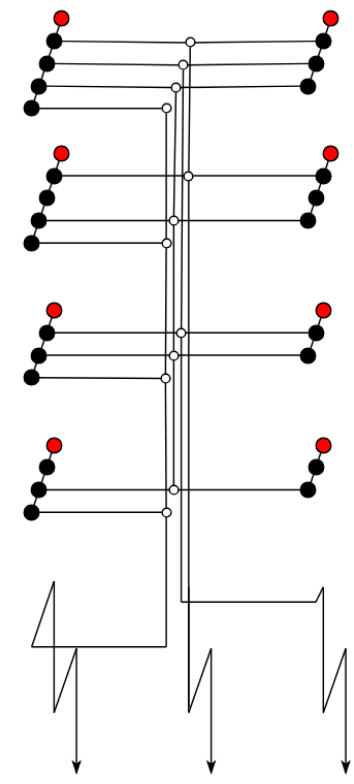}
	\caption{A representation of a CCZ factory in 3D (left) and (something similar to) the ZX calculus (right),
		as presented in Ref.~\cite[Figure~9]{MagicStateFactory}.\label{figCCZFactory}}
\end{figure}

Graphical calculi are not new.
The ones we discus in this thesis flow, conceptually, from Roger Penrose's
diagrammatic representations of tensor networks,
but as Ref.~\cite{PQP} points out
`many similar diagrammatic languages were invented prior to this or reinvented later'.
Here we shall be using these diagrams to consider Categorical Quantum Mechanics
in the sense introduced in Ref.~\cite{CQM},
a paradigm that lead to the introduction
of what would become the ZX calculus \cite{Coecke08}.
In the words of Ref.~\cite{CQM}:
`The arguments for the benefits of a high-level, conceptual approach to designing and reasoning about
quantum computational systems are just as compelling as for classical computation.'
The ZX calculus has grown over time,	
in appearance, generators, scope and rules, with Figure~\ref{figZXLooks} showing some of the changes in appearance.
(The debate as to whether there is `a' or `the' ZX calculus is still ongoing\footnote{Private communication during a ZX-calculus workshop},
but we shall cover the different varieties in \S\ref{secZoo}.)
The story behind the development of the ZX calculus is probably best told by the titles of the various papers
in Figure~\ref{figZXStory}.
In that list of titles we can see issues of `incompleteness' highlighted,
and then resolved, for different fragments of the calculus.
With each challenge the calculus has grown, changed,
and created off-shoots.

ZX, although probably the best known of the graphical calculi discussed here,
is not the only calculus built for qubit quantum computing.
The completeness result `A Complete Axiomatisation of the ZX-Calculus for Clifford+T
Quantum Mechanics' \cite{SimonCompleteness}
is constructed from the completeness result of another calculus
called ZW \cite{AmarThesis}.
ZX is built from the Z and X rotations of the Bloch Sphere;
ZW is built from the W and GHZ states (it was originally called the GHZ/W calculus \cite[p~ii]{AmarThesis}).
ZH, introduced in 2018 \cite{ZH}, is built from Z spiders and H-boxes.
There is no need to stop here,
and indeed this thesis introduces two new graphical calculi of its own.
All these calculi are built to be powerful
and intuitive ways of reasoning about a chosen domain.
The utility, compared to standard matrix notation, is down to three factors:
dimension, scale, and connectivity \cite[{\S}2.1.2]{Backens16Thesis}.
It is doubtless for these reasons that we see the use of graphical calculi
spread not just into wider academia \cite{Horsman2017Surgery, MagicStateFactory}
but also into industry,
with papers such as Refs.~\cite{CQCZX} and \cite{CQCShallow}
appearing with support from the company Cambridge Quantum Computing~\cite{CQCSite}.

\begin{example}[Interchange Law {\cite[p43]{Maclane2013}}]  \label{exaInterchangeLaw}
	The Interchange Law (given below) exemplifies the way that a graphical notation can convey meaning.
	Although we have yet to give precise meanings for the diagrams below
	we hope the reader can appreciate the way that the algebraic equation on the right
	is a tautology inherent in the notation on the left.
	\begin{align}
		\vc{\InputIfFileExists{./figures/wire/interchange.tikz}{}{Missing file!}} \equiv \vc{\InputIfFileExists{./figures/wire/interchange.tikz}{}{Missing file!}} \qquad \sim \qquad (a \tensor b) \comp (c \tensor d) =  (a \comp c) \tensor (b \comp d)
	\end{align}
	An analogy would be the way that the expression $f \comp g \comp h$
	expresses the associativity of the composition by virtue of the lack of brackets.
\end{example}

\begin{figure}[p]
	\center
	\includegraphics[width=0.3\linewidth]{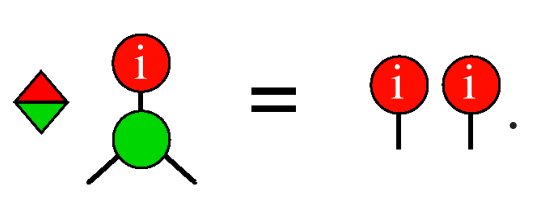}~
	\includegraphics[width=0.3\linewidth]{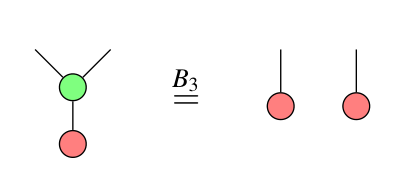}~
	\raisebox{1.3\height}{\vc{\InputIfFileExists{./figures/ZX/copy.tikz}{}{Missing file!}}}
	\caption{The stylistic change in the ZX calculus between 2008 \cite{Coecke08} (left),
		2011 \cite{Coecke11Entanglement} (middle), and this thesis (right).
		The colours used in this thesis have been chosen to
		be easier to distinguish between for anyone with red / green colour vision deficiency,
		or when printed in greyscale, with the green node appearing lighter than the red.
		\cite{ZXAccessibility} \label{figZXLooks}}
\end{figure}

\begin{figure}[p]
	\begin{tabular}{l l p{10cm}}
		Reference                        & Published & Title                                                                                         \\ \hline
		\cite{ChristianMSC}              & 2013      & The ZX calculus is incomplete for non-stabilizer quantum mechanics                            \\
		\cite{BackensStabilizerComplete} & 2014      & The ZX-calculus is complete for stabilizer quantum mechanics                                  \\
		\cite{Duncan13}                  & 2014      & Pivoting makes the ZX-calculus complete for real stabilizers                                  \\
		\cite{SdW14}                     & 2014      & The ZX-calculus is incomplete for quantum mechanics                                           \\
		\cite{BackensSingleQubit}        & 2014      & The {ZX-calculus} is complete for the single-qubit {Clifford+T} group                         \\
		\cite{Backens15}                 & 2015      & Making the stabilizer ZX-calculus complete for scalars                                        \\
		\cite{BackensSimplified}         & 2017      & A Simplified Stabilizer ZX-calculus                                                           \\
		\cite{cyclo}                     & 2017      & {ZX-Calculus: Cyclotomic Supplementarity and Incompleteness for Clifford+T Quantum Mechanics} \\
		\cite{SimonCompleteness}         & 2018      & {A Complete Axiomatisation of the ZX-Calculus for Clifford+T Quantum Mechanics}               \\
		\cite{UniversalComplete}         & 2017      & A universal completion of the ZX-calculus
	\end{tabular}
	\caption{
		A table detailing the path towards the completeness results
		for several fragments of the ZX calculus.
		The year shown is the year of publication,
		but it should be noted that the work of \cite{UniversalComplete}
		builds on the work of \cite{SimonCompleteness},
		despite being published beforehand. \label{figZXStory}
	}
\end{figure}

There is a second strand to this thesis,
and that is Conjecture Synthesis.
Isabelle Conjecture Synthesis (IsaCoSy)
was created with the
`[...] aim to automatically produce a useful set
of theorems' \cite{ISACOSY}.
What made IsaCoSy, expanded upon in Ref.~\cite{ISACOSY2},
different from previous theory formation systems is that of method.
Previous examples had used a \emph{deductive} approach (from the known facts make more facts)
whereas IsaCoSy uses a \emph{generative} approach (generate all hypotheses,
weed out false ones).
While this method is, on the face of it,
intractable,
the authors of Ref.~\cite{ISACOSY} realised that they only needed to generate hypotheses that didn't
reduce to any other hypothesis.
As a simple example the hypothesis `$x = 1+0$'
is automatically redundant if you are also going to generate the hypothesis `$x = 1$'.
We shall get into the graphical calculus version of conjecture synthesis in more detail
in \S\ref{chapCoSy}, but it was this idea that spurred the research presented in this thesis:
`Can we make conjecture synthesis work well for graphical calculi?'

The implicit first part of the question (`Can we make conjecture synthesis work \emph{at all}
for graphical calculi?') had been answered by Kissinger \cite{AleksCoSy}
in a working algorithm written as a module for Quantomatic \cite{Quantomatic}.
Quantomatic then went through a rebirth from ML into Scala
and a Scala-based implementation of a similar algorithm was implemented by
Kissinger and the author of this thesis.
While this produced many results
(over the course of the research for this thesis there were over 100,000 files produced,
all with conjecture synthesis in mind)
it did not produce, to the author's knowledge, any particularly \emph{useful} results.
The reason for this is that even after running the algorithm for two weeks
on the group's computing cluster
the algorithm was still only producing theorems that a researcher could have shown by hand,
of at most four vertices per diagram.

The theorems produced by QuantoCoSy (as this program had now been named)
often showed high levels of similarity between them.
These were not redundancies like the aforementioned `$x = 1+0$' reducing to `$x=1$'.
Instead these similarities represented `higher'
structure inherent in the phase group
or in repeated sections of the diagrams.
The direction of the research then became
`how can we infer and verify conjectures in this higher structure?'

This thesis comes in two parts.
In the first we generalise on, investigate, and create
new graphical calculi,
exhibiting new ways of comparing and relating these calculi,
and examining diagrammatic maps that preserve equational soundness.
In the second we discover
a deep link between phase algebras and algebraic geometry,
allowing us to infer and then verify parameterised equations.
Perhaps the clearest example of these strands coming together is in \S\ref{secPhaseVariablesOverQubits},
where we use what we have learned
about soundness-preserving maps,
the geometric structure of parameter spaces,
and some Galois Theory,
to show that a ZX, ZW, ZH, or $\ring$ diagram with phase variables
can be verified by a \emph{single} equation without phase variables.
This single, simple equation implies
and is implied by the large, parameterised family of equations.
Our choice of phase algebra directly dictates the geometry of the parameter space,
which in turn dictates the generalisations of our theorems.

As quantum computing algorithms get steadily more complicated,
and as the diagrams get larger
we shall need to be able to automatically produce the useful set
of theorems that conjecture synthesis promises us.
We shall need a toolbox of useful theorems
for proof assistants (like Quantomatic or PyZX \cite{PyZX}),
or for intermediate representations and optimisers of quantum algorithms \cite{TRIQ, CQCShallow, Backens2020circuit}.
Our researchers will also need new tools and new calculi.
This thesis provides advances in all of these areas.
The more efficiently we can reason about quantum computing
the more efficient we can make our circuits,
and the faster we shall achieve the world where quantum computing is directly helpful.

\section{The contributions of this thesis}
\label{secContributions}

The author hopes that the reader finds the entirety of this thesis
useful, but this section exists to draw attention to the most likely candidates of interest,
grouped loosely by theme.
The main results chapters come after two background chapters:
\S\ref{chapGraphicalCalculi}
for graphical calculi,
and \S\ref{chapCoSy} for conjecture synthesis.
The author would like to highlight the applications
of Algebraic Geometry and Galois Theory as particularly novel
to the study of graphical calculi or quantum circuits,
as well as the construction of the two new calculi $\ring$ and $\ZQ$.

\subsection{The calculus RING}

$\ring$ is a universal, sound, complete graphical calculus built just from ring operations and compact closure.
As such it acts like a unifier: Any graphical calculus with a similar phase ring structure contains a copy of $\ring$.
Examples of such calculi are ZW and ZH.
$\ring$ and ZX form the bulk of the examples used in this thesis.

\begin{itemize}
	\item[\brag] The calculus $\ring_R$ is introduced and shown to be complete (\S\ref{chapRingR})
	\item[\brag] We justify referring to the calculus $\ring_\bbC$ as `the' generic phase ring quantum graphical calculus (Remark~\ref{remRingKGeneric})
\end{itemize}

\subsection{Parameterised rules as geometric surfaces}

We demonstrate how to present parameterised rules as geometric surfaces,
preserving existing ideas of rule linearity from the literature and
introducing an algebraic geometric approach to graphical calculi and conjecture inference.

\begin{itemize}
	\item[\brag] We provide a geometric framework for parameterised rules (\S\ref{secSkeletons})
	\item[\brag] We demonstrate how to convert linear rules
	      to this geometric framework and back again (\S\ref{secInferringPhaseVariables})
	\item[\brag] We use this link to infer the existence of linear rules
	      for conjecture synthesis (\S\ref{secFindingSubmodules})
\end{itemize}

\subsection{Phase homomorphism pairs}

Phase homomorphisms are maps that act on the phases of diagrams.
We define here phase homomorphism \emph{pairs};
phase homomorphisms paired with an additional functor
that then commutes with the interpretation.
An example of a phase homomorphism pair negates phases in the ZX calculus,
and in doing so creates the complex conjugate of the corresponding matrix.
In this thesis we codify, explore, and classify this notion,
showing that phase homomorphism pairs preserve soundness,
and even preserve proofs, in some of the calculi we consider.

\begin{itemize}
	\item[\brag] We show that the space of equational theorems is closed under
	      the action of phase homomorphisms for $\ring$, ZW, and ZH (Theorem~\ref{thmLiftRingHomSoundness})
	\item[\brag] We classify the phase homomorphism pairs for $\ring$, $\ZW$,
	      $\ZH$, and the finite fragments of $\ZX$ containing $\piby{4}$ (\S\ref{secPhaseGroupHomomorphismsEtc})
	\item[\brag] We link phase homomorphism pairs to geometric symmetries of parameter spaces (\S\ref{secParameterSymmetries})
\end{itemize}

\subsection{!-box verification}

We introduce a method for verifying equations containing !-boxes
by verifying a finite number of equations without !-boxes.
This verification depends on a property we call \emph{separability}.
This is possible thanks to a new way of presenting
successive !-box iterations
in a manner that depends on $\comp$ rather than $\tensor$ products.

\begin{itemize}
	\item[\brag] We define the notion of separability in \S\ref{secnesting}
	\item[\brag] Series !-box form is defined in \S\ref{secSeriesBBoxForm}
	\item[\brag] Finite verification of !-boxes is shown in Theorem~\ref{thmbbox2}
\end{itemize}

\subsection{Phase variable verification}

We show how to verify families of equations parameterised by phase variables
using a method built on polynomial interpolation.
This method requires checking a finite number of sample values for the variables.

\begin{itemize}
	\item[\brag] Phase variables and !-boxes can be verified
	in a finite manner without reference to either (Theorem~\ref{thmfinite})
\end{itemize}

\subsection{Applications of Galois Theory}

For certain languages we can encode automorphisms of field extensions as phase homomorphism pairs.
In doing so we can use the machinery of Galois Theory
to manipulate certain phases in a quantum circuit while preserving others.
This is combined with the symmetries provided by phase homomorphism pairs
to further reduce the number of equations needed for verification.

\begin{itemize}
	\item[\brag] Over qubits we can verify all the phase variables in a diagram using a
	      single equation containing no phase variables (\S\ref{secPhaseVariablesOverQubits})
\end{itemize}

\subsection{The generalisation step}

Conjecture synthesis saw a recent advancement with the generative model
discussed in \S\ref{chapCoSy}.
In this thesis we introduce an new step to this process;
generalisation of theorems at the point of synthesis.
This step requires a conjecture inference framework
such as the one introduced in this thesis in order to be possible.

\begin{itemize}
	\item[\brag] We alter the generative conjecture synthesis method
	with the addition of a generalisation step (Remark~\ref{remGeneralisationStep})
\end{itemize}

\subsection{The graphical calculus ZQ}

The foundation of the calculus ZX is the Z and X spiders
that represent rotations of the Bloch Sphere.
These rotations generate all possible rotations of the sphere via Euler angles.
Rather than take a generating subset of these rotations as vertex-labels,
ZQ allows arbitrary rotations (expressed as quaternions) as edge-labels.
This leads to a calculus that the author feels is clearer and more intuitive than ZX.

\begin{itemize}
	\item[\brag] The calculus ZQ is introduced and shown to be complete (\S\ref{chapZQ})
	\item[\brag] The spiders of ZX are shown to not be viable for representing a non-commutative group (Corollary~\ref{corUQuatNotMonoid})
	\item[\brag] ZQ is the first qubit graphical calculus to use a non-commutative phase group
	\item[\brag] The rules of ZQ are linear, whereas the rules of ZX are not (Remark~\ref{remLinearRulesOfZQ})
\end{itemize}

\subsection{Generalising ZH}

The calculus ZH treats phase-free Z spiders and generalised Hadamard nodes as fundamental building blocks \cite{ZH}.
It was announced in 2018 and was shown to be complete via the construction of a normal form.
For the rest of this thesis we shall refer to this version as $\ZHC$,
because the phases of these generalised Hadamard nodes are elements of $\bbC$.
It can be shown, however, that by simply changing this parameterisation to any commutative ring with a half, $R$,
that a complete and universal calculus is then exhibited for $\bits{R}$.

\begin{itemize}
	\item[\brag] The calculus $\ZH_R$ is generalised from $\ZH_\bbC$ and shown to be complete (\S\ref{secCompleteTranslationZHQR})
\end{itemize}

\section{Overview}
\label{secOverview}

The eventual conclusion of this thesis
is that phase algebras are a powerful and important
aspect of a graphical calculus,
allowing us to apply techniques and knowledge from previously
unrelated disciplines 
to the field of quantum graphical calculi.
We will use this knowledge,
among other new results,
to also further the field of conjecture synthesis
with the introduction of a new `generalisation' step.
Along the way we will introduce new calculi,
novel links to algebraic geometry,
and novel inference and verification methods.

Following the two background chapters
we will introduce
the complete, universal calculus $\ring$ in \S\ref{chapRingR},
which will serve as an important example in later chapters.
\S\ref{chapPhaseRingCalculi} will then
discuss the place of $\ring$ in the wider context
of phase-ring graphical calculi,
and the ways in which phase algebra homomorphisms
interact with diagram equations.

The thesis will then explore another new calculus in \S\ref{chapZQ},
explaining the necessity behind its novel presentation,
and the way in which it solves the issues with the (EU) rule of ZX.
\S\ref{chapConjectureInference} establishes
the novel link between parameter spaces and geometric surfaces,
building on the phase homomorphism pairs from \S\ref{chapPhaseRingCalculi}.
We also cover why this framework works for ZQ and not ZX.
This link forms the core of our conjecture inference
and is directly linked to the verification methods
of \S\ref{chapConjectureVerification}.
\S\ref{chapConjectureVerification} 
introduces new verification results
and then brings these together with the phase homomorphism pairs
of \S\ref{chapPhaseRingCalculi}
and some Galois Theory
in what the author feels is quite a beautiful way.
This allows the verification of an equation involving phase variables
by verifying a single equation not involving phase variables.
In \S\ref{chapConclusion} we conclude the thesis.

\subsection{An Algebraic Preview} \label{secAlgebraicOverview}

Given that quantum graphical calculi exist
at the intersection of mathematics, physics, and computer science
it is likely that some readers may have never encountered the topics of
Algebraic Geometry, Galois Theory, or Laurent Polynomials.
Since we mention these ideas before giving a proper treatment of them
an overview may be appreciated.
In their broadest terms:
\begin{itemize}
	\item Algebraic Geometry is based on the idea that there is no difference between
	a geometric surface in some space,
	and the collection of polynomial equations that define that surface.
	For example the points of the circle described by $x^2 + y^2 - 1 = 0$
	are still solutions to the equation $x^3 + xy^2 - x = 0$.
	This observation allows us to switch between considering
	collections of polynomials and the surface itself.
	\item Galois Theory concerns pairs of fields, one inside the other,
	and the homomorphisms of the larger field that act as the identity on the smaller.
	A clear example is complex conjugation acting on the fields $\bbR$ and $\bbC$:
	The real numbers are unaffected by complex conjugation,
	but the elements of $\bbC \less \bbR$ are moved.
	We will use Galois Theory to construct maps that fix certain phases in a diagram,
	but move others.
	\item Laurent Polynomials are a generalisation of traditional polynomials
	by allowing positive \emph{and negative} powers of the indeterminants.
	Importantly they can still be added, multiplied, and factorised in the same manner
	as traditional polynomials. Once we introduce them we will be working to
	convert our equations involving Laurent polynomials into equations using traditional polynomials instead.
\end{itemize}
These topics are covered in far more detail when their uses arise in this thesis.

\subsection{Ethical considerations}
\label{secEthics}

There is justified concern about the impact
of quantum computing on humanity.
The impact on classical cryptography is probably the most important \cite{Shor},
but the author hopes that, as with classical computing,
the net effect will be undoubtedly positive.
An example of a near-term positive is advances in chemical synthesis
and protein folding for medical research \cite{Perdomo-Ortiz2012}.

\subsection{Chapter dependencies}

Each of the results chapters
contain a short introduction and gives an indication
of which earlier chapters they are dependent upon.
We also include a dependency diagram
in Figure~\ref{figDependencies}.
Although \S\ref{chapZQ} is not strictly a dependency for \S\ref{chapConjectureInference}
there are examples and remarks
beyond the core results that will require knowledge of ZQ.

\begin{figure}[ht]
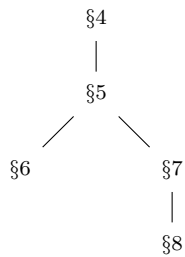

	\center
	\vc{\InputIfFileExists{./figures/wire/dependencies.tikz}{}{Missing file!}}
	\caption{\label{figDependencies}The dependency structure of the chapters in this thesis.}
\end{figure}

Before embarking on any new results
we will spend the next two chapters covering the motivation
and context for this thesis.
The final note for this introduction is that unless stated otherwise all results were the work of the author. \thispagestyle{empty}
\chapter{Graphical Calculi} \label{chapGraphicalCalculi}
\thispagestyle{plain}

\noindent\emph{In this chapter:}
\begin{itemize}
	\item[\chapterbullet] We cover graphical calculi, graph rewriting and reduction systems
	\item[\chapterbullet] We cover the calculi $\ZX_G$, $\ZW_R$, and $\ZH_\bbC$, among others
	\item[\chapterbullet] We cover phase variables and !-boxes
\end{itemize}

Graphical calculi, in the context of quantum computing, stem from Penrose's diagrammatic notation for tensors
\cite{Penrose, PQP}.
The benefit of graphical notation for tensors is that these diagrams are inherently two-dimensional,
in contrast to the single dimension used for writing out algebraic terms.
This allows us to use one dimension (horizontal composition) to represent $\tensor$ and the other
(vertical composition) to represent $\comp$.
Since Penrose's work the idea of tensors has been generalised to that of a monoidal category \cite{JoyalStreet},
and with this generalisation has come a host of graphical calculi \cite{SelingerSurvey}.
Graphical calculi are not just representations to help intuition,
but are powerful tools for reasoning about monoidal categories (see the coherence theorems of Ref.~\cite{SelingerSurvey}).

We present the desiderata of a graphical calculus in Remark~\ref{remGraphicalCalculus},
with a concrete definition and many examples to follow.
Importantly, graphical calculi
need to be both graphical (i.e. based on diagrams)
and calculi (i.e. can be used for calculation).
The second part is the interesting one;
we use the diagrams to represent some other object of study,
and then manipulate the diagrams to represent manipulations
of the object of study.
The important properties of calculi,
which will receive formal definitions in a moment, are
\emph{soundness} (the calculus can only show things that are true)
and \emph{completeness} (anything that is true and can be stated can also be shown).
The `and can be stated' part introduces a third property:
A graphical calculus is \emph{universal} with regards to some collection of objects
if every one of those objects can be represented in the graphical calculus.

Almost certainly the best known example of a graphical calculus
in quantum computing is simply that of quantum circuit diagrams (see Ref.~\cite{NielsonChuang}).
Probably the best known example of the calculi considered here is the ZX-calculus
(Definition~\ref{defZX})
which uses diagrams built from red and green vertices to represent matrices.
It differs from quantum circuit diagrams not just in terms of generators,
but also because ZX represents morphisms in a compact closed, dagger, symmetric monoidal category,
compared to the `plain' monoidal category of circuit diagrams.
In a very real sense the power of ZX comes from this extra structure.

\begin{remark}[Desiderata for Graphical Calculi]  \label{remGraphicalCalculus}
	We desire the following aspects for a graphical calculus:
	\begin{itemize}
		\item A description of valid diagrams to be used in the calculus
		\item An interpretation from diagrams into some object of study
		\item Rules for manipulating these diagrams
	\end{itemize}
\end{remark}

Every graphical calculus we consider in this thesis
will be presented in the form of a PROP\footnote{A \textbf{PRO}ducts and \textbf{P}ermutations category},
an idea that originated with Maclane \cite[\S5]{Maclane1965},
and neatly described as a symmetric monoidal category
where the objects are the natural numbers,
and the tensor product is addition.
Diagrams in our graphical calculi are morphisms in such a PROP.
For our diagrams this means we have a single wire-type,
and boxes (representing morphisms)
that take some number of wires as inputs and some number of wires as outputs.
Because we are dealing with PROPs we can cross wires over,
but we cannot always bend wires around:
Every wire flows from an output of a box (or boundary) to an input of a box (or boundary).

\begin{definition}[Graphical Calculus]  \label{defGraphicalCalculus}
	For this thesis, a graphical calculus is a PROP
	with additional morphisms called \emph{generators}.
	\emph{Diagrams} are morphisms in this PROP.
	The calculus comes equipped with a set of \emph{rules}
	which extend to an equivalence relation $\sim$ on diagrams.
	This PROP has an \emph{interpretation} into some symmetric monoidal category $\cal{C}$,
	which is a strict symmetric monoidal functor $\idot{}$
	that sends a diagram $D$ to a morphism in $\cal{C}$.
\end{definition}

\begin{definition}[Morphisms of graphical calculi]  \label{defMorphismGraphicalCalculi}
A morphism of graphical calculi
is a strict symmetric monoidal functor from
one PROP of diagrams to the other.
There are no requirements on the interpretations of the two calculi,
and it will be made clear whether the functor respects the equivalence relation $\sim$.
\end{definition}

The rules of a graphical calculus
come in two types.
There are \emph{rewrite rules}
which are expressed as `replace this subdiagram with this other diagram'
(see \S\ref{secRewriting})
and there are \emph{meta rules},
for example `only topology matters' in ZX.
Meta rules are often used to reflect extra categorical structure
that we impose on the PROP;
for example the \emph{wire-bending} rules implicit in the cups and caps
of ZX
are just reflections of compact closed structure on the PROP.
The `all these rewrite rules hold with the colours red and green swapped'
meta-rule from the early versions of ZX is not a rewrite rule
itself but a rule that constructs rewrite rules. 
When we consider rules as equivalence relations
we really mean the smallest equivalence relation
containing all the rewrite rules \emph{and} the meta rules.
Meta rules are usually made explicit once
(e.g. the (T) rule of Ref.~\cite{SdW14})
and then used implicitly from then on.
Two diagrams are often referred to as being \emph{isomorphic}
if they are equivalent using only the meta rules.
By contrast the rewrite rules are (ideally) explicitly mentioned
every time they are used.
In this thesis we indicate the rewrite rules used by
the notation $A \by{R} B$, indicating that the rule $R$ was used
in transforming $A$ into $B$.

\begin{definition}[Soundness {\cite{SelingerSurvey}}]  \label{defSoundness} 
	A graphical calculus is sound if the interpretation $\idot{}$ respects the equivalence relation $\sim$.
	That is:
	\begin{align}
		A \sim B \implies \interpret{A} = \interpret{B}
	\end{align}
\end{definition}

\begin{definition}[Syntactic and semantic entailment]  \label{defEntailmentRules}
	We will have cause in this thesis to consider multiple versions of a graphical
	calculus with the same generators (and therefore the same diagrams) but
	with different rules.
	For a graphical calculus $\bb{G}$ with a given set of rules $\Rules$ we may write:
	\begin{align}
		\Rules \syntactic A = B  \qquad & \text{or, if unambiguous, just} \qquad \bb{G} \syntactic A = B
	\end{align}
	for when $A$ is equivalent to $B$ using the rules $\Rules$ (syntactic entailment).
	Likewise we may write:
	\begin{align}
		\bb{G} \semantic A = B
	\end{align}
	if $\interpret{A} = \interpret{B}$ using the interpretation of $\bb{G}$ (semantic entailment).
	Note that semantic entailment is independent of any choice of rules,
	and syntactic entailment is independent of any choice of interpretation.
\end{definition}

Soundness can therefore be seen as the idea that the rules (syntax)
preserve the interpretation (semantics): If $\Rules \syntactic A = B$ then $\interpret{A} = \interpret{B}$.
The converse of this is \emph{completeness}:
If the semantics of two diagrams are equal,
then they should also be equal from the point of view of the syntax.

\begin{definition}[Completeness {\cite{SelingerSurvey}}]  \label{defCompleteness} 
	A graphical calculus is complete
	if the equivalence relation $\sim$ reflects any equality under the interpretation. That is:
	\begin{align}
		\interpret{A} = \interpret{B} \implies A \sim B
	\end{align}
\end{definition}

Combined theorems showing both soundness
and completeness of a graphical calculus may be referred to as
\emph{coherence theorems}
\cite{SelingerSurvey},
but we will just refer to soundness and completeness as separate properties
in this work.
The final property mentioned in the introduction
to this chapter is that of \emph{universality},
the ability for a graphical calculus to represent everything it needs to.
The idea is analogous to that of a universal gate set in computer science \cite{NielsonChuang}.

\begin{definition}[Universality]  \label{defUniversal}
	A graphical calculus is called universal
	for a collection of morphisms if each of those morphisms can be represented
	by a diagram. I.e. for a collection of morphisms $F$:
	\begin{align}
		\forall f \in F\ \exists D_f \text{ such that } \interpret{D_f} = f
	\end{align}
	Note that just because we require a diagram to exist does not mean we have an efficient
	method of constructing it.
\end{definition}

The origin of these languages was for them to be graphical calculi
for categorical quantum computing \cite{Coecke08}.
The ZX calculus was created as a representation
of manipulations of qubits
(the qubits themselves represented as points in $\bbC^2$).
Since then there has been some generalisation
both beyond qubits
(e.g. a qu\textbf{d}it version of ZX \cite{Ranchin14}, a qudit being a point in $\bbC^d$)
and into systems over rings other than $\bbC$ (e.g. $\ZWR$ in \S\ref{secZW}).
This work remains focussed on the `-bit' case,
but is not restricted to qubits:

\begin{definition}[R-bits and qubits {\cite[Definition 5.1]{AmarThesis}}]  \label{defRBits} 
	Let R be a commutative ring.
	The PROP of $\bit{R}$ is the full monoidal subcategory
	of $\rmod{R}$ whose objects are tensor products of a finite number of
	copies of $R \oplus R$.
	Morphisms are complex matrices with $2^n$ columns and $2^m$ rows.
	The category of qubits is simply $\bit{\bbC}$ and written $\Qubit$.
\end{definition}

When talking about $\bit{R}$
we will sometimes use bra-ket notation
(also called Dirac notation) to save space,
as it allows us to clearly express the interpretations
of those generators with arbitrary arity.

\begin{definition}[Bra-kets {\cite[p62]{NielsonChuang}}]  \label{defBraKet} 
	A ket $\ket{\cdot}$ is used to indicate a (column) vector,
	with a bra $\bra{\cdot}$ used to indicate
	the dual linear functional.
	Tensor products of the vectors and functionals is shown by concatenation:
	$\bra{xy} = \bra{x} \tensor \bra{y}$
	and $\ket{xy} = \ket{x} \tensor \ket{y}$,
	and the composition (inner product) of the vector and functional is shown as
	$\braket{x}{y}$.
	In our work over $\bit{R}$ we will sometimes
	refer to the following two bases:
	\begin{align}
		\begin{split}
			\ket{0} & := \begin{pmatrix}1 \\ 0\end{pmatrix} \\
			\ket{1} & := \begin{pmatrix}0 \\ 1\end{pmatrix}
		\end{split}
		  &
		\text{and}
		  &
		\begin{split}
			\ket{+} & := \begin{pmatrix}\frac{1}{\sqrt{2}} \\ \frac{1}{\sqrt{2}}\end{pmatrix} \\
			\ket{-} & := \begin{pmatrix}\frac{1}{\sqrt{2}} \\ \frac{-1}{\sqrt{2}}\end{pmatrix}
		\end{split}
	\end{align}
\end{definition}

\bit{R} sets the scene for the calculi presented in this thesis.
When considering quantum computing
it is tempting to just look at qubits,
but it turns out that this more general setting of $\bit{R}$
still has relevance to quantum computing.
For example the fragment of Clifford+T
quantum computing represents a subcategory of $\bit{\bbC}$.
Also, when talking about the morphisms in \bit{R}, we will
use the term \emph{matrix}
rather than \emph{morphism} or \emph{tensor}
to indicate that we have an ordering on the inputs and outputs of each linear map,
which allows us to explicitly write out the matrix (using the choice of basis implied by $R\oplus R$).
Every diagram that is interpreted into $\bit{R}$ therefore represents a matrix
with elements from the ring $R$.

\begin{remark} \label{remSZXAsDiagramsForDiagrams}
	There is a graphical calculus called scalable ZX (SZX from here on) \cite{SZX}
	that is notable here because
	one can view SZX as a graphical calculus that represents (potentially enormous) ZX diagrams.
	That is to say that one can view SZX as a graphical calculus with
	diagrams and diagrammatic rules as laid out in Ref.~\cite{SZX},
	with an interpretation that turns a SZX diagram into a ZX diagram.
	We highlight this because all the other graphical calculi considered in this thesis
	have interpretations that turn diagrams into matrices,
	as indeed SZX does in its original presentation.
\end{remark}

\section{Rewriting and reduction} \label{secRewriting}

First we shall establish notation.
We will use horizontal,
side-by-side placement to indicate the tensor product $\tensor$,
and vertical composition (with appropriate edge connections) to indicate the usual $\comp$ composition.
We shall be working \textbf{bottom-to-top} in this thesis.
I.e. inputs will be at the bottom of the diagram, and outputs will be at the top.
In fact there are good reasons\footnote{Admittedly the author seems to be alone in believing that
	the argument for using the right-to-left direction for SZX or ZQ qualifies as a good reason}
for any choice of direction, but in this thesis we are using bottom to top.

All the diagrams we will be considering in this thesis
will (implicitly) be in the form of \emph{pattern graphs} \cite{Kissinger2012PatternGR}.
Pattern graphs evolved from \emph{string diagrams}
\cite{Penrose},
via \emph{open graphs} \cite{DixonKissingerOpenGraphs}
(later renamed to \emph{string graphs}).
The benefit of string graphs over string diagrams
is that string graphs give a combinatoric representation
that is amenable to Double Pushout Rewriting
(Definition~\ref{defDPO}).
The benefit of pattern graphs over string graphs
is that pattern graphs allow for !-boxes (Definition~\ref{defBBox}).
While pattern graphs and double pushout rewriting
are the backbone of graphical calculi
for quantum computing,
pattern graphs will not be used explicitly
beyond this section of the thesis.
We will instead favour the representation
of !-boxes as the blue boxes found in Ref.~\cite{ZH}.
We will also use ellipses to indicate repeated structure where clear.
Ref.~\cite{AleksThesis} contains a wealth of information
and results about string diagrams, string graphs and rewriting,
and the following definitions are all in the style of Ref.~\cite[\S4.2]{AleksThesis}.

\begin{definition}[Rewrite Rule]  \label{defRewriteRule}
	A graphical rewrite rule (with fixed boundary) is a pair of string graphs $L$ and $R$,
	with a third graph $I$,
	called the \emph{invariant subgraph} and monomorphisms $I \into L$ and $I \into R$.
	In this thesis we will insist that the image of $I$ is the boundary
	(i.e. inputs and outputs) of $L$ and $R$.
	In particular this means that $L$ and $R$ have isomorphic boundaries.
\end{definition}

\begin{definition}[Matching]  \label{defMatching}
	Given a rewrite rule $L \from I \to R$
	we say that $L$ \emph{matches} onto a graph $G$ if there
	is a monomorphism $m : L \into G$
	such that the set of vertices $m(I)$ disconnects $m(L) \setminus m(I)$
	from $G \setminus m(L)$.
	This is called the \emph{no dangling wires} condition.
\end{definition}

The purpose of rewriting is to cut out an instance of the graph $L$
(found by matching $L$ onto $G$)
and then replace it with an instance of the graph $R$.
We can make this idea rigorous in the following way:

\begin{definition}[Double Pushout Rewriting {\cite{DPO}}]  \label{defDPO} 
	A Double Pushout Rewriting (DPO rewriting) of the rule
	$L \from I \to R$,
	with $L$ matched onto $G$ by $m$
	produces $H$ in the below diagram.
	Both squares are pushouts.

	\begin{align}
		\begin{tikzcd}[ampersand replacement=\&]
			L \arrow[d, "m"]  \& \arrow[l] I \arrow[r] \arrow[d] \& R \arrow[d] \\
			G \arrow[ur, phantom, "\urcorner", very near start] \& \arrow[l] G' \arrow[r] \& \arrow[ul, phantom, "\ulcorner", very near start]H
		\end{tikzcd}
	\end{align}
	The existence and uniqueness of such a rewrite is covered in Ref.~\cite{Kissinger2012PatternGR},
	based on the results of Ref.~\cite{DixonKissingerOpenGraphs}.
\end{definition}

Now that we can rewrite graphs
we shall touch on how to perform \emph{reductions}.
This idea originates in that of \emph{term rewriting},
see Ref.~\cite{Terms}.
Term rewriting combines logic and algebra
to give a system of applying modifications
to an algebraic term (e.g. $X^2 + 2 + 2X^2 -1$)
in order to find an equivalent term with
certain properties (e.g. $3X^2 + 1$,
which has grouped similar monomials together).
Term rewriting is an essential part of conjecture synthesis,
but this work does not make use of any notable changes to the
term reduction system already present in earlier work \cite{AleksCoSy, AleksThesis}.

\begin{definition}[Graphical reduction and reducibility {\cite{AleksThesis}}]  \label{defReduction} 
	Given an ordering $\omega$ on the diagrams of a graphical calculus
	we say a rule application $A \mapsto B$ is a reduction if $\omega(A) > \omega(B)$.
	A diagram $A$ is called reducible (or a \emph{redex}) by a set of rules $\Rules$ if
	a single application of one of the rules in $\Rules$ yields a reduction of $A$.
\end{definition}

\section{A zoo of calculi} \label{secZoo}

We will now define the generators and interpretations
of the families of graphical calculi we will use in this thesis.
We will not provide all the rules for each calculus here,
just the ones we will have reason to reference on occasion.
This is because several different rulesets exist for each calculus,
but also because any time we will need a particular ruleset we will
state the rules explicitly.
Note that the calculi ZQ and $\ring_R$
(introduced in this thesis) are not covered in this background chapter,
and are instead given their own chapters (\S\ref{chapRingR} and \S\ref{chapZQ}).

\subsection{Wires} \label{secWires}

\begin{definition}[Wires, cups, and caps {\cite[\S4, \S5]{PQP}}]  \label{defWires} 
	(with thanks to Ref.~\cite{Selinger07} for the standardisation of much of the language and notation).
	All of the calculi considered here are based on the same
	underlying structure of wires.
	These wires represent the structure
	of a compact closed PROP over $\bit{R}$ without any additional morphisms.
	There will be a section in this thesis
	that does not assume compact closure,
	but we have sign-posted this clearly
	(see Definition~\ref{defRingPropS}).
	In ZX, ZW and ZH (and, when we get to them, $\ring$ and ZQ)
	we have the following graphical elements:
	\begin{align}
		\text{wire}  &   & \node{none}{}    \quad & \interpretedas \begin{pmatrix}1 & 0 \\ 0 & 1\end{pmatrix}   \\[\rowgap]
		\text{swap}  &   & \vc{\InputIfFileExists{./figures/wire/swap.tikz}{}{Missing file!}}    & \interpretedas \begin{pmatrix}1 & 0 & 0 & 0 \\ 0 & 0 & 1 & 0\\ 0 & 1 & 0 & 0 \\ 0 & 0 & 0 & 1\end{pmatrix}   \\[\rowgap]
		\text{cap}   &   & \dcap                  & \interpretedas \begin{pmatrix}
			1 & 0 & 0 & 1
		\end{pmatrix}   \\[\rowgap]
		\text{cup}   &   & \dcup                  & \interpretedas \begin{pmatrix}
			1 & 0 & 0 & 1
		\end{pmatrix}^T \\[\rowgap]
		\text{empty} &   & \vc{\InputIfFileExists{./figures/wire/empty.tikz}{}{Missing file!}}\  & \interpretedas \begin{pmatrix} 1	\end{pmatrix}
	\end{align}

	Note that the swap, in bra-ket notation, has the action of $\ket{xy} \mapsto \ket{yx}$,
	and that our cups and caps obey the \emph{snake equations}:
	\begin{align}
		\vc{\InputIfFileExists{./figures/wire/snakel.tikz}{}{Missing file!}} = \vc{\begin{tikzpicture}
	\begin{pgfonlayer}{nodelayer}
		\node [style=none] (0) at (-0.5, -1) {};
		\node [style=none] (5) at (-0.5, 1) {};
	\end{pgfonlayer}
	\begin{pgfonlayer}{edgelayer}
		\draw (5.center) to (0.center);
	\end{pgfonlayer}
\end{tikzpicture}
} = \vc{\InputIfFileExists{./figures/wire/snaker.tikz}{}{Missing file!}}
	\end{align}
\end{definition}

For the rest of this section we won't mention the wires when defining the calculi.

\subsection{ZX} \label{secZX}

Officially ZX is called `the ZX-calculus'
but since there is no risk of ambiguity,
and it is common to hear people do so,
we shall just use the short form
(likewise `ZW' and `ZH').
We shall, however,
indicate \emph{fragments}
of each calculus
using either words or subscripts.
For example `Clifford+T ZX', or `$\ZX_{\piby{4}}$'.
ZX originated in Ref.~\cite{Coecke08},
and since then has flourished into the most recognisable of the calculi
presented here.
We covered part of its story in \S\ref{chapIntroduction},
and given the range of definitions and fragments there is still debate as to whether there is `a'
or `the' ZX calculus, or even what requirements a graphical
calculus would need to fulfil to count as a ZX calculus\footnote{Private communication
with Niel de Beaudrap and Harny Wang}.

\begin{definition}[ZX]  \label{defZX}
	The ZX-calculus comes in several fragments.
	Each fragment is determined by the phases (labels) allowed on the nodes.
	In all cases the phases form a group
	that is a subgroup of $[0, 2\pi)$ under addition.
	Certain special subgroups have special names (given below)
	and only certain subgroups have been the subject of study.

	\begin{center}
		\begin{tabular}{c l l c}
			Group        & Name       & Shorthand        & Completeness reference   \\ \hline
			$[0, 2\pi)$  & Universal  & $\ZX_U$          & \cite{UniversalComplete} \\
			$<\piby{2}>$ & Stabilizer & $\ZX_{\piby{2}}$ & \cite{Backens16Thesis}   \\
			$<\piby{4}>$ & Clifford+T & $\ZX_{\piby{4}}$ & \cite{SimonCompleteness} \\
			$G$          & -          & $\ZX_{G}$        & \cite{JPVNormalForm}     \\
		\end{tabular}
	\end{center}

	Note that Ref.~\cite{Backens16Thesis} combines the scalar and scalar-free completeness proofs
	by the same author.
	The interpretation of a ZX diagram is understood as a complex matrix,
	even though only the Universal fragment of ZX is universal onto all complex matrices.
	When a generic group $G$ is considered it is always assumed to contain $\piby{2}$,
	and usually assumed to contain $\piby{4}$. We now present the generators, interpretation and sample rules for ZX:

	\begin{itemize}
		\item Green (Z) and red (X) spiders with any number of inputs and outputs (including 0):
		      \begin{align}
			      \spider{gn}{\alpha} & \interpretedas \ket{0\dots 0}\bra{0\dots 0 } + e^{i \alpha}  \ket{1\dots 1}\bra{1\dots 1} \\
			      \spider{rn}{\alpha} & \interpretedas \ket{+\dots +}\bra{+\dots +} + e^{i \alpha}  \ket{-\dots -}\bra{-\dots -}
		      \end{align}
		      Phases of $0$ are instead left blank.
		\item Hadamard gates \cite[p19]{NielsonChuang} have one input and one output:
		      \begin{align}
			      \node{h}{} & \interpretedas \frac{1}{\sqrt{2}} \begin{pmatrix}
				      1 & 1 \\ 1 & -1
			      \end{pmatrix}
		      \end{align}
		      There is a trend, welcomed by the author for its clarity,
		      to draw edges decorated with Hadamard gates as dashed, blue lines
		      (in situations where there is neither ambiguity nor loss of rigour).
		      \begin{align}
			      \vc{\begin{tikzpicture}
	\begin{pgfonlayer}{nodelayer}
		\node [style=none] (6) at (0, 0.75) {};
		\node [style=none] (7) at (0, -0.25) {};
	\end{pgfonlayer}
	\begin{pgfonlayer}{edgelayer}
		\draw [style=hadamard edge] (6.center) to (7.center);
	\end{pgfonlayer}
\end{tikzpicture}
} :=\node{h}{}
		      \end{align}
		      Hadamard gates are derivable from red and green spiders \cite{Duncan09}
		      but, as we shall see below,
		      red spiders can be constructed from green spiders and Hadamard gates.
		      It has become a matter of authorial preference whether a given paper will consider ZX
		      to be constructed from red and green spiders, or from green spiders and Hadamard gates
		      (in a manner similar to that of graph states), or even all three.
		\item Meta-rules:
		      \begin{itemize}
			      \item Only topology matters
			            (the underlying structure is an open graph)
			      \item All rules hold with red and green swapped
		      \end{itemize}
		\item Example rewrite rules. We present them here using !-boxes,
		as is standard,
		and a formal definition is given later as Definition~\ref{defBBox}.
		These light blue boxes indicate `0 or more copies of the subdiagram',
		requiring the same number of copies for the matching box on the other side of the equation.
		For example if the !-box surrounding the input wire on the left hand side of the rule Spider 1
		is expanded as having 3 copies, then the box around the inputs on the right hand side of Spider 1 is also
		understood to indicate 3 copies.
		      \begin{align}
			      \vc{\InputIfFileExists{./figures/ZX/S1bb_l.tikz}{}{Missing file!}} & \by{\text{Spider } 1} \vc{\InputIfFileExists{./figures/ZX/S1bb_r.tikz}{}{Missing file!}} & \node{gn}{}         & \by{\text{Spider } 2} \node{none}{}  \\
			      \vc{\InputIfFileExists{./figures/ZX/BAbb_l.tikz}{}{Missing file!}} & \by{\text{Bialgebra}} \vc{\InputIfFileExists{./figures/ZX/BAbb_r.tikz}{}{Missing file!}} & \vc{\InputIfFileExists{./figures/ZX/Hopf_l.tikz}{}{Missing file!}} & \by{\text{Hopf}} \vc{\InputIfFileExists{./figures/ZX/Hopf_r.tikz}{}{Missing file!}} \\
			      \vc{\InputIfFileExists{./figures/ZX/ColourSwapbb_l.tikz}{}{Missing file!}} & \by{\text{Colour swap}} \vc{\InputIfFileExists{./figures/ZX/ColourSwapbb_r.tikz}{}{Missing file!}}
		      \end{align}
	\end{itemize}
\end{definition}

\subsection{ZW} \label{secZW}

The ZW calculus \cite{ZW, AmarThesis} was created for manipulating the GHZ and W entanglement states,
representing the two classes of connected entanglement available for three qubits.
The W state (generalised to any arity) is represented as a black spider.
The GHZ state, too, is generalised in terms of arity,
but also given a phase, becoming a Z spider
similar to the Z spider of ZX.
We will introduce a convention of drawing nodes
that never contain a phase
(e.g. the black, rounded spider of ZW)
smaller than nodes that can contain a phase
(e.g. the white, rounded spider of ZW).

In contrast with ZX, which aims to represent only qubit quantum computing,
ZW has more freedom.
By allowing the phases of the Z spider to come from some choice of ring $R$,
the diagrams in $\ZWR$ (as that fragment is then called)
are universal over $\bit{R}$.
Not only that, but the rules of $\ZWR$ remain sound and complete,
independent of the choice of $R$.
To recover qubit quantum computing one simply sets $R$ to
$\bbC$, but it should be noted that completeness of the Clifford+T \cite{SimonCompleteness}
and Universal \cite{UniversalComplete} fragments of ZX were shown by equivalences with two
different fragments of ZW.
ZW's own completeness comes from the existence of a normal form \cite{ZW}.

\begin{definition}[ZW {\cite{AmarThesis}}]  \label{defZW} 
	The ZW calculus is generated by Z (white)
	and W (black) spiders.
	It is parameterised by a choice of commutative ring $R$,
	with an interpretation into $\bit{R}$.
	\begin{itemize}

		\item The Z (white) spider, parameterised by $r \in R$,
		      and the W (black) spider:

		      \begin{align}
			      \spider{white}{r} & \interpretedas \ket{0\dots 0}\bra{0\dots 0 } + r\ket{1\dots 1}\bra{1\dots 1}          \\
			      \wspider{smallblack}{}{n}   & \interpretedas \sum\limits_{k=1}^n \ket{\underbrace{0\ldots 0}_{k-1}1\underbrace{0\ldots 0}_{n-k}}
		      \end{align}
		      When the phase on the white node is $1$ it is left blank.
		\item The Crossing $x$ is a non-commutative, derived generator of fixed arity used in the rules:

		      \begin{align}
			      \vc{\InputIfFileExists{./figures/ZW/crossing.tikz}{}{Missing file!}} \interpretedas
			      \begin{pmatrix}
				      1 & 0 & 0 & 0  \\
				      0 & 0 & 1 & 0  \\
				      0 & 1 & 0 & 0  \\
				      0 & 0 & 0 & -1
			      \end{pmatrix}
		      \end{align}
		\item Meta-rules:
		      \begin{itemize}
			      \item Only topology matters
			            (the underlying structure is an open graph)
		      \end{itemize}
		\item Example rewrite rules:
		      \begin{align}
			      \vc{\InputIfFileExists{./figures/ZW/cutw_l.tikz}{}{Missing file!}}  & \by{\text{cut W}} \vc{\InputIfFileExists{./figures/ZW/cutw_r.tikz}{}{Missing file!}}  \label{eqnZWcutW} & \vc{\InputIfFileExists{./figures/ZW/cutz_l.tikz}{}{Missing file!}}      & \by{\text{cut Z}} \vc{\InputIfFileExists{./figures/ZW/cutz_r.tikz}{}{Missing file!}}     \\
			      \vc{\InputIfFileExists{./figures/ZW/rei3x_l.tikz}{}{Missing file!}} & \by{\text{Reidemeister 3}} \vc{\InputIfFileExists{./figures/ZW/rei3x_r.tikz}{}{Missing file!}}          & \vc{\InputIfFileExists{./figures/ZW/rngplusrs_l.tikz}{}{Missing file!}} & \by{\text{Plus}} \vc{\begin{tikzpicture}
	\begin{pgfonlayer}{nodelayer}
		\node [style=none] (1) at (0, 1.5) {};
		\node [style=none] (26) at (0, -1.5) {};
		\node [style=white] (31) at (0, 0) {r+s};
	\end{pgfonlayer}
	\begin{pgfonlayer}{edgelayer}
		\draw (1.center) to (31.center);
		\draw (31.center) to (26.center);
	\end{pgfonlayer}
\end{tikzpicture}
} \\
		      \end{align}
	\end{itemize}
\end{definition}

\subsection{ZH}

The `Z' of ZW, ZX, and ZH (and indeed ZQ) are all minor modifications on the same theme.
The `H' of ZH refers to the H-box;
a generalisation of the Hadamard gate in both arity and phase.
One important notational aspect is that the Hadamard node of ZX
(a yellow rectangle with no phase)
is represented in ZH as a white rectangle with a phase of $-1$,
which can feel counter-intuitive when moving between the calculi.
Similar to ZW, the completeness of ZH was proven via
the existence of a normal form for its diagrams.

\begin{definition}[ZH {\cite{ZH}}]  \label{defZH} 
	The ZH calculus is generated by phase free Z spiders
	and generalised (in both arity and phase) Hadamard nodes called `H-boxes' \cite{ZH}.
	The original paper only allows the H-boxes to have phases from the complex numbers,
	but one of the results in this thesis (\S\ref{secZHR}) is that this can be extended
	to any commutative ring containing the element $\half$,
	preserving soundness and completeness.
	We will give the version of ZH from Ref.~\cite{ZH} below,
	our only modification being that we refer to it as $\ZH_\bbC$:

	\begin{itemize}
		\item The Z (white) spider and H-box (rectangle with $m$ inputs and $n$ outputs) generate $\ZH_\bbC$:
		      \begin{align}
			      \spider{smallZ}{} & \interpretedas \ket{0\dots 0}\bra{0\dots 0 } + e^{i \alpha}  \ket{1\dots 1}\bra{1\dots 1}                    \\
			      \spider{ZH}{c}    & \interpretedas \sum_{\text{bitstrings}} c^{i_1\dots i_m j_1 \dots j_n} \ket{j_1\dots j_n}\bra{i_1 \dots i_m}
		      \end{align}

		      The sum in the interpretation of the H-box is over all bitstrings $i_1 \dots i_m j_1 \dots j_n$.
		      The reader may find it easier to view this interpretation
		      as a matrix where every entry is $1$, with the exception of the bottom right entry which is $c$.
		      In the case where this matrix only has one row and one column (i.e. a scalar)
		      the sole entry is $1 + c$.
		      When the phase is $-1$ it is left off the H-box.
		\item The derived generators of ZH include the grey spider and the NOT gate:
		      \begin{align}
			      \spider{smallgrey}{} & := \vc{\InputIfFileExists{./figures/ZH/grey_spider.tikz}{}{Missing file!}} \\
			      \node{smallgrey}{\neg}    & := \vc{\InputIfFileExists{./figures/ZH/neg_node.tikz}{}{Missing file!}}
		      \end{align}
		\item Meta-rules:
		      \begin{itemize}
			      \item Only topology matters
			            (the underlying structure is an open graph)
		      \end{itemize}
		\item Example rewrite rules:
		      \begin{align}
			      \vc{\InputIfFileExists{./figures/ZH/HS1_l.tikz}{}{Missing file!}} & \by{\text{H spider}} \vc{\InputIfFileExists{./figures/ZH/HS1_r.tikz}{}{Missing file!}} & \state{smallwhite}{}  & \by{\text{Unit}} \state{ZH}{1}              \\
			      \vc{\InputIfFileExists{./figures/ZH/a_l.tikz}{}{Missing file!}}   & \by{\text{Addition}} \vc{\begin{tikzpicture}
	\begin{pgfonlayer}{nodelayer}
		\node [style=ZH] (5) at (1.75, 0) {$\frac{a+b}{2}$};
		\node [style=none] (6) at (1.75, 1) {};
		\node [style=ZH] (8) at (1, 0.5) {$2$};
	\end{pgfonlayer}
	\begin{pgfonlayer}{edgelayer}
		\draw (6.center) to (5.center);
	\end{pgfonlayer}
\end{tikzpicture}
}   & \vc{\InputIfFileExists{./figures/ZH/m_l.tikz}{}{Missing file!}} & \by{\text{Multiplication}} \vc{\begin{tikzpicture}
	\begin{pgfonlayer}{nodelayer}
		\node [style=none] (6) at (1.5, 0.75) {};
		\node [style=ZH] (7) at (1.5, 0) {$a \times b$};
	\end{pgfonlayer}
	\begin{pgfonlayer}{edgelayer}
		\draw (7.center) to (6.center);
	\end{pgfonlayer}
\end{tikzpicture}
}
		      \end{align}
	\end{itemize}
\end{definition}

\section{Decorated and simple diagrams} \label{secSimpleDiagrams}

The presentation of the calculi above (\S\ref{secZoo}) already involved
using variables to indicate arbitrary phases.
The meaning of this is usually considered to be so clear that it is not worth remarking upon:
The calculus $\ZH_\bbC$ was introduced in Ref.~\cite{ZH}
using variables on the phases.
We will, however, have cause to be rather pedantic about this nuance in this thesis.
Our reason for caring about this nuance is that when we later try to generate
and generalise theorems it is important to be able to introduce and verify
these useful components (see \S\ref{chapCoSy}).
First we will clearly define phase variables and !-boxes, and then declare a diagram to be \emph{simple}
if it does not use them.

\begin{definition}[Phase Algebra]  \label{defPhaseAlgebra}
	The phase algebra for a vertex in a graphical calculus
	is the collection of terms allowed to be written as a phase for that vertex.
\end{definition}

\begin{example}[ZX Phase Algbera]\label{exaZXPhase}
	ZX uses a phase \emph{group} algebra,
	which reflects the fragment of the calculus (Definition~\ref{defZX}).
	For example, Universal ZX has phase group $[0, 2\pi)$ under addition,
	and the red and green spiders can be decorated with terms from that algebra.
	The following are therefore valid Universal ZX diagrams:
	\begin{align}
		\spider{gn}{\pi/2} \qquad \spider{rn}{\pi/4 + \pi/6+12 \pi/3}
	\end{align}
	The following equation is not a Universal ZX diagram, because the group algebra does not contain the symbol `$\times$':
	\begin{align}
		\spider{gn}{\pi/2 \times \pi/6}
	\end{align}
\end{example}

\begin{definition}[Phase variables]  \label{defPhaseVariable}
	A phase variable is a formal variable introduced to the phase algebra.
	For ZX (i.e. phase group variables) we tend to use $\alpha, \beta, \dots$,
	and for ZH and ZW (i.e. phase ring variables) use $a,b,c, \dots$ or $r,s,\dots$.
	We can \emph{evaluate} a phase variable to some element of the phase algebra
	by replacing all instances of that phase variable in the diagram with that element.
\end{definition}

\begin{example}[ZX Phase Variable] \label{exaZXPhaseVariable}
	The following are considered to be Clifford+T ZX diagrams
	but with the phase variable $\alpha$:
	\begin{align}
		\spider{rn}{\alpha} \qquad \spider{gn}{\pi/2-3\alpha}
	\end{align}
	Note that this is indistinguishable
	from a ZX diagram where the phase group
	is generated by $\pi/4$ and $\alpha$.
	We will use this observation when talking about parameter spaces in \S\ref{chapConjectureInference}.
\end{example}

Now that we have covered phase variables let us look at !-boxes.
!-boxes indicate repeated elements in a diagram,
and are notational sugar\footnote{The region indicated by a !-box is less visually cluttered than drawing out the !-vertices and directed edges of a pattern graph.} for the concept of
pattern graphs introduced in Ref.~\cite{Kissinger2012PatternGR}.

\begin{definition}[!-box {\cite{Kissinger2012PatternGR}}]  \label{defBBox} 
	A !-box (pronounced `bang box') indicates a collection of nodes (potentially including wire vertices)
	in a diagram. By an instantiation of a !-box we mean choosing a number $n \in \bbN$
	and making $n$ copies of those nodes, preserving their connectivity.
	A !-box is drawn as a blue box surrounding the nodes inside the !-box.
\end{definition}

\begin{example}[!-boxes in ZH {\cite[\S2.3]{ZH}}]  \label{exaZHBangBox}
	The diagram below represents an entire family of ZH diagrams:
	\begin{align}
		\vc{\InputIfFileExists{./figures/ZH/bang_box_example.tikz}{}{Missing file!}} \leftrightarrow
		\set{
		\ \vc{\begin{tikzpicture}
	\begin{pgfonlayer}{nodelayer}
		\node [style=smallZ] (0) at (-0.25, 0) {};
		\node [style=smallGrey] (7) at (0.25, 0) {};
	\end{pgfonlayer}
\end{tikzpicture}
}\ ,
		\ \vc{\InputIfFileExists{./figures/ZH/bang_box_example1.tikz}{}{Missing file!}}\ ,
		\ \vc{\InputIfFileExists{./figures/ZH/bang_box_example2.tikz}{}{Missing file!}}\ ,
		\ \vc{\InputIfFileExists{./figures/ZH/bang_box_example3.tikz}{}{Missing file!}}\ ,
		\ \dots\
		}
	\end{align}
\end{example}

\begin{definition}[Simple and decorated]  \label{defSimpleDiagram}
	A diagram that does not contain any phase variables
	or !-boxes is called simple.
\end{definition}

Any diagram that isn't simple
represents a family of simple diagrams,
and it is important to be able to recover any particular simple diagram
from this family.
When talking about these families we will use the term
`parameterised family',
and be explicit, where possible, about what these parameters are.
This isn't just to say `parameterised by !-boxes and phase variables',
but to give each of these a name;
phase variables are simply referred to by their own name
(e.g. the phase variable $\alpha$ is still referred to as $\alpha$)
and for !-boxes we just ascribe to each !-box the name $\delta_1, \delta_2, \dots$\ .
With explicit names for the parameters we can recover a simple diagram
by \emph{instantiating} each !-box some number of times,
and \emph{evaluating} each phase variable to some element of the phase algebra.
We will use the notation $\alpha | \alpha = \pi$ to indicate
that $\alpha$ is a parameter, and that it should be evaluated to the value $\pi$.
When we are talking about equations it is understood that the evaluation and instantiation
is to be applied to both the left and right hand side of the equation.

\begin{example}[An explicit evaluation and instantiation] \label{exaExplicitInstantiation}
	The (slightly simplified) spider law in Universal ZX is parameterised over:
	\begin{itemize}
		\item !-boxes: $\delta_1 \in \bb{N} $ inputs and $\delta_2 \in \bb{N} $ outputs
		\item phase variables: $\alpha_1, \alpha_2 \in [0, 2\pi)$
	\end{itemize}
	We write this parameterised family of equations as:
	\begin{align}
		\family{
		\vc{\InputIfFileExists{./figures/ZX/param_s1_l.tikz}{}{Missing file!}} = \; \vc{\InputIfFileExists{./figures/ZX/param_s1_r.tikz}{}{Missing file!}}
		}_{\alpha_1, \alpha_2, \delta_1, \delta_2}
	\end{align}
	We now instantiate some of the parameters of the spider law,
	resulting in what is still an infinite, parameterised family:
	\begin{align}
		  & \family{
			\vc{\InputIfFileExists{./figures/ZX/param_s1_l.tikz}{}{Missing file!}} = \; \vc{\InputIfFileExists{./figures/ZX/param_s1_r.tikz}{}{Missing file!}}
		}_{\alpha_1, \alpha_2, \delta_1, \delta_2 | \alpha_1 = \pi, \delta_1 = 2} \\
		= & \family{
			\vc{\InputIfFileExists{./figures/ZX/param_s1_l_pi.tikz}{}{Missing file!}} = \; \vc{\InputIfFileExists{./figures/ZX/param_s1_r_pi.tikz}{}{Missing file!}}
		}_{\alpha_2, \delta_1, \delta_2 | \delta_1 = 2} \\
		= & \family{
			\vc{\InputIfFileExists{./figures/ZX/param_s1_l_pi_2.tikz}{}{Missing file!}} = \; \vc{\InputIfFileExists{./figures/ZX/param_s1_r_pi_2.tikz}{}{Missing file!}}
		}_{\alpha_2, \delta_2}
	\end{align}
\end{example}

When the topic of verification comes up later (\S\ref{chapConjectureVerification})
the goal is to be able to say that some family of equations holds true,
just by checking whether some subset of that family holds true.

\begin{definition}[Soundness for an equation]  \label{defSimpleSound}
	An equation of diagrams $\bbD_1 = \bbD_2$ is sound if $\interpret{\bbD_1} = \interpret{\bbD_2}$.
	A family of equations is sound if every equation in the family is sound.
\end{definition}

\begin{definition}[Verifying family]  \label{defFamilySound}
	A family of equations $S$ is said to verify the family of equations $S'$
	if $S$ being sound implies that $S'$ is sound:
	\begin{align}
		e \text{ sound } \forall e \in S \quad\implies\quad e' \text{ sound } \forall e' \in S'
	\end{align}
\end{definition}

One of the questions considered later in this thesis is
`when can a parameterised family of equations be verified by
a small set of simple equations?'
Our reasons for wanting to answer this
will hopefully be made clear in the next chapter,
where we give the background for conjecture synthesis.  \thispagestyle{empty} 
\chapter{Conjecture Synthesis}
\label{chapCoSy}
\thispagestyle{plain}
\noindent\emph{In this chapter:}
\begin{itemize}
	\item[\chapterbullet] We discuss conjecture synthesis for terms and string graphs
	\item[\chapterbullet] We introduce a novel step in the conjecture synthesis process
	\item[\chapterbullet] We discuss the difficulties associated with our graphical calculi
	\item[\chapterbullet] We give an example from a conjecture synthesis run
\end{itemize}

It was mentioned in the introduction
that IsaCoSy
(the Isabelle Conjecture Synthesis)
was created with the
`[...] aim to automatically produce a useful set
of theorems'.
The full quotation, however, is:

\begin{displayquote}[IsaCoSy: Synthesis of inductive theorems \cite{ISACOSY}]
	Given a set of initial definitions of
	recursive datatypes and functions, we aim to automatically produce a useful set
	of theorems, that will be useful as lemmas in further proofs, by either a human or
	an automated theorem prover.
\end{displayquote}

Where this thesis deviates from that work is in the data structures considered
(for us it will be diagrams rather than recursive datatypes)
and in the introduction of a \emph{generalisation} step	
(to be discussed later).
Where that work deviated from its predecessors was in terms of method.
IsaCoSy follows a \emph{generative} approach,
where hypotheses are generated and then assessed,
compared to a \emph{deductive} approach
where new theorems are deduced from old
(see Figure~\ref{figDeductive} for a diagram indicating the deductive approach).

\begin{figure}[ht]
	\center
	\includegraphics[width=0.5\textwidth]{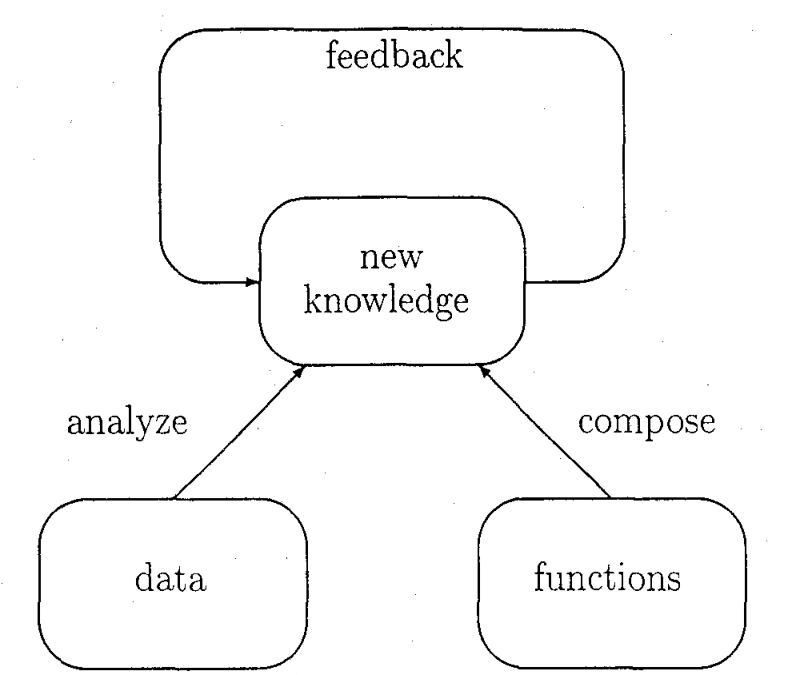}
	\caption{A diagram captioned `The construction of new knowledge'
		from Ref.~\cite{Ammon1993},
		indicating the deductive approach to automated theorem generation \label{figDeductive}}
\end{figure}

The generative approach relies on the concept of reduction,
a concept that was briefly covered in \S\ref{secRewriting}.
For IsaCoSy this meant term reduction,
but for this work it will mean string graph reduction.
The key in the generative approach is to not generate
(or to generate and immediately discard) any \emph{reducible} objects,
be they terms or string graphs.

\begin{example}[Adding 0 is pointless {\cite[Example 1]{ISACOSY}}]  \label{exaIsaCoSyAdd0}
	The following is a term reduction describing part of the nature of $+$ acting on the natural numbers:
	\begin{align}
		0 + x \rto x
	\end{align}
	We therefore know that if we generate an expression containing $0 + x$
	it would have been both simpler and equivalent to generate the same expression just containing $x$ instead.
\end{example}

While Example~\ref{exaIsaCoSyAdd0} feels entirely obvious this simple idea
works for any abstract reduction system.
What's more, conjecture synthesis can generate new reductions:
Every orientable theorem
results in a new reduction that will then further limit the space that needs to be searched.
In contrast to the deductive model of Figure~\ref{figDeductive}
we can therefore use the generative model of Figure~\ref{figGenerative}.
This model has been implemented in the QuantoCoSy part of Quantomatic \cite{Quantomatic}.
The test for how well the model works is (counter-intuitively)
to see how many \emph{fewer} theorems it creates for a given system.
Any theorems it doesn't provide to the user as output
were theorems that were reducible.
Earlier work with an earlier version of QuantoCoSy \cite{AleksCoSy}
claimed that it produced exponentially fewer rules
when considering GHZ/W (later to become ZW) diagrams up to a certain size:

\begin{displayquote}[Synthesising Graphical Theories \cite{AleksCoSy}]
	The results there were promising, as we
	demonstrated an exponential drop-off in the number of extraneous rules generated when using
		[a reducible object]-eliminating routine as compared to a naive synthesis routine
\end{displayquote}

\begin{figure}[ht]
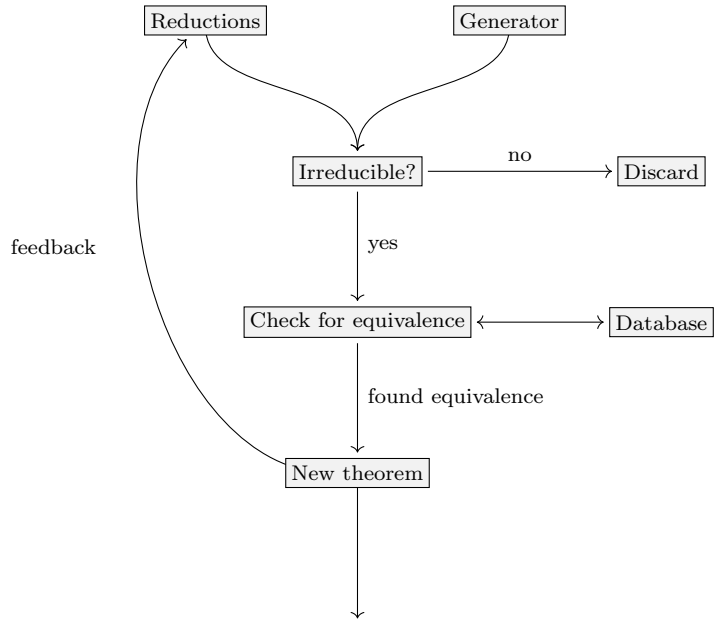

	\center
	\vc{\InputIfFileExists{./figures/wire/generative.tikz}{}{Missing file!}}
	\caption{The generative model, as used in QuantoCoSy.
		The database stores the matrix representation of every irreducible diagram generated.
		This is then used for comparison with each newly generated diagram to determine whether
		a new equivalence has been found. \label{figGenerative}}
\end{figure}

The starting point for this research was clear:
Use QuantoCoSy to generate theorems in ZX or ZW.
The big question remaining was `how should we measure success?'
For this the author took the goal
of designing a system that could
`feasibly discover a complete ruleset for ZX or ZW, given elementary domain knowledge'.
This goal was chosen before the complete rulesets for Clifford+T or Universal ZX were discovered,
but this goal, together with the question of implementation,
led to the research in this thesis.

One final note is that in our conjecture synthesis runs
we usually considered two diagrams
to be equivalent if they agreed up to non-zero scalar.
This is because scalars are comparatively easy to keep track of
using a single complex number rather than a diagram.

\section{Generating diagrams} \label{sexZXCoSy}

As research for this thesis started, the ZX-calculus (see \S\ref{secZX}) was
already gaining attention.
The book Ref.~\cite{PQP}
was finished, focussing almost exclusively on ZX,
and the quantum computing lectures at the University of Oxford
(attended by the author) were lectured using ZX.
ZX had already been shown to be complete for stabilizer quantum mechanics,
although not for the other fragments,
and the expectation and hope was that it would be shown to be complete
for Clifford+T and universal quantum computation soon (see Table~\ref{figZXStory}).
ZX, for our purposes, is generated by red and green spiders,
and these spiders are very relevant for CoSy.
We shall take the following as our initial `elementary domain knowledge':

\begin{itemize}
	\item Spiders have arbitrary arity
	\item Adjacent spiders of the same colour fuse
	\item If two edges join a pair of spiders, then those spiders either fuse by the spider law,
	      or the edges disconnect by the Hopf law
\end{itemize}

Combining these three points results in the realisation that ZX diagrams
can be reduced by the spider and Hopf laws to a labelled, bipartite graph.
Since any reducible diagram will be ignored by the model in Figure~\ref{figGenerative}
we can instead simply generate labelled bipartite graphs
(the label just contains the information on the colour and angle of the spider).
For this QuantoCoSy uses the algorithm given by Ref.~\cite{ColbournRead}.
ZW also contains two colours of spider,
but crucially we cannot use the spider laws of ZW (Equation~\eqref{eqnZWcutW}) to reduce all diagrams to bipartite graphs.
Likewise when the calculus ZH (Definition~\ref{defZH}) was introduced
it too had spider laws, similar to ZW's, that do not provide us with a reduction to bipartite graphs.
For this reason the first implementations focused on ZX.

\begin{remark} \label{remGeneratingPatterns}
	Thanks to the combinatoric model of pattern graphs
	it is possible to generate diagrams containing !-boxes.
	One of the contributions of this thesis
	is that we can now verify equations involving !-boxes
	without needing to directly reason with !-boxes, see \S\ref{secBBoxes}.
	This verification depends on a property called `separability',
	covered in \S\ref{secSeparability},
	which is a condition one could build into the pattern graph generator.
\end{remark}

\section{Generalising theorems} \label{secGeneralisingTheorems}

Figure~\ref{figZXCoSyFirst} contains
one of the first ZX rules generated by QuantoCoSy
running on Duvel.\footnote{The computing cluster made available for this research, with 24 cores at 2.9MHz.}
Note that at this point the feedback loop from Figure~\ref{figGenerative}
had not been implemented, and scalars were still being considered.
This theorem was, however, towards the upper end in terms of the size of diagram considered
for those runs.\footnote{
The results of these QuantoCoSy runs can be found as additional files
provided along with this thesis, although we do not rely upon them
for any of the results in this thesis.
Results are in folders labelled by a date and target hardware,
and there is a file `summary.txt' in each folder recording the notes for that run.}

\begin{figure}[ht]
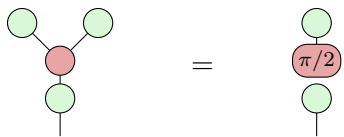

	\center
	\vc{\InputIfFileExists{./figures/ZX/CoSyExample_l.tikz}{}{Missing file!}}\qquad = \qquad\vc{\InputIfFileExists{./figures/ZX/CoSyExample_r.tikz}{}{Missing file!}}
	\caption{An early example of a theorem generated by Duvel.
		In this case the theorem is `-1001866680\_-1614245455.qrule'
		from run `18-04-04-10-55-19-duvel',
		chosen because it is first in the list when displayed by Quantomatic.\label{figZXCoSyFirst}}
\end{figure}

\begin{remark}[The generalisation step] \label{remGeneralisationStep}
The first contribution of this thesis is that of adding a \emph{generalisation step}
to the method of generative conjecture synthesis.
It is performed during the feedback step of Figure~\ref{figGenerative},
and during this step we take our newly synthesised result
and try to generalise it as much as possible
before creating new reductions from it.
This inner generalisation algorithm could itself be deductive
(we cover examples in \S\ref{secInferenceDeductive})
or produce informed guesses which are then verified (see \S\ref{secFindingSubmodules}).
In both cases the processes are novel to this thesis,
and the non-deductive process is only possible due to a novel
link between phase algebras and algebraic geometry.
\end{remark}

\begin{example}[Spotting a generalisation] \label{exaWiderpattern}the theorems in Figure~\ref{figCopy2AndCopy3} form part of a wider pattern,
illustrated by the use of !-boxes in Figure~\ref{figCopyBang}.
Likewise the two `$\pi$-commutations' exhibited in Figure~\ref{figPiCommExamples}
are both instances of the more general rule shown in Figure~\ref{figPiCommGeneral}.
\end{example}

\begin{figure}[p]
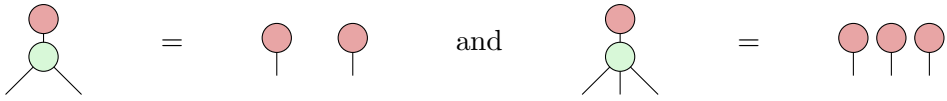

	\center
	\vc{\InputIfFileExists{./figures/ZX/Copy2_l.tikz}{}{Missing file!}}\qquad = \qquad\vc{\InputIfFileExists{./figures/ZX/Copy2_r.tikz}{}{Missing file!}} \qquad and \qquad
	\vc{\InputIfFileExists{./figures/ZX/Copy3_l.tikz}{}{Missing file!}}\qquad = \qquad\vc{\InputIfFileExists{./figures/ZX/Copy3_r.tikz}{}{Missing file!}}
	\caption{Two examples of (scalar-free) copying\label{figCopy2AndCopy3}}
\end{figure}

\begin{figure}[p]
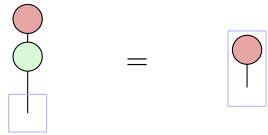

	\center
	\vc{\InputIfFileExists{./figures/ZX/CopyBang_l.tikz}{}{Missing file!}}\qquad = \qquad\vc{\InputIfFileExists{./figures/ZX/CopyBang_r.tikz}{}{Missing file!}}
	\caption{The more general version of (scalar-free) copying, using !-boxes\label{figCopyBang}}
\end{figure}

\begin{figure}[p]
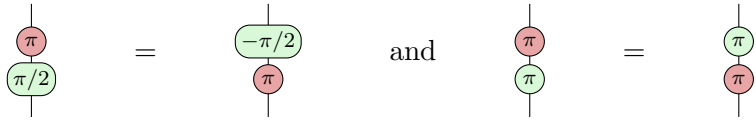

	\center
	\vc{\InputIfFileExists{./figures/ZX/PiCommExample1_l.tikz}{}{Missing file!}}\qquad = \qquad\vc{\InputIfFileExists{./figures/ZX/PiCommExample1_r.tikz}{}{Missing file!}} \qquad and \qquad
	\vc{\InputIfFileExists{./figures/ZX/PiCommExample2_l.tikz}{}{Missing file!}}\qquad = \qquad\vc{\InputIfFileExists{./figures/ZX/PiCommExample2_r.tikz}{}{Missing file!}}
	\caption{Two examples of (scalar-free) $\pi$-commutation\label{figPiCommExamples}}
\end{figure}

\begin{figure}[p]
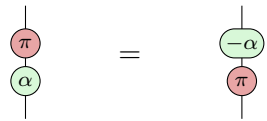

	\center
	\vc{\InputIfFileExists{./figures/ZX/PiCommGeneral_l.tikz}{}{Missing file!}}\qquad = \qquad\vc{\InputIfFileExists{./figures/ZX/PiCommGeneral_r.tikz}{}{Missing file!}}
	\caption{The more general version of (scalar-free) $\pi$-commutation\label{figPiCommGeneral}}
\end{figure}

This question of how to generalise theorems is not new.
The philosophical `Problem of Induction' of creating ideas from evidence dates back to Hume in 1739 
\cite{Hume}.
Example~\ref{exaAmmonQuantifiers} gives a more recent example from the field of conjecture synthesis,
where the task of term generalisation is handled by attempting to replace existential quantifiers
with universal quantifiers,
and if that proves to be too general then successively greater restrictions are placed on the universal quantifier.
For graphical calculi we have two ways to generalise.
We can generalise these graph patterns with !-boxes (as in Figure~\ref{figCopyBang})
and generalise the phase patterns with `phase variables' (as in Figure~\ref{figPiCommGeneral}).

\begin{example}[Term generalisation using quantifiers {\cite[p416]{Ammon1993}}]  \label{exaAmmonQuantifiers}
	The progression of first attempting to generalise the variable $n$
	to universal quantification, and then to restricted universal quantification
	in a conjecture about natural numbers and primes.
	\begin{align}
		\exists n,p,q(n \in N \land p \in P \land q \in P \land n                              & = p+q)    \\
		\forall n(n \in N \to \exists p,q(p\in P \land q \in P \land n                         & = p+q))   \\
		\forall (n(n\in N \land \text{is-even}(n) \to \exists p,q(p\in P \land q \in P \land n & = p + q))
	\end{align}
\end{example}

The question of how to introduce phase variables (and, to a lesser degree, !-boxes)
is covered in \S\ref{chapConjectureInference},
and relies on a novel link between phase algebras and algebraic geometry.
We will discuss how to then verify these conjectures in \S\ref{chapConjectureVerification}.
One of the interesting results of combining these two ideas arises in \S\ref{secPhaseVariablesOverQubits},
where we note that some phases can, almost automatically, be replaced with more general phase variables.
In terms of Example~\ref{exaAmmonQuantifiers}
this would amount to saying that certain existential quantifiers could automatically be
replaced with universal quantifiers.
The operational mantra for this is that there is a bound on the algebraic
complexity for the phases in an equation;
if a phase appears less often than the degree of its minimal polynomial
then it can be replaced by a phase variable (Remark~\ref{remNonQImpliesPhaseVariable}).

The framework of \S\ref{chapConjectureInference} is, however, limited by the choice of phase algebra.
For ZX the framework
can infer all the rules in the complete ruleset of Stabilizer ZX from Ref.~\cite{BackensStabilizerComplete}
and Clifford+T ZX in Ref.~\cite{SimonCompleteness},
but it cannot reason about the rules needed for completeness of Universal ZX:
The original completeness result contained the (AD) rule
\cite{UniversalComplete} which requires the user to compute the magnitude and argument of $\lambda_1 e^{i \beta} + \lambda_2 e^{i\alpha}$.
With considerable work that ruleset was reduced to that of Ref.~\cite{VilmartZX},
which requires the user to compute the following:

\begin{displayquote}[Figure~2, A Near-Minimal Axiomatisation of ZX-Calculus
		for Pure Qubit Quantum Mechanics \cite{VilmartZX}]
	In rule (EU'), $\beta_1$, $\beta_2$, $\beta_3$ and $\gamma$
	can be determined as follows: $x^+ := \frac{\alpha_1 + \alpha_2}{2}$,
	$x^- := x^-\alpha_2$, $z := - \sin (x^+) + i \cos(x^-)$
	and $z' := \cos(x^+) - i \sin (x^-)$,
	then $\beta_1 = \arg z + \arg z'$, $\beta_2 = 2 \arg(i + \frac{\abs{z}}{\abs{z'}})$,
	$\beta_3 = \arg z - \arg z'$, $\gamma = x^+ - \arg(z) + \frac{\pi - \beta_2}{2}$
	where by convention $\arg(0) := 0$ and $z' = 0 \implies \beta_2 = 0$.
\end{displayquote}

Since we are always striving to find more refined, or even just different,
complete rulesets for these calculi it is possible that
there will exist some complete ruleset for ZX where every rule
\emph{does} fit into the framework presented in \S\ref{chapConjectureInference}.
The authors of Ref.~\cite{BeyondCliffordT}, however, are pessimistic
about this in their comment: `Notice
however that \cite{SdW14} and \cite{cyclo} are two different kinds of evidence that such a finite
complete axiomatisation [using linear rules] may not exist'.
(See Remark~\ref{remConverseToLemLinearRules}
for how the linear rules of Ref.~\cite{BeyondCliffordT} coincide with the framework for ZX described in \S\ref{chapConjectureInference}.)

Constructing a framework that allows the direct inference of the (EU') rule's external computations,
also called \emph{side conditions},
seems infeasible to the author.
This in no small part prompted the wider research of this thesis.
Fortunately another avenue opened up,
for despite the apparent complexity of the (EU') rule of Ref.~\cite{VilmartZX}
it actually hints at something far simpler:
It is analogous to the Euler angle decomposition of rotations in 3D space
(this is the origin of the name (EU)).
In response to this the author constructed a new, complete,
graphical calculus called ZQ (Definition~\ref{defZQ})
which took all rotations to be primitive, and not just the Z and X rotations of ZX.
In contrast to Universal ZX it \emph{is} possible to express a complete ruleset of
(the equally universal and complete calculus) ZQ using the linear framework of \S\ref{chapConjectureInference}.

\section{Generalising ring calculi}

By this point in the narrative the paper introducing ZH (Ref.~\cite{ZH})
had been published. The author was then presented with the question
of which of ZH and ZW, both being ring-based, spider calculi, would be better to investigate.
Both ZH and ZW are built from two types of spider,
with the rules of ZH involving a third, derived, type of spider.
Neither calculus relies on side conditions for its complete ruleset.
ZW, unlike ZH, had the advantage of having been used directly for the completeness results
for Clifford+T and Universal ZX, but the completeness results of both ZH and ZW were shown via normal forms.

It struck the author that the similarities between ZH and ZW
were more interesting than their differences,
and so at this point the research diverted from synthesising conjectures
in these ring-based languages to trying to find a unifying system for discussing them.
From this came the $\ring_R$ graphical calculus (\S\ref{secRingR});
a calculus designed to be as generic a phase ring calculus as possible.
By construction it is easy to translate $\ring_R$ diagrams into ZW or ZH diagrams,
but by considering the category of phase ring graphical calculi
we can actually say that $\ring_R$ is essentially initial for graphical calculi
with phase ring $R$ (Corollary~\ref{corGivenIsoRingRInitial}).
$\ring_\bbC$ has even been described as `solving the field'
for phase ring qubit graphical calculi\footnote{Aleks Kissinger, private communication}.

The investigation of $\ring_R$ and its categorical context (\S\ref{chapPhaseRingCalculi})
gave rise to the idea of phase homomorphisms
and phase homomorphism pairs;
structures that acted directly on the phase of the diagrams
while also commuting with the interpretations.
These phase homomorphism pairs
also arise as symmetries in the framework introduced in \S\ref{chapConjectureInference},
and are instrumental in the finite verification result
of \S\ref{secPhaseVariablesOverQubits}.
At the end of \S\ref{chapPhaseRingCalculi}
we look at the corresponding structures for phase group calculi,
and then generalise further to $\Sigma$-calculi.

This ends our introductory chapters,
and so we shall begin by introducing the calculus $\ring$.  \thispagestyle{empty} 
\chapter{The Graphical Calculus \textsc{RING}} \label{chapRingR}
\thispagestyle{plain}
\noindent\emph{In this chapter:}
\begin{itemize}
	\item[\chapterbullet] We introduce a complete, universal phase ring calculus, $\ring_R$
	\item[\chapterbullet] We extend $\ZHC$ to the more general $\ZHR$
	\item[\chapterbullet] We exhibit translations between \ZWR, \ZHR, and $\ring_R$
	\item[\chapterbullet] The calculus $\ring_R$ introduced in this chapter is both a motivation and an example used throughout the rest of this thesis
\end{itemize}

The results of this chapter were born from the idea of a generic phase ring calculus for conjecture synthesis.
The ZW calculus `contains rational arithmetic' in the words of Ref.~\cite{ZWArithmetic},
meaning that one can use the W and GHZ states of ZW to perform addition and multiplication.
This is not to say that ZW is the only way of encoding a ring-like structure in a graphical calculus,
or even that there are not other ways of presenting ring-like structures inside ZW.
Just because ZW contains rational arithmetic does not mean in embodies rational arithmetic.

ZW is not the only calculus with a phase ring.
We were introduced to ZH$_\bbC$ in Definition~\ref{defZH},
which is universal and complete for \Qubit\ \cite{ZH},
but there is also the newly introduced
Algebraic ZX\footnote{The author would have suggested a different name, such as `Ring ZX'.}
\cite{AZX},
which coerces the phase group of ZX into a phase ring.

Our aim for this chapter is to introduce a phase ring calculus
that lays bare the ring structure of the phase algebra.
This will allow us to explore the uses and limitations of such a phase algebra,
especially over qubits, in later chapters.

We introduce our own calculus $\ring_R$ in \S\ref{secRingR},
which is designed with the ring operations first and foremost,
and to work for any commutative ring $R$.
We exhibit a complete ruleset for $\ring_R$ in \S\ref{secRingRRules}, found via a translation with $\ZWR$.
We then extend ZH to work for arbitrary
commutative rings with $\half{}$ in \S\ref{secZHR},
and also a translation between $\ring_R$ and
this more general version of ZH.
Although the translation with $\ZHR$ is covered in far briefer terms than the translation with \ZWR\
we use it to produce a different, complete, ruleset for $\ring_R$ in \S\ref{secRingRZHRules}
(provided the ring $R$ contains $\half{}$).

As mentioned earlier the calculus $\ring$ is used as both motivation
and example throughout the rest of this thesis,
but the only knowledge required for those future chapters
is familiarity with the definition and main results of \S\ref{secRingR},
and awareness of the rules in Figure~\ref{figRingRRules1}.

\section{The definition of RING} \label{secRingR}

The graphical calculus $\ring_R$ is created to be
`as generic a phase ring graphical calculus over states as possible'.
With this in mind our generators are precisely the ring operations of some ring $R$,
a set of states representing the elements of $R$,
and the ambient wires, swaps, cups and caps of a compact closed graphical calculus.
There are complete rulesets available via translation from both \ZWR\ (\S\ref{secRingRRules})
and also the extended ZH calculus, \ZHR,
which we introduce in \S\ref{secCompleteTranslationZHQR}.

\begin{definition}[$\ring_R$]  \label{defRingR}
	The graphical calculus $\ring_R$ for a commutative ring $R$
	is a (compact closed)
	PROP with generators and interpretation given by Figure~\ref{figGeneratorsRingR},
	with derived generators given in Figure~\ref{figDerivedGeneratorsRingR},
	and rules given in Figures~\ref{figRingRRules1} and \ref{figRingRRules2}.
\end{definition}

Let's unpack this definition a little.
The generators in Figure~\ref{figGeneratorsRingR}
form the core of the calculus.
Every diagram in $\ring$ is built from these three components.
Since we allow wire bending (i.e. this is a compact closed calculus)
these generators can be plugged together
with no regard for concepts of `input' or `output'.
The multiplication gate is the familiar (phase-free) Z spider from ZX, ZW, and ZH,
and so retains all of its spider-like qualities
(commutativity, associativity, etc.).
Where $\ring$ differs from some of the other calculi
is that the addition gate has a distinguished wire,
drawn at the top in Figure~\ref{figGeneratorsRingR}.
This means that addition in $\ring$ forms a `one-sided' spider:
wires connected to the rounded side of the gate act like wires
connected to a spider,
but the distinguished wire (attached to the flat, bold side)
does not.

\begin{figure}[p]
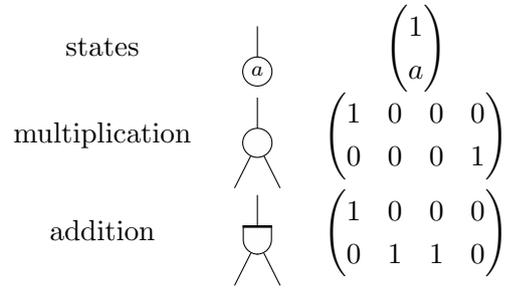

	\centering
	\begin{tabular}{ccc}
		states         & $\state{white}{a}$ & $\begin{pmatrix}
				1 \\ a
			\end{pmatrix}$ \\
		multiplication & $\binary[white]{}$ & $\begin{pmatrix}
				1 & 0 & 0 & 0 \\
				0 & 0 & 0 & 1
			\end{pmatrix}$ \\
		addition       & $\binary[poly]{}$  & $\begin{pmatrix}
				1 & 0 & 0 & 0 \\
				0 & 1 & 1 & 0
			\end{pmatrix}$
	\end{tabular}
	\caption{
		The generators of $\ring_R$
		\label{figGeneratorsRingR}}
\end{figure}

Readers familiar with monoids should note that
the multiplication gate forms a monoid with the state 1 as a unit,
and the addition gate forms a monoid with the state 0 as a unit.
Knowing this allows us to create the derived generators
of Figure~\ref{figDerivedGeneratorsRingR}.
Derived generators are built from our core generators,
and represent motifs that appear often enough in our diagrams
to warrant creating these shorthand notations.
The Hadamard Gate in ZX is a derived generator:
It is derivable from Z and X spiders,
and so is technically redundant,
but is so common and useful a gate that
it is afforded the same status as a generator.
In $\ring$ the multiplication gate
can be generalised as the familiar, phased Z spider
(as in ZW or ZX),
and the addition gate as a phased, one-side spider.
In either case the label on the derived generator is an element of
the ring $S$.
To summarise:

\begin{remark}[Topology matters a bit] \label{remTopologyMattersABit}
	The white multiplication node is indeed a spider, in the sense of Ref.~\cite{Coecke08},
	but addition has a distinguished output
	indicated by the bold, flat edge.
\end{remark}

The two other derived generators in 
Figure~\ref{figDerivedGeneratorsRingR}
are the NOT gate and Hadamard gates respectively.
These derived generators use the `flipped'
version of the addition gate,
drawn upside down.
This is simply a space-saving convention:
The distinguished wire comes from below,
and so the bold, flat edge of the gate is drawn at the bottom.
To calculate the interpretation we would use bent wires
to allow us to draw the addition gate the same way up as it is in Figure~\ref{figGeneratorsRingR}.

\begin{figure}[p]
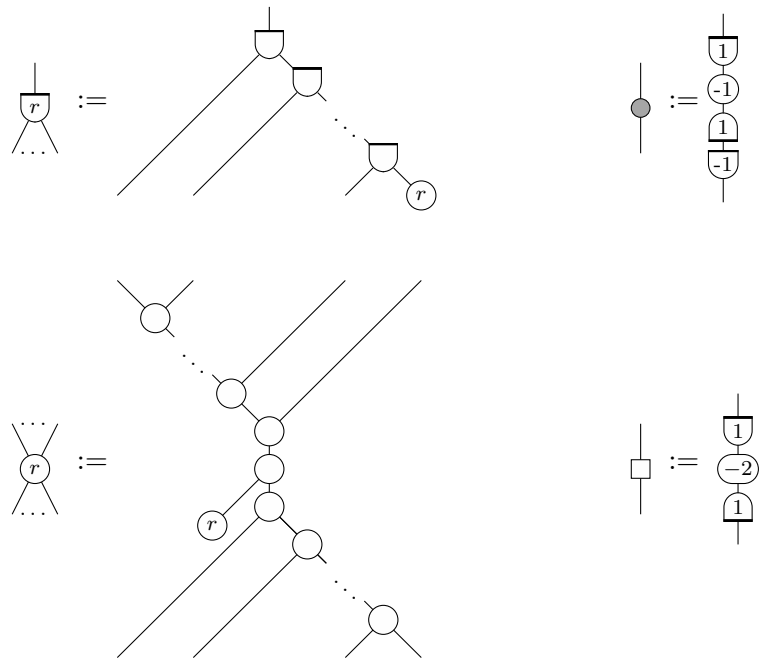

	\centering
	\begin{align*}
		\nary[poly]{r}    & := \vc{\InputIfFileExists{./figures/RingR/plus_spider.tikz}{}{Missing file!}} & \unary[smallblack]{} & := \vc{\InputIfFileExists{./figures/QR/X.tikz}{}{Missing file!}}
		\\[2em]
		\spider{white}{r} & := \vc{\InputIfFileExists{./figures/RingR/times_r.tikz}{}{Missing file!}}        & \unary[h2]{}            & := \QRHadamard \label{eqnRingXGate}
	\end{align*}
	\caption{Derived generators of $\ring_R$:
		The one-sided spider for addition, the NOT gate,
		the multiplication spider and the Hadamard gate \label{figDerivedGeneratorsRingR}.
	}
\end{figure}

Now that we have defined our language we shall
explore its important properties:
Completeness, Soundness and Universality.

\begin{theorem}[$\ring$ is complete] \label{thmRingRCompleteZW}
	The language $\ring_R$ is complete for any commutative ring $R$.
\end{theorem}

\begin{proof} \label{prfThmRingRCompleteZW}
	We include in Figures~\ref{figZWRingRTranslation}
	and \ref{figRingRZWTranslation}
	a translation between the generators of $\ring_R$ and $\ZW_R$
	that preserves the interpretation.
	In \S\ref{secTranslatedRulesZWRingR}
	we derive, in $\ring_R$, each equality $F_{\ring}L=F_{\ring}R$ 
	for each rule $L=R$ in \cite{AmarThesis, ZW}.
	In Lemma~\ref{lemRingRZWTranslatedGenerators} we derive the
	equation $g = F_{\ring} F_{\ZW} (g)$ for each generator $g$ of $\ring_R$.
	Therefore for each sound diagrammatic equation $D_1 = D_2$ we will be able to derive:
	\begin{align}
		\ZW_R \entails                                       & F_{\ZW} D_1 = F_{\ZW} D_2                                                        \\
		\ring_R \entails & F_{\ring} F_{\ZW} D_1 = F_{\ring} F_{\ZW} D_2                                    \\
		\ring_R \entails     & D_1 = F_{\ring} F_{\ZW} D_1 \qquad \text{and} \qquad D_2 = F_{\ring} F_{\ZW} D_2 \\
		\therefore \ring_R \entails                          & D_1 = F_{\ring} F_{\ZW} D_1 = F_{\ring} F_{\ZW} D_2 = D_2
	\end{align}
\end{proof}

\begin{theorem}[$\ring$ is sound] \label{thmRingRRulesSound}
	The rules presented for $\ring$ are sound
\end{theorem}

\begin{proof} \label{prfThmRingRRulesSound}
	This is can be checked directly,
	made simpler when combined with the knowledge that the translation
	preserves the interpretation.
\end{proof}

\begin{theorem}[$\ring$ is universal] \label{thmRingRUniversal}
	$\ring$ is universal over $\bit{R}$.
\end{theorem}

\begin{proof} \label{prfThmRingRUniversal}
	The same interpretation-preserving translation from ZW used in Theorem~\ref{thmRingRRulesSound}
	shows that any matrix expressible in ZW$_R$ is expressible in $\ring_R$,
	and since ZW$_R$ is universal for $\bit{R}$ so is $\ring_R$.
\end{proof}

\begin{remark}[Comparison to Quantum Arithemtic] \label{remRingRAndAddition}
Quantum Arithmetic refers to performing
elementary arithmetic operations to numbers
stored as registers of qubits.
For example one can store the number six,
110 in binary, using three qubits represented in the computational basis as:
\begin{align}
\ket{1} \tensor \ket{1} \tensor \ket{0}
\end{align}
A demanding part of Shor's algorithm \cite{VBE, Shor}
requires being able to perform the following elementary arithmetic operation:
\begin{align}
U_{a,N} : \ket{x} \tensor \ket {0} &\mapsto \ket {x} \tensor \ket { a^x \mod N}
\end{align}
While $\ring$ is constructed from the two ring operations
and compact closure,
its native way of handling numbers is fundamentally different
to that of quantum arithmetic.
For example the state bearing the label 6 
represents a single qubit, expressed
as follows in the computational basis:
\begin{align}
\interpret{\state{white}{6}} &= \ket{0} + 6 \ket{1}
\end{align}
The state representing the bit-string 110 is more complicated:
\begin{align}
\interpret{\vc{\InputIfFileExists{./figures/RingR/ket1.tikz}{}{Missing file!}}\quad \vc{\InputIfFileExists{./figures/RingR/ket1.tikz}{}{Missing file!}}\quad \state{white}{0}} &= \ket{1} \tensor \ket{1} \tensor \ket{0}
\end{align}
Since $\ring$ is universal it is still capable
of representing (and reasoning about) quantum arithmetic operations,
but the use of complex numbers as labels in $\ring$
should not be confused with the presentation
of numbers as bitstrings in bra-ket notation.
\end{remark}

\begin{remark}[The power of Compact Closure] \label{remRingRAndCC}
Compact Closure was a fundamental property highlighted
in Ref.~\cite{CQM}, the seminal work on Categorical Quantum Mechanics.
That work also highlighted the importance of
algebraic considerations in the design of process systems:
\begin{displayquote}[Categorical Quantum Mechanics, \cite{CQM}]
In particular, we have in mind the hard-learned lessons from Computer Science of
the importance of compositionality, types, abstraction, and the use of tools from
algebra and logic in the design and analysis of complex informatic processes.
\end{displayquote}
This new calculus demonstrates that it is possible
to construct a sound, universal and complete diagrammatic calculus
for quantum computing using just a single algebraic structure and compact closure.
Indeed $\ring$ could be seen as further evidence
of the power of compact closed structures,
since without cups and caps the calculus $\ring$
is less expressive than open term-trees of rings.
This idea is explored further in \S\ref{secPhaseRingGeneric}.
\end{remark}

\section{The rules of RING}
\label{secRingRRules}

We break the rules of $\ring$ into two parts.
The first part are the rules relating to
addition and multiplication;
in particular multiplication forms a spider
(the familiar Z spider).
The addition gate does not form a spider,
although since it is associative and commutative on inputs
we use the rounded side of the gate to indicate the inputs,
and use the straight, bold edge to indicate the single output.
These are found in Figure~\ref{figRingRRules1}.
We label these rules as simply $+$ or $\times$,
since these are the effects the gates have on states.

\begin{figure}[H]
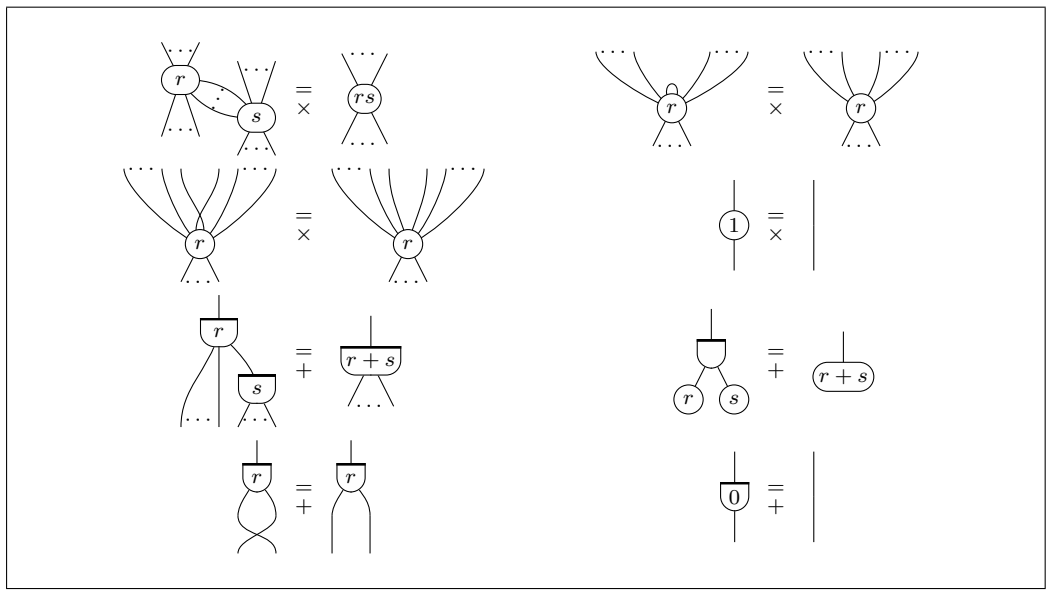

	\centering\begin{tabular}{|c|}
		\hline
		\begin{minipage}{0.9\textwidth}{
				\begin{align*}
					\vc{\InputIfFileExists{./figures/QR/ZS.l.tikz}{}{Missing file!}}       & \by{\times} \spider{white}{rs}      & \vc{\InputIfFileExists{./figures/RingR/tr_z.l.tikz}{}{Missing file!}} & \by{\times} \vc{\InputIfFileExists{./figures/RingR/tr_z.r.tikz}{}{Missing file!}} \\
					\vc{\InputIfFileExists{./figures/RingR/sym_z.l.tikz}{}{Missing file!}} & \by{\times} \vc{\InputIfFileExists{./figures/RingR/sym_z.r.tikz}{}{Missing file!}} & \unary[white]{1}       & \by{\times} \unary[none]{}         \\
					\vc{\InputIfFileExists{./figures/RingR/PS.l.tikz}{}{Missing file!}}    & \by{+} \nary[poly]{r+s}             & \bap{poly}{}{r}{s}     & \by{+} \state{white}{r+s}          \\
					\vc{\InputIfFileExists{./figures/RingR/sym_p.l.tikz}{}{Missing file!}} & \by{+} \vc{\InputIfFileExists{./figures/RingR/sym_p.r.tikz}{}{Missing file!}}      & \unary[poly]{0}        & \by{+} \unary[none]{}
				\end{align*}}
		\end{minipage} \\
		\hline
	\end{tabular}
	\caption{
		The rules of $\ring$\label{figRingRRules1} concerning addition and multiplication.
	}
\end{figure}

In Figure~\ref{figRingRRules2}
we show the remaining rules of $\ring$.
These can be considered as rules containing the NOT gate
(A) and (N),
the distributive law (D),
scalar identity (I),
the Hopf law (Hopf),
a law for forming a loop with addition and multiplication (L),
two rules for the interactions between addition and Controlled Z (CZA1) and (CZA2),
and three bialgebra laws (B1), (B2) and (CP).

\begin{figure}[H]
	\centering\begin{tabular}{|c|}
		\hline
		\begin{minipage}{0.9\textwidth}{
				\begin{align*}
					\vc{\InputIfFileExists{./figures/RingR/A.l.tikz}{}{Missing file!}}     & \by{A} \vc{\InputIfFileExists{./figures/RingR/A.r.tikz}{}{Missing file!}}      & \vc{\InputIfFileExists{./figures/RingR/NOT.l.tikz}{}{Missing file!}}     & \by{N} \vc{\InputIfFileExists{./figures/RingR/NOT.r.tikz}{}{Missing file!}}      \\
					\vc{\InputIfFileExists{./figures/RingR/D.l.tikz}{}{Missing file!}}     & \by{D} \vc{\InputIfFileExists{./figures/RingR/D.r.tikz}{}{Missing file!}}      & \scalar{white}{0}         & \by{I} \vc{\InputIfFileExists{./figures/wire/empty.tikz}{}{Missing file!}}       \\
					\vc{\InputIfFileExists{./figures/RingR/Hopf.l.tikz}{}{Missing file!}}  & \by{Hopf}\vc{\InputIfFileExists{./figures/RingR/Hopf.r.tikz}{}{Missing file!}} & \vc{\InputIfFileExists{./figures/RingR/L.l.tikz}{}{Missing file!}}       & \by{L} \vc{\InputIfFileExists{./figures/RingR/L.r.tikz}{}{Missing file!}}        \\
					\vc{\InputIfFileExists{./figures/RingR/HP1.l.tikz}{}{Missing file!}}   & \by{CZA1} \vc{\InputIfFileExists{./figures/RingR/HP1.r.tikz}{}{Missing file!}} & \vc{\InputIfFileExists{./figures/RingR/HP2.l.tikz}{}{Missing file!}}     & \by{CZA2} \vc{\InputIfFileExists{./figures/RingR/HP2.r.tikz}{}{Missing file!}}   \\
					\vc{\InputIfFileExists{./figures/RingR/ZW_5a_l.tikz}{}{Missing file!}} & \by{B1} \vc{\InputIfFileExists{./figures/RingR/ZW_5a_r.tikz}{}{Missing file!}} & \vc{\InputIfFileExists{./figures/RingR/ZW_bazw_l.tikz}{}{Missing file!}} & \by{B2} \vc{\InputIfFileExists{./figures/RingR/ZW_bazw_r.tikz}{}{Missing file!}} \\
					\vc{\InputIfFileExists{./figures/RingR/CP.l.tikz}{}{Missing file!}} & \by{CP} \vc{\InputIfFileExists{./figures/RingR/CP.r.tikz}{}{Missing file!}}
				\end{align*}}
		\end{minipage} \\
		\hline
	\end{tabular}
	\caption{
		The remaining rules of $\ring$.\label{figRingRRules2}
	}
\end{figure}

The first thing we must do with our rules is justify the notational symmetry
used in defining the NOT derived generator.

\begin{lemma} \label{lemRingRNotSymmetrical}
	The derived generator NOT (Figure~\ref{figDerivedGeneratorsRingR}) is justified in its graphical symmetry.
\end{lemma}

\begin{proof} \label{prfLemRingRNotSymmetrical}
	\begin{align}
		\vc{\begin{tikzpicture}
	\begin{pgfonlayer}{nodelayer}
		\node [style=smallblack] (1) at (0, 0) {};
		\node [style=none] (3) at (-0.25, 0.5) {};
		\node [style=none] (4) at (0.25, 0.5) {};
	\end{pgfonlayer}
	\begin{pgfonlayer}{edgelayer}
		\draw (3.center) to (1.center);
		\draw (1.center) to (4.center);
	\end{pgfonlayer}
\end{tikzpicture}
}  \by{N} \vc{\InputIfFileExists{./figures/RingR/NOT.r.tikz}{}{Missing file!}} \by{+} \vc{\InputIfFileExists{./figures/RingR/notsym_1.tikz}{}{Missing file!}} \by{N} \vc{\InputIfFileExists{./figures/RingR/notsym_2.tikz}{}{Missing file!}}
	\end{align}
\end{proof}

\section{Completeness via ZW}

\label{secTranslatedRulesZWRingR}

We construct a complete ruleset for $\ring_R$ by exhibiting a translation
to and from ZW$_R$.
This translation acts as the identity on bare wires, swaps, cups and caps.
This means that the diagram components in \eqref{eqnWires}
are considered to be the same in both calculi:

\begin{align} \label{eqnWires}
	  & \set{\dcup, \dcap, \vc{\InputIfFileExists{./figures/wire/empty.tikz}{}{Missing file!}}, \vc{\InputIfFileExists{./figures/wire/swap.tikz}{}{Missing file!}}, \;\unary[none]{}\;}
\end{align} \label{secZWRingRTranslation}

The translation is given in Figures~\ref{figZWRingRTranslation} and \ref{figRingRZWTranslation}.
We will now apply the translation to each of the rules of ZW
then derive these translated rules using the rules given in \S\ref{secRingRRules}.
In order to ease this process we shall first
show that three commonly used components of ZW
are translated into simple $\ring$ counterparts,
namely the NOT gate, plus gate and negation.
The visual similarities between the two languages
can make it hard to tell at a glance which language is being considered.
In this section the notation $\mapsto$ indicates
the translation of a ZW diagram or equation (on the left)
to a $\ring$ diagram or equation (on the right).

\begin{figure}
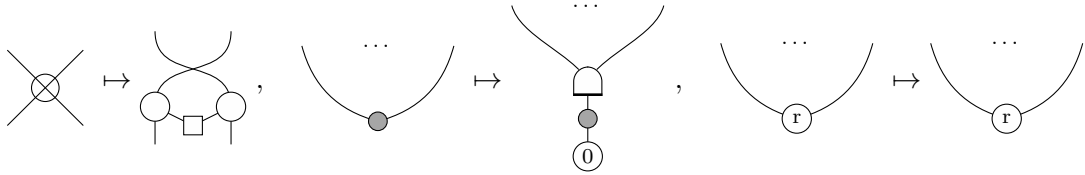

	\begin{align*}
		\vc{\InputIfFileExists{./figures/ZW/crossing.tikz}{}{Missing file!}} & \mapsto \vc{\InputIfFileExists{./figures/QR/ZW_crossing.tikz}{}{Missing file!}},\  &
		\vc{\InputIfFileExists{./figures/ZW/wn.tikz}{}{Missing file!}}       & \mapsto \vc{\InputIfFileExists{./figures/QR/ZW_wn.tikz}{}{Missing file!}},\        &
		\vc{\InputIfFileExists{./figures/ZW/zn.tikz}{}{Missing file!}}       & \mapsto \vc{\InputIfFileExists{./figures/QR/ZW_zn.tikz}{}{Missing file!}}
	\end{align*}
	\caption{\label{figZWRingRTranslation}
		Translation from ZW$_R$ to $\ring_R$.
	}
\end{figure}

\begin{figure}
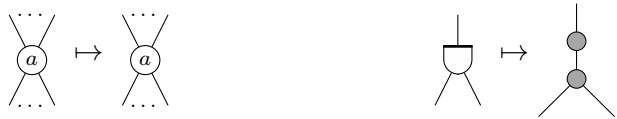

	\begin{align*}
		\spider{white}{a} & \mapsto \spider{white}{a} &
		\binary[poly]{}   & \mapsto\vc{\InputIfFileExists{./figures/ZW/plus.tikz}{}{Missing file!}}
	\end{align*}
	\caption{\label{figRingRZWTranslation}
		Translation from $\ring_R$ to \ZWR
	}
\end{figure}

\begin{lemma} \label{lemRingREntails2wIs2w}
	The ruleset $\ring$ entails that the translation of the binary $w$ spider is the NOT gate
	with an input wire bent upwards:
	\begin{align}
		\vc{\begin{tikzpicture}
	\begin{pgfonlayer}{nodelayer}
		\node [style=smallblack] (1) at (0, 0) {};
		\node [style=none] (3) at (-0.25, 0.5) {};
		\node [style=none] (4) at (0.25, 0.5) {};
	\end{pgfonlayer}
	\begin{pgfonlayer}{edgelayer}
		\draw (3.center) to (1.center);
		\draw (1.center) to (4.center);
	\end{pgfonlayer}
\end{tikzpicture}
} \mapsto \vc{\InputIfFileExists{./figures/RingR/NOT.r.tikz}{}{Missing file!}} = \vc{\begin{tikzpicture}
	\begin{pgfonlayer}{nodelayer}
		\node [style=smallblack] (1) at (0, 0) {};
		\node [style=none] (3) at (-0.25, 0.5) {};
		\node [style=none] (4) at (0.25, 0.5) {};
	\end{pgfonlayer}
	\begin{pgfonlayer}{edgelayer}
		\draw (3.center) to (1.center);
		\draw (1.center) to (4.center);
	\end{pgfonlayer}
\end{tikzpicture}
}
	\end{align}
\end{lemma}

\begin{proof} \label{prfLemRingREntails2wIs2w}
	This is just the rule $N$
\end{proof}

From here on we will use Lemma~\ref{lemRingREntails2wIs2w}
to translate the binary $w$ spider.
This NOT gate is involutive:

\begin{lemma} \label{lemRingRDoubleNOT}
	The ruleset $\ring$ entails the translated rule $inv$ from ZW:
	\begin{align}
		\vc{\begin{tikzpicture}
	\begin{pgfonlayer}{nodelayer}
		\node [style=none] (1) at (0, -2.5) {};
		\node [style=none] (4) at (0, -4) {};
		\node [style=smallblack] (5) at (0, -3) {};
		\node [style=smallblack] (6) at (0, -3.5) {};
	\end{pgfonlayer}
	\begin{pgfonlayer}{edgelayer}
		\draw (1.center) to (4.center);
	\end{pgfonlayer}
\end{tikzpicture}
} =\ \vc{\begin{tikzpicture}
	\begin{pgfonlayer}{nodelayer}
		\node [style=none] (1) at (0, -3) {};
		\node [style=none] (4) at (0, -4) {};
	\end{pgfonlayer}
	\begin{pgfonlayer}{edgelayer}
		\draw (1.center) to (4.center);
	\end{pgfonlayer}
\end{tikzpicture}
}
		\mapsto
		\vc{\begin{tikzpicture}
	\begin{pgfonlayer}{nodelayer}
		\node [style=none] (1) at (0, -2.5) {};
		\node [style=none] (4) at (0, -4) {};
		\node [style=smallblack] (5) at (0, -3) {};
		\node [style=smallblack] (6) at (0, -3.5) {};
	\end{pgfonlayer}
	\begin{pgfonlayer}{edgelayer}
		\draw (1.center) to (4.center);
	\end{pgfonlayer}
\end{tikzpicture}
} =\ \vc{\begin{tikzpicture}
	\begin{pgfonlayer}{nodelayer}
		\node [style=none] (1) at (0, -3) {};
		\node [style=none] (4) at (0, -4) {};
	\end{pgfonlayer}
	\begin{pgfonlayer}{edgelayer}
		\draw (1.center) to (4.center);
	\end{pgfonlayer}
\end{tikzpicture}
}
	\end{align}
\end{lemma}

\begin{proof} \label{prfLemRingRDoubleNOT}
	\begin{align}
		\vc{\InputIfFileExists{./figures/RingR/nn_1.tikz}{}{Missing file!}} \by{+} \vc{\InputIfFileExists{./figures/RingR/nn_2.tikz}{}{Missing file!}} \by{D} \vc{\InputIfFileExists{./figures/RingR/nn_3.tikz}{}{Missing file!}} \by{+,\times} \vc{\InputIfFileExists{./figures/RingR/nn_4.tikz}{}{Missing file!}} \by{+} \vc{\begin{tikzpicture}
	\begin{pgfonlayer}{nodelayer}
		\node [style=none] (1) at (0, -3) {};
		\node [style=none] (4) at (0, -4) {};
	\end{pgfonlayer}
	\begin{pgfonlayer}{edgelayer}
		\draw (1.center) to (4.center);
	\end{pgfonlayer}
\end{tikzpicture}
}
	\end{align}
\end{proof}

Having shown that the NOT gate is involutive
we now use this lemma to show that the commonly appearing combination
of a binary $w$ spider following a larger $w$ spider is in fact the plus gate:

\begin{lemma} \label{lemRingRTranslatedPlus}
	The translation of the commonly appearing ZW diagram below is equal to the plus gate in $\ring$:
	\begin{align}
		\vc{\InputIfFileExists{./figures/ZW/plusdots.tikz}{}{Missing file!}} \mapsto \vc{\InputIfFileExists{./figures/RingR/ZW_plusdots.tikz}{}{Missing file!}} = \nary[poly]{}
	\end{align}
\end{lemma}

\begin{proof} \label{prfLemRingRTranslatedPlus}
	\begin{align}
		\vc{\InputIfFileExists{./figures/RingR/ZW_plusdots.tikz}{}{Missing file!}} \by{+} \vc{\InputIfFileExists{./figures/RingR/plus_1.tikz}{}{Missing file!}}  \by{N} \vc{\InputIfFileExists{./figures/RingR/plus_2.tikz}{}{Missing file!}} \by{\ref{lemRingRDoubleNOT}} \nary[poly]{}
	\end{align}
\end{proof}

From here on we will use Lemma~\ref{lemRingRTranslatedPlus}
to translate that composition of the $1\to 1$ and $1 \to 2$ $w$ spiders.
To show that the traced crossing is equal to negation
we will first need the following lemmas:

\begin{lemma} \label{lemRingREntailsPNOT}
	The following equations are derivable in $\ring$:
	\begin{align}
		\vc{\InputIfFileExists{./figures/RingR/pnot.l.tikz}{}{Missing file!}} = \vc{\InputIfFileExists{./figures/RingR/pnot.r.tikz}{}{Missing file!}} \qquad
		\vc{\begin{tikzpicture}
	\begin{pgfonlayer}{nodelayer}
		\node [style=none] (13) at (-1.5, 0) {};
		\node [style=none] (18) at (-1.5, -1) {};
		\node [style=polynomial] (19) at (-1.5, -0.5) {$a$};
	\end{pgfonlayer}
	\begin{pgfonlayer}{edgelayer}
		\draw (19.center) to (18.center);
		\draw (13.center) to (19.center);
	\end{pgfonlayer}
\end{tikzpicture}
} = \vc{\InputIfFileExists{./figures/RingR/pbig.r.tikz}{}{Missing file!}}
	\end{align}
\end{lemma}
\begin{proof} \label{prfLemRingREntailsPNOT}
	\begin{align}
		  & \vc{\InputIfFileExists{./figures/RingR/pnot.l.tikz}{}{Missing file!}} \by{N}
		\vc{\InputIfFileExists{./figures/RingR/pnot_1.tikz}{}{Missing file!}} \by{+}
		\vc{\InputIfFileExists{./figures/RingR/pnot_2.tikz}{}{Missing file!}} \by{N}
		\vc{\InputIfFileExists{./figures/RingR/pnot.r.tikz}{}{Missing file!}} \\
		  & \vc{\begin{tikzpicture}
	\begin{pgfonlayer}{nodelayer}
		\node [style=none] (13) at (-1.5, 0) {};
		\node [style=none] (18) at (-1.5, -1) {};
		\node [style=polynomial] (19) at (-1.5, -0.5) {$a$};
	\end{pgfonlayer}
	\begin{pgfonlayer}{edgelayer}
		\draw (19.center) to (18.center);
		\draw (13.center) to (19.center);
	\end{pgfonlayer}
\end{tikzpicture}
} \by{\ref{lemRingRDoubleNOT}}
		\vc{\InputIfFileExists{./figures/RingR/pbig_1.tikz}{}{Missing file!}}\by{N}
		\vc{\InputIfFileExists{./figures/RingR/pbig_2.tikz}{}{Missing file!}}\by{+}
		\vc{\InputIfFileExists{./figures/RingR/pbig.r.tikz}{}{Missing file!}}
	\end{align}
\end{proof}

\begin{lemma} \label{lemRingREntailsB3}
	The following equation is derivable in $\ring$:
	\begin{align}
		\vc{\InputIfFileExists{./figures/RingR/B3.l.tikz}{}{Missing file!}} = \node{poly}{\Pi\ a_1 \dots a_n}
	\end{align}
\end{lemma}

\begin{proof} \label{prfLemRingREntailsB3}
	\begin{align}
		       & \vc{\InputIfFileExists{./figures/RingR/B3.l.tikz}{}{Missing file!}} \by{\ref{lemRingREntailsPNOT}, A} \vc{\InputIfFileExists{./figures/RingR/B3_1.tikz}{}{Missing file!}} \by{B2}\vc{\InputIfFileExists{./figures/RingR/B3_2.tikz}{}{Missing file!}} \by{\ref{lemRingREntailsPNOT}}\vc{\InputIfFileExists{./figures/RingR/B3_3.tikz}{}{Missing file!}} \\
		\by{+, A} & \vc{\InputIfFileExists{./figures/RingR/B3_4.tikz}{}{Missing file!}} \by{}\vc{\InputIfFileExists{./figures/RingR/B3_5.tikz}{}{Missing file!}} \by{inv, \times, +}\node{poly}{\Pi\ a_1 \dots a_n}
	\end{align}
\end{proof}

\begin{lemma} \label{lemRingREntailsCrossingNegation}
	The ruleset $\ring$ entails the translation of the traced crossing is multiplication by $-1$:
	\begin{align}
		\vc{\InputIfFileExists{./figures/ZW/rei1x_l.tikz}{}{Missing file!}} \mapsto \vc{\InputIfFileExists{./figures/RingR/ZW_rei1x_l.tikz}{}{Missing file!}} = \node{white}{-1}
	\end{align}
\end{lemma}

\begin{proof} \label{prfLemRingREntailsCrossingNegation}
	\begin{align}
		\vc{\InputIfFileExists{./figures/RingR/ZW_rei1x_l.tikz}{}{Missing file!}} \by{\times} \vc{\InputIfFileExists{./figures/RingR/neg_1.tikz}{}{Missing file!}} = \vc{\InputIfFileExists{./figures/RingR/neg_2.tikz}{}{Missing file!}}\by{\times}\vc{\InputIfFileExists{./figures/RingR/neg_3.tikz}{}{Missing file!}} \by{\ref{lemRingREntailsB3}} \vc{\InputIfFileExists{./figures/RingR/neg_4.tikz}{}{Missing file!}} \by{+,\times} \node{white}{-1}
	\end{align}
\end{proof}

From here on we will use Lemma~\ref{lemRingREntailsCrossingNegation}
to translate the traced crossing from ZW to $\ring$.
Having derived these implicit simplifications in the translation
we will now translate and derive the rules of ZW in the order presented in \cite{AmarThesis},
with the exception of the $inv$ rule,
which was shown in Lemma~\ref{lemRingRDoubleNOT}.

\begin{lemma} \label{lemRingEntailsZWWires}
	The ruleset $\ring$ entails the wire bending rules of ZW:
	\begin{align}
		\vc{\InputIfFileExists{./figures/wire/snakel.tikz}{}{Missing file!}} & = \vc{\begin{tikzpicture}
	\begin{pgfonlayer}{nodelayer}
		\node [style=none] (0) at (-0.5, -1) {};
		\node [style=none] (5) at (-0.5, 1) {};
	\end{pgfonlayer}
	\begin{pgfonlayer}{edgelayer}
		\draw (5.center) to (0.center);
	\end{pgfonlayer}
\end{tikzpicture}
} = \vc{\InputIfFileExists{./figures/wire/snaker.tikz}{}{Missing file!}} &
		\vc{\InputIfFileExists{./figures/wire/cupsym.tikz}{}{Missing file!}} & = \dcup                                         &
		\vc{\InputIfFileExists{./figures/wire/capsym.tikz}{}{Missing file!}} & = \dcap
	\end{align}
\end{lemma}

\begin{proof} \label{prfLemRingEntailsZWWires}
	This is just the requirement that $\ring_R$ is compact closed.
\end{proof}

\begin{lemma} \label{lemRingEntailsZWrei2x}
	The ruleset $\ring$ entails the translated rule $rei_2^x$ from ZW:
	\begin{align} \label{eqnZW-QR-rei2x}
		\vc{\InputIfFileExists{./figures/ZW/rei2x_l.tikz}{}{Missing file!}} = \vc{\begin{tikzpicture}
	\begin{pgfonlayer}{nodelayer}
		\node [style=none] (1) at (-0.5, 1) {};
		\node [style=none] (4) at (0.5, 1) {};
		\node [style=none] (9) at (-0.5, -1) {};
		\node [style=none] (10) at (0.5, -1) {};
	\end{pgfonlayer}
	\begin{pgfonlayer}{edgelayer}
		\draw (1.center) to (9.center);
		\draw (4.center) to (10.center);
	\end{pgfonlayer}
\end{tikzpicture}
}
		\mapsto
		\vc{\InputIfFileExists{./figures/RingR/ZW_rei2xl.tikz}{}{Missing file!}} = \vc{\begin{tikzpicture}
	\begin{pgfonlayer}{nodelayer}
		\node [style=none] (17) at (-1, 2) {};
		\node [style=none] (18) at (0, 2) {};
		\node [style=none] (19) at (-1, 0.5) {};
		\node [style=none] (20) at (0, 0.5) {};
	\end{pgfonlayer}
	\begin{pgfonlayer}{edgelayer}
		\draw (17.center) to (19.center);
		\draw (18.center) to (20.center);
	\end{pgfonlayer}
\end{tikzpicture}
}
	\end{align}
\end{lemma}

\begin{proof} \label{prfLemRingEntailsZWrei2x}
	Apply rules $\times$ and $Hopf$.
\end{proof}

\begin{lemma} \label{lemRingEntailsZWrei3x}
	The ruleset $\ring$ entails the translated rule $rei_3^x$ from ZW:
	\begin{align} \label{eqnZW-QR-rei3x}
		\vc{\InputIfFileExists{./figures/ZW/rei3x_l.tikz}{}{Missing file!}} = \vc{\InputIfFileExists{./figures/ZW/rei3x_r.tikz}{}{Missing file!}}
		\mapsto
		\vc{\InputIfFileExists{./figures/RingR/ZW_rei3x_l.tikz}{}{Missing file!}} = \vc{\InputIfFileExists{./figures/RingR/ZW_rei3x_r.tikz}{}{Missing file!}}
	\end{align}
\end{lemma}

\begin{proof} \label{prfLemRingEntailsZWrei3x}
	After applying rule $\times$ both sides are isometric to:
	\begin{align}
		\vc{\InputIfFileExists{./figures/RingR/ZW_rei3x.tikz}{}{Missing file!}}
	\end{align}
\end{proof}

\begin{lemma} \label{lemRingEntailsZWrei1x}
	The ruleset $\ring$ entails the translated rule $rei_1^x$ from ZW:
	\begin{align} \label{eqnZW-QR-rei1x}
		\vc{\InputIfFileExists{./figures/ZW/rei1x_l.tikz}{}{Missing file!}} = \vc{\InputIfFileExists{./figures/ZW/rei1x_r.tikz}{}{Missing file!}}
		\mapsto
		\node{white}{-1} = \node{white}{-1}
	\end{align}
\end{lemma}

\begin{proof} \label{prfLemRingEntailsZWrei1x}
	Tautological, thanks to Lemma~\ref{lemRingREntailsCrossingNegation}.
\end{proof}

\begin{lemma} \label{lemRingEntailsZWnatxeta}
	The ruleset $\ring$ entails the translated rule $nat_x^\eta$ from ZW:
	\begin{align} \label{eqnZW-QR-nat}
		\vc{\InputIfFileExists{./figures/ZW/natetax_l.tikz}{}{Missing file!}} = \vc{\InputIfFileExists{./figures/ZW/natetax_r.tikz}{}{Missing file!}}
		\mapsto
		\vc{\InputIfFileExists{./figures/RingR/ZW_natetax_l.tikz}{}{Missing file!}} = \vc{\InputIfFileExists{./figures/ZW/natetax_r.tikz}{}{Missing file!}}
	\end{align}
\end{lemma}

\begin{proof} \label{prfLemRingEntailsZWnatxeta}
	This is identical to $rei^2_x$ from Lemma~\ref{lemRingEntailsZWrei2x}.
\end{proof}

\begin{lemma} \label{lemRingEntailsZWepsilonx}
	The ruleset $\ring$ entails the translated rule $nat^{\epsilon}_x$ from ZW:
	\begin{align} \label{eqnZW-QR-natepsilonx}
		\vc{\InputIfFileExists{./figures/ZW/natepsilonx_l.tikz}{}{Missing file!}} = \vc{\InputIfFileExists{./figures/ZW/natepsilonx_r.tikz}{}{Missing file!}}
		\mapsto
		\vc{\InputIfFileExists{./figures/ZW/natepsilonx_l.tikz}{}{Missing file!}} = \vc{\InputIfFileExists{./figures/RingR/ZW_natepsilonx_r.tikz}{}{Missing file!}}
	\end{align}
\end{lemma}

\begin{proof} \label{prfLemRingEntailsZWepsilonx}
	This is identical to $rei^2_x$ from Lemma~\ref{lemRingEntailsZWrei2x}.
\end{proof}

\begin{lemma} \label{lemRingEntailsZWnatwx}
	The ruleset $\ring$ entails the translated rule $nat^w_x$ from ZW:
	\begin{align} \label{eqnZW-QR-natwx}
		\vc{\InputIfFileExists{./figures/ZW/natwx_l.tikz}{}{Missing file!}} = \vc{\InputIfFileExists{./figures/ZW/natwx_r.tikz}{}{Missing file!}}
		\mapsto
		\vc{\InputIfFileExists{./figures/RingR/ZW_natwx_l.tikz}{}{Missing file!}} = \vc{\InputIfFileExists{./figures/RingR/ZW_natwx_r.tikz}{}{Missing file!}}
	\end{align}
\end{lemma}

\begin{proof} \label{prfLemRingEntailsZWnatwx}
	This is just rule CZA1.
\end{proof}

\begin{lemma} \label{lemRingEntailsZWcutw}
	The ruleset $\ring$ entails the translated rule $cut_w$ from ZW:
	\begin{align} \label{eqnZW-QR-cutw}
		\vc{\InputIfFileExists{./figures/ZW/cutw_l.tikz}{}{Missing file!}} = \vc{\InputIfFileExists{./figures/ZW/cutw_r.tikz}{}{Missing file!}}
		\mapsto
		\vc{\InputIfFileExists{./figures/RingR/ZW_cutw_l.tikz}{}{Missing file!}} = \vc{\InputIfFileExists{./figures/RingR/ZW_cutw_r.tikz}{}{Missing file!}}
	\end{align}
\end{lemma}

\begin{proof} \label{prfLemRingEntailsZWcutw}
	\begin{align}
		LHS \by{+}\vc{\InputIfFileExists{./figures/RingR/cut_1.tikz}{}{Missing file!}} \by{N} \vc{\InputIfFileExists{./figures/RingR/cut_2.tikz}{}{Missing file!}} \by{\ref{lemRingRDoubleNOT}} \vc{\InputIfFileExists{./figures/RingR/cut_3.tikz}{}{Missing file!}} \by{+} RHS
	\end{align}
\end{proof}

The next rule to translate (Lemma~\ref{lemRingEntailsZWtrw}) requires us
to be able to apply a cup as an input to the plus gate.
We will first provide a quick lemma to show that this
results in the zero state:

\begin{lemma} \label{lemRingRTracePlus0}
	The following is derivable in $\ring$:
	\begin{align}
		\vc{\InputIfFileExists{./figures/RingR/PZ1.l.tikz}{}{Missing file!}} = \vc{\begin{tikzpicture}
	\begin{pgfonlayer}{nodelayer}
		\node [style=none] (10) at (0.5, -1.5) {};
		\node [style=white] (11) at (0.5, -1) {$0$};
	\end{pgfonlayer}
	\begin{pgfonlayer}{edgelayer}
		\draw (11.center) to (10.center);
	\end{pgfonlayer}
\end{tikzpicture}
}
	\end{align}
\end{lemma}

\begin{proof} \label{prfLemRingRTracePlus0}
	\begin{align}
		\vc{\InputIfFileExists{./figures/RingR/PZ1.l.tikz}{}{Missing file!}} \by{\times}
		\vc{\InputIfFileExists{./figures/RingR/trplus_1.tikz}{}{Missing file!}} \by{\times}
		\vc{\InputIfFileExists{./figures/RingR/trplus_2.tikz}{}{Missing file!}}\by{L}
		\vc{\begin{tikzpicture}
	\begin{pgfonlayer}{nodelayer}
		\node [style=none] (10) at (0.5, -1.5) {};
		\node [style=white] (15) at (0.5, -1) {$0$};
		\node [style=white] (16) at (0.5, -0.5) {$0$};
		\node [style=white] (17) at (0.5, 0) {};
	\end{pgfonlayer}
	\begin{pgfonlayer}{edgelayer}
		\draw (15.center) to (10.center);
		\draw (17.center) to (16.center);
	\end{pgfonlayer}
\end{tikzpicture}
}\by{\times, I}
		\vc{\begin{tikzpicture}
	\begin{pgfonlayer}{nodelayer}
		\node [style=none] (10) at (0.5, -1.5) {};
		\node [style=white] (11) at (0.5, -1) {$0$};
	\end{pgfonlayer}
	\begin{pgfonlayer}{edgelayer}
		\draw (11.center) to (10.center);
	\end{pgfonlayer}
\end{tikzpicture}
}
	\end{align}
\end{proof}

\begin{lemma} \label{lemRingEntailsZWtrw}
	The ruleset $\ring$ entails the translated rule $tr_w$ from ZW:
	\begin{align} \label{eqnZW-QR-trw}
		\vc{\InputIfFileExists{./figures/ZW/trw_l.tikz}{}{Missing file!}} = \vc{\InputIfFileExists{./figures/ZW/trw_r.tikz}{}{Missing file!}}
		\mapsto
		\vc{\InputIfFileExists{./figures/RingR/ZW_trw_l.tikz}{}{Missing file!}} = \vc{\InputIfFileExists{./figures/RingR/ZW_trw_r.tikz}{}{Missing file!}}
	\end{align}
\end{lemma}

\begin{proof} \label{prfLemRingEntailsZWtrw}
	\begin{align}
		\vc{\InputIfFileExists{./figures/RingR/ZW_trw_l.tikz}{}{Missing file!}} \by{+} \vc{\InputIfFileExists{./figures/RingR/trw_1.tikz}{}{Missing file!}} \by{\ref{lemRingRTracePlus0}} \vc{\InputIfFileExists{./figures/RingR/trw_2.tikz}{}{Missing file!}} \by{+} \vc{\InputIfFileExists{./figures/RingR/ZW_trw_r.tikz}{}{Missing file!}}
	\end{align}
\end{proof}

\begin{lemma} \label{lemRingEntailsZWsymw}
	The ruleset $\ring$ entails the translated rule $sym_w$ from ZW:
	\begin{align} \label{eqnZW-QR-symw}
		\vc{\InputIfFileExists{./figures/ZW/symw_l.tikz}{}{Missing file!}} = \vc{\InputIfFileExists{./figures/ZW/symw_r.tikz}{}{Missing file!}}
		\mapsto
		\vc{\InputIfFileExists{./figures/QR/ZW_symw_l.tikz}{}{Missing file!}} = \vc{\InputIfFileExists{./figures/QR/ZW_symw_r.tikz}{}{Missing file!}}
	\end{align}
\end{lemma}

\begin{proof} \label{prfLemRingEntailsZWsymw}
	This is symmetry of addition in $+$.
\end{proof}

The rule $ba_w$ from \cite{AmarThesis} is, frankly, too visually complicated once translated.
Instead we rely on the fact that $ba_w$ can be derived from $2a$, $5a$, $5b$ and $5c$ of \cite{ZW} as shown in \cite[proposition 4]{ZW}.
Rule $2a$ of Ref.~\cite{ZW} happens to be the same as the rule $inv$ of Ref.~\cite{AmarThesis},
and the translation of $inv$ was already derived in Lemma~\ref{lemRingRDoubleNOT}.
The remaining three rules we modify slightly and take as axioms B1, CP and I in $\ring$.

\begin{lemma} \label{lemRingREntailsZWBialg1}
	The translations of three rules $5a$, $5b$ and $5c$ of Ref.~\cite{ZW} into $\ring$
	are derivable:
	\begin{align}
		\label{eqnZW-QR-5a}
		\vc{\InputIfFileExists{./figures/ZW/5a_l.tikz}{}{Missing file!}}    = \vc{\InputIfFileExists{./figures/ZW/5a_r.tikz}{}{Missing file!}}
		  & \mapsto
		\vc{\InputIfFileExists{./figures/RingR/ZW_5a_l.tikz}{}{Missing file!}} = \vc{\InputIfFileExists{./figures/RingR/ZW_5a_r.tikz}{}{Missing file!}} \\
		\label{eqnZW-QR-5b}
		\vc{\InputIfFileExists{./figures/ZW/5b_l.tikz}{}{Missing file!}} = \vc{\InputIfFileExists{./figures/ZW/5b_r.tikz}{}{Missing file!}}
		  & \mapsto
		\vc{\InputIfFileExists{./figures/RingR/ZW_5b_l.tikz}{}{Missing file!}} = \vc{\InputIfFileExists{./figures/RingR/ZW_5b_r.tikz}{}{Missing file!}} \\
		\label{eqnZW-QR-5c}
		\vc{\InputIfFileExists{./figures/ZW/5c_l.tikz}{}{Missing file!}} = \vc{\InputIfFileExists{./figures/wire/empty.tikz}{}{Missing file!}}
		  & \mapsto
		\vc{\InputIfFileExists{./figures/RingR/ZW_5c_l.tikz}{}{Missing file!}} = \vc{\InputIfFileExists{./figures/wire/empty.tikz}{}{Missing file!}}
	\end{align}
\end{lemma}

\begin{proof} \label{prfLemRingREntailsZWBialg1}
	\begin{align}
		\vc{\InputIfFileExists{./figures/RingR/ZW_5a_l.tikz}{}{Missing file!}}                             & \by{B1} \vc{\InputIfFileExists{./figures/RingR/ZW_5a_r.tikz}{}{Missing file!}}                              \\
		\vc{\InputIfFileExists{./figures/RingR/ZW_5b_l.tikz}{}{Missing file!}} \by{+} \vc{\InputIfFileExists{./figures/RingR/CP.l.tikz}{}{Missing file!}} & \by{CP}  \vc{\begin{tikzpicture}
	\begin{pgfonlayer}{nodelayer}
		\node [style=white] (20) at (1.75, 0.25) {$0$};
		\node [style=none] (21) at (1.75, 1) {};
		\node [style=none] (22) at (2.25, 1) {};
		\node [style=white] (23) at (2.25, 0.25) {$0$};
	\end{pgfonlayer}
	\begin{pgfonlayer}{edgelayer}
		\draw (21.center) to (20.center);
		\draw (22.center) to (23.center);
	\end{pgfonlayer}
\end{tikzpicture}
} \by{+,\ref{lemRingRTracePlus0}} \vc{\InputIfFileExists{./figures/RingR/ZW_5b_r.tikz}{}{Missing file!}} \\
		\vc{\InputIfFileExists{./figures/RingR/ZW_5c_l.tikz}{}{Missing file!}}                             & \by{+,\times} \scalar{white}{0} \by{I} \vc{\InputIfFileExists{./figures/wire/empty.tikz}{}{Missing file!}}
	\end{align}
\end{proof}

Before our next derivation we need a lemma that shows
the interaction between the NOT and H derived generators.

\begin{lemma} \label{lemRingRNOTH}
	The composition of the NOT and H derived generators in $\ring$ gives the equality:
	\begin{align}
		\vc{\InputIfFileExists{./figures/RingR/noth_1.tikz}{}{Missing file!}} = \vc{\begin{tikzpicture}
	\begin{pgfonlayer}{nodelayer}
		\node [style=none] (0) at (-0.5, -0.5) {};
		\node [style=none] (6) at (-0.5, 1) {};
		\node [style=h2] (7) at (-0.5, 0.5) {};
		\node [style=white] (14) at (-0.5, 0) {$-1$};
	\end{pgfonlayer}
	\begin{pgfonlayer}{edgelayer}
		\draw (6.center) to (0.center);
	\end{pgfonlayer}
\end{tikzpicture}
}
	\end{align}
\end{lemma}

\begin{proof} \label{prfLemRingRNOTH}
	\begin{align}
		\vc{\InputIfFileExists{./figures/RingR/noth_1.tikz}{}{Missing file!}} \by{} \vc{\InputIfFileExists{./figures/RingR/noth_2.tikz}{}{Missing file!}} \by{+} \vc{\InputIfFileExists{./figures/RingR/noth_3.tikz}{}{Missing file!}} \by{\times} \vc{\InputIfFileExists{./figures/RingR/noth_4.tikz}{}{Missing file!}} \by{D} \vc{\InputIfFileExists{./figures/RingR/noth_5.tikz}{}{Missing file!}} \by{+,\times} \vc{\InputIfFileExists{./figures/RingR/noth_6.tikz}{}{Missing file!}} \by{H}  \vc{\begin{tikzpicture}
	\begin{pgfonlayer}{nodelayer}
		\node [style=none] (0) at (-0.5, -0.5) {};
		\node [style=none] (6) at (-0.5, 1) {};
		\node [style=h2] (7) at (-0.5, 0.5) {};
		\node [style=white] (14) at (-0.5, 0) {$-1$};
	\end{pgfonlayer}
	\begin{pgfonlayer}{edgelayer}
		\draw (6.center) to (0.center);
	\end{pgfonlayer}
\end{tikzpicture}
}
	\end{align}
\end{proof}

\begin{lemma} \label{lemRingEntailsZWantxn}
	The ruleset $\ring$ entails the translated rule $ant_x^n$ from ZW:
	\begin{align} \label{eqnZW-QR-antxn}
		\vc{\InputIfFileExists{./figures/ZW/antxn_l.tikz}{}{Missing file!}} = \vc{\InputIfFileExists{./figures/ZW/antxn_r.tikz}{}{Missing file!}}
		\mapsto
		\vc{\InputIfFileExists{./figures/RingR/ZW_antxn_l.tikz}{}{Missing file!}} = \vc{\InputIfFileExists{./figures/RingR/ZW_antxn_r.tikz}{}{Missing file!}}
	\end{align}
\end{lemma}

\begin{proof} \label{prfLemRingEntailsZWantxn}
	\begin{align}
		\vc{\InputIfFileExists{./figures/RingR/ZW_antxn_l.tikz}{}{Missing file!}} \by{A} \vc{\InputIfFileExists{./figures/RingR/antxn_1.tikz}{}{Missing file!}} \by{\ref{lemRingRNOTH}} \vc{\InputIfFileExists{./figures/RingR/antxn_2.tikz}{}{Missing file!}} \by{\times} \vc{\InputIfFileExists{./figures/RingR/ZW_antxn_r.tikz}{}{Missing file!}}
	\end{align}
\end{proof}

\begin{lemma} \label{lemRingEntailsZWZRules}
	The following rules are translated verbatim from ZW,
	and are all derivable in $\ring$:
	\begin{align} \label{eqnZW-QR-cutz}
		\vc{\InputIfFileExists{./figures/ZW/cutz_l.tikz}{}{Missing file!}} &= \vc{\InputIfFileExists{./figures/ZW/cutz_r.tikz}{}{Missing file!}}
		&
		\vc{\InputIfFileExists{./figures/QR/ZW_trz_l.tikz}{}{Missing file!}} &= \vc{\InputIfFileExists{./figures/QR/ZW_trz_r.tikz}{}{Missing file!}} \\
		\vc{\InputIfFileExists{./figures/QR/ZW_symz_l.tikz}{}{Missing file!}} & = \vc{\InputIfFileExists{./figures/QR/ZW_symz_r.tikz}{}{Missing file!}} &
		\unary[white]{} &= \unary[none]{}
	\end{align}
\end{lemma}

\begin{proof}\label{prfLemRingEntailsZWZRules}
	These are all derivable from $\times$.
\end{proof}

\begin{lemma} \label{lemRingEntailsZWbazw}
	The ruleset $\ring$ entails the translated rule $ba_{zw}$ from ZW,
	noting there must be at least one output (boundary at the top of the diagram):
	\begin{align} \label{eqnZW-QR-bazw}
		\vc{\InputIfFileExists{./figures/ZW/bazw_l.tikz}{}{Missing file!}}       = \vc{\InputIfFileExists{./figures/ZW/bazw_r.tikz}{}{Missing file!}}
		\mapsto
		\vc{\InputIfFileExists{./figures/RingR/ZW_bazw_l.tikz}{}{Missing file!}} = \vc{\InputIfFileExists{./figures/RingR/ZW_bazw_r.tikz}{}{Missing file!}}
	\end{align}
\end{lemma}

\begin{proof} \label{prfLemRingEntailsZWbazw}
	This is rule $B2$
\end{proof}
\begin{lemma} \label{lemRingEntailsZWloop}
	The ruleset $\ring$ entails the translated rule $loop$ from ZW:
	\begin{align} \label{eqnZW-QR-loop}
		\vc{\InputIfFileExists{./figures/ZW/loop_l.tikz}{}{Missing file!}} = \vc{\InputIfFileExists{./figures/ZW/loop_r.tikz}{}{Missing file!}}
		\mapsto
		\vc{\InputIfFileExists{./figures/RingR/L.l2.tikz}{}{Missing file!}} = \vc{\begin{tikzpicture}
	\begin{pgfonlayer}{nodelayer}
		\node [style=none] (4) at (0, 0.5) {};
		\node [style=white] (5) at (0, 0) {$0$};
		\node [style=none] (10) at (0, -1.5) {};
		\node [style=white] (11) at (0, -1) {$0$};
	\end{pgfonlayer}
	\begin{pgfonlayer}{edgelayer}
		\draw (4.center) to (5.center);
		\draw (11.center) to (10.center);
	\end{pgfonlayer}
\end{tikzpicture}
}
	\end{align}
\end{lemma}

\begin{proof} \label{prfLemRingEntailsZWloop}
	This is rule $L$ and $\times$.
\end{proof}
\begin{lemma} \label{lemRingEntailsZWph}
	The ruleset $\ring$ entails the translated rule $ph$ from ZW:
	\begin{align} \label{eqnZW-QR-ph}
		\vc{\InputIfFileExists{./figures/ZW/ph_l.tikz}{}{Missing file!}} = \vc{\InputIfFileExists{./figures/ZW/ph_r.tikz}{}{Missing file!}}
		\mapsto
		\vc{\InputIfFileExists{./figures/RingR/ZW_ph_l.tikz}{}{Missing file!}} = \vc{\InputIfFileExists{./figures/RingR/ZW_ph_r.tikz}{}{Missing file!}}
	\end{align}
\end{lemma}

\begin{proof} \label{prfLemRingEntailsZWph}
	This is spider fusion and un-fusion from rule $\times$.
\end{proof}
\begin{lemma} \label{lemRingEntailsZWnatncr}
	The ruleset $\ring$ entails the translated rule $nat^n_c$ from ZW:
	\begin{align} \label{eqnZW-QR-natnc}
		\vc{\InputIfFileExists{./figures/RingR/A.l.tikz}{}{Missing file!}} = \vc{\InputIfFileExists{./figures/RingR/A.r.tikz}{}{Missing file!}}
		\mapsto
		\vc{\InputIfFileExists{./figures/RingR/A.l.tikz}{}{Missing file!}} = \vc{\InputIfFileExists{./figures/RingR/A.r.tikz}{}{Missing file!}}
	\end{align}
\end{lemma}

\begin{proof} \label{prfLemRingEntailsZWnatncr}
	This is rule $A$.
\end{proof}

\begin{lemma} \label{lemRingEntailsZWunx}
	The ruleset $\ring$ entails the translated rule $unx$ from ZW:
	\begin{align} \label{eqnZW-QR-unx}
		\vc{\InputIfFileExists{./figures/ZW/unx_l.tikz}{}{Missing file!}} = \vc{\InputIfFileExists{./figures/ZW/unx_r.tikz}{}{Missing file!}}
		\mapsto
		\vc{\InputIfFileExists{./figures/RingR/ZW_unx_l.tikz}{}{Missing file!}} = \vc{\InputIfFileExists{./figures/RingR/ZW_unx_r.tikz}{}{Missing file!}}
	\end{align}
\end{lemma}

\begin{proof} \label{prfLemRingEntailsZWunx}
	\begin{align}
		\vc{\InputIfFileExists{./figures/RingR/ZW_unx_l.tikz}{}{Missing file!}} \by{\times}\vc{\InputIfFileExists{./figures/RingR/unx_1.tikz}{}{Missing file!}} \by{CZA2}  \vc{\InputIfFileExists{./figures/RingR/ZW_unx_r.tikz}{}{Missing file!}}
	\end{align}
\end{proof}

\begin{lemma} \label{lemRingEntailsZWrngminus1}
	The ruleset $\ring$ entails the translated rule $rng_{-1}$ from ZW:
	\begin{align} \label{eqnZW-QR-rngminus1}
		\vc{\InputIfFileExists{./figures/ZW/rngminus1_l.tikz}{}{Missing file!}} = \unary[white]{-1}
		\mapsto
		\unary[white]{-1} = \unary[white]{-1}
	\end{align}
\end{lemma}

\begin{proof} \label{prfLemRingEntailsZWrngminus1}
	This is tautological,
	because of Lemma~\ref{lemRingREntailsCrossingNegation}.
\end{proof}

\begin{lemma} \label{lemRingEntailsZWrngplusrs}
	The ruleset $\ring$ entails the translated rule $rng_+^{r,s}$ from ZW:
	\begin{align} \label{eqnZW-QR-rngplusrs}
		\vc{\InputIfFileExists{./figures/ZW/rngplusrs_l.tikz}{}{Missing file!}} = \vc{\begin{tikzpicture}
	\begin{pgfonlayer}{nodelayer}
		\node [style=none] (1) at (0, 1.5) {};
		\node [style=none] (26) at (0, -1.5) {};
		\node [style=white] (31) at (0, 0) {r+s};
	\end{pgfonlayer}
	\begin{pgfonlayer}{edgelayer}
		\draw (1.center) to (31.center);
		\draw (31.center) to (26.center);
	\end{pgfonlayer}
\end{tikzpicture}
}
		\mapsto
		\vc{\InputIfFileExists{./figures/RingR/ZW_rngplusrs_l.tikz}{}{Missing file!}} = \vc{\begin{tikzpicture}
	\begin{pgfonlayer}{nodelayer}
		\node [style=none] (1) at (0, 1) {};
		\node [style=none] (26) at (0, -1) {};
		\node [style=white] (31) at (0, 0) {r+s};
	\end{pgfonlayer}
	\begin{pgfonlayer}{edgelayer}
		\draw (1.center) to (31.center);
		\draw (31.center) to (26.center);
	\end{pgfonlayer}
\end{tikzpicture}
}
	\end{align}
\end{lemma}

\begin{proof} \label{prfLemRingEntailsZWrngplusrs}
	\begin{align}
		\vc{\InputIfFileExists{./figures/RingR/ZW_rngplusrs_l.tikz}{}{Missing file!}} \by{\times} \vc{\InputIfFileExists{./figures/RingR/rngplusrs_1.tikz}{}{Missing file!}} \by{B2} \vc{\InputIfFileExists{./figures/RingR/rngplusrs_2.tikz}{}{Missing file!}} \by{+,\times} \vc{\begin{tikzpicture}
	\begin{pgfonlayer}{nodelayer}
		\node [style=none] (1) at (0, 1) {};
		\node [style=none] (26) at (0, -1) {};
		\node [style=white] (31) at (0, 0) {r+s};
	\end{pgfonlayer}
	\begin{pgfonlayer}{edgelayer}
		\draw (1.center) to (31.center);
		\draw (31.center) to (26.center);
	\end{pgfonlayer}
\end{tikzpicture}
}
	\end{align}
\end{proof}

 \label{secZWRingRRetranslate}

\begin{lemma} \label{lemRingRZWTranslatedGenerators}
	The equality $\interpret{\interpret{g}_{ZW}}_{\ring_R} = g$ is derivable in $\ring$ for each generator $g$.
\end{lemma}

\begin{proof} \label{prfLemRingRZWTranslatedGenerators}
	\begin{align} \label{eqnQR-ZW-QR-mult}
		\interpret{\interpret{\binary[white]{}}_{ZW}}_{\ring_R} = \interpret{\binary[white]{}}_{\ring_R} = \binary[white]{}                                \\
		\interpret{\interpret{\binary[poly]{}}_{ZW}}_{\ring_R} = \interpret{\vc{\InputIfFileExists{./figures/ZW/plus.tikz}{}{Missing file!}}}_{\ring_R} \by{\ref{lemRingRTranslatedPlus}} \binary[poly]{} \\
		\interpret{\interpret{\state{white}{a}}_{ZW}}_{\ring_R} = \interpret{\state{white}{a}}_{\ring_R} = \state{white}{a}
	\end{align}
\end{proof}

This concludes our collection of proofs showing that $\ring_R$ is complete.

\section{Extending ZH, and relating it to RING}
\label{secCompleteTranslationZHQR}

The calculus $\ZHC$ was shown to be complete via the existence of its own normal form,
rather than via equivalence with ZW \cite{ZH}.
In this section will give a generalisation of ZH, called $\ZHR$,
and we will then exhibit translations between $\ZHR$ and $\ring_R$.
It must be noted that while $\ZWR$ exists and is well defined for
any commutative ring $R$,
the calculus $\ZHR$ is only defined for commutative rings $R$ with a half.

Although we have already provided a complete ruleset for $\ring_R$ from \ZWR\ in \S\ref{secRingRRules},
ZW and ZH have very different rulesets of their own \cite{ZH, AmarThesis}.
The translations of these rulesets into $\ring$
remain distinct, and therefore one may be more useful than the other in different situations.
The translation from $\ring_R$ to ZH$_R$ was rather easy to find;
one simply had to find a ring substructure in ZH (an idea continued in \S\ref{chapPhaseRingCalculi}),
which required looking at the rules of ZH to see where addition and multiplication were used.
Translating from ZH to $\ring_R$ was similarly easy;
since we have access to elementary matrices it was simple enough to find
a translation for low-arity H-boxes,
and then generalise from there.

Since \ZHR\ has the additional requirement on $R$
that it be a commutative ring with~$\frac{1}{2}$ (compared to
\ZWR, where $R$ simply needs to be a commutative ring) we
simply state here a smaller ruleset (\S\ref{secRingRZHRules}) for $\ring$
and translation to and from \ZHR\ (Figure~\ref{figRingRZHTranslation}),
relegating the proof of completeness to Appendix~\ref{appZHRing}.

\subsection{The graphical calculus $\ZHR$} \label{secZHR}

The graphical calculus $\ZHR$ is almost identical to that of $\ZHC$,
with the exception that the labels of the H-boxes are now from
some arbitrary commutative ring with a half, $R$.
Explicitly this means that $R$ has a characteristic other than 2,
and that there is an element $\half{}$ which is the multiplicative inverse of 2.

\begin{definition}[Generalised ZH]  \label{defZHR}
	\ZHR\ has generators and interpretation as given in Figure~\ref{figZHRGenerators},
	and a ruleset given in Figure~\ref{figZHRRules}.
\end{definition}

\begin{restatable}[\ZHR\ is complete]{theorem}{thmZHRComplete} \label{thmZHRComplete}
	The rules presented in Figure~\ref{figZHRRules},
	are sound and complete for the universal calculus \ZHR,
	where $R$ is a commutative ring with a half.
\end{restatable}

\begin{proof} \label{prfThmZHRComplete}
	The proof is given in Appendix~\ref{chapZHR},
	because it follows an almost identical proof to that of \cite{ZH}.
\end{proof}

\begin{restatable}[$\ring$ is complete, via \ZHR]{theorem}{thmRingRCompleteZH} \label{thmRingRCompleteZH}
	The rules of \S\ref{secRingRZHRules} are complete for the graphical calculus $\ring_R$,
	where $R$ is a commutative ring with a half.
\end{restatable}

\begin{proof} \label{prfThmRingRCompleteZH}
	The proof is given in Appendix~\ref{appZHRing}.
\end{proof}

\begin{figure}[H]
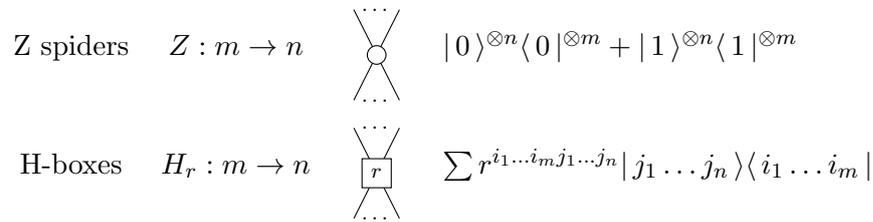

	\centering
	\begin{tabular}{cccl}
		Z spiders & $Z : m \to n$   & $\spider{smallZ}{}$ & $\ket{0}^{\otimes n}\bra{0}^{\otimes m} + \ket{1}^{\otimes n}\bra{1}^{\otimes m}$\\[\rowgap]
		H-boxes   & $H_r : m \to n$ & $\spider{ZH}{r}$ & $\sum r^{i_1\ldots i_m j_1\ldots j_n} \ket{j_1\ldots j_n}\bra{i_1\ldots i_m}$
	\end{tabular}
	\caption{
		The generators and interpretation of $\ZHR$
		\label{figZHRGenerators}
		\label{figZHRInterpretation}}
\end{figure}

\begin{figure}[p]
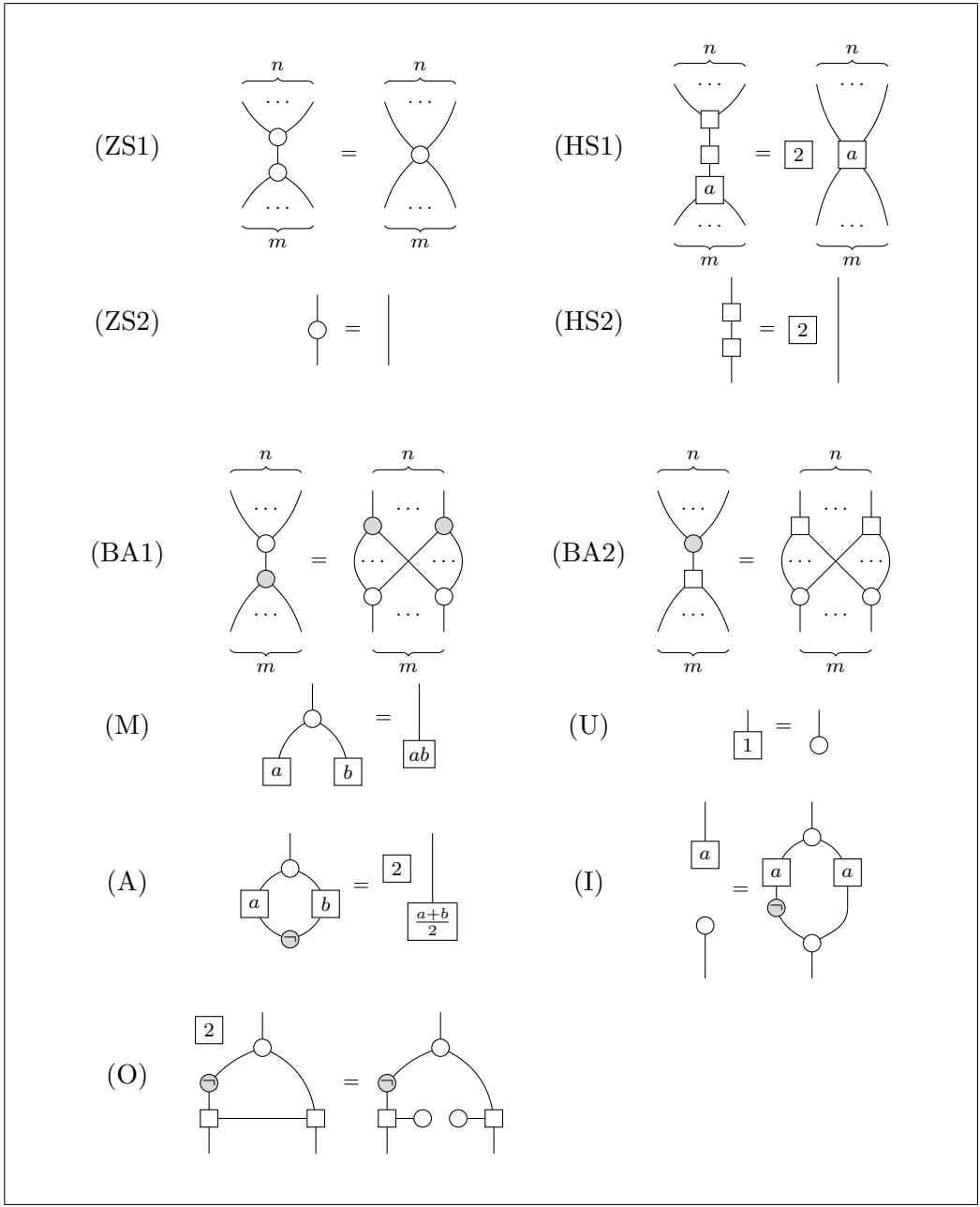

	\centering\begin{tabular}{|c|}
		\hline
		\begin{minipage}{0.9\textwidth}{
				\begin{tabular}{cccc} \\
			\qquad(ZS1) & \vc{\InputIfFileExists{./figures/ZH/Z_spider_rule.tikz}{}{Missing file!}} &   (HS1) & \vc{\InputIfFileExists{./figures/ZH/H_spider_rule_phased.tikz}{}{Missing file!}} \\[4em]
			\qquad(ZS2) & \vc{\InputIfFileExists{./figures/ZH/Z_special.tikz}{}{Missing file!}} &  (HS2) & \vc{\InputIfFileExists{./figures/ZH/H_identity_phased.tikz}{}{Missing file!}} \\[4em]
			\qquad(BA1) & \vc{\InputIfFileExists{./figures/ZH/ZX_bialgebra.tikz}{}{Missing file!}} &  (BA2) & \vc{\InputIfFileExists{./figures/ZH/ZH_bialgebra.tikz}{}{Missing file!}} \qquad\\[4em]
			\qquad(M) & \vc{\InputIfFileExists{./figures/ZH/multiply_rule_phased.tikz}{}{Missing file!}} &  (U) & \vc{\InputIfFileExists{./figures/ZH/unit_rule.tikz}{}{Missing file!}} \\[2em]
			\qquad(A) & \vc{\InputIfFileExists{./figures/ZH/average_rule.tikz}{}{Missing file!}} &  (I) & \vc{\InputIfFileExists{./figures/ZH/intro_rule.tikz}{}{Missing file!}} \qquad\\[4em]
			\qquad(O) & \vc{\InputIfFileExists{./figures/ZH/ortho_rule_phased.tikz}{}{Missing file!}} &  & \\[4em]
		\end{tabular}}
		\end{minipage} \\
		\hline
	\end{tabular}
	\caption{\label{figZHRRules}
		The rules for \ZHR,
		note that these are almost identical to the rules for $\ZH_\bbC$
		as presented in \cite{ZH}.
		Throughout, $m,n$ are non-negative integers and in the original paper $a,b$ were arbitrary complex numbers.
		In \ZHR\ we allow $a,b$ to be arbitrary elements of the ring $R$, where $R$ is commutative with a half.
		The right-hand sides of both \textit{bialgebra} rules (BA1) and (BA2) are complete bipartite graphs on $(m+n)$ vertices, with an additional input or output for each vertex.
		The horizontal edges in equation (O) are well-defined because only the topology matters and we do not need to distinguish between inputs and outputs of generators.
		The rules (M), (A), (U), (I), and (O) are pronounced \textit{multiply}, \textit{average}, \textit{unit}, \textit{intro}, and \textit{ortho}, respectively.}
\end{figure}

\FloatBarrier
\subsection{Translation between RING and ZH}

The translation between $\ring_R$ and $\ZHR$ is given as
Figure~\ref{figRingRZHTranslation}.
This translation preserves the interpretation into $\bit{R}$.

\begin{figure}[H]
	\def\rowgap{2em}
	\begin{align*}
		\ZHR              & \to \ring_R \nonumber     &
		\spider{smallZ}{} & \mapsto \spider{smallZ}{} &
		\spider{ZH}{a}    & \mapsto
		\vc{\InputIfFileExists{./figures/QR/QR_of_HBox.tikz}{}{Missing file!}} \\
		\ring_R           & \to \ZHR \nonumber &
		\spider{white}{a} & \mapsto
		\vc{\InputIfFileExists{./figures/ZH/ZH_of_QRingSpider.tikz}{}{Missing file!}} &
		\binary[poly]{}   & \mapsto
		\vc{\InputIfFileExists{./figures/ZH/ZH_plus.tikz}{}{Missing file!}}
	\end{align*}
	\caption{\label{figRingRZHTranslation}
		Translations between $\ring_R$ and \ZHR
	}
\end{figure}

\subsection{A small, complete ruleset of RING from ZH} \label{secRingRZHRules}

This ruleset is the result of translating the rules from ZH into $\ring$.
We have taken some effort to remove redundant rules from the list below,
although we do not make any claims as to necessity of individual rules.
This list is included to provide a contrast against the rules of \S\ref{secRingRRules}.
We have also left the redundant ring rules in since they are a foundational aspect of $\ring_R$.
We present these somewhat reduced rules here, divided into sections:

\begin{itemize}
	\item Ring Rules. These rules have clear interpretations as the ring operations in $\ring_R$.
	      \begin{align*}
		      \vc{\InputIfFileExists{./figures/QR/QR_mult.tikz}{}{Missing file!}}
		        & \by{\times}
		      \state{white}{a \times b}
		      &
		      \vc{\InputIfFileExists{./figures/QR/comm_plus.tikz}{}{Missing file!}}
		        & \by{+_c}
		      \binary[poly]{} 
		      &
		      \vc{\InputIfFileExists{./figures/QR/ass_plusl.tikz}{}{Missing file!}}
		        & \by{+_a}
		      \vc{\InputIfFileExists{./figures/QR/ass_plusr.tikz}{}{Missing file!}} 
		      \\[\rowgap]
		      \vc{\InputIfFileExists{./figures/QR/comm_times.tikz}{}{Missing file!}}
		        & \by{\times_c}
		      \binary[white]{} 
		      &
		      \vc{\InputIfFileExists{./figures/QR/QR_plus.tikz}{}{Missing file!}}
		        & \by{+}
		      \state{white}{a + b}
		      &
		      \vc{\InputIfFileExists{./figures/QR/ass_timesl.tikz}{}{Missing file!}}
		        & \by{\times_a}
		      \vc{\InputIfFileExists{./figures/QR/ass_timesr.tikz}{}{Missing file!}}
		      \\[\rowgap]
		      \vc{\InputIfFileExists{./figures/QR/QR_UnitPlus.tikz}{}{Missing file!}}
		        & \by{+_0}
		      \unary[none]{}
		      &
		      \vc{\InputIfFileExists{./figures/QR/QR_UnitTimes.tikz}{}{Missing file!}}
		        & \by{\times_1}
		      \unary[none]{}
		      &
		      \vc{\InputIfFileExists{./figures/QR/QR_Dl.tikz}{}{Missing file!}}
		        & \by{D}
		      \vc{\InputIfFileExists{./figures/QR/QR_Dr.tikz}{}{Missing file!}}
	      \end{align*}

	\item The multiplication spider, scalar multiplication, and Hadamard.
	      As always the diagonal dots indicate one or more wires, and the horizontal dots indicate zero or more wires.
	      \begin{align}
		      \vc{\InputIfFileExists{./figures/QR/Z_spiderl.tikz}{}{Missing file!}} &\by{S}
		      \spider{white}{a \times b}, & 
		      \scalar{white}{} \; \scalar{white}{-\frac{1}{2}}
		        & \by{sc} \vc{\InputIfFileExists{./figures/wire/empty.tikz}{}{Missing file!}}, &
		      \vc{\InputIfFileExists{./figures/QR/QR_Hl.tikz}{}{Missing file!}}
		        & \by{H}
		      \vc{\begin{tikzpicture}
	\begin{pgfonlayer}{nodelayer}
		\node [style=none] (5) at (1.25, 1.75) {};
		\node [style=none] (6) at (1.25, -1.75) {};
		\node [style=white] (14) at (0.75, 0) {};
	\end{pgfonlayer}
	\begin{pgfonlayer}{edgelayer}
		\draw (5.center) to (6.center);
	\end{pgfonlayer}
\end{tikzpicture}
}
	      \end{align}
	\item Bialgebra rules.
	      \begin{align}
		      \vc{\InputIfFileExists{./figures/RingR/QR_BA1l.tikz}{}{Missing file!}}
		        & \by{BA1}
		      \vc{\InputIfFileExists{./figures/RingR/QR_BA1r.tikz}{}{Missing file!}},&
		      \vc{\InputIfFileExists{./figures/RingR/QR_BA2l.tikz}{}{Missing file!}}
		        & \by{BA2}
		      \vc{\InputIfFileExists{./figures/RingR/QR_BA2r.tikz}{}{Missing file!}}
	      \end{align}
	\item The $intro$ and $ortho$ rules.
	      \begin{align}
		      \vc{\begin{tikzpicture}
	\begin{pgfonlayer}{nodelayer}
		\node [style=white] (9) at (-1.5, 0.5) {$a$};
		\node [style=white] (10) at (-1.5, -0.75) {};
		\node [style=none] (11) at (-1.5, 1.5) {};
		\node [style=none] (12) at (-1.5, -1.5) {};
	\end{pgfonlayer}
	\begin{pgfonlayer}{edgelayer}
		\draw (11.center) to (9.center);
		\draw (10.center) to (12.center);
	\end{pgfonlayer}
\end{tikzpicture}
}
		        & \by{I}
		      \vc{\InputIfFileExists{./figures/RingR/QR_Ir.tikz}{}{Missing file!}} &
		      \vc{\InputIfFileExists{./figures/RingR/QR_Ol.tikz}{}{Missing file!}}
		        & \by{O}
		      \vc{\InputIfFileExists{./figures/RingR/QR_Or.tikz}{}{Missing file!}}
			  \end{align}
	\item Finally the rule that links the plus gate with its image after being translated into \ZHR\ and back again.
	\begin{align}
		      \binary[poly]{}
		        & \by{+'}
		      \vc{\InputIfFileExists{./figures/RingR/QR_plus_prime_r.tikz}{}{Missing file!}}
	      \end{align}
\end{itemize}

\section{Summary}

In our first results chapter we have introduced the new graphical calculus $\ring$,
showed its completeness using ZW,
and extended ZH.
$\ring$ itself is generated only by ring operations
(acting on states) and compact closed structure.
This sparsity of generators will allow us to use $\ring$
as a `generic' phase ring calculus in the next chapter,
and so investigate the limitations and uses of phase ring graphical calculi.
The focus of the next chapter is
on the category of phase ring graphical calculi,
more precisely the category of graphical calculi that model ring structure
in the manner of ZW, ZH and $\ring$,
and the morphisms that preserve soundness.
  \thispagestyle{empty} 
\chapter{Phase Ring Graphical Calculi and Phase Homomorphisms} \label{chapPhaseRingCalculi}
\thispagestyle{plain}
\noindent\emph{In this chapter:}
\begin{itemize}
	\item[\chapterbullet] We describe the category of phase ring graphical calculi
	\item[\chapterbullet] We justify $\ring_R$ being `initial, up to ring isomorphism'
	\item[\chapterbullet] We justify $\ring_K$ being `the generic phase field calculus'
	\item[\chapterbullet] We investigate when algebra homomorphisms lift to soundness- and proof- preserving maps
	\item[\chapterbullet] This chapter requires knowledge of the calculus $\ring$ from \S\ref{chapRingR},
	and the idea of phase algebra homomorphisms will be key to parts of \S\ref{chapConjectureInference} and \S\ref{chapConjectureVerification}.
\end{itemize}

The calculus $\ring_R$, introduced in \S\ref{chapRingR},
resulted from investigating the topics we shall cover in this chapter.
We would like to know what extra structure we achieve by describing 
a calculus as having a phase ring (or phase algebra in general)
over just a set of labels.
We shall phrase this as a category of phase ring graphical calculi,
and then investigate the morphisms of this category
that arise from morphisms of rings.
We will only be looking at phase rings that act as monoids on states of arity $0 \to 1$,
i.e. situations where elements of arity $0 \to 1$ act like elements in a ring,
with some gates of arity $2 \to 1$ acting like monoids on those states.
This is how ZW and ZH act as phase rings, with ZX acting analogously but on a phase group instead.
The phase group calculus ZQ (\S\ref{chapZQ}) does not act in this manner, for reasons covered in that chapter.

In \S\ref{secPhaseRingGeneric} we look at the category
of Phase Ring Graphical Calculi as Monoids
before exhibiting the calculi $\ring_R$, $\ZW_R$ and $\ZHR$ as objects in this category.
This is compared to the earlier work performed in Ref.~\cite{DixonKissingerOpenGraphs},
which also looked at models of algebraic structures in graphical calculi.
We then investigate how ring homomorphisms $\phi: S \to S'$
relate to morphisms out of $\ring_S$ in this category,
laying the groundwork for \S\ref{secPhaseFieldGeneric} and \S\ref{secPhaseRingHomomorphisms}.

In \S\ref{secPhaseFieldGeneric} we look at the calculus
$\ring_K$ where $K$ is a field,
showing that $\ring_K$ has an essential uniqueness in its interpretation,
as well as the property of being `initial, up to field automorphism'
for $K$-phase-field calculi into $\bit{K}$.
In \S\ref{secPhaseRingHomomorphisms}
we look at how ring morphisms can lift to
morphisms of certain graphical calculi,
and how ring homomorphisms can lift to
soundness- and proof- preserving maps of theorems.
Finally in \S\ref{secPhaseGroupHomomorphismsEtc} we give a similar treatment
to phase \emph{group} homomorphisms,
classifying the phase group endomorphisms available
in certain finite fragments of ZX,
and then show how to generalise these ideas to $\Sigma$-algebras.

We derive these results in the generality afforded to us by
the ring structure of both the phase ring
and the category $\bit{R}$.
For quantum graphical calculi, however,
we only need to work over $\bit{\bbC}$,
and we shall show later on that
this severely restricts the phase ring structure of the graphical calculus as well.
The effect of this is that
the calculus $\ring_\bbC$
will embed naturally into any universal qubit graphical calculus with a phase ring.
$\ring_\bbC$ therefore acts as a common (universal and complete) core
for all phase ring qubit graphical calculi.

\section{The category of phase ring graphical calculi} \label{secPhaseRingGeneric}

We begin this chapter with an intention of finding a `generic' phase ring calculus.
In order to tackle this problem we
build on the idea of models of PROPs
from Ref.~\cite{Maclane1965}
and introduce a category
of Phase Ring Graphical Calculi (Definition~\ref{defPRGCR}).
From there we show that $\ring_R$
has certain desirable properties,
such as the fact that
given any graphical calculus $G$ with phase ring structure $R$
there is exactly one map from $\ring_R$ to $G$ that preserves
that structure (Corollary~\ref{corGivenIsoRingRInitial}).

The author did not feel they could justify calling $\ring_R$
`the' generic phase ring calculus;
there are too many ways to model phase ring structure over $\bits{R}$.
They do, however, feel justified in calling $\ring_K$
`the' generic phase field calculus over $\bit{K}$ for
reasons that we show in \S\ref{secPhaseFieldGeneric}.
The idea of a phase ring homomorphism that
we begin in this section (Definition~\ref{defPhaseRingHomomorphism})
is continued in \S\ref{secPhaseRingHomomorphisms},
and then generalised to other algebras in \S\ref{secPhaseGroupHomomorphismsEtc}.

\subsection{Objects in \PRGC{$R$}}

We begin with as generic a notion of phase ring structure as we can.
Consider a commutative ring which is a set $S$ equipped with multiplication
$\times : 2 \to 1$, addition $+: 2 \to 1$, and with the distinguished elements $0$ and $1$.

\begin{definition}[$\rprop_S$]  \label{defRingPropS}
	We construct our generic ring PROP $\rprop_S$ for the ring $(S, +, \times)$ with the following generators:
	\begin{align}
		\forall s \in S \text{ the element} \quad & \; \state{white}{s} & 0 \to 1 \\
		\text{addition gate} \quad                & \binary{+}          & 2 \to 1 \\
		\text{multiplication gate} \quad          & \binary{\times}     & 2 \to 1
	\end{align}
	We then quotient out these two rewrite rules:
	\begin{align} \label{eqnRingPropRewrite}
		\bap{ZH}{\times}{a}{b} & \by{\text{multiplication}} \state{white}{a \times b} &
		\bap{ZH}{+}{a}{b} & \by{\text{addition}} \state{white}{a + b}
	\end{align}

	There is a category of all such PROPs,
	defined by being the image of the functor from the category of rings
	that sends $S$ to $\rprop_S$,
	and sends a ring homomorphism $\phi$ to a strict symmetric monoidal functor $\bar \phi$:
	\begin{align}
		\bar \phi: \rprop_S  & \to \rprop_{S'}                \\
		\state{white}{s} & \mapsto \state{white}{\phi(s)} &
		\binary{+}       & \mapsto \binary{+}             &
		\binary{\times}  & \mapsto \binary{\times}
	\end{align}
\end{definition}

\begin{remark} \label{remRingPropGates}
	These two rewrite rules in \eqref{eqnRingPropRewrite} imply that our addition and multiplication gates
	are commutative and associative \emph{when applied directly to states},
	purely because those properties are already present in the ring $S$.
	For example:
	\begin{align}
	\bap{ZH}{\times}{a}{b} \by{\text{multiplication}} \state{white}{a \times b} 
	= \state{white}{b \times a} \by{\text{multiplication}} \bap{ZH}{\times}{b}{a}
	\end{align}
	This is a weaker property than saying that these gates are commutative etc. in general.
\end{remark}

\begin{definition}[Models of $\rprop$]  \label{defModelOfRPROP}
	A model for $\rprop_S$
	in some graphical calculus $G$
	is a morphism of PROPs $\rprop_S \to G$.
\end{definition}

A model of $\rprop_S$ into a graphical calculus
can be seen as choosing diagrams $\set{g_s}$, $g_+$ and $g_\times$
to represent states, an addition gate, and a multiplication gate
with the syntactic behaviour suggested by their names.

\begin{example}[Degenerate model] \label{exaDegenerateModel}
	We construct the graphical calculus $\Zero$
	as having generators $O_m^n$, one for each arity $m \to n$,
	interpreted as the zero matrix of size $2^m \to 2^n$,
	and the following rewrite rules:
	\begin{align}
		O_m^n \comp O_l^m   & = O_l^n         \\
		O_a^b \tensor O_c^d & = O_{a+c}^{b+d}
	\end{align}
	For any PROP $\bbP$ there is a strict symmetric monoidal
	functor from $\bbP$ to $\Zero$;
	in particular there is a unique model of $\rprop_S$ in $\Zero$
	for any $S$, but it is not faithful.
\end{example}

\begin{definition}[\PRGC{$R$}]  \label{defPRGCR}
	The category of Phase Ring Graphical Calculi (as monoids) into $\bit{R}$,
	shortened to \PRGC{$R$},
	has as objects graphical calculi over $\bit{R}$ that contain a model of $\rprop_{S}$ for some $S$,
	and morphisms that preserve this $\rprop_S$ structure. Specifically:
	\begin{itemize}
		\item Objects are pairs $(G, \Model_S)$,
		      where $G$ is a (compact closed) graphical calculus given as a collection of generators,
		      an interpretation $\idot{G}$, and rewrite rules,
		      and $\Model_S$ is a model of $\rprop_S$ in $G$.
		\item A morphism
		      \begin{align}
			      f : (G, \Model_S) \to (G', \Model_{S'})
		      \end{align}
		      is a strict symmetric monoidal functor
		      $f_\bbD: G \to G'$ and also a ring homomorphism $f_\bbR : S \to S'$ such that
		      the right hand diagram commutes:
		      \begin{align} \label{eqnCommutativePRGC}
			      \begin{tikzcd}[ampersand replacement=\&]
				      S \arrow[d,"f_\bbR"] \& \rprop_S \arrow[d,"\bar f_\bbR"] \arrow[r,"\Model_S"] \& G \arrow[d,"f_\bbD"] \\
				      S' \& \rprop_{S'} \arrow[r,"\Model_{S'}"] \& G'
			      \end{tikzcd}
		      \end{align}
	\end{itemize}
\end{definition}

\begin{proposition} \label{propPRGCWellFormed}
	The above definition of \PRGC{$R$} is indeed a category.
\end{proposition}

\begin{proof} \label{prfPropPRGCWellFormed}
	We just need to show the following properties:
	\begin{itemize}
		\item Identity morphisms $f$ are given by the identity morphism $f_\bbD$
		      of compact closed PROPs and $f_\bbR$ of commutative rings respectively,
		      and this satisfies the condition of preserving the distinguished diagrams $\set{g_s}, g_+, g_\times$
		\item Composition of morphisms exists for $f_\bbD$ and $f_\bbR$,
		      and the commutative diagram of \eqref{eqnCommutativePRGC} composes vertically
		\item Associativity of composition of morphisms follows from the associativity of composition of the functors
		      $f_\bbD$ and the ring homomorphisms $f_\bbR$
	\end{itemize}
\end{proof}

\begin{remark} \label{remSemanticModel}
	In this thesis we are using a \emph{syntactic} version
	of the models of $\rprop_S$.
	In Definition~\ref{defRingPropS},
	where we stated that addition and multiplication
	should be rules (i.e. equivalences)
	in $\rprop_S$ we could instead
	have given a requirement on the interpretation
	of any model, such as:
	\begin{align}
		\interpret{\Model_S\left(\bap{ZH}{\times}{a}{b}\right)} = \interpret{\Model_S\left(\state{white}{a \times b}\right)} \\
		\interpret{\Model_S\left(\bap{ZH}{+}{a}{b}\right)} = \interpret{\Model_S\left(\state{white}{a + b}\right)}
	\end{align}

	The author feels that this semantic version
	better reflects the way these calculi are researched:
	First by constructing generators with certain interpretations
	and then finding rules.
	The syntactic version we use, however,
	is far neater.
	All the results in this chapter
	can be adapted to semantic versions with little effort.
\end{remark}

\begin{definition}  \label{defRingREverywhere}
	We will regard the calculus $\ring_S$ as being an object of \PRGC{$R$}
	whenever $S$ is a subring of $R$, exhibited as
	$\ring_S$ with
	interpretation $\idot{} : \ring_S \to \bit{S}$ naturally extending to $\idot{} : \ring_S \to \bit{R}$,
	and with model $\Model_S$ the identification of the ring structure inherent in $\ring_S$:
	\begin{align}
		\state{white}{s} & \mapsto \state{white}{s} & \binary{+} & \mapsto \binary[poly]{} & \binary{\times} & \mapsto \binary[white]{}
	\end{align}
	The addition and multiplication rules of $\ring_S$
	are the image of the addition and multiplication rules of $\rprop_S$ under $\Model_S$.
\end{definition}

\begin{remark} \label{remNotNecSubring}
	Definition~\ref{defRingREverywhere}
	shows that whenever $S \subset R$
	we have a model of $\rprop_S$ in $\PRGC{\bit{R}}$.
	It could well be the case that
	there are faithful models of some $\rprop_S$
	in $\PRGC{\bit{R}}$ where $R \subsetneq S$,
	but we will show in \S\ref{secPhaseFieldGeneric}
	that this cannot be the case when $R$ is a field.
\end{remark}

\begin{lemma} \label{lemPRGCZWZH}
	The calculi \ZW$_S$ and \ZH$_S$
	(with the usual phase ring structure)
	can be exhibited as objects of $\PRGC{\bit{R}}$
	whenever $S$ is a subring of $R$
\end{lemma}

\begin{proof} \label{prfLemPRGCZWZH}
	For $\ZW_S$ we extend its interpretation to be into $\bit{R}$,
	and use the model $\Model_S$ given by:
	\begin{align}
		\Model_S : \rprop_S & \to \ZW_S                                        \\
		\state{white}{s}    & \mapsto \state{white}{s} \quad \forall s \in S \\
		\binary{+}          & \mapsto \vc{\InputIfFileExists{./figures/ZW/plus.tikz}{}{Missing file!}}                      \\
		\binary{\times}     & \mapsto \binary[white]{}
	\end{align}

	Likewise for $\ZH_S$ we extend the interpretation and use the model given by:
	\begin{align}
		\Model_S : \rprop_S & \to \ZH_S                                     \\
		\state{white}{s}    & \mapsto \state{ZH}{s} \quad \forall s \in S \\
		\binary{+}          & \mapsto \vc{\InputIfFileExists{./figures/ZH/ZH_plus.tikz}{}{Missing file!}}                \\
		\binary{\times}     & \mapsto \binary[white]{}
	\end{align}
	Since these are the `obvious'
	models of $\rprop_S$ suggested by the addition and multiplication
	rules of $\ZW$ and $\ZH$
	we shall simply refer to these objects of $\PRGC{\bit{R}}$
	as $\ZW_S$ and $\ZH_S$.
	The requirements of Definition~\ref{defPRGCR}
	are shown by checking:
	\begin{align}
		ZW \entails \vc{\InputIfFileExists{./figures/ZW/plusab.tikz}{}{Missing file!}}    & \by{cut_z,ba_{zw},rng_+} \state{white}{a+b} &
		ZW \entails \bap{white}{}{a}{b} & \by{cut_z} \state{white}{a\times b} \\
		ZH \entails \vc{\InputIfFileExists{./figures/ZH/ZH_plusab.tikz}{}{Missing file!}} & \by{HS1,A,M!} \state{ZH}{a+b}               &
		ZH \entails \vc{\InputIfFileExists{./figures/ZH/timesab.tikz}{}{Missing file!}} &\by{M} \state{ZH}{a\times b}
	\end{align}
\end{proof}

\begin{definition}[Generic object in $\PRGC{\bit{R}}$]  \label{defGenericObjectInPRGC}
	We will simply write $G_S$ to indicate a generic object
	$(G, \Model_S)$ in $\PRGC{\bit{R}}$.
	We will also distinguish the image
	of the generators of $\rprop_S$ under $\Model_S$ as:
	\begin{align}
		g_s & := \Model_S\left(\state{white}{s}\right) &
		g_+ & := \Model_S\left(\binary{+}\right)       &
		g_\times & := \Model_S\left(\binary{\times}\right)
	\end{align}
\end{definition}

The idea behind the definition we give for $\PRGC{\bit{R}}$
is similar to that of MacLane's definition of algebras for a PROP (Definition~\ref{defPROPAlg}).
In the usual definition a PROP was held fixed and all possible
models of that PROP were considered.
In this work we want to be able to consider
multiple phase rings (i.e. a related family of PROPs)
while also restricting the models
into only graphical calculi over $\bit{R}$,
and only those models that are the identity on objects.

\begin{definition}[Algebras for a PROP {\cite[\S5]{Maclane1965}}]  \label{defPROPAlg} 
	For a PROP $\bb{P}$ representing an algebraic theory $\cal{A}$
	an \emph{algebra of type $\cal{A}$} is a product-preserving functor from $\bb{P}$ to
	the category of sets.
\end{definition}

\subsection{Some morphisms in \PRGC{$R$}}

Now that we know the objects in this category,
let us give some example morphisms.
Of the phase ring graphical calculi
we have discussed $\ring_S$
has the least `extra' structure
beyond that which is necessary for phase rings.
It is perhaps unsurprising that the morphisms
out of $\ring_S$ are determined by
the ring morphisms out of $S$.

\begin{proposition}[Morphisms out of $\ring_S$] \label{propHomInPRGCHomInRing}
	\begin{align}
		\Hom_{\PRGC{\bit{R}}}[\ring_S, G_{S'}] & \iso \Hom_{\text{Ring}}[S, S']
	\end{align}
\end{proposition}

\begin{proof} \label{prfPropHomInPRGCHomInRing}
	We first note that every map $\phi$ in $\Hom_{\text{Ring}}[S, S']$
	lifts to the map
	\begin{align}
		\ring_S \xrightarrow{(\phi_\bbD, \phi)} G_{S'}
	\end{align}
	where
	\begin{align}
		\phi_\bbD : \ring_S & \to G               \\
		\state{white}{a} & \mapsto g_{\phi(a)} &
		\binary{+}       & \mapsto g_+         &
		\binary{\times}     & \mapsto g_\times
	\end{align}
	which satisfies the conditions of Definition~\ref{defPRGCR}.
	Next we note that
	for a generic morphism $f = (f_\bbD, f_\bbR)$
	of $\Hom_{\PRGC{\bit{R}}}[\ring_S, G_{S'}]$,
	the diagram map $f_\bbD$ is uniquely determined by $f_\bbR$,
	because $\ring_S$ is generated by
	\begin{align}
		\state{white}{a},\ \binary{+},\ \binary{\times},
	\end{align}
	and the conditions of Definition~\ref{defPRGCR} determine $f_\bbD$ as:
	\begin{align}
		f_\bbD : \ring_S   & \to G                 \\
		\state{white}{s}\  & \mapsto g_{f_\bbR(s)} &
		\binary{+}         & \mapsto g_+           &
		\binary{\times}    & \mapsto g_\times
	\end{align}
	Therefore every map out of $\ring_S$ in $\PRGC{\bit{R}}$ is of the form $f = (\phi_\bbD, \phi)$.
	We therefore have the equivalence:
	\begin{align}
		\Hom_{\PRGC{\bit{R}}}[\ring_S, G_{S'}] & \iso \Hom_{\text{Ring}}[S, S']     \\
		(\phi_\bbD, \phi)                      & \between \phi \label{eqnLiftedPhi}
	\end{align}
\end{proof}

Suppose that instead of just looking at morphisms
out of $\ring_R$ we also wanted to look at morphisms
into some generic object $G_S$ that preserved ring structure.

\begin{corollary} \label{corGivenIsoRingRInitial}
	Given a generic object $G_{S'}$ and an isomorphism $\phi : S \iso S'$
	there is exactly one morphism $(f_\bbD, f_\bbR)$
	from $\ring_S$ to $G_{S'}$
	such that $f_\bbR = \phi$
\end{corollary}

\begin{remark} \label{remRingInitial}
	Corollary~\ref{corGivenIsoRingRInitial}
	is the justification for saying that $\ring_S$
	is `initial for graphical calculi with phase ring $S$'.
	The degenerate example
	$\Zero$ can be seen as similarly terminal;
	for any ring $S$ there is a unique model $\Model_S$
	of $\rprop_S$ in $\Zero$,
	and for any object $G_S$ of $\PRGC{\bit{R}}$
	there is a unique morphism from $G_S$ to $(\Zero, \Model_S)$
	formed by sending every diagram in $G_S$ to the unique diagram of
	the same arity in $\Zero$.
\end{remark}

Another example of a morphism in $\PRGC{\bit{R}}$
is that of the phase ring homomorphism
we are about to define.
Phase ring homomorphisms are
in essence the lifting of a ring homomorphism
$\phi : S \to S'$ not to $\bar \phi : \rprop_S \to \rprop_{S'}$
but to a functor between two members of
a parameterised family of calculi
(such as $\ring_S \to \ring_{S'}$).
Such a lifting is not guaranteed to exist,
and depends on both the calculus and the model.
We will define these phase ring homomorphisms for
$\ring_S$, $\ZH_S$ and $\ZW_S$ in \S\ref{secPhaseRingHomomorphisms},
but the definition relies on
the way phases are presented in the calculus.

\begin{definition}[Phase ring homomorphism for $\ring$]  \label{defPhaseRingHomomorphism}
	Given a ring homomorphism $\phi: S \to S'$,
	we can lift $\phi$ to the symmetric monoidal functor $\hat \phi : \ring_S \to \ring_{S'}$
	where we let $\phi$ act on each phase:
	\begin{align}
		\state{white}{a} & \mapsto \state{white}{\phi(a)} & \binary[poly]{} & \mapsto \binary[poly]{} & \binary[white]{} & \mapsto \binary[white]{}
	\end{align}
	This $\hat \phi$ is called a `phase ring homomorphism',
	with the pair $(\hat \phi, \phi)$ being a morphism in $\PRGC{\bit{R}}$.
	There is a similar lifting for semantics,
	where $\phi$ is lifted to the symmetric monoidal functor
	$\tilde \phi : \bits{S} \to \bits{S'}$ which is the identity on objects,
	and acts on matrices as:
	\begin{align}
		\begin{pmatrix}
			a_{11} & a_{12} & \dots \\
			a_{21} & a_{22} & \dots \\
			\vdots
		\end{pmatrix} \mapsto
		\begin{pmatrix}
			\phi(a_{11}) & \phi(a_{12}) & \dots \\
			\phi(a_{21}) & \phi(a_{22}) & \dots \\
			\vdots
		\end{pmatrix}
	\end{align}
	We touch on this structure of $(\phi, \hat \phi, \tilde \phi)$ again in \S\ref{secPhaseRingHomomorphisms}.
\end{definition}

Note that in \PRGC{$R$} we have no restrictions
on the ring substructures we are modelling with $\rprop_S$.
There is also no interplay between the interpretation \idot{} and the ring structure.
In the next subsection we limit ourselves to \PRGC{$K$},
for $K$ a field,
and show the very strict requirements that result from \emph{faithful} interpretations of $\rprop_S$.
After we have examined fields we shall return to the idea of a phase ring homomorphism.

\section{The generic phase field calculus} \label{secPhaseFieldGeneric}

In this section $K$ is a field, and $R$ is a commutative ring.
Faithfulness of the interpretation for $\rprop_R$
(Definition~\ref{defFaithfulInterpretation})
not only allows us to avoid degenerate cases,
but also imposes profound restrictions on the interpretation which we will cover in this section.
We first note that the restriction forces two important properties for the states
$\smallstate{white}{0}$ and $\smallstate{white}{1}$;
in particular that they are both non-zero, and non-colinear,
and so form a basis for $K^2$ (Lemma~\ref{lemUniquenessBasis})
which in turn forces the interpretation of the other generators of $\rprop_R$
(Lemma~\ref{lemRingRepFunctions}).
From there we show several properties
of $\ring_K$, which we summarise in Remark~\ref{remRingKGeneric}.

\begin{definition}[Faithful interpretation, and faithful model]  \label{defFaithfulInterpretation}
	An interpretation $\idot{G}$ of a graphical calculus
	is called faithful if no two generators are sent to the same
	morphism.
	A model of $\rprop_S$ in $G$ is faithful if no two generators of $\rprop_S$
	are sent to the same morphism by the composition of $\Model_S$ and $\idot{G}$.
\end{definition}

\begin{definition}[Treating $\rprop_R$ as a graphical calculus]  \label{defRingPropInterpreted}
	We create the (explicitly not compact-closed) graphical calculus $\rpropint_R$
	as the PROP $\rprop_R$
	of Definition~\ref{defRingPropS}
	with a faithful interpretation $\idot{}$ into $\bit{K}$.
\end{definition}

\begin{lemma}\label{lemRingRepNonzero}
	In $\rpropint_R$
	\begin{align}
		\forall r \in R \quad \interpret{\state{white}{r}} \neq 0
	\end{align}
\end{lemma}

\begin{proof}
	We rely on the fact that for any matrix $M$, $0 \tensor M = 0$.
	\begin{align}
		\text{assume\ }  & \interpret{\state{white}{r}} = 0         \\
		\implies & \interpret{\bap{ZH}{\times}{r}{0}} = 0 & \text{by\ } 0 \tensor M = 0 \\
		\implies &\interpret{\bap{ZH}{\times}{r}{0}} \by{\times}   \interpret{\state{white}{0}}   = 0     \\
		\implies & \interpret{\bap{ZH}{+}{0}{1}} = 0      & \text{by\ } 0 \tensor M = 0 \\
		\implies & \interpret{\bap{ZH}{+}{0}{1}} \by{+} \interpret{\state{white}{1}    } = 0 \\
		\implies & \interpret{\state{white}{0}} = \interpret{\state{white}{1}}
	\end{align}
	Which contradicts faithfulness of the generators.
\end{proof}

\begin{lemma}  \label{lemRingBasis}
	In $\rpropint_R$
	\begin{align}
		\interpret{\smallstate{white}{0}} \neq \interpret{\smallstate{white}{1}} \; \lambda, \quad \lambda \in K
	\end{align}
\end{lemma}

\begin{proof} \label{prfLemRingBasis} By contradiction
	\begin{align}
		\interpret{\state{white}{0}}  \by{\times} \interpret{\bap{ZH}{\times}{0}{0}}
		\by{hyp.} \interpret{\bap{ZH}{\times}{0}{1}}  \lambda
		\by{\times} \interpret{\state{white}{0}} \lambda
	\end{align}
	Note that by Lemma~\ref{lemRingRepNonzero} we know that $\interpret{\smallstate{white}{0}} \neq 0$,
	and so because $K$ is a field we know that $\lambda=1$.
	Therefore $\interpret{\smallstate{white}{0}} = \interpret{\smallstate{white}{1}}$,
	which contradicts faithfulness.
\end{proof}

\begin{lemma} \label{lemUniquenessBasis}
	The interpretations $\interpret{\smallstate{white}{0}}$ and $\interpret{\smallstate{white}{1}}$ are

	\begin{align}
		\interpret{\state{white}{0}} & := \begin{pmatrix}
			1 \\ 0
		\end{pmatrix} &
		\interpret{\state{white}{1}} & := \begin{pmatrix}
			1 \\ 1
		\end{pmatrix}
	\end{align}

	up to change of basis.
\end{lemma}

\begin{proof}  \label{prfLemUniquenessBasis}
	Follows immediately from the non-colinearity shown in Lemma~\ref{lemRingBasis}
\end{proof}

\begin{remark} \label{remWhyNotModules}
	We limit ourselves to interpretations over $\bit{K}$ rather than $\bit{R}$s,
	and so vector spaces rather than modules,
	because we want to be able to talk about `up to change of basis'.
\end{remark}

\begin{lemma} \label{lemRingRepFunctions}
	The choice of basis in Lemma~\ref{lemUniquenessBasis} forces the following interpretations of $\times$ and $+$
	in $\rpropint_R$:

	\begin{align}
		\interpret{\binary{\times}} & = \begin{pmatrix}
			1 & 0 & 0 & 0 \\
			0 & 0 & 0 & 1
		\end{pmatrix} &
		\interpret{\binary{+}}      & = \begin{pmatrix}
			1 & 0 & 0 & 0 \\
			0 & 1 & 1 & 0
		\end{pmatrix}
	\end{align}
\end{lemma}

\begin{proof}  \label{prfLemRingRepFunctions}
	The following products form a basis for $K^2 \tensor K^2$
	\begin{align}
		\interpret{\state{white}{0} \tensor \state{white}{0}},\;
		\interpret{\state{white}{1} \tensor \state{white}{0}},\;
		\interpret{\state{white}{0} \tensor \state{white}{1}},\;
		\interpret{\state{white}{1} \tensor \state{white}{1}}
	\end{align}
	We will use these basis vectors to determine the entries of $\interpret{\smallbinary{+}}$ and $\interpret{\smallbinary{\times}}$.
	We do this by expressing the following equations as diagrams, which should hold under $\interpret{\cdot}$
	by soundness of the rules in definition~\ref{defRingPropS}:
	\begin{align}
		0 \times 0 & = 0 & 0 \times 1 & = 0 \nonumber \\
		1 \times 0 & = 0 & 1 \times 1 & = 1 \nonumber \\
		0+0        & = 0 & 0+1        & = 1 \nonumber \\
		1+1 &= 1 \label{eqnListGraphicalEasy}
	\end{align}

	We perform the first calculation (which determines the first column in the matrix
	interpretation of $\times$) in full as an example:
	\begin{align}
		\interpret{\binary{\times}} \comp \begin{pmatrix}
			1 \\ 0 \\ 0 \\ 0
		\end{pmatrix}
		=
		\interpret{\bap{ZH}{\times}{0}{0}}
		\by{\times}
		\interpret{\state{white}{0}}=
		\begin{pmatrix}
			1 \\ 0
		\end{pmatrix}
	\end{align}

	In this manner we are able to determine the entries for the matrix interpretation of $\times$ as:

	\begin{align}
		\interpret{\binary{\times}} = \begin{pmatrix}
			1 & 0 & 0 & 0 \\
			0 & 0 & 0 & 1
		\end{pmatrix} \label{eqnRingStateTimes}
	\end{align}

	We do not, however, have enough equations involving addition in \eqref{eqnListGraphicalEasy}
	to determine all the entries in the matrix interpretation of $+$.
	We have merely determined the first three columns,
	expressed as:
	\begin{align}
		\interpret{\binary{+}} = \begin{pmatrix}
			1 & 0 & 0 & b \\
			0 & 1 & 1 & c
		\end{pmatrix} \label{eqnRingStatePlusPartial}
	\end{align}
	Where $b$ and $c$ are elements of $K$. Since $R$ is a ring;
	even if we don't know its characteristic we can still define $2 := 1 + 1$.
	\begin{align}
		\interpret{\state{white}{2}} \by{+} \interpret{\bap{ZH}{+}{1}{1}}
		\by{\eqref{eqnRingStatePlusPartial}} \begin{pmatrix}
			1 + b \\ 2 + c
		\end{pmatrix}
	\end{align}
	We then use the equation $2 \times 0 = 0$ to determine $b$:
	\begin{align}
		\begin{pmatrix}
			1 \\ 0
		\end{pmatrix} & = \interpret{\state{ZH}{0}}  \by{\times} \interpret{
			\bap{ZH}{\times}{2}{0}
		}
		= \begin{pmatrix}
			1 & 0 & 0 & 0 \\
			0 & 0 & 0 & 1
		\end{pmatrix}
		\comp
		\begin{pmatrix}
			1+b \\ 0 \\ 2+c \\ 0
		\end{pmatrix} \\
		\\
		\implies b                 & =0
	\end{align}

	We can now determine that $c \in \set{0, -1, -2}$ by checking $2+2 = 2\times 2$:
	\begin{align}
		                & \interpret{
			\bap{ZH}{+}{2}{2}
		}
		\by{\eqref{eqnRingStatePlusPartial}} \begin{pmatrix}
			1 & 0 & 0 & 0 \\
			0 & 1 & 1 & c
		\end{pmatrix}
		\comp
		\begin{pmatrix}
			1 \\ 2+c \\ 2+c \\ (2+c)^2
		\end{pmatrix} \\
		\by{\times}     & \interpret{
			\bap{ZH}{\times}{2}{2}
		}
		= \begin{pmatrix}
			1 & 0 & 0 & 0 \\
			0 & 0 & 0 & 1
		\end{pmatrix}
		\comp
		\begin{pmatrix}
			1 \\ 2+c \\ 2+c \\ (2+c)^2
		\end{pmatrix} \\
		\therefore\quad &
		2(2+c) + c(2+c)^2 =
		(2+c)^2
		\\
		\therefore\quad & (2+c)(c^2 + 2+c)  =0  \\
		\therefore\quad & c \in \set{0, -1, -2}
	\end{align}

	Finally we show that $c=0$ by contradiction:
	\begin{itemize}
		\item If $c = -1$
		      \begin{align}
			      \interpret{
				      \bap{ZH}{+}{1}{1}
			      }
			                   &
			      = \begin{pmatrix}
				      1 & 0 & 0 & 0  \\
				      0 & 1 & 1 & -1
			      \end{pmatrix}
			      \comp
			      \begin{pmatrix}
				      1 \\ 1 \\ 1 \\ 1
			      \end{pmatrix} =
			      \begin{pmatrix}
				      1 \\ 1
			      \end{pmatrix}
			      = \interpret{\state{white}{1}} \\
			      \implies 1+1 & = 1                      \\
			      \implies 1   & = 0 \quad \contradiction
		      \end{align}
		\item If $c=-2$:
		      \begin{align}
			      \interpret{
				      \bap{ZH}{+}{1}{1}
			      }
			                   &
			      = \begin{pmatrix}
				      1 & 0 & 0 & 0  \\
				      0 & 1 & 1 & -2
			      \end{pmatrix}
			      \comp
			      \begin{pmatrix}
				      1 \\ 1 \\ 1 \\ 1
			      \end{pmatrix} =
			      \begin{pmatrix}
				      1 \\ 0
			      \end{pmatrix}
			      = \interpret{\state{white}{0}} \\
			      \implies 1+1 & = 0 \\
		      \end{align}
		      If $1+1=0$ then either we are working in characteristic 2 (and $c = -2 = 0$) or we have a contradiction.
	\end{itemize}
	Therefore $c = 0$.
\end{proof}

\begin{remark} \label{remWhyThisBasis}
	We chose this basis so that multiplication would coincide with the usual interpretation of the Z-spider.
	There is an argument for using the computational basis $\interpret{\smallstate{white}{0}} = \ket{0}$ and $\interpret{\smallstate{white}{1}} = \ket{1}$
	because it makes the actions of $\times$ and $+$ easier to read in bra-ket notation,
	but it also makes the results of the next section harder to see.
	Under the computational basis the interpretations of $\times$ and $+$ are:
	\begin{align}
		\interpret{\binary{\times}}_{\text{computational}} &= \begin{pmatrix}
			1 & 1 & 1 & 0 \\
			0 & 0 & 0 & 1
		\end{pmatrix} \\
		\interpret{\binary{+}}_{\text{computational}} &= \begin{pmatrix}
			1 & 1 & 1 & -1 \\
			0 & 1 & 1 & 2
		\end{pmatrix}
	\end{align}
\end{remark}

Now that we have established the interpretation of $\rpropint_R$ up to a change of basis,
we shall look at the relationship this forces between $R$ and $K$ as rings.
First we look at the \emph{sets} $R$ and $K$ (Lemma~\ref{lemSetEmbedding}),
and then find an injective ring homomorphism from $R$ to $K$ (Lemma~\ref{propSubringC}).

\begin{lemma} \label{lemSetEmbedding}
	In $\rpropint_R$, recalling that we assume the interpretation into $\bit{K}$ is faithful,
	$R$ embeds into $K$ as sets.
\end{lemma}

\begin{proof}
	Using the interpretation from lemma \ref{lemUniquenessBasis} (i.e. up to change of basis) we know that:
	\begin{align}
		         & r \in R \nonumber                                                                                                                        \\
		         & \interpret{\state{white}{r}} = \begin{pmatrix}
			b \\ c
		\end{pmatrix}                                                                                \\
		\implies & \begin{pmatrix}
			1 \\ 0
		\end{pmatrix} = \interpret{\state{white}{0}} \by{\times}  \interpret{ \bap{ZH}{\times}{r}{0} } = \begin{pmatrix}
			b \\ 0
		\end{pmatrix} \\
		\implies & b = 1
	\end{align}
	Therefore we can construct the function $f$ that sends an element of the set $R$ to its second component in $\interpret{\state{white}{r}}$.
	\begin{align}
		f : R                          & \to K     \\
		r = \begin{pmatrix}
			1 \\ c
		\end{pmatrix} & \mapsto c
	\end{align}
	Since the interpretation is faithful, and the first component is always 1, $f$ must be injective.
\end{proof}

\begin{proposition}\label{propSubringC}
	For $\rpropint_R$ over $\bit{K}$ there is an injective ring homomorphism $R \injects K$
\end{proposition}

\begin{proof}
	The restriction of $K$ to $R$ in Lemma~\ref{lemSetEmbedding} preserves all of ($\times$, +, 0, 1) so
	$f$ is a ring homomorphism.
\end{proof}

\begin{remark} \label{remPropRSubringK}
	This, in a sense, classifies phase ring substructures of a graphical calculus over $\bit{K}$.
	That is whenever one finds a model for $\rprop_R$ (Definition~\ref{defRingPropS})
	then Proposition~\ref{propSubringC} shows that $R$ must be a subring of $K$.
	Note that the calculus $\ring_K$ contains a model of $\rprop_R$ for every subring $R$ of $K$,
	so this subring requirement is the strictest subring requirement we will be able to find.
\end{remark}

Now that we have shown how restrictive $\rpropint_R$ is in terms of its requirements on $R$,
let us look at how expressive it is as a calculus,
and compare it to the earlier parts of this chapter.

\begin{lemma} \label{lemPropRNotUniversal}
	$\rpropint_R$ is not universal over $\bit{K}$
\end{lemma}
\begin{proof} \label{prfLemPropRNotUniversal}
	We cannot form morphisms of the shape $1 \to 0$
\end{proof}

This is unsurprising;
the category $\bit{K}$ is compact closed, but $\rprop_R$ is not.
What is surprising is that adding compact closure, and setting $R$ to $K$,
\emph{is} enough to achieve universality.
Not only this but we will show that our newly added generators again have their interpretations determined
by our choice of basis.

\begin{definition}[$\rprop_R^{T,\ \interpret{}}$]  \label{defPropRT}
	The graphical calculus $\rprop_R^{T,\ \interpret{}}$ is defined by adding the generators
	\begin{align}
		\text{cup} \quad & \dcup \\
		\text{cap} \quad & \dcap
	\end{align}
	to $\rpropint_R$, requiring that they obey the snake equations
	\begin{align}
		\vc{\InputIfFileExists{./figures/wire/snakel.tikz}{}{Missing file!}} = \vc{\begin{tikzpicture}
	\begin{pgfonlayer}{nodelayer}
		\node [style=none] (0) at (-0.5, -1) {};
		\node [style=none] (5) at (-0.5, 1) {};
	\end{pgfonlayer}
	\begin{pgfonlayer}{edgelayer}
		\draw (5.center) to (0.center);
	\end{pgfonlayer}
\end{tikzpicture}
} = \vc{\InputIfFileExists{./figures/wire/snaker.tikz}{}{Missing file!}}
	\end{align}
	and also enact the transpose:
	\begin{align}
		\interpret{\vc{\InputIfFileExists{./figures/wire/transpose.tikz}{}{Missing file!}}} = \interpret{D}^T \label{eqnTranspose}
	\end{align}
\end{definition}

\begin{lemma} \label{lemPropTransposeInterpretation}
	The requirement given in Definition~\ref{defPropRT} forces the following interpretation for the cup and cap
	in $\rprop_R^{T,\ \interpret{}}$:
	\begin{align}
		\interpret{\dcup} & = \begin{pmatrix}
			1 \\ 0 \\ 0 \\1
		\end{pmatrix} &
		\interpret{\dcap} & = \begin{pmatrix}
			1 & 0 & 0 & 1
		\end{pmatrix}
	\end{align}
\end{lemma}

\begin{proof} \label{prfLemPropTransposeInterpretation}
	\begin{align}
		\interpret{\vc{\begin{tikzpicture}
	\begin{pgfonlayer}{nodelayer}
		\node [style=white] (0) at (0, 0) {a};
		\node [style=none] (1) at (1, 0) {};
		\node [style=none] (2) at (1, -1) {};
	\end{pgfonlayer}
	\begin{pgfonlayer}{edgelayer}
		\draw [in=90, out=90, looseness=2.50] (0.center) to (1.center);
		\draw (1.center) to (2.center);
	\end{pgfonlayer}
\end{tikzpicture}
}} & = \interpret{\dcap} \comp \begin{pmatrix}
			1 & 0 \\
			0 & 1 \\
			a & 0 \\
			0 & a
		\end{pmatrix} \by{\eqref{eqnTranspose}}  \begin{pmatrix}
			1 & a
		\end{pmatrix} \\
		\implies \interpret{\dcap}           & = \begin{pmatrix}
			1 & 0 & 0 & 1
		\end{pmatrix} \label{eqnInterpretCap}
	\end{align}
	We now use the requirement that the cup and cap obey the snake equation
	alongside \eqref{eqnInterpretCap} to find the interpretation for the cup:
	\begin{align}
		\vc{\InputIfFileExists{./figures/wire/snaker.tikz}{}{Missing file!}} & = \vc{\begin{tikzpicture}
	\begin{pgfonlayer}{nodelayer}
		\node [style=none] (0) at (-0.5, -1) {};
		\node [style=none] (5) at (-0.5, 1) {};
	\end{pgfonlayer}
	\begin{pgfonlayer}{edgelayer}
		\draw (5.center) to (0.center);
	\end{pgfonlayer}
\end{tikzpicture}
} &
		\implies \interpret{\dcup} & = \begin{pmatrix}
			1 \\ 0 \\ 0 \\ 1
		\end{pmatrix}
	\end{align}
\end{proof}

\begin{proposition} \label{propRingRPropR}
	The graphical calculus $\rprop_K^{T,\ \interpret{}}$ has the same generators and interpretation as $\ring_K$,
	up to change of basis.
\end{proposition}

\begin{proof} \label{prfPropRingRPropR}
	This is a combination of Lemmas~\ref{lemUniquenessBasis}, \ref{lemRingRepFunctions} and \ref{lemPropTransposeInterpretation},
	and Definitions~\ref{defRingR}, \ref{defRingPropS} and \ref{defPropRT}.
\end{proof}

\begin{proposition} \label{propRingKInterpretationDeterminedByBasis}
	For every generic object $(G,\Model_R)$ of $\PRGC{\bit{K}}$ there is a
	unique map $\psi : \bit{K} \to \bit{K}$
	dependent only on $(G,\Model_R)$
	such that for any morphism $f: \ring_K \to (G,\Model_R)$
	the following diagram commutes
	\begin{align}
		\begin{tikzcd}[ampersand replacement=\&]
			\ring_K \arrow[d, "\idot{\ring_K}"] \arrow[r, "f"] \& (G,\Model_R) \arrow[d, "\idot{G}"] \\
			\bit{K} \arrow[r, "\psi"] \& \bit{K}
		\end{tikzcd}
	\end{align}
	What's more $\psi_G$ is given by a change of basis.
\end{proposition}

\begin{proof} \label{prfPropRingKInterpretationDeterminedByBasis}
	We construct $\psi$ as simply the change of basis that sends
	\begin{align}
		\interpret{\state{white}{0}}_{\ring_K} & \mapsto \interpret{g_0}_G &
		\interpret{\state{white}{1}}_{\ring_K} &\mapsto \interpret{g_1}_G
	\end{align}
	which by Lemma~\ref{lemRingBasis} is indeed a change of basis
	(because these states form a basis),
	and by Lemma~\ref{lemRingRepFunctions} the diagram commutes for all generators of $\ring_K$.
\end{proof}

It is easy to construct trivial, or at least dull, graphical calculi in \PRGC{$K$}.
For example for $R \subsetneq K$ the calculus $\ring_R$ is an object of \PRGC{$K$}
(see Definition~\ref{defRingREverywhere}).
What we would like to show is a relationship between $\ring_K$ and any
`suitably expressive' object of \PRGC{$K$}.
It turns out that completeness and universality of the graphical calculus is suitable for our needs.

\begin{lemma} \label{lemUniversalGraphicalKBitExtends}
	Any object $(G, \Model_R)$ of \PRGC{$K$} that is a complete universal graphical calculus (over $\bit{K}$)
	extends naturally to an object $(G, \Model_K)$.
\end{lemma}

\begin{proof} \label{prfLemUniversalGraphicalKBitExtends}
	Consider the object $( G, \Model_R)$.
	Let $f:R \to K$ be the $f$ from Proposition~\ref{propSubringC} that identifies $R$ with a subring of $K$.
	For every $k \in K \less f(R)$ there is a representative $D$ in $G$ of the sole equivalence class
	with \begin{align}
		\interpret{D}_G = \begin{pmatrix}
			1 \\ k
		\end{pmatrix}
	\end{align} (with respect to the basis
	of Lemma~\ref{lemUniquenessBasis}) because the calculus $G$ is universal and complete.

	Choose, for each $k \in K \less f(R)$ such a diagram and label it $g_k$.
	In doing so we have constructed $( G, \Model_K )$;
	an object in \PRGC{$K$} with the same generators as the original object but with phase ring $K$.
	The multiplication and addition rules of $\rprop_K$ hold in $G$ because
	the semantics of $g_+$ and $g_\times$ have already been determined by Lemma~\ref{lemRingRepFunctions},
	those semantics preserve addition and multiplication with $g_k$,
	and $G$ is complete.
\end{proof}

Note that we needed to be able to determine the interpretation
(up to change of basis)
in order to prove Lemma~\ref{lemUniversalGraphicalKBitExtends}.
Since we can extend complete, universal graphical calculi
over $\bit{K}$ naturally to ones with models $\Model_K$,
let us look at the category of just those graphical
calculi with model $\Model_K$,
and where the morphisms preserve this `fullness' of the model.

\begin{definition}[$K$-Phase-Field Calculi]  \label{defKPhaseCalculi}
	We define the category of $K$-Phase-Field Graphical Calculi
	as the subcategory of \PRGC{$K$}
	where the objects $G_S$ are restricted by requiring $S\iso K$,
	and the ring homomorphisms $f_\bbR$ are restricted to being automorphisms of $K$.
\end{definition}

\begin{proposition} \label{propRingKInitial}
	The object $\ring_K$
	is `initial, up to ring isomorphism',
	for the category of $K$-Phase Field Graphical Calculi.
	I.e. every morphism $f$ from $\ring_K$ to $G_K$
	is given by a phase ring isomorphism $\hat \phi : \ring_K \to \ring_K$
	(see Definition~\ref{defPhaseRingHomomorphism})
	followed by an unique morphism $(f_{\bbD,2}, f_{\bbR,2}) :\ring_K \to G$ which is the identity on the ring structure:
	\begin{align}
		f_{\bbD, 2} : \ring_K & \to G   \\
		\state{white}{k} & \mapsto g_k &
		\binary{+}       & \mapsto g_+ &
		\binary{\times}       & \mapsto g_\times \\
		f_{\bbR, 2}  = id_K : K & \to K
	\end{align}
	What's more, every morphism of this form exists.
\end{proposition}

\begin{proof} \label{prfPropRingKInitial}
	We first show that every morphism of the form
	$(f_{\bbD,2}, f_{\bbR,2}) :\ring_K \to G$,
	as constructed in the statement, obeys the requirements of our category.
	Given a field automorphism $\phi : K \to K$:
	\begin{align}
		f_{\bbD,2} \phi_\bbD \left(\state{white}{s}\right) & = f_{\bbD,2} \left( \state{white}{\phi(s)} \right) = g_{id_K \phi(s)} = g_{\phi(s)} \\
		f_{\bbD,2} \phi_\bbD \left(\binary{+}\right)       & = f_{\bbD,2} \left(\binary{+}\right) = g_+                                          \\
		f_{\bbD,2} \phi_\bbD \left(\binary{\times}\right)  & = f_{\bbD,2} \left(\binary{\times}\right) = g_\times
	\end{align}
	The morphism $(f_{\bbD,2}, f_{\bbR,2})$ is unique because $f_{\bbD,2}$ is entirely determined by the action on the generators of $\ring_K$,
	which is entirely determined by the model $\Model_K$ in $G_K$.
\end{proof}

\begin{remark} \label{remRingKGeneric}
	We have therefore shown that:
	\begin{itemize}
		\item $\ring_K$ contains only those generators needed to be a phase ring graphical calculus over monoids (Definition~\ref{defRingR}
		      and Proposition~\ref{propRingRPropR})
		\item $\ring_K$ is universal and complete (Theorem~\ref{thmRingRCompleteZW})
		\item $\ring_K$ has its interpretation fully determined, up to choice of basis (Lemma~\ref{lemRingRepFunctions})
		\item $\ring_K$ is `initial, up to field automorphism' for $K$-Phase-Field Graphical Calculi (Proposition~\ref{propRingKInitial})
		\item Every complete universal graphical calculus in \PRGC{$K$} extends to a $K$-Phase-Field Graphical Calculus (Lemma~\ref{lemUniversalGraphicalKBitExtends})
	\end{itemize}
	Which, in the view of the author,
	is sufficient reason for calling $\ring_K$
	\emph{the} generic phase field calculus for the field $K$ into \bits{K}.
\end{remark}

It is tempting to quotient out this `up to isomorphism' aspect of Proposition~\ref{propRingKInitial},
but doing so would remove phase ring homomorphisms.
These special morphisms will turn out to have important
properties, in particular that they send sound equations to sound equations,
i.e. `theorems are closed under phase ring homomorphism'.

\section{Phase ring homomorphisms} \label{secPhaseRingHomomorphisms}

In Definition~\ref{defPhaseRingHomomorphism}
we showed that any ring homomorphism $\phi: S \to S'$
lifts to the maps
\begin{align}
	\hat \phi : \ring_S   & \to \ring_{S'} \\
	\tilde \phi : \bit{S} & \to \bit{S'}
\end{align}
In this section we continue this idea,
showing first how to perform the same operation
for $\ZWR$ and $\ZHR$,
and then how this lifting interacts
with the semantics (Theorem~\ref{thmLiftRingHomSoundness})
and syntax (Theorem~\ref{thmRingAutomorphismEntailment}) of all three calculi.
We will assume, when speaking of $\ZH_R$, that
the ring $R$ has $\half$.
Exploration of algebras other than rings is delayed until \S\ref{secPhaseGroupHomomorphismsEtc}.

\begin{proposition} \label{propLiftingRingHoms}
	For $\hat \phi$ and $\tilde \phi$ from Definition~\ref{defPhaseRingHomomorphism} the following diagram commutes
	\begin{align} \label{eqnLiftedRingHomDiagram}
		\begin{tikzcd}[ampersand replacement=\&]
			\ring_S \arrow[d, "\idot{}"] \arrow[r, "\hat \phi"] \& \ring_{S'} \arrow[d, "\idot{}"] \\
			\bit{S} \arrow[r, "\tilde \phi"] \& \bit{S'}
		\end{tikzcd}
	\end{align}
\end{proposition}

\begin{proof} \label{prfPropLiftingRingHoms}
	We verify that the diagram \eqref{eqnLiftedRingHomDiagram} commutes for all the generators of $\ring_S$:
	\begin{align}
		  & \begin{tikzcd}[ampersand replacement=\&]
			\state{white}{s} \arrow[d, "\idot{S}"] \arrow[r, "\hat \phi"] \& \state{white}{\phi(s)} \arrow[d, "\idot{S'}"] \\
			\begin{pmatrix}
				1 \\ s
			\end{pmatrix} \arrow[r, "\tilde \phi"] \& \begin{pmatrix}
				1 \\ \phi(s)
			\end{pmatrix}
		\end{tikzcd}
	\end{align}
	\begin{align}
		\begin{tikzcd}[ampersand replacement=\&]
			\binary[poly]{} \arrow[d, "\idot{S}"] \arrow[r, "\hat \phi"] \& \binary[poly]{} \arrow[d, "\idot{S'}"] \\
			\begin{pmatrix}
				1 & 0 & 0 & 0 \\
				0 & 1 & 1 & 0
			\end{pmatrix} \arrow[r, "\tilde \phi"] \& \begin{pmatrix}
				1 & 0 & 0 & 0 \\
				0 & 1 & 1 & 0
			\end{pmatrix}
		\end{tikzcd}
	\end{align}
	\begin{align}
		\begin{tikzcd}[ampersand replacement=\&]
			\binary[white]{} \arrow[d, "\idot{S}"] \arrow[r, "\hat \phi"] \& \binary[white]{} \arrow[d, "\idot{S'}"] \\
			\begin{pmatrix}
				1 & 0 & 0 & 0 \\
				0 & 0 & 0 & 1
			\end{pmatrix} \arrow[r, "\tilde \phi"] \& \begin{pmatrix}
				1 & 0 & 0 & 0 \\
				0 & 0 & 0 & 1
			\end{pmatrix}
		\end{tikzcd}
	\end{align}
\end{proof}

\begin{proposition} \label{propLiftedRingHomZWZH}
	The analogous lifting of $\phi$ to $\hat \phi$ for $\ZW_R$ and $\ZH_R$ are as follows:
	\begin{align}
		\phi : S              & \to S'                          \\
		\hat \phi : \ZW_S     & \to \ZW_{S'}                    \\
		\spider{smallblack}{} & \mapsto \spider{smallblack}{}   \\
		\spider{white}{s}     & \mapsto \spider{white}{\phi(s)}
	\end{align}
	\begin{align}
		\hat \phi : \ZH_S     & \to \ZH_{S'}                  \\
		\spider{ZH}{s}        & \mapsto \spider{ZH}{\phi(s)}  \\
		\spider{smallwhite}{} & \mapsto \spider{smallwhite}{}
	\end{align}
	And the following diagrams commute
	\begin{align}
		\begin{tikzcd}[ampersand replacement=\&]
			\ZW_S \arrow[d, "\idot{}"] \arrow[r, "\hat \phi"] \& \ZW_{S'} \arrow[d, "\idot{}"] \\
			\bit{S} \arrow[r, "\tilde \phi"] \& \bit{S'}
		\end{tikzcd}
		  & \qquad
		\begin{tikzcd}[ampersand replacement=\&]
			\ZH_S \arrow[d, "\idot{}"] \arrow[r, "\hat \phi"] \& \ZH_{S'} \arrow[d, "\idot{}"] \\
			\bit{S} \arrow[r, "\tilde \phi"] \& \bit{S'}
		\end{tikzcd}
	\end{align}
	Note that we do not make it notationally explicit which base language is being considered.
	The construction of $\hat \phi$ and $\tilde \phi$ are the same for $\ring_R$,
	$\ZH$, and $\ZW$, with the context making clear which language we are considering.
\end{proposition}

\begin{proof} \label{prfPropLiftedRingHomZWZH}
	In both cases $\tilde \phi$ is defined analogously to Definition~\ref{defPhaseRingHomomorphism},
	and one can verify that for the generators of $\ZH$ and $\ZW$ that in both cases we have
	$\interpret{\hat \phi D} = \tilde \phi \interpret{D}$.
	It is also clear from the definitions
	that $(\hat \phi, \phi)$ satisfies the conditions for being
	a morphism in $\PRGC{\bit{R}}$.
\end{proof}

\begin{remark} \label{remUniversalGeneratorsDetermineTilde}
Since $\ring$, $\ZH$ and $\ZW$ are all universal graphical calculi we note that $\tilde \phi$
is entirely determined by its action on the generators.
\end{remark}

Since $\phi$ determines $\hat \phi$
we propose the name `phase homomorphism pair' for this structure of $\phi, \hat \phi, \tilde \phi$
and in \S\ref{secPhaseGroupHomomorphismsEtc} we
will give a broader definition of such a pair.
\S\ref{secPhaseGroupHomomorphismsEtc} also gives explicit examples in the phase \emph{group} case.
The primary property of these pairs is preservation of \emph{soundness}:

\begin{proposition} \label{propRingHomPreservesSoundness}
	If the diagrammatic equation $A = B$ in $\ring_S$ is sound,
	then the diagrammatic equation $\hat \phi(A) = \hat \phi(B)$ in $\ring_{S'}$ is sound
\end{proposition}

\begin{proof} \label{prfPropRingHomPreservesSoundness}
	\begin{align}
		\interpret{A}             = \interpret{B}           \implies
		\tilde \phi \interpret{A} = \tilde \phi \interpret{B}   \overset{\ref{propLiftingRingHoms}}{\implies}
		\interpret{\hat \phi A}   = \interpret{\hat \phi B}
	\end{align}
\end{proof}

\begin{example}[Soundness of field phase homomorphisms] \label{exaLiftedRingHomSemantics}
	Consider the field automorphisms of $\bbQ[i, \sqrt{2}]$
	\begin{align}
		\sigma: a + ib + \sqrt{2}c + i\sqrt{2}d & \mapsto a - ib + \sqrt{2}c - i\sqrt{2}d \\
		\tau: a + ib + \sqrt{2}c + i\sqrt{2}d   & \mapsto a + ib - \sqrt{2}c - i\sqrt{2}d
	\end{align}

	The following $\ring_{\bbQ[i, \sqrt{2}]}$ equation is sound:

	\begin{align}
		\exaGalois{i}{-\sqrt{2}i}{\sqrt{2}} = \state{white}{-\sqrt{2}(1+i)}
	\end{align}

	Therefore the following $\ring_{\bbQ[i, \sqrt{2}]}$ equations,
	found by applying the automorphisms $\hat \sigma$, $\hat \tau$ and $\hat \sigma \hat \tau$, are sound:
	\begin{align}
		\hat \sigma           &   & \exaGalois{-i}{\sqrt{2}i}{\sqrt{2}}   & =  \state{white}{-\sqrt{2}(1-i)} \\
		\hat \tau             &   & \exaGalois{i}{\sqrt{2}i}{-\sqrt{2}}   & = \state{white}{\sqrt{2}(1+i)}   \\
		\hat \sigma \hat \tau &   & \exaGalois{-i}{-\sqrt{2}i}{-\sqrt{2}} & = \state{white}{\sqrt{2}(1-i)}   \\
	\end{align}
\end{example}

\begin{theorem}[Phase homomorphisms preserve soundness]\label{thmLiftRingHomSoundness}
	if $\phi: R \to S$ is a ring homomorphism then
	$\ring_R \semantic A = B$ implies $\ring_S \semantic \hat \phi A = \hat \phi B$.
	The same is true for $\ZW_R$ and $\ZH_R$.
\end{theorem}

\begin{proof} \label{prfThmLiftRingHomSoundness}
	The statement for $\ring$ is just Proposition~\ref{propRingHomPreservesSoundness}.
	For $\ZW$ and $\ZH$ use the same argument as Proposition~\ref{propRingHomPreservesSoundness}
	in conjunction with Proposition~\ref{prfPropLiftedRingHomZWZH}.
\end{proof}

We can even go beyond the level of semantics,
and show that syntactic entailment also commutes with $\hat \phi$
(for these three calculi).
We will need to show two things:
That rules are preserved under the action of $\hat \phi$,
and that rule application is preserved under the action of $\hat \phi$.
Preservation of the rules is done by checking all the rules for every calculus mentioned by hand.
This requires no additional calculation,
and so we include just one example below.
There are rules in ZX that do \emph{not}
satisfy the analogous requirements for phase group homomorphisms,
which we will mention explicitly in the proof of Proposition~\ref{propZXPhaseGroupProofPreservingCliffordT}.

\begin{example}[A rule preserved by $\hat \phi$] \label{exaRulePreservingMap}
	Given a ring homomorphism $\phi: R \to S$, the spider rule of $\ring_R$ is mapped by $\hat \phi$
	to a restriction of the same rule in $\ring_S$:
	\begin{align}
		\vc{\InputIfFileExists{./figures/QR/phihat_times_l.tikz}{}{Missing file!}} = \vc{\InputIfFileExists{./figures/QR/phihat_times_r.tikz}{}{Missing file!}}
	\end{align}
	The restriction is that the above equation only matches onto
	phases that are in the image of $\phi$ in $S$.
\end{example}

\begin{proposition} \label{propRingHomPreservesSyntax}
	Using the rules from the completeness theorem for $\ring_\cdot$ (Theorem~\ref{thmRingRCompleteZW})
	and given a ring homomorphism $\phi: R \to S$,
	if $\ring_R \entails A = B$ then $\ring_{S} \entails \hat \phi A = \hat \phi B$,
	using the same proof steps.
\end{proposition}

\begin{proof} \label{prfPropRingHomPreservesSyntax}
	We first note that for any rule $A = B$ in $\ring_R$,
	the equation $\hat \phi A = \hat \phi B$
	expresses a restriction of the rule $A=B$ in $\ring_{S}$
	to the image of $R$ under $\phi$.
	This is because every phase used in a rule in $\ring_R$ is either an integer or a free variable,
	and so under the image of $\phi$ is either fixed (for an integer)
	or expressed as the image of a free variable (for a free variable).
	For one example see Example~\ref{exaRulePreservingMap},
	or for another: Applying $\hat \phi$ to the rule $(L)$ (Figure~\ref{figRingRRules2}) in $\ring_\bbZ$
	yields the identical rule $(L)$ in $\ring_\bb{Q}$,
	because $tr_z$ only involves the phases 0 and 1 (recall that we leave the white spider blank when it has phase 1),
	which are fixed by $\hat \phi$.

	The following diagram\footnote{A \emph{double double pushout diagram}.},
	where both the front and back wide rectangles are double-pushout rewrites, commutes,
	and $\hat \phi A = \hat \phi B$ is a valid rule application in $\ring_{R}$:

	\begin{align*}
		\begin{tikzcd}[ampersand replacement=\&]
			A \arrow[dd] \arrow[from=rr] \arrow[dr, "\hat \phi"] \& \& I \arrow[dd] \arrow[rr]  \arrow[dr, "\hat \phi"]\& \& B \arrow[dd]  \arrow[dr, "\hat \phi"]\& \\
			\& \hat \phi A \arrow[from=rr, crossing over] \& \&  \hat \phi I   \arrow[rr, crossing over] \& \& \hat \phi B \arrow[dd] \\
			G   \arrow[from=rr]  \arrow[dr, "\hat \phi"] \& \& G' \arrow[rr]  \arrow[dr, "\hat \phi"]\& \&   H  \arrow[dr, "\hat \phi"]\& \\
			\& \hat \phi G \arrow[from=rr]  \arrow[from=uu, crossing over]   \& \&  \hat \phi G' \arrow[rr]  \arrow[from=uu, crossing over]\& \&  \hat \phi H \\
		\end{tikzcd}
	\end{align*}

	In particular note that
	\begin{align}
		\begin{tikzcd}[ampersand replacement=\&]
			\hat \phi A \arrow[d] \arrow[from=r] \& \hat \phi I   \arrow[d]  \\
			\hat \phi G \arrow[from=r]   \&  \hat \phi G'
		\end{tikzcd}
		\quad \text{and} \quad
		\begin{tikzcd}[ampersand replacement=\&]
			\hat \phi I \arrow[d] \arrow[r] \& \hat \phi B   \arrow[d]  \\
			\hat \phi G' \arrow[r]   \&  \hat \phi H
		\end{tikzcd}
	\end{align}
	are both still pushouts, since $\hat \phi$ changes none of the underlying graph and only the labels.
	Therefore a sequence of DPO rewrites from $C$ to $D$ in $\ring_R$
	becomes a sequence of DPO rewrites from $\hat \phi C$ to $\hat \phi D$ in $\ring_{S}$
	under the action of $\hat \phi$.
\end{proof}

\begin{example}[Proof preservation of field phase homomorphisms] \label{exaLiftedRingHomSyntax}
	Using the same equation as in Example~\ref{exaLiftedRingHomSemantics},
	we see that the proof
	\begin{align}
		\ring_\bbC \entails & A = B                                                                                                               \\
		                    & \exaGalois{i}{\sqrt{2}i}{-\sqrt{2}} \by{+} \bapwide{white}{}{i}{\sqrt{2}(i-1)} \by{\times} \state{white}{-\sqrt{2}(1+i)}
	\end{align}
	is translated, using the field automorphism $\sigma: \bbC \to \bbC$,
	$\sigma: i \mapsto -i$, to the proof
	\begin{align}
		\ring_\bbC \entails & \hat \sigma A = \hat \sigma B                                                                                           \\
		                    & \exaGalois{-i}{-\sqrt{2}i}{-\sqrt{2}} \by{+} \bapwide{white}{}{-i}{\sqrt{2}(-i-1)} \by{\times} \state{white}{-\sqrt{2}(1-i)}
	\end{align}
\end{example}

\begin{theorem}[Phase homomorphisms preserve proofs] \label{thmRingAutomorphismEntailment}
	If $\ring_R \entails A = B$, then $\ring_S \entails \hat \phi(A) = \hat \phi(B)$
	for every ring homomorphism $\phi : R \to S$.
	The proof of $\ring_S \entails \hat \phi(A) = \hat \phi(B)$ is constructed
	as the one-to-one translation via $\hat \phi$ of the proof of $\ring_R \entails A = B$.
	The same is true for $\ZW_R$ and $\ZH_R$.
\end{theorem}

\begin{proof} \label{prfThmRingAutomorphismEntailment}
	The claim about $\ring_R$ follows from Proposition~\ref{propRingHomPreservesSyntax}.
	One then goes through each rule of $\ZH_R$ and $\ZW_R$
	and shows that the action of $\hat \phi$ on any instance of that rule is
	to create an instance of a rule already present in $\ZW_S$ or $\ZH_S$ accordingly.
	We give the $(A)$ rule in $\ZH_R$ as an example:
	\begin{align}
		\phi : R                                 & \to S \qquad \half{} \in R            \\
		1                                        & \mapsto 1                             \\
		2                                        & \mapsto 2                             \\
		\therefore \phi : \half{}                & \mapsto \half{}                       \\
		\therefore \hat \phi : \node{grey}{\neg} & \mapsto \node{grey}{\neg}             \\
		\vc{\InputIfFileExists{./figures/ZH/average_rule.tikz}{}{Missing file!}}                & \mapsto \vc{\InputIfFileExists{./figures/ZH/average_rule_phi.tikz}{}{Missing file!}}
	\end{align}
	Any assignment of elements of $R$ to $a$ or $b$ in the rule $(A)$
	gives a valid instance of the rule $\hat \phi(A)$,
	because $\phi(a) \in S$,
	$\phi(b) \in S$,
	$\phi(2) = 2$ and $\phi\left(\frac{a + b}{2}\right)
		= \frac{\phi(a) + \phi(b)}{2}$.
\end{proof}

Phase ring homomorphisms shall reappear
in the chapter on Conjecture Inference (\S\ref{chapConjectureInference}),
where we note that these lifted homomorphisms correspond to certain symmetries
in a geometric space.
They then come up again in the Conjecture Verification chapter
(\S\ref{chapConjectureVerification}) where we
use Theorem~\ref{thmLiftRingHomSoundness} to `generate' new sound equations from old,
and in doing so construct sufficient evidence to then prove hypotheses.
Before moving on to conjecture synthesis we shall talk about algebras other than rings.

\section{Phase algebra homomorphisms and beyond} \label{secPhaseGroupHomomorphismsEtc}

While ZX is \emph{not} a phase ring graphical calculus
we can still consider phase group homomorphisms,
just as we did for phase ring homomorphisms earlier.
In this section we will classify the phase group homomorphism pairs
for the finite subgroups of $[0, 2\pi)$ that contain $\pi / 4$.
Note that completeness of these finite subgroups
has been shown in Ref.~\cite[Figure~4]{JPVNormalForm}.
We shall be considering the presentation of ZX
where the generators are green spiders with phases from our chosen fragment,
and Hadamard gates.
\begin{align}
	\set{\spider{gn}{\alpha}, \node{h}{}}_{\alpha \in G}
\end{align}

\begin{definition}[ZX phase homomorphism]  \label{defZXPhaseHomomorphism}
	For a given group homomorphism $\phi:A \to B$ we can define the $\ZX_A$ phase homomorphism $\hat \phi$ as
	\begin{align}
		\hat \phi : \spider{gn}{\alpha} & \mapsto \spider{gn}{\phi(\alpha)} \\
		\node{h}{}                      & \mapsto \node{h}{}
	\end{align}
\end{definition}

\begin{remark} \label{remPhaseGroupHomomorphismsNotNew}
	Something similar to phase group homomorphisms
	has already been used in the ZX calculus.
	Some of the early incompleteness proofs used `non-standard interpretations',
	such as in Refs~\cite[Lemma~8]{Duncan09}, \cite[Lemma~1.5]{Duncan13}
	or \cite[\S2]{SdW14}.
	These papers do not define something akin to our $\hat \phi$,
	but instead jump straight to defining and using $\interpret{\hat \phi(\cdot)}$.
	The idea of $\hat \phi$ is already apparent in the definition of $\idot{k}$ of Ref.~\cite{SdW14}.
\end{remark}

At the end of the chapter we will have the language for defining
the general case of a phase homomorphism
and phase homomorphism pairs
(Definition~\ref{defPhaseHomomorphismPair})
but for now we define a ZX phase homomorphism pair as:

\begin{definition}[ZX phase homomorphism pair]  \label{defPhaseGroupHomomorphismPair}
	A ZX phase homomorphism pair $(\phi, \tilde \phi)$ is a group homomorphism $\phi: A \to B$
	and a functor $\tilde \phi : \bit{R} \to \bit{S}$ such that the following diagram commutes:
	\begin{align}
		\begin{tikzcd}[ampersand replacement=\&] \label{eqnPhaseGroupHomomorphismPair}
			A \arrow[d,"\phi"] 	\&  ZX_A \arrow[d,"\hat \phi"] \arrow[r,"\idot{}"] 	\& \bit{R} \arrow[d,"\tilde \phi"] \\
			B 					\&  ZX_B \arrow[r,"\idot{}"] 						\& \bit{S}
		\end{tikzcd}
	\end{align}
\end{definition}

\begin{proposition} \label{propEndomorphismsOfFiniteSubgroupOfZeroTwoPi}
	The finite subgroups of $[0, 2\pi)$ under addition are cyclic groups
	of the form $< 2\pi i k/n >$,
	and have endomorphisms of the form
	\begin{align}
		\phi_j: 2 \pi i k/n \mapsto j \times 2 \pi i k/n
	\end{align}
	These endomorphisms lift to a phase homomorphism pair (Definition~\ref{defPhaseGroupHomomorphismPair})
	if and only if the map $\tilde \phi$
	\begin{align}
		\tilde \phi  : \bit{\bbZ[2^{-\frac{1}{2}},e^{2\pi i k / n}]} & \to \bit{\bbZ[2^{-\frac{1}{2}},e^{2\pi i k / n}]}                 \\
		e^{2 \pi i k/n}                                        & \mapsto e^{2 \pi i j k/n} \label{eqnZXPhaseGroupLiftToField}
	\end{align}
	is a well-defined ring homomorphism (when restricted to scalars) that fixes the subring $\bbZ[2^{-\frac{1}{2}}]$.
\end{proposition}

\begin{proof} \label{prfPropEndomorphismsOfFiniteSubgroupOfZeroTwoPi}
	The statements on the presentation of the finite subgroups and their
	endomorphisms follow from clearing denominators in the generators.
	The statement about lifting to phase homomorphisms is shown by noting that:
	\begin{itemize}
		\item $\tilde \phi$ must be a ring homomorphism when restricted to scalars
		\item For the green spider the diagram in \eqref{eqnPhaseGroupHomomorphismPair} commutes
		      for any $\phi$ that fixes $1$ and $0$,
		      but for the Hadamard node we additionally require $\tilde \phi$ to fix $-1$ and $\frac{1}{\sqrt{2}}$
	\end{itemize}
\end{proof}

\begin{example}[Clifford+T ZX phase homomorphism pairs] \label{exaPhaseGroupHomomorphisms}
	Consider the Clifford+T phase group $<e^{2\pi i/8}>$
	and its relationship to the field automorphisms $\sigma$ and $\tau$ from Example~\ref{exaLiftedRingHomSemantics}
	\begin{align}
		\sigma : i      & \mapsto -i        \\
		\tau : \sqrt{2} & \mapsto -\sqrt{2}
	\end{align}
	Note that $\tau$ does not fix $\sqrt{2}$.
	In Figure~\ref{figCliffordTPhaseGroupLift}
	we show all the endomorphisms of the Clifford+T phase group,
	the associated ring homomorphism $\tilde \phi$ from \eqref{eqnZXPhaseGroupLiftToField},
	and whether $(\phi, \tilde\phi)$ forms a ZX phase homomorphism pair.
	If no ring homomorphism is recorded in the column then the $\tilde \phi$ constructed
	is not a ring homomorphism at all.
	If $\tilde \phi$ is a valid ring homomorphism and also fixes $\sqrt{2}$,
	then $(\phi, \tilde \phi)$
	forms a phase homomorphism pair,
	and we record the effect of the pair in the final column.
	Note that this table shows that the only phase homomorphism pairs of
	the Clifford+T ZX phase group are the identity and complex conjugation (i.e. $\phi_1$ and $\phi_7$).
	\begin{figure}[h] \center
		\begin{tabular}{c l l}
			Group endomorphism & $\bbZ[i,\frac{1}{\sqrt{2}}]$ automorphism & Pair                \\
			\hline
			$\phi_0$           & --                                        & --                  \\
			$\phi_1$           & Identity                                  & Identity            \\
			$\phi_2$           & --                                        & --                  \\
			$\phi_3$           & $\tau$                                    & --                  \\
			$\phi_4$           & --                                        & --                  \\
			$\phi_5$           & $\sigma \tau$                             & --                  \\
			$\phi_6$           & --                                        & --                  \\
			$\phi_7$           & $\sigma$                                  & Complex conjugation \\
		\end{tabular}
		\caption{The group endomorphisms of the Clifford+T phases,
			details on whether they lift to phase group homomorphisms,\label{figCliffordTPhaseGroupLift}
			and the associated ring automorphism of $\bbZ[i, \frac{1}{\sqrt{2}}]$}
	\end{figure}
\end{example}

\begin{proposition} \label{propZXPhaseGroupHomomorphismsMod8}
	For a finite fragment of ZX that contains $\pi/4$, a group endomorphism $\phi_j$
	lifts to a phase homomorphism pair if and only if $j \equiv 1$ or $7$ modulo $8$.
\end{proposition}

\begin{proof} \label{prfPropZXPhaseGroupHomomorphismsMod8}
	Restrict $\phi_j$ to the subgroup $<\pi/4>$,
	and note that by Figure~\ref{figCliffordTPhaseGroupLift} $j$ must be $1$ or $7$ modulo $8$.
	Conversely if $j \equiv 1$ then $\phi_j$ fixes $\sqrt{2}$ and $i$,
	and if $j \equiv 7$ then $\phi_j$ moves $i$ but fixes $\sqrt{2}$ (as seen by the action on
	$\frac{1 + i}{\sqrt{2}}$),
	in both cases satisfying the conditions of \ref{propEndomorphismsOfFiniteSubgroupOfZeroTwoPi}.
\end{proof}

\begin{remark} \label{remClassifyingRingHomPairs}
We have already classified the phase ring homomorphisms
for $\ZW$, $\ZH$ and $\ring$ by showing
that every ring homomorphism lifts to a phase ring homomorphism pair.
\end{remark}

Looking now at the phase group version of Theorem~\ref{thmRingAutomorphismEntailment}
(phase ring homomorphisms preserve proofs)
we need a set of rules for these fragments of ZX, for which we take the rules
of Figure~1 of Ref.~\cite{JPVNormalForm},
along with the rule (cancel) given in their Definition~10. We call this set of rules `RZX'.

\begin{proposition} \label{propZXPhaseGroupProofPreservingCliffordT}
	For a finite fragment of ZX that contains $\pi/4$,
	a group endomorphism $\phi_j$
	lifts to a $RZX$-proof preserving map
	if and only if
	$j \equiv 1$ mod $8$
	\begin{align}
		RZX \syntactic A = B \quad \implies \quad RZX \syntactic \hat \phi A = \hat \phi B
	\end{align}
\end{proposition}

\begin{proof} \label{prfPropZXPhaseGroupProofPreservingCliffordT}
	Following the same structure as the proof of Theorem~\ref{thmRingAutomorphismEntailment}
	we note that the rules of RZX contain phases that are multiples of $\pi/4$,
	and in particular any $\phi_j$ with $j \equiv 7$ mod 8 will send the rule (BW)
	(shown in Figure~\ref{figBW}) to an equation that is not a rule.
	In contrast if $j \equiv 1$ mod 8 then the rules are preserved by $\hat \phi$.
	By Proposition~\ref{propZXPhaseGroupHomomorphismsMod8} the only suitable $\phi_j$
	are those where $j \equiv 1$ mod 8.
\end{proof}

\begin{figure}[ht]
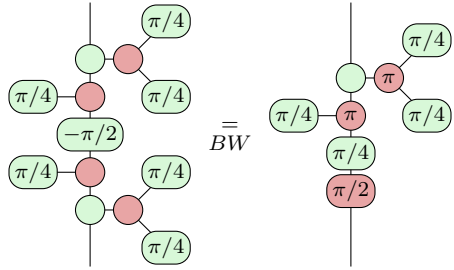

	\center
	\begin{align*}
		\vc{\InputIfFileExists{./figures/ZX/BW_l.tikz}{}{Missing file!}} \by{BW} \vc{\InputIfFileExists{./figures/ZX/BW_r.tikz}{}{Missing file!}}
	\end{align*}
	\caption{The rule (BW) from Ref.~\cite{JPVNormalForm}\label{figBW}}
\end{figure}

\begin{remark} \label{remOnlyFiniteSubgroupsZX}
	While the classification of ZX$_G$ phase homomorphism pairs
	in Proposition~\ref{propZXPhaseGroupHomomorphismsMod8}
	only applied to finite subgroups of $[0, 2\pi)$ containing $\piby{4}$,
	it is worth remembering that any given
	ZX diagram only contains finitely many phases,
	and so this result applies in all situations
	where the phases are rational multiples of $\pi$.
\end{remark}

\begin{remark} \label{remZXNotDeficient}
The result of Proposition~\ref{propZXPhaseGroupProofPreservingCliffordT}
does not imply that RZX is in any way deficient.
RZX is, after all, still sound, universal, and complete.
Rather this result highlights that the notation
of ZX does not match perfectly onto the available algebraic structure.
There is a sense in which phase ring homomorphisms `just work' for ZW, $\ring$, and ZH,
but the same does not apply to phase group homomorphisms for ZX.
\end{remark}

\begin{remark} \label{remUniversalAlgebras}
	It is possible to generalise the above work from
	rings and groups to $\Sigma$-algebras.
	$\Sigma$-algebras are at the heart of Universal Algebra,
	allowing us to represent algebraic concepts without
	restricting ourselves to a particular algebraic system.
	Part of our reason for not doing so from the start
	lies in the next chapter:
	ZQ is based on group structure
	but cannot be represented as a model of
	the $\Sigma$-algebra for groups (Corollary~\ref{corUQuatNotMonoid}).
\end{remark}

We present the general method of turning
a signature and equational theory into a
parameterised family of PROPs below,
as well as the general definition of a phase homomorphism pair
for $\Sigma$-algebras.
As just mentioned this framework
is not sufficient for all our needs,
which is why it is covered in such brief terms.
It is, however, a bridging between the concepts
of universal algebras and graphical calculi.

\begin{definition}[Signatures, terms, algebras, homomorphisms {\cite[Section~3]{Terms}}]  \label{defSignature} 
	\begin{itemize}
		\item A signature $\Sigma$ is a set of function symbols with fixed arities.
		      We write $f \in \Sigma^{(n)}$ if $f$ is a function symbol with arity $n$.
		\item A $\Sigma$-term is inductively defined using a set of variables:
		      Every variable is a term, and
		      for every function symbol $f\in\Sigma^{(n)}$ and
		      terms $t_1, \dots, t_n$ the application $f(t_1, \dots, t_n)$ is a term.
		\item A $\Sigma$-algebra $\cal{A}$ is a set $A$,
		      called the carrier set,
		      and an assignment of each function variable $f \in \Sigma^{(n)}$
		      to a function $f^{\cal{A}}: A^n \to A$.
		\item A $\Sigma$-homomorphism $\phi : \cal{A} \to \cal{B}$ is a function
		      on the carrier sets $A \to B$ such that
		      \begin{align}
			      \phi(f^{\cal{A}}(x_1, \dots, x_n)) = f^{\cal{B}}(\phi(x_1), \dots \phi(x_n))
		      \end{align}
	\end{itemize}
\end{definition}

With these definitions we can form a PROP analogue of the same concept.
This is very similar to the idea of a \emph{term tree} \cite[Figure~3.1]{Terms}
except that the PROP structure is also present,
allowing arbitrary horizontal products and wire permutations.

\begin{definition}[PROP form of a $\Sigma$-algebra]  \label{defPROPUniversalAlgebra}
	Given a signature $\Sigma$
	we create the PROP $\Sigma\bbP$ generated by the morphisms
	\begin{align}
		\nary{f} \quad : n \to 1 \quad \forall f \in \Sigma^{(n)}
	\end{align}
	For any $\Sigma$-algebra $\cal{A}$ with carrier set $A$
	and function assignments $f^{\cal{A}}$ we
	generate the PROP $\Sigma\bbP_\cal{A}$ by also including
	the diagrams
	\begin{align}
		\state{white}{a} \quad \forall a \in A
	\end{align}
	Given a set of variables $X$ (symbols distinct from $A$ and $\Sigma$)
	we likewise represent them as
	\begin{align}
		\state{white}{x} \quad \forall x \in X
	\end{align}
	For each function symbol $f$ we have the rewrite rule $app_f$
	which applies the function $f^{\cal{A}}$ to its inputs:
	\begin{align}
		\vc{\InputIfFileExists{./figures/wire/apply_f.tikz}{}{Missing file!}} \by{app_f} \state{white}{f^{\cal{A}}(x_1, \dots, x_n)}
	\end{align}
	Any $\Sigma$-homomorphism $\phi : \cal{A} \to \cal{A'}$ lifts to a morphism
	$\bar \phi : \Sigma\bbP_A \to \Sigma\bbP_B$
	of $\Sigma$-PROPs which is the identity on function symbols
	and sends
	\begin{align}
		\state{white}{a} \mapsto \state{white}{\phi(a)}
	\end{align}
	Note that any equational theory
	one wishes to apply to a signature can be applied \emph{inside} the phases;
	i.e. if $E \syntactic a = b$ then we would consider the following to be
	an isomorphism of diagrams:
	\begin{align}
		E \syntactic a & = b & \therefore \state{white}{a} & \iso \state{white}{b}
	\end{align}
\end{definition}

Definition~\ref{defRingPropS}, where we defined $\rprop$, is an example of such a $\Sigma$-PROP.
As we did for $\rprop$ we can define a \emph{model} of $\Sigma\bbP_{\cal{A}}$
as a strict symmetric monoidal functor $\Model_{\cal{A}}$ out of $\Sigma\bbP_{\cal{A}}$,
and form the following category:

\begin{definition}[$\Sigma$ Graphical Calculi]  \label{defSigmaGraphicalCalculi}
	The category of $\Sigma$ Graphical Calculi
	over the category $\cal{C}$, written $\SigmaGC{\cal{C}}$,
	is the category given by:
	\begin{itemize}
		\item Objects are pairs $(G, \Model_{\cal{A}})$
		      where $G$ is a (compact closed) graphical calculus
		      given as a collection of generators,
		      an interpretation $\idot{G} : G \to \cal{C}$, and rewrite rules,
		      and $\Model_{\cal{A}}$ is a model of $\Sigma\bbP_\cal{A}$
		\item Morphisms
		      \begin{align}
			      f: (G, \Model_{\cal{A}}) \to (G', \Model_{\cal{A'}})
		      \end{align}
		      are given by a strict symmetric monoidal functor $f_\bbD : G \to G'$
		      as well as a $\Sigma$-homomorphism $\phi : \cal{A} \to \cal{A'}$ such that
		      the right hand diagram commutes:
		      \begin{align}
			      \begin{tikzcd}[ampersand replacement=\&]
				      \cal{A} \arrow[d,"\phi"] \& \Sigma\bbP_{\cal{A}} \arrow[d,"\bar \phi"] \arrow[r,"\Model_\cal{A}"] \& G \arrow[d,"f_\bbD"] \\
				      \cal{A'} \& \Sigma\bbP_{\cal{A'}} \arrow[r,"\Model_{\cal{A}'}"] \& G'
			      \end{tikzcd}
		      \end{align}
	\end{itemize}
\end{definition}

The category $\PRGC{\bit{R}}$ given in Definition~\ref{defPRGCR}
is an example of a category of $\Sigma$ Graphical Calculi,
where $\Sigma$ is the usual signature for rings.
Our final definitions of this chapter
are the general definitions for phase homomorphisms
and phase homomorphism pairs.
This extends the notion first seen in Definition~\ref{defPhaseRingHomomorphism}
of finding an algebra homomorphism $\phi$
that lifts to a map $\hat \phi$ `acting on phases' for diagrams
and also commutes with the interpretation via some third map $\tilde \phi$.

\begin{definition}[Phase homomorphism]  \label{defPhaseAlgebraHomomorphism}
	The phase homomorphism for a $\Sigma$-homo\-morphism $\phi : \cal{A} \to \cal{A'}$
	and objects $G_\cal{A}$ and $G'_\cal{A'}$
	is the $\SigmaGC{\cal{C}}$ morphism   $G_\cal{A} \to G'_\cal{A'}$ (if it exists) given by $\phi$ and $f_\bbD$ such that:
	\begin{itemize}
		\item $f_\bbD$ acts as the identity on those vertices not labelled by symbols from $\Sigma$ or $A$
		\item Otherwise the action of $f_\bbD$ on the label is:
		      \begin{align}
			      a & \mapsto \phi(a) & a & \in A            \\
			      f & \mapsto f       & f & \in \Sigma^{(n)}
		      \end{align}
	\end{itemize}
	We refer to such a $f_\bbD$ as $\hat \phi$.
\end{definition}

\begin{definition}[Phase Homomorphism Pair]  \label{defPhaseHomomorphismPair}
	A phase $\Sigma$-homomorphism pair
	between $G_{\cal{A}}$ and $G'_{\cal{A'}}$ in $\SigmaGC{\cal{C}}$
	is a pair of maps $(\phi, \tilde \phi)$
	such that:
	\begin{itemize}
		\item $\phi: \cal{A} \to \cal{A'}$ is a $\Sigma$-homomorphism
		\item The phase homomorphism $\hat \phi : G \to G'$ exists
		\item $\tilde \phi : \cal{C} \to \cal{C}$ is a strict symmetric monoidal functor
		\item The following diagram commutes:
		      \begin{align}
			      \begin{tikzcd}[ampersand replacement=\&] \label{eqnPhaseHomomorphismPair}
				      \cal{A} \arrow[d,"\phi"] 	\&  G \arrow[d,"\hat \phi"] \arrow[r,"\idot{G}"] 	\& \cal{C} \arrow[d,"\tilde \phi"] \\
				      \cal{A'} 					\& G' \arrow[r,"\idot{G'}"] 						\& \cal{C}
			      \end{tikzcd}
		      \end{align}
	\end{itemize}
\end{definition}

The commutative square of Definition~\ref{defPhaseHomomorphismPair} gives us the soundness-preserving property
used in Theorem~\ref{thmLiftRingHomSoundness}.

\section{Summary}

The most important result
from this chapter (from the point of view of the later chapters)
will be that of phase ring homomorphisms
preserving soundness (Theorem~\ref{thmLiftRingHomSoundness}),
which for the ring calculi we have considered can be paraphrased as
`theorems are closed under phase homomorphism'.
These pairs only arise
because we view phase algebras as algebras and not sets.
We shall rephrase these pairs
as symmetries of a geometric space in \S\ref{chapConjectureInference}.

We also showed that $\ring_\bbC$
acts as a form of unifier for phase ring qubit graphical calculi.
While different choices of generators will yield graphical calculi
suited for different situations it is also important to know
that any phase ring qubit graphical calculus will have a
relationship to $\ring_\bbC$ as described in Remark~\ref{remRingKGeneric}.

In the last subsection we also provided a classification of phase group homomorphism pairs
for the finite fragments of ZX containing $\piby{4}$.
Future work would include extending this classification of
phase homomorphism pairs to other fragments and calculi,
but also the investigation of any other category of $\Sigma$ Graphical Calculi.  \thispagestyle{empty} 
\chapter{The Graphical Calculus ZQ} \label{chapZQ}
\thispagestyle{plain}

\noindent\emph{In this chapter:}
\begin{itemize}
	\item[\chapterbullet] We introduce the graphical calculus ZQ
	\item[\chapterbullet] We justify the structural difference between ZQ and ZX
	\item[\chapterbullet] We prove completeness of ZQ
	\item[\chapterbullet] This chapter touches on the ideas of \S\ref{chapPhaseRingCalculi}, but stands alone, and is not required for the understanding of later chapters.
\end{itemize}

The ZX calculus is built from the Z and X classical structures of quantum computing \cite{Coecke08}.
Even in that earliest paper the Z `phase shift' is illustrated as a rotation of the Bloch Sphere \cite[\S 4]{Coecke08}.
By the time of \cite{Backens16Thesis}, eight years later,
language has changed to that of Z `rotations' or `angles' \cite[Lemma~3.1.7]{Backens16Thesis},
and explicit use is made of the Euler Angle Decomposition result;
that any rotation in \SO3 can be broken down into rotations about the Z then X then Z axes.
The idea behind ZQ is to represent not just the Z and X rotations of the Bloch Sphere,
but represent arbitrary rotations via unit-length quaternions.
In order to do so we shall have to fundamentally alter the structure of how these diagrams are constructed:
Composition of rotations along different axes is not commutative, something that the spider-mediated monoid of ZX cannot model
(\S\ref{secZQRecapAndStructure}).
While there have been graphical calculi with phase groups outside of $([0,2\pi),+)$ before,
such as Backens' calculus representing Spekkens' toy bit theory \cite{Spekkens},
to the author's knowledge the language ZQ is the first with a non-commutative phase group.

To emphasise this idea further:
Much study has been done on fragments of ZX \cite{JPVNormalForm},
but these fragments all have a phase group that sits inside the unit circle
of the complex plane.
Likewise ZW$_\bbC$ and ZH$_\bbC$
have phase rings of the entire complex plane.
Indeed we showed in Proposition~\ref{propSubringC}
that a phase ring for a qubit graphical calculus must inject into $\bbC$.
ZQ is therefore interesting as a qubit graphical calculus
with a phase algebra that does not neatly fit into the complex numbers,
quite apart from its practical uses.

The Bloch Sphere, which we cover in more detail in \S\ref{secZQRecapAndStructure}, is not a perfect analogy.
Although it provides us with useful intuition and a way to consider a single qubit in real Euclidean space,
its group of rotations, \SO3, is a subgroup of the group of unitary evolutions, \SU2,
which the standard model of quantum computing actually uses \cite{NielsonChuang}.
The group \SU2 itself is isomorphic to the group of unit-length quaternions,
and it is this group of quaternions that we shall use as rotations,
giving us the `Q' in ZQ.
This use of quaternions to represent rotations is nothing new in the world of engineering and computer graphics \cite{Shoemake},
but has recently surfaced as a useful component of intermediate representations for quantum circuits.
Intermediate representations sit between the user's specification of an algorithm
and the actual implementation on a specific piece of hardware.
The system TriQ \cite{TRIQ}
provides such an intermediate representation, targeting existing quantum computers run by IBM, Rigetti, and the University of Maryland.
The authors claim a speed-up in execution of their benchmarks on the seven quantum computers considered,
in part because of TriQ's use of quaternions in the optimisation process \cite[\S4]{TRIQ}:
Any sequence of single qubit rotations can be combined into just one quaternion,
then decomposed into the most efficient sequence of gates for the target hardware architecture.
Our aims in making ZQ are the following:
\begin{itemize}
  \item Construct a graphical calculus that expresses all qubit rotations
  \item Construct a qubit graphical calculus with a non-commutative phase group
  \item Provide a complete ruleset that can condense sequences of single qubit rotations
  \item Provide a complete ruleset for a qubit phase group calculus that only uses linear side conditions
\end{itemize}

This last point deserves extra attention:
The (EU') rule of ZX, given with the rest of the Universal ZX rules in Figure~\ref{figVilmartZXRules},
has side conditions that are fiddly to compute,
making them unsuitable for researchers to use or remember,
and extremely hard to find through conjecture synthesis
(we explore this idea further in \S\ref{secInferringPhaseVariables}).
Neither of these objections apply to ZQ,
with the Euler Angle composition captured neatly by the multiplication of quaternions.
Before we give the definition of ZQ we first give a brief overview of
the Bloch Sphere, the groups \SU2 and \SO3, and unit quaternions.

\section{Quaternions, Rotations and the Bloch Sphere}
\label{secZQRecapAndStructure}

In this section we will introduce the Bloch Sphere and quaternions,
then demonstrate the incompatibility of the group of unit-length quaternions with the spider notation of ZX.

\begin{definition}[Bloch Sphere {\cite[Figure~1.3]{NielsonChuang}}]  \label{defBlochSphere} 
	A qubit state, up to global phase, can be represented as
	\begin{align}
		\ket{\psi} = \cos \half[\theta] \ket{0} + e^{i \phi} \sin \half[\theta] \ket{1}
	\end{align}
	The parameters $\theta$ and $\phi$ describe a point on the unit sphere in $\bbR^3$,
	known as the Bloch Sphere.
	The group of rotations of this sphere is \SO3.
\end{definition}

\begin{definition}[Quaternions {\cite[p12]{Hazewinkel2004algebras}}]  \label{defQuaternions} 
	The quaternions are a non-commutative, four-dimensional, real algebra:

	\begin{align}
		  & \bbR + i\bbR + j \bbR + k \bbR \\
		  & i^2 = j^2 = k^2 = ijk = -1
	\end{align}

	For ZQ we are only interested in unit-length quaternions, forming the group \UQuat\ under multiplication.
\end{definition}

\begin{remark} \label{remQuaternionsSU2}
	The group \UQuat\ is isomorphic with \SU2, via the isomorphism:

	\begin{align}
		\phi: \UQuat             & \to \SU2     \label{eqnUQuatSU2}  \\
		q_w + iq_x + jq_y + kq_z & \mapsto \begin{pmatrix}
			q_w - iq_z & -q_y + iq_x \\
			-q_y-iq_x  & q_w + iq_z
		\end{pmatrix}
	\end{align}
	The fact that the groups are isomorphic is well known,
	but we include a proof that this particular map is an isomorphism in Proposition~\ref{propZQPhi}.
\end{remark}

There is another way to represent unit-length quaternions, and that is by an angle and a unit vector.
It is important to note that this is not the same thing as `an angle rotation along a unit vector';
the angle-vector pair $(\alpha, \hat v)$ and the angle-vector pair $(-\alpha, -\hat v)$ are different as pairs,
but would constitute the same rotation. This, in fact, describes the relationship between \UQuat\ and \SO3:

\begin{definition}[Quaternions as rotations]  \label{defUQUatSO3}
	There is a canonical homomorphism from \UQuat\ to \SO3, given by
	\begin{align}
		(\alpha, v)  & := \cos \half[\alpha] + \sin\half[\alpha](iv_x + jv_y + kv_z) \\ \nonumber \\
		\psi: \UQuat & \to \SO3                                                      \\
		(\alpha,v)   & \mapsto \text{ rotation by angle $\alpha$ along vector $v$}   \\ \nonumber \\
		\ker \psi    & = \set{1, -1}
	\end{align}
\end{definition}

\begin{remark} \label{remQuaternionPauli}
	This presentation of quaternions as angle-vector pairs is directly linked to the Pauli matrices and the Hadamard map
	via the homomorphism $\phi$ of \eqref{eqnUQuatSU2},
	up to a scalar factor of $-i$:

	\begin{align}
		\phi((\pi, x))                         & =
		-i
		\begin{pmatrix}
			0 & 1 \\
			1 & 0
		\end{pmatrix}
		\\
		\phi((\pi, y))                         & =
		-i
		\begin{pmatrix}
			0  & i \\
			-i & 0
		\end{pmatrix}
		\\
		\phi((\pi, z))                         & =
		-i
		\begin{pmatrix}
			1 & 0  \\
			0 & -1
		\end{pmatrix}
		\\
		\phi((\pi, \frac{1}{\sqrt{2}}(x + z))) & =
		\frac{-i}{\sqrt{2}}
		\begin{pmatrix}
			1 & 1  \\
			1 & -1
		\end{pmatrix}
	\end{align}
	Generally
	\begin{align}
		\phi((\pi, v)) & = -i
		\begin{pmatrix}
			v_z         & v_x + i v_y \\
			v_x - i v_y & - v_z
		\end{pmatrix}
	\end{align}

\end{remark}

\begin{definition}[The H quaternion]  \label{defQuaternionH}
	The quaternion $(\pi, \frac{1}{\sqrt{2}}(x + z))$ will often just be referred to as $H$.
	It corresponds to the rotation that maps the Z axis to the X axis and vice versa.
\end{definition}

\begin{example}[Quaternions in TriQ] \label{exaQuaternionsInTriQ}
The compiler TriQ uses quaternions as rotations as part of its optimisation process\footnote{1Q: One qubit, 2Q: Two qubits, IR: Intermediate representation}:
\begin{displayquote}[Full-Stack, Real-System Quantum Computer Studies:
Architectural Comparisons and Design Insights \cite{TRIQ}]
Since 1Q operations are rotations, each
1Q gate in the IR can be expressed using a unit rotation quaternion
which is a canonical representation using a 4D complex number.
TriQ composes rotation operations by multiplying the corresponding
quaternions and creates a single arbitrary rotation. This rotation is
expressed in terms of the input gate set. Furthermore, on all three
vendors, Z-axis rotations are special operations that are implemented
in classical hardware and are therefore error-free. TriQ expresses
the multiplied quaternion as a series of two Z-axis rotations and one
rotation along either X or Y axis, thereby maximizing the
number of error-free operations.
\end{displayquote}
We shall explicitly construct this decomposition of a quaternion
into a Z-X-Z rotation in Proposition~\ref{propZQQDecomp}
when we explore how to translate from ZQ to ZX.
\end{example}

Unit-length quaternions form a very different group to the rotations of just the Z axis.
We will now show that the usual presentation of ZX (spiders and states)
is fundamentally incompatible with the group \UQuat.
In order to do so we introduce the notion of a generic monoid,
similar to how in Definition~\ref{defRingPropS}
we constructed a generic PROP for rings, $\rprop_R$.

\begin{definition}[Monoid {\cite[p2]{Maclane2013}}]  \label{defMonoid} 
	A monoid is a set $M$ equipped with a binary operation:
	\begin{align}
		\mu : M \times M & \to M
	\end{align}
	referred to as the multiplication,
	which is required to be associative.
	There is also a unit, $e$,
	which is the left and right unit for $\mu$.
	We can also view $e$ as the image of a function $\eta : 1 \to M$.
	Graphically we represent these as:
	\begin{align}
		\mu := & \binary[white]{} & \eta := & \state{white}{}
	\end{align} such that the following hold:
	\begin{align}
		\vc{\InputIfFileExists{./figures/monoid/assoc.l.tikz}{}{Missing file!}} & =\vc{\InputIfFileExists{./figures/monoid/assoc.r.tikz}{}{Missing file!}} &
		\vc{\InputIfFileExists{./figures/monoid/unitl.tikz}{}{Missing file!}} = \node{none}{} &=\vc{\InputIfFileExists{./figures/monoid/unitr.tikz}{}{Missing file!}}
	\end{align}
\end{definition}

\begin{definition}[A generic monoid PROP]  \label{defGenericMonoid}
	We construct the generic monoid PROP, $\monprop_M$,
	for the monoid $(M, \mu, e)$ as the PROP generated by
	\begin{align}
		\binary[white]{}\quad \text{and} \quad \set{\state{white}{a}}_{a \in M}
	\end{align}
	With the following rewrite rules:
	\begin{align}
		\vc{\InputIfFileExists{./figures/monoid/unital.tikz}{}{Missing file!}} \by{unit_l} \state{white}{a} & \by{unit_r} \vc{\InputIfFileExists{./figures/monoid/unitar.tikz}{}{Missing file!}}    \\
		\vc{\InputIfFileExists{./figures/monoid/assocl_abc.tikz}{}{Missing file!}}                          & \by{assoc} \vc{\InputIfFileExists{./figures/monoid/assocr_abc.tikz}{}{Missing file!}}
	\end{align}
	A model for $\monprop_M$
	in a graphical calculus $G$ is
	a strict symmetric monoidal functor $\monprop_M \to G$ (just as in Definition~\ref{defModelOfRPROP}).
\end{definition}

\begin{proposition} \label{propQubitModelsCommutative}
	Every faithful model of $\monprop_M$ into a graphical calculus over \Qubit\ is commutative.
\end{proposition}

\begin{proof} \label{prfPropQubitModelsCommutative}
	Taking a generic model of $\monprop_M$
	we look at the interpretation of the image of the generators of $\monprop_M$.
	For brevity we will just write $\interpret{D}$
	for $D$ a diagram in $\monprop_M$ to mean the interpretation
	of the image of $D$ in the model.
	We proceed by looking at the span of the interpretations of the elements of $M$.

	\begin{align}
		W & := span \set{\interpret{\state{white}{e}},\interpret{\state{white}{a}},\interpret{\state{white}{a'}}, \dots}
	\end{align}
	\begin{itemize}
		\item If $\dim W = 0$ then the monoid has only one element, $e$,
		      and so is commutative.
		\item If $\dim W > 0$ then
		      there either $M = \set{e}$ (and so commutative),
		      or there is some other element $a \in M$.
		      This implies that $\interpret{\state{white}{e}} \neq 0$:
		      \begin{align}
			      \text{assume} \quad \interpret{\state{white}{e}} & = 0                                                                       \\
			      \therefore \interpret{\vc{\InputIfFileExists{./figures/monoid/md_0.tikz}{}{Missing file!}}}     & = 0 \quad \text{(any element $a$)}                                        \\
			      \interpret{\vc{\InputIfFileExists{./figures/monoid/md_0.tikz}{}{Missing file!}}}                & = \interpret{\state{white}{a}}  \quad  \text{($e$ is the unit for $m$)}   \\
			      \therefore \interpret{\state{white}{a}}          & =\interpret{\state{white}{e}}  \contradiction \quad \text{(faithfulness)}
		      \end{align}
		\item If $\dim W = 1$ then without loss of generality:
		      \begin{align}\interpret{\state{white}{e}}                 & = \begin{pmatrix}
				      1 \\ 0
			      \end{pmatrix}                                                       \\
			      \therefore \interpret{\state{white}{a}}      & = \begin{pmatrix}
				      \lambda_a \\ 0
			      \end{pmatrix} \; \forall a\; \quad \text{ some } \lambda_a \in \bbC \\
			      \interpret{\vc{\InputIfFileExists{./figures/monoid/md_0.tikz}{}{Missing file!}}}            & = \interpret{\state{white}{a}}  \quad  \text{$e$ is the unit for $m$}              \\
			      \therefore \interpret{\binary[white]{}}      & = \begin{pmatrix}
				      1 & \cdot & \cdot & \cdot \\
				      0 & \cdot & \cdot & \cdot
			      \end{pmatrix} \text{ where $\cdot$ represents unknowns}             \\
			      \therefore \interpret{\vc{\InputIfFileExists{./figures/monoid/md_1.tikz}{}{Missing file!}}} & = \begin{pmatrix}
				      \lambda_a \lambda_b \\ 0
			      \end{pmatrix}= \interpret{\vc{\InputIfFileExists{./figures/monoid/md_2.tikz}{}{Missing file!}}}
		      \end{align}
		\item if $\dim W = 2$ then the states span all of $\bbC^2$:
		      \begin{align}
			      \interpret{\state{white}{e}}                 & = \begin{pmatrix}
				      1 \\ 0
			      \end{pmatrix}        \quad \text{w.l.o.g}                                                                \\
			      \vc{\InputIfFileExists{./figures/monoid/md_0.tikz}{}{Missing file!}}                        & \by{unit_l} \state{white}{a} \quad \forall a                                                                            \\
			      \therefore \interpret{\vc{\InputIfFileExists{./figures/monoid/m_2.tikz}{}{Missing file!}}}  & = \interpret{\node{none}{}} = \interpret{\vc{\InputIfFileExists{./figures/monoid/m_1.tikz}{}{Missing file!}}} \quad \text{since } span \set{\state{white}{a}} = \bbC^2 \\
			      \therefore \interpret{\binary[white]{}}      & = \begin{pmatrix}
				      1 & 0 & 0 & \cdot \\
				      0 & 1 & 1 & \cdot
			      \end{pmatrix} \text{ where $\cdot$ represents unknowns}                                                  \\
			      \therefore \interpret{\vc{\InputIfFileExists{./figures/monoid/m_3.tikz}{}{Missing file!}}}  & = \interpret{\binary[white]{}}                                                                                          \\
			      \therefore \interpret{\vc{\InputIfFileExists{./figures/monoid/md_1.tikz}{}{Missing file!}}} & = \interpret{\vc{\InputIfFileExists{./figures/monoid/md_2.tikz}{}{Missing file!}}}
		      \end{align}
	\end{itemize}
\end{proof}

\begin{corollary} \label{corUQuatNotMonoid}
	There is no qubit graphical calculus
	that faithfully models $\monprop_M$ (for any $M$)
	such that this monoid action is isomorphic to the group multiplication of \UQuat,
	because \UQuat\ is not commutative.
\end{corollary}

\begin{remark} \label{remUQuatNotAsSpiders}
	Since $\UQuat$ cannot be modelled as a monoid over states,
	it cannot obey the fundamental spider rules of \cite{Coecke08}.
	We therefore note that the presentation of ZQ (when we come to it)
	will be different by necessity rather than whim.
\end{remark}

We will deviate from the structure of ZX by having \emph{labelled, directed edges} carry the phase group of ZQ,
but keep phase-free Z spiders to mediate the graph structure.
Although ZQ is presented here in the form of a PROP we encourage the reader to think
of ZQ diagrams as directed, multi, open graphs, where each directed edge is labelled by a unit-length quaternion.
A different approach to this monoid-based presentation of groups
(as suggested by our general notion of
graphical calculi based on $\Sigma$-algebras in \S\ref{secPhaseGroupHomomorphismsEtc})
is to instead consider a group as a category with one object and where all morphisms
are invertible.
This point of view lends us to the following definition:

\begin{definition}[Group PROP with vertical composition]  \label{defGroupPROP}
	We construct the (vertically composed) group PROP, $\groupprop_G$,
	for a group $(G, \cdot)$
	as the PROP with generators
	\begin{align}
		\node{white}{g} \quad \forall g \in G
	\end{align}
	It additionally satisfies the rewrite rules
	\begin{align}
		\vc{\begin{tikzpicture}
	\begin{pgfonlayer}{nodelayer}
		\node [style=white] (0) at (0, 0) {f};
		\node [style=none] (4) at (0, 0.5) {};
		\node [style=white] (5) at (0, -0.5) {g};
		\node [style=none] (6) at (0, -1) {};
	\end{pgfonlayer}
	\begin{pgfonlayer}{edgelayer}
		\draw (4.center) to (0);
		\draw (0) to (5);
		\draw (5) to (6.center);
	\end{pgfonlayer}
\end{tikzpicture}
} \by{M} \node{white}{f \cdot g}
	\end{align}
\end{definition}

Both ZQ and ZX will be models of this PROP;
in each case the image of the morphisms in $\groupprop_G$
are the single qubit rotations of the phase group.
For ZQ this group is \SU2, for ZX this group is $[0,2\pi)$ under addition.

\section{The definition of the ZQ-calculus}

We present the graphical calculus ZQ as a compact closed PROP
generated by the morphisms in Figure~\ref{figZQGenProp}
and then present the interpretation of these generators in Figure~\ref{figZQInterpretation}.
We build the transpose of the $Q_q$ node in the usual way, as shown in Figure~\ref{figQTranspose}.
A unit-length quaternion $q$ can be expressed as an element of $\bbR^4$, i.e. $q_w + i q_x + j q_y + k q_z$,
or as an angle-vector pair, e.g. $(\alpha, v)$ (see \S\ref{secZQRecapAndStructure}).
Finally we give the rules of ZQ in Figure~\ref{figZQRules}.

\begin{definition}[ZQ]  \label{defZQ}
	The graphical calculus ZQ is formed by:
	\begin{itemize}
		\item The generators of Figure~\ref{figZQGenProp}
		\item The interpretation of Figure~\ref{figZQInterpretation}
		\item The rules of Figure~\ref{figZQRules}
	\end{itemize}
\end{definition}

\begin{theorem}[ZQ is sound] \label{thmZQSound}
	The rules of ZQ (Figure \ref{figZQRules}) are sound with respect to the standard interpretation
	(Figure~\ref{figZQInterpretation})
\end{theorem}

\begin{proof}
	\label{prfThmZQSound}
	This proof is covered in \S \ref{secZQSound}
\end{proof}

\begin{theorem}[ZQ is complete] \label{thmZQComplete}
	ZQ with the rules of Figure \ref{figZQRules} is complete with respect to the standard interpretation
	(Figure~\ref{figZQInterpretation})
\end{theorem}

\begin{proof}
	\label{prfThmZQComplete}
	This proof is covered in \S \ref{secZQComplete},
	and is performed by an equivalence with the ZX calculus.
\end{proof}

Diagrams of ZQ are elements of the compact closed PROP generated by the
morphisms in Figure~\ref{figZQGenProp}.
We add to these generators a shorthand for the transpose of $Q_q$, written as a trapezium pointing downwards
(as in Figure~\ref{figQTranspose}),
and defined using the cup and cap.
The interpretations of these generators are given both
in Figure~\ref{figZQInterpretation},
or alternatively using angle-vector pair notation (Definition~\ref{defUQUatSO3}) we have the interpretation given in Figure~\ref{figZQInterpretAngleVectorPair}.

\begin{figure}
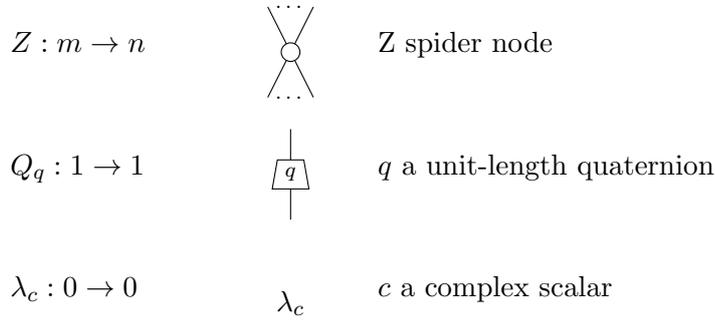

	\begin{center}
		\renewcommand{\arraystretch}{3}
		\tabcolsep=10pt
		\begin{tabular}[h!]{m{1in} c m{2in}}
			$Z : m \to n$         & $\spider{smallZ}{}$ & Z spider node                \\
			$Q_{q} : 1 \to 1$     & $\node{qn}{q}$      & $q$ a unit-length quaternion \\
			$\lambda_c : 0 \to 0$ & \raisebox{-0.5em}{$\lambda_c$}         & $c$ a complex scalar
		\end{tabular}
	\end{center}
	\caption{\label{figZQGenProp}The generators of ZQ as a PROP}
\end{figure}

\begin{figure}
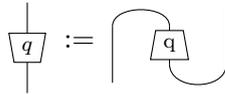

	\begin{align*}
		\node{qnt}{q} & :=  \; \vc{\InputIfFileExists{./figures/ZQ/q_y.tikz}{}{Missing file!}}
	\end{align*}
	\caption{The transpose of $q$ in ZQ\label{figQTranspose}}
\end{figure}

\begin{figure}
	\begin{center}
		\begin{align*}
			\interpret{\spider{smallZ}{}} & =
			\begin{pmatrix}
				1      & 0 & \dots  &   & 0      \\
				0      & 0 & \dots  &   & 0      \\
				\vdots &   & \ddots &   & \vdots \\
				0      &   & \dots  & 0 & 0      \\
				0      &   & \dots  & 0 & 1
			\end{pmatrix} \\[\rowgap]
			\interpret{\node{qn}{q}}      & =
			\begin{pmatrix}
				q_w - iq_z & -q_y + iq_x \\
				-q_y-iq_x  & q_w + iq_z
			\end{pmatrix} \\[\rowgap]
			\interpret{\lambda_c}         & =
			c \\[\rowgap]
			\interpret{\dcup}             & = \begin{pmatrix}
				1 \\ 0 \\ 0 \\ 1
			\end{pmatrix} \\[\rowgap]
			\interpret{\dcap}             & = \begin{pmatrix}
				1 & 0 & 0 & 1
			\end{pmatrix}
		\end{align*}
	\end{center}
	\caption{Interpretations of the generators of ZQ \label{figZQInterpretation}}
\end{figure}

\begin{figure}
	\begin{center}
\begin{align*}
	\interpret{\node{qn}{(\alpha, v)}} & =
	\begin{pmatrix}
		\cos \frac{\alpha}{2} - i \sin \frac{\alpha}{2} v_z & -i \sin\frac{\alpha}{2} (v_x + i v_y)                \\
		- i \sin\frac{\alpha}{2} (v_x - i v_y)              & \cos{\frac{\alpha}{2}} + i \sin\frac{\alpha}{2}{v_z}
	\end{pmatrix}  \label{eqnZQInterpretAlpha}
\end{align*}
\end{center}
	\caption{Interpretation of the Q generator using angle-vector pair notation \label{figZQInterpretAngleVectorPair}}
\end{figure}

\FloatBarrier
In order to decrease diagrammatic clutter we shall use the following notation:
\begin{align}
	\vc{\begin{tikzpicture}
	\begin{pgfonlayer}{nodelayer}
		\node [style=none] (6) at (0, 0.75) {};
		\node [style=none] (7) at (0, -0.25) {};
	\end{pgfonlayer}
	\begin{pgfonlayer}{edgelayer}
		\draw [style=hadamard edge] (6.center) to (7.center);
	\end{pgfonlayer}
\end{tikzpicture}
} & := \node{trap}{(\pi, \frac{1}{\sqrt{2}}(x + z))} = \node{trap}{H}
\end{align}
This is the familiar `Hadamard edge' from e.g. \cite{Duncan2020}.
Note that the Hadamard edge is symmetrical,
but the $H$ quaternion edge decoration is not.
As soon as we have the tools to show that
this is well defined we shall do so (Lemma~\ref{lemHadamardEdgeWellDefinedZQ}).

\subsection{The rules of ZQ}
\label{secZQRules}

We present in Figure~\ref{figZQRules} a ruleset of ZQ.
The soundness of this ruleset is proved in subsection \ref{secZQSound}
and completeness is proved in section \ref{secZQComplete}.

\begin{figure}
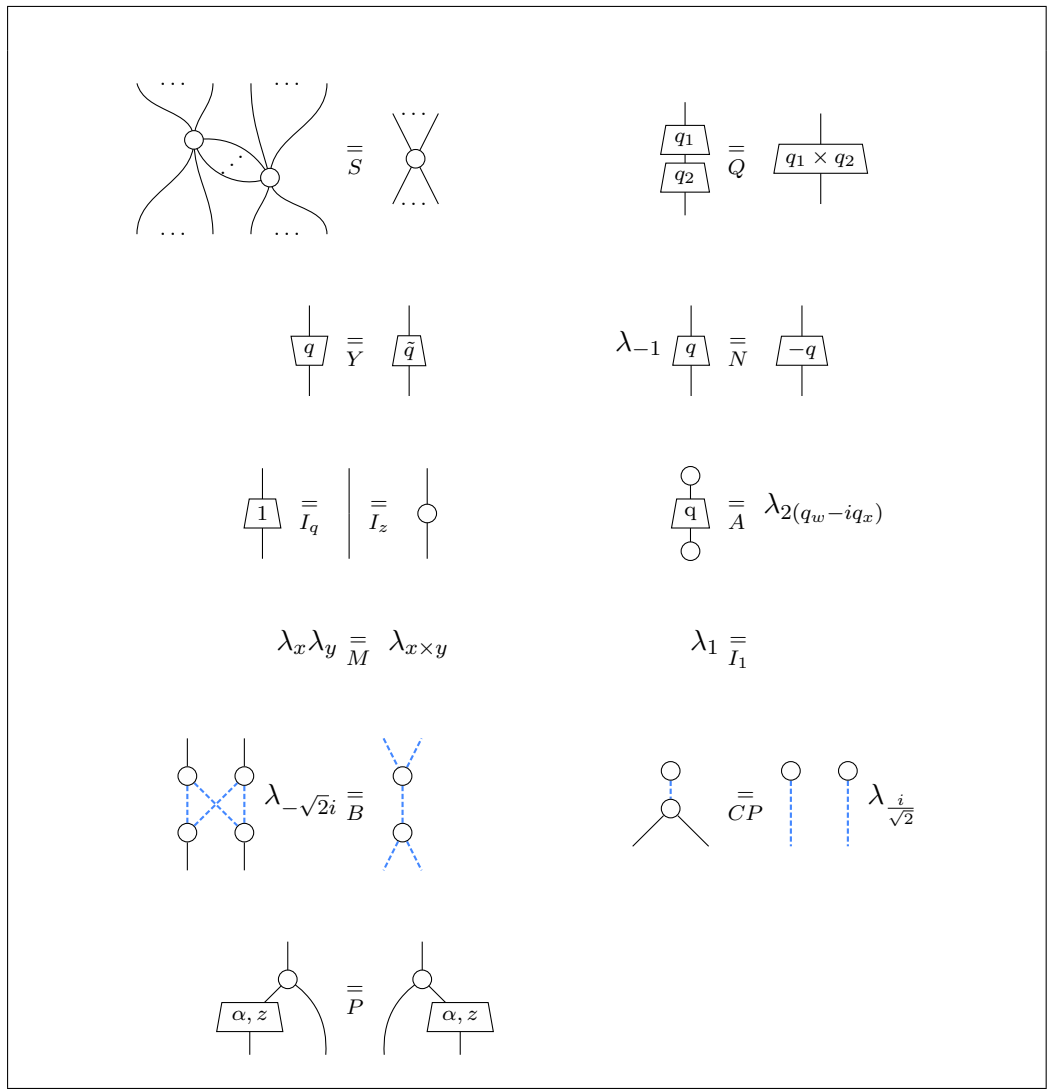

	\centering\begin{tabular}{|c|}
		\hline                     \\
		\begin{minipage}{0.9\textwidth}{
		\begin{align*}
			{\vc{\InputIfFileExists{./figures/ZQ/q_s_lhs.tikz}{}{Missing file!}}}                      & \by{S} \spider{smallZ}{}               &
			{\vc{\InputIfFileExists{./figures/ZQ/q_lhs.tikz}{}{Missing file!}}} & \by{Q} \node{qn}{q_1 \times q_2} \\[2em]
			\node{qnt}{q}                               & \by{Y} \node{qn}{\tilde q}             &
			\lambda_{-1} \node{qn}{q} & \by{N} \node{qn}{-q}\\[2em]
			\node{qn}{1} \by{I_q}                       & \node{none}{} \by{I_z} \node{smallZ}{} &
			{\vc{\InputIfFileExists{./figures/ZQ/q_l3_lhs.tikz}{}{Missing file!}}} & \by{A} \lambda_{2(q_w-iq_x)}\\[2em]
			\lambda_x \lambda_y                         & \by{M} \lambda_{x \times y}            &
			\lambda_1 & \by{I_1}  \\[2em]
			{\vc{\InputIfFileExists{./figures/ZQ/q_b_lhs.tikz}{}{Missing file!}}} \lambda_{-\sqrt{2}i} & \by{B} {\vc{\InputIfFileExists{./figures/ZQ/q_b_rhs.tikz}{}{Missing file!}}}          &
			{\vc{\InputIfFileExists{./figures/ZQ/q_cp_lhs.tikz}{}{Missing file!}}} &\by{CP} {\vc{\InputIfFileExists{./figures/ZQ/q_cp_rhs.tikz}{}{Missing file!}}}  \lambda_{\frac{i}{\sqrt{2}}} \\[2em]
			{\vc{\InputIfFileExists{./figures/ZQ/q_phase_lhs.tikz}{}{Missing file!}}}                  & \by{P} {\vc{\InputIfFileExists{./figures/ZQ/q_phase_rhs.tikz}{}{Missing file!}}}      &
		\end{align*}}
		\end{minipage} \\
		\hline
	\end{tabular}
	\caption{
		The rules of ZQ.\label{figZQRules}
		In rule S the diagonal dots indicate one or more wires, horizontal dots indicate zero or more wires.
		The right hand side of rule $I_1$ is the empty diagram,
		and $\tilde q$ is the quaternion $q$ reflected in the map $j \mapsto -j$.
	}
\end{figure}

\begin{lemma} \label{lemHadamardEdgeWellDefinedZQ}
	The Hadamard edge is well defined in ZQ,
	in that:
	\begin{align}
		ZQ \semantic \node{qn}{H} &= \node{qnt}{H} = \vc{\begin{tikzpicture}
	\begin{pgfonlayer}{nodelayer}
		\node [style=none] (6) at (0, 0.75) {};
		\node [style=none] (7) at (0, -0.25) {};
	\end{pgfonlayer}
	\begin{pgfonlayer}{edgelayer}
		\draw [style=hadamard edge] (6.center) to (7.center);
	\end{pgfonlayer}
\end{tikzpicture}
}  & ZQ \entails \node{qn}{H} \by{Y} \node{qnt}{H} &= \vc{\begin{tikzpicture}
	\begin{pgfonlayer}{nodelayer}
		\node [style=none] (6) at (0, 0.75) {};
		\node [style=none] (7) at (0, -0.25) {};
	\end{pgfonlayer}
	\begin{pgfonlayer}{edgelayer}
		\draw [style=hadamard edge] (6.center) to (7.center);
	\end{pgfonlayer}
\end{tikzpicture}
}
	\end{align}
\end{lemma}

\begin{proof} \label{prfLemHadamardEdgeWellDefinedZQ}
	For the semantics:
	\begin{align}
		\interpret{\node{qn}{H}} = \frac{-i}{\sqrt{2}}
		\begin{pmatrix}
			1 & 1  \\
			1 & -1
		\end{pmatrix} =\interpret{\node{qn}{H}}^T = \interpret{\left(\node{qn}{H}\right)^T}
		= \interpret{\node{qnt}{H}}
	\end{align}
	Syntactically
	\begin{align}
		\node{qnt}{\cos \half[\pi] + \sin\half[\pi](i + k)}
		\by{Y}
		\node{qn}{\cos \half[\pi] + \sin\half[\pi](i + k)}
	\end{align}
\end{proof}

The author feels that the rules of ZQ are easier to grasp
than the rules of Universal ZX given in Figure~\ref{figVilmartZXRules}.
(Other rulesets for Universal ZX exist but all of
them require trigonometric side conditions for their variation on the EU rule.)
By contrast the rules of ZQ involve a single, linear side condition,
it abstracts scalar diagrams into simply scalar numbers $\lambda_c$,
and it neatly divides graphically into entanglement (vertices)
and rotations (edges).
The `complexity' of the (EU') rule in ZX
is simply a shadow of the group multiplication of rotations
shown in the (Q) rule of ZQ.
This is best illustrated by noting that
to simplify a chain of single qubit unitaries
in ZX one must successively apply the (EU') rule,
at each point calculating $x^+$, $x^-$, $z$, and $z'$,
as well as various arguments, moduli, sines, and cosines,
eventually reaching three rotations (and a choice of ZXZ or XZX);
in ZQ one simply applies the group multiplication.

\begin{figure}
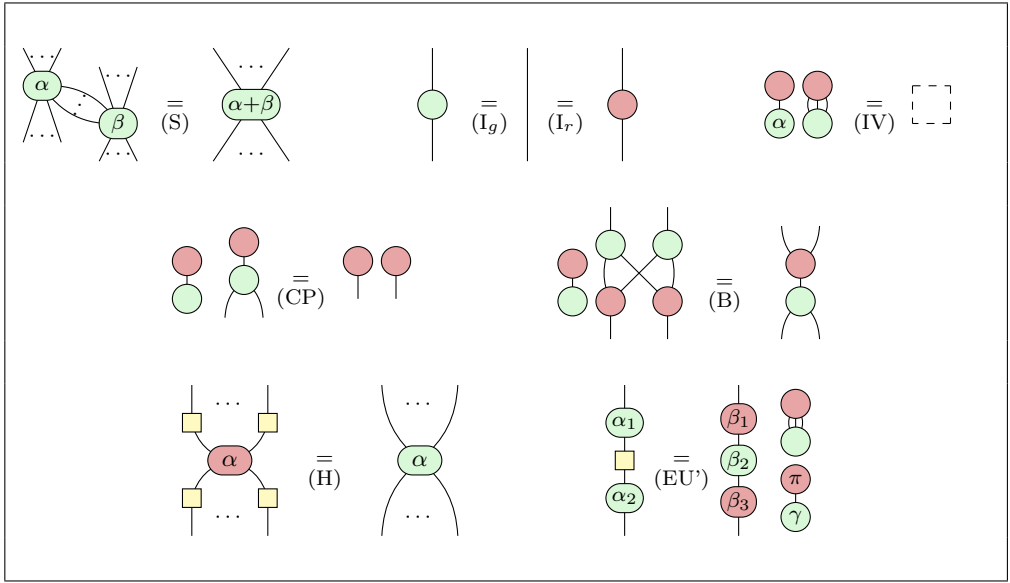

	\centering
	\scalebox{1}{\begin{tabular}{|c|}
			\hline                                                                                                                         \\
			\vc{\InputIfFileExists{./figures/ZX/Vilmart/spider-1.tikz}{}{Missing file!}}$\qquad\quad~$\vc{\InputIfFileExists{./figures/ZX/Vilmart/s2-green-red.tikz}{}{Missing file!}}$\qquad\quad~$\vc{\InputIfFileExists{./figures/ZX/Vilmart/inverse-param.tikz}{}{Missing file!}}~~ \\\\
			\vc{\InputIfFileExists{./figures/ZX/Vilmart/b1s.tikz}{}{Missing file!}}$\qquad\qquad$\vc{\InputIfFileExists{./figures/ZX/Vilmart/b2s.tikz}{}{Missing file!}}                                                                 \\\\
			\vc{\InputIfFileExists{./figures/ZX/Vilmart/h2.tikz}{}{Missing file!}}$\qquad\qquad$\vc{\InputIfFileExists{./figures/ZX/Vilmart/Euler-Had.tikz}{}{Missing file!}}                                                            \\\\
			\hline
		\end{tabular}}
	\caption[]{Set of rules ZX for the ZX-Calculus with scalars from Ref.~\cite{VilmartZX}. The right-hand side of (IV) is an empty diagram. (...) denote zero or more wires, while (\protect\rotatebox{45}{\raisebox{-0.4em}{$\cdots$}}) denote one or more wires. In rule (EU'), $\beta_1,\beta_2,\beta_3$ and $\gamma$ can be determined as follows:
		$x^+:=\frac{\alpha_1+\alpha_2}{2}$, $x^-:=x^+-\alpha_2$, $z := -\sin{x^+}+i\cos{x^-}$ and $
			z' := \cos{x^+}-i\sin{x^-}$, then
		$\beta_1 = \arg z + \arg z',
			\beta_2 = 2\arg\left(i+\left|\frac{z}{z'}\right|\right),
			\beta_3 = \arg z - \arg z',
			\gamma = x^+-\arg(z)+\frac{\pi-\beta_2}{2}$
		where by convention $\arg(0):=0$ and $z'=0\implies \beta_2=0$.
	}
	\label{figVilmartZXRules}
\end{figure}

\section{Soundness of ZQ} \label{secZQSound}

In this section we go through each of the rules given in Section~\ref{secZQRules},
showing that the interpretations of the left and right hand sides of the rules are equal.

\begin{proposition} \label{propZQSSound}
	The rule S is sound:
	\begin{align}
		\interpret{{\vc{\InputIfFileExists{./figures/ZQ/q_s_lhs_dots.tikz}{}{Missing file!}}}} & = \interpret{\spidermn{smallZ}{}{c+d}{a+b}}
	\end{align}
	Where there are $k \geq 1$ wires represented by \reflectbox{$\ddots$} in the middle of the left hand side.
\end{proposition}

\begin{proof} \label{prfPropZQSSound}
	This is simply a restating of the original Z spider law from \cite[Theorem~6.12]{Coecke08}.
\end{proof}

\begin{proposition} \label{propZQQSound}
	The rule $Q$ is sound:
	\begin{align}
		\interpret{{\vc{\begin{tikzpicture}
	\begin{pgfonlayer}{nodelayer}
		\node [style=qn] (0) at (0, -0.25) {$q_2$};
		\node [style=qn] (1) at (0, 0.25) {$q_1$};
		\node [style=none] (2) at (0, 0.75) {};
		\node [style=none] (3) at (0, -0.75) {};
	\end{pgfonlayer}
	\begin{pgfonlayer}{edgelayer}
		\draw (2.center) to (3.center);
	\end{pgfonlayer}
\end{tikzpicture}
}}} & = \interpret{\node{qn}{q_1 \times q_2}}
	\end{align}
\end{proposition}

\begin{proof}  \label{prfPropZQQSound}
	Follows from $\phi$ (see Remark~\ref{remQuaternionsSU2}) being a group isomorphism.
	The left hand side is multiplication in \SU2, the right hand side is multiplication in \UQuat.
\end{proof}

\begin{proposition} \label{propZQYSound}
	The rule $Y$ is sound:
	\begin{align}
		\interpret{\node{qn}{q_w + iq_x - jq_y + kq_z}} & = \interpret{\node{qnt}{q_w + iq_x + jq_y + kq_z}}
	\end{align}
\end{proposition}

\begin{proof} \label{prfPropZQYSound}
	The action of the cups and caps in Figure~\ref{figQTranspose}
	(where we defined the diagrammatic transpose),
	is to enact the transpose in the interpretation:
	\begin{align}
		\interpret{{\vc{\InputIfFileExists{./figures/ZQ/q_y.tikz}{}{Missing file!}}}} =  \begin{pmatrix}
			q_w - iq_z & -q_y - iq_x \\
			q_y-iq_x   & q_w + iq_z
		\end{pmatrix}
		=                                 \interpret{\node{qn}{q_w + iq_x - jq_y + kq_z}}
	\end{align}
\end{proof}

\begin{proposition} \label{propZQNSound}
	The rule $N$ is sound:
	\begin{align}
		\interpret{\lambda_{-1} \node{qn}{q}}= \interpret{\node{qn}{-q}}
	\end{align}
\end{proposition}

\begin{proof} \label{prfPropZQNSound}
	\begin{align}
		LHS =  -1 \begin{pmatrix}
			q_w - iq_z & q_y - iq_x \\
			-q_y-iq_x  & q_w + iq_z
		\end{pmatrix}
		=      \begin{pmatrix}
			-q_w + iq_z & -q_y + iq_x \\
			q_y+iq_x    & -q_w - iq_z
		\end{pmatrix} = \, RHS
	\end{align}
\end{proof}

\begin{proposition} \label{propZQISound}
	The rules $I_q$ and $I_z$ are sound:
	\begin{align}
		\interpret{\node{qn}{1}} = \interpret{\node{none}{}} = \interpret{\node{smallZ}{}}
	\end{align}
\end{proposition}

\begin{proof} \label{prfPropZQISound}
	They all have the interpretation $\begin{pmatrix}
			1 & 0 \\ 0 & 1
		\end{pmatrix}$.
\end{proof}

\begin{proposition} \label{propZQASound}
	The rule $A$ is sound:
	\begin{align}
		\interpret{{\vc{\begin{tikzpicture}
	\begin{pgfonlayer}{nodelayer}
		\node [style=qn] (7) at (0, 0) {q};
		\node [style=smallZ] (8) at (0, 0.5) {};
		\node [style=smallZ] (9) at (0, -0.5) {};
	\end{pgfonlayer}
	\begin{pgfonlayer}{edgelayer}
		\draw (8.center) to (7.center);
		\draw (7.center) to (9.center);
	\end{pgfonlayer}
\end{tikzpicture}
}}} = \interpret{\lambda_{2(q_w-iq_x)}}
	\end{align}
\end{proposition}

\begin{proof} \label{prfPropZQASound}
	\begin{align}
		\interpret{{\vc{\begin{tikzpicture}
	\begin{pgfonlayer}{nodelayer}
		\node [style=qn] (7) at (0, 0) {q};
		\node [style=smallZ] (8) at (0, 0.5) {};
		\node [style=smallZ] (9) at (0, -0.5) {};
	\end{pgfonlayer}
	\begin{pgfonlayer}{edgelayer}
		\draw (8.center) to (7.center);
		\draw (7.center) to (9.center);
	\end{pgfonlayer}
\end{tikzpicture}
}}} = &
		\begin{pmatrix}
			1 & 1
		\end{pmatrix} \comp
		\begin{pmatrix}
			q_w - iq_z & q_y - iq_x \\
			-q_y-iq_x  & q_w + iq_z
		\end{pmatrix} \comp
		\begin{pmatrix}
			1 \\ 1
		\end{pmatrix} \\
		=                                     & q_w - iq_z + q_y - iq_x -q_y-iq_x + q_w + iq_z \\
		=                                     & 2(q_w - iq_x)                                  \\
		=                                     & \interpret{\lambda_{2(q_w - iq_x)}}
	\end{align}
\end{proof}

\begin{proposition} \label{propZQMSound}
	The rule $M$ is sound:
	\begin{align}
		\interpret{\lambda_x \lambda_y} = \interpret{\lambda_{x \times y}}
	\end{align}
\end{proposition}

\begin{proof} \label{prfPropZQMSound}
	Both sides have interpretation $x \times y$.
\end{proof}

\begin{proposition} \label{propZQL'Sound}
	The rule $I_\lambda$ is sound:
	\begin{align}
		\interpret{\lambda_1} = \interpret{\epsilon}
	\end{align}
	Where $\epsilon$ is the empty diagram.
\end{proposition}

\begin{proof} \label{prfPropZQL'Sound}
	Both sides have interpretation $1$.
\end{proof}

\begin{proposition} \label{propZQBSound}
	The rule $B$ is sound:
	\begin{align}
		\interpret{{\vc{\InputIfFileExists{./figures/ZQ/q_b_lhs.tikz}{}{Missing file!}}} \lambda_{-\sqrt{2}i}} = \interpret{{\vc{\InputIfFileExists{./figures/ZQ/q_b_rhs.tikz}{}{Missing file!}}}}
	\end{align}
\end{proposition}

\begin{proof} \label{prfPropZQBSound}
	\begin{align}
		LHS = & -\sqrt{2} i \times \begin{pmatrix}
			1 & 0 & 0 & 0 \\
			0 & 0 & 0 & 1
		\end{pmatrix}^{\tensor 2} \comp
		\left(\frac{-i}{\sqrt{2}}\begin{pmatrix}
			1 & 1 \\ 1 & -1
		\end{pmatrix}\right)^{\tensor 4} \comp \\
		      & \left(id_2 \tensor \begin{pmatrix}
				1 & 0 & 0 & 0 \\
				0 & 0 & 1 & 0 \\
				0 & 1 & 0 & 0 \\
				0 & 0 & 0 & 1
			\end{pmatrix} \tensor \id_2\right) \comp
		\begin{pmatrix}
			1 & 0 \\
			0 & 0 \\
			0 & 0 \\
			0 & 1
		\end{pmatrix}^{\tensor 2} \\
		=     & \frac{-i}{\sqrt{2}^3} \begin{pmatrix}
			1 & 1  & 1  & 1 \\
			1 & -1 & -1 & 1 \\
			1 & -1 & -1 & 1 \\
			1 & 1  & 1  & 1
		\end{pmatrix}                                 \\
		RHS = & \left(\frac{-i}{\sqrt{2}}\right)^5 \times \begin{pmatrix}
			1 & 1 \\ 1 & -1
		\end{pmatrix}^{\tensor 2}
		\comp
		\begin{pmatrix}
			1 & 0 \\
			0 & 0 \\
			0 & 0 \\
			0 & 1
		\end{pmatrix} \comp \\ &  \begin{pmatrix}
			1 & 1 \\ 1 & -1
		\end{pmatrix} \comp
		\begin{pmatrix}
			1 & 0 & 0 & 0 \\
			0 & 0 & 0 & 1
		\end{pmatrix} \comp  \begin{pmatrix}
			1 & 1 \\ 1 & -1
		\end{pmatrix}^{\tensor 2} \\
		=     & \left(\frac{-i}{\sqrt{2}}\right)^5 \times \begin{pmatrix}
			2 & 2  & 2  & 2 \\
			2 & -2 & -2 & 2 \\
			2 & -2 & -2 & 2 \\
			2 & 2  & 2  & 2
		\end{pmatrix}
		= \frac{-i}{\sqrt{2}^3} \times \begin{pmatrix}
			1 & 1  & 1  & 1 \\
			1 & -1 & -1 & 1 \\
			1 & -1 & -1 & 1 \\
			1 & 1  & 1  & 1
		\end{pmatrix}
	\end{align}
\end{proof}

\begin{proposition} \label{propZQCPSound}
	The rule $CP$ is sound:
	\begin{align}
		\interpret{{\vc{\InputIfFileExists{./figures/ZQ/q_cp_lhs.tikz}{}{Missing file!}}}} = \interpret{{\vc{\InputIfFileExists{./figures/ZQ/q_cp_rhs.tikz}{}{Missing file!}}}  \lambda_{\frac{i}{\sqrt{2}}}}
	\end{align}
\end{proposition}

\begin{proof} \label{prfPropZQCPSound}
	\begin{align}
		LHS = & \begin{pmatrix}
			1 & 1
		\end{pmatrix} \comp \frac{-i}{\sqrt{2}} \begin{pmatrix}
			1 & 1 \\ 1 & -1
		\end{pmatrix} \comp \begin{pmatrix}
			1 & 0 & 0 & 0 \\ 0 & 0 & 0 & 1
		\end{pmatrix}
		=  \frac{-i}{\sqrt{2}} \begin{pmatrix}
			2 & 0 & 0 & 0
		\end{pmatrix} \\
		RHS = & \frac{i}{\sqrt{2}} \left(\begin{pmatrix}
			1 & 1
		\end{pmatrix} \comp \frac{-i}{\sqrt{2}} \begin{pmatrix}
			1 & 1 \\ 1 & -1
		\end{pmatrix} \right)^{\tensor 2}
		=  \frac{-i}{\sqrt{2}} \begin{pmatrix}
			2 & 0 & 0 & 0
		\end{pmatrix}
	\end{align}
\end{proof}

\begin{proposition} \label{propZQPSound}
	The rule $P$ is sound:
	\begin{align}
		\interpret{{\vc{\InputIfFileExists{./figures/ZQ/q_phase_lhs.tikz}{}{Missing file!}}}} = \interpret{{\vc{\InputIfFileExists{./figures/ZQ/q_phase_rhs.tikz}{}{Missing file!}}}}
	\end{align}
\end{proposition}

\begin{proof} \label{prfPropZQPSound}
	\begin{align}
		LHS = & \begin{pmatrix}
			1 & 0 & 0 & 0 \\
			0 & 0 & 0 & 1
		\end{pmatrix} \comp \left(\begin{pmatrix}
				q_w - iq_z & 0           \\
				0          & q_w + i q_z
			\end{pmatrix} \tensor id_2 \right) \\
		=     & \begin{pmatrix}
			q_w - iq_z & 0 & 0 & 0           \\
			0          & 0 & 0 & q_w + i q_z
		\end{pmatrix}                                                             \\
		RHS = & \begin{pmatrix}
			1 & 0 & 0 & 0 \\
			0 & 0 & 0 & 1
		\end{pmatrix} \comp \left(id_2 \tensor \begin{pmatrix}
				q_w - iq_z & 0           \\
				0          & q_w + i q_z
			\end{pmatrix}\right)  \\
		=     & \begin{pmatrix}
			q_w - iq_z & 0 & 0 & 0           \\
			0          & 0 & 0 & q_w + i q_z
		\end{pmatrix}
	\end{align}
\end{proof}

\section{Completeness of ZQ} \label{secZQComplete}

The completion of ZQ is achieved by finding an equivalence between ZQ and ZX as PROPs.
We already know that ZX is complete \cite{UniversalComplete}
and this proof was by a similar equivalence with ZW, which was shown to be complete in Ref.~\cite{ZW}.
Equivalence is shown by finding a translation of the generators from ZX to ZQ and vice versa
(\S\ref{secZQZXTranslation}),
before then translating all of the rules from ZX into ZQ (\S\ref{secZXZQTranslatedRules}),
and keeping these as rules in ZQ.
Finally one has to ensure that any diagram translated from ZQ to ZX and back again
can be proven to be equivalent to the original ZQ diagram (\S\ref{secCompletenessZXZQBack}).
In symbols this is:
\begin{align}
	                                              & \interpret{D_1} = \interpret{D_2} \quad \text{Two diagrams in ZQ} \\
	ZX \entails                                   & F_X D_1 = F_X D_2                                                 \\
	\S\ref{secZXZQTranslatedRules} : ZQ \entails  & F_Q F_X D_1 = F_Q F_X D_2                                         \\
	\S\ref{secCompletenessZXZQBack} : ZQ \entails & D_1 = F_Q F_X D_1 \qquad \text{and} \qquad D_2 = F_Q F_X D_2      \\
	\therefore ZQ \entails                        & D_1 = F_Q F_X D_1 = F_Q F_X D_2 = D_2
\end{align}

\subsection{Translation to and from ZX} \label{secZQZXTranslation}

We define the strict monoidal functors $F_X$ and $F_Q$ on generators in Figure~\ref{figZQZXTranslation}.
In defining this translation we make use of two facts;
firstly that we can decompose any unit quaternion into
Z then X then Z rotations. This is tantamount to Euler Angle Decomposition
and is performed explicitly in Proposition~\ref{propZQQDecomp}.
Secondly we need to be able to express any complex number
in a rather particular form, which is shown in Lemma~\ref{lemRealNumberCanonicalForm}.

\begin{figure}[p]
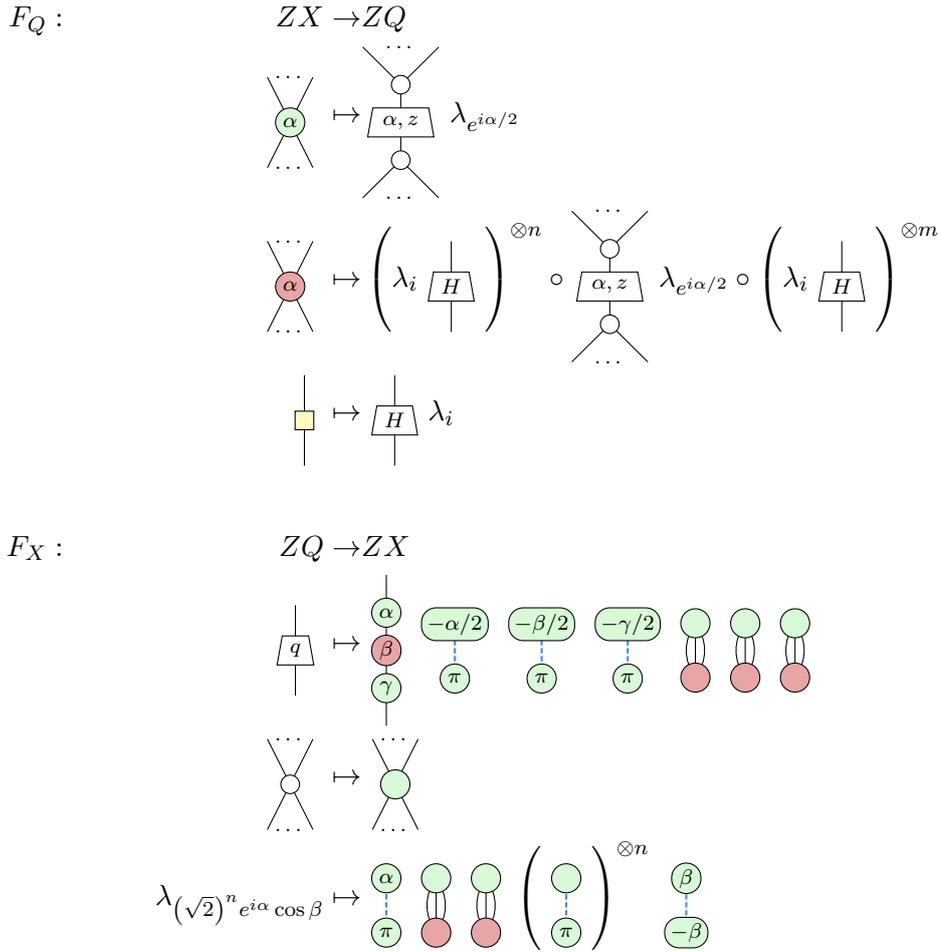

	\begin{align*}
		F_Q : &   & ZX \to                                                                                            & ZQ  \nonumber                                                                                                                                \\
		      &   & \spider{gn}{\alpha} \mapsto                                                                       & \vc{\InputIfFileExists{./figures/ZQ/q_spider_expand.tikz}{}{Missing file!}} \lambda_{e^{i\alpha/2}}                                                                                         \\
		      &   & \spider{rn}{\alpha} \mapsto                                                                       & \left(\lambda_i \node{qn}{H}\right)^{\tensor n}
		\comp {\vc{\InputIfFileExists{./figures/ZQ/q_spider_expand.tikz}{}{Missing file!}}} \lambda_{e^{i\alpha/2}}
		\comp \left(\lambda_i \node{qn}{H}\right)^{\tensor m} \\
		      &   & \node{h}{} \mapsto                                                                                & \node{qn}{H} \lambda_i                                                                                                                       \\[\rowgap]
		F_X : &   & ZQ \to                                                                                            & ZX \nonumber                                                                                                                                 \\
		      &   & \node{qn}{q} \mapsto & \ \vc{\InputIfFileExists{./figures/ZX/abc.tikz}{}{Missing file!}} \galpharpi{-\alpha/2}\nonumber\galpharpi{-\beta/2}\galpharpi{-\gamma/2} \tripleblobs\tripleblobs\tripleblobs    \nonumber \\
		      &   & \spider{smallZ}{} \mapsto                                                                         & \spider{gn}{}                                                                                                                                \\
		      &   & \lambda_{\left(\sqrt{2}\right)^n e^{i\alpha} \cos \beta} \mapsto                                  & \galpharpi{\alpha} \halfblobs \left(\rg{}{\pi}\right)^{\tensor n} \rg{\beta}{-\beta}
	\end{align*}
	\caption{Translation from ZX to ZQ and back again.\label{figZQZXTranslation}
		The existence of $\alpha$, $\beta$ and $\gamma$ when translating the Q node is shown in Proposition~\ref{propZQQDecomp},
		likewise the decomposition of any complex number as $\left(\sqrt{2}\right)^n e^{i\alpha} \cos \beta$ is shown in Lemma~\ref{lemRealNumberCanonicalForm}.
	}
\end{figure}

\begin{proposition} \label{propZQQDecomp}
	There exist $\alpha$ and $\gamma \in [0, 2\pi)$, and $\beta \in [0,\pi]$ such that:
	\begin{align}
		q_w + i q_x + j q_y + k q_z = \left(\cos \frac{\alpha}{2} + k \sin \frac{\alpha}{2}\right)\left(\cos \frac{\beta}{2} + i \sin \frac{\beta}{2}\right)\left(\cos \frac{\gamma}{2} + k \sin \frac{\gamma}{2}\right)
	\end{align}
\end{proposition}

\begin{proof} \label{prfPropZQQDecomp}
	\begin{align}
		RHS = & \left(\cos \frac{\alpha}{2} + k \sin \frac{\alpha}{2}\right)\left(\cos \frac{\beta}{2} + i \sin \frac{\beta}{2}\right)\left(\cos \frac{\gamma}{2} + k \sin \frac{\gamma}{2}\right) \\
		=     & \left(\cos \frac{\alpha}{2} \cos \frac{\beta}{2} \cos \frac{\gamma}{2} - \sin \frac{\alpha}{2} \cos \frac{\beta}{2} \sin \frac{\gamma}{2}\right) +                                 \\
		      & i \left(\cos \frac{\alpha}{2} \sin \frac{\beta}{2} \cos \frac{\gamma}{2} + \sin \frac{\alpha}{2} \sin \frac{\beta}{2} \sin \frac{\gamma}{2}\right) \nonumber                       \\
		      & j \left(\cos \frac{\alpha}{2} \sin \frac{\beta}{2} \sin \frac{\gamma}{2} - \sin \frac{\alpha}{2} \sin \frac{\beta}{2} \cos \frac{\gamma}{2}\right) + \nonumber                     \\
		      & k \left(\cos \frac{\alpha}{2} \cos \frac{\beta}{2} \sin \frac{\gamma}{2} + \sin \frac{\alpha}{2} \cos \frac{\beta}{2} \cos \frac{\gamma}{2}\right) + \nonumber                     \\
		=     & \cos \frac{\beta}{2} \left(\cos \frac{\alpha + \gamma}{2} + i \sin \frac{\alpha + \gamma}{2}\right) +
		j \sin \frac{\beta}{2} \left(\sin \frac{\gamma - \alpha}{2} -  i \cos \frac{\gamma - \alpha}{2}\right)\nonumber
	\end{align}
	From this we gather:
	\begin{align}
		q_w = & \cos \frac{\beta}{2} \cos \frac{\alpha + \gamma}{2} &
		q_x =& \cos \frac{\beta}{2} \sin \frac{\alpha + \gamma}{2} \\
		q_y = & \sin \frac{\beta}{2} \sin \frac{\gamma - \alpha}{2} &
		q_z =&\sin \frac{\beta}{2} \cos \frac{\gamma - \alpha}{2}
	\end{align}
	And finally use these to determine values of $\alpha$, $\beta$ and $\gamma:$
	\begin{itemize}
		\item $q_w^2 + q_x^2 = \cos^2 \frac{\beta}{2}$ determines up to two different possibilities of $\beta \in [0, 2\pi)$.
		      We will enforce $\beta \in [0, \pi]$ to make this unique and $\cos \frac{\beta}{2}$ non-negative.
		\item If $\beta = 0$ then set $\gamma = 0$, use $q_w$ and $q_x$ to determine $\alpha$
		\item Likewise if $\beta = \pi$ set $\gamma = 0$, use $q_y$ and $q_z$ to determine $\alpha$
		\item Otherwise determine $\alpha + \gamma / 2$ from $q_w$ and $q_x$,
		      and $\alpha - \gamma / 2$ from $q_y$ and $q_z$; their sum and difference give $\gamma$ and $\alpha$ respectively.
	\end{itemize}

	The choices we made in this proof we justify by noting that we can represent these choices
	by certain applications of the spider rule (in the case $\beta = 0$) and $\pi$-commutativity rules
	(relating $(\alpha, \beta,\gamma) \sim (\alpha + \pi, -\beta, \gamma + \pi)$) in ZX.
\end{proof}

\begin{lemma} \label{lemRealNumberCanonicalForm}
	Any non-zero complex number $c$ can be expressed uniquely as $\left(\sqrt{2}\right)^n e^{i\alpha} \cos \beta$ where
	$n \in \bbN$, $\alpha \in [0,2\pi)$, $\beta \in [0, \pi)$
	and where $n$ is chosen to be the least $n$ such that $\sqrt{2}^n \geq \abs{c}$.
\end{lemma}

\begin{proof} \label{prfLemRealNumberCanonicalForm}
	Express the complex number $c$ as $re^{i\alpha}$, where $r \in \bbR_{\geq 0}$.
	This matches our choice of $\alpha \in [0, 2\pi)$.
	For all $r$ there is at least one $n$ where $\sqrt{2}^n \geq r$ and so we can find a least such $n$.
	Once we know $n$ there is a unique $\beta \in [0,\pi)$ such that $\cos{\beta} \sqrt{2}^n = r$.
\end{proof}

If $c$ is zero then we set $\alpha$ and $n$ to 0 and $\beta$ to $\pi/2$.

\subsection{Proving the translated ZX rules} \label{secZXZQTranslatedRules}

We aim to show that the rules translated from ZX are all derivable by the rules in \S\ref{secZQRules}, which we will refer to as $\ZQ$.
We will use the ZX ruleset from \cite[Figure~2]{VilmartZX},
and refer to individual ZX rules as $\ZX_{\text{rule name}}$.
To save space, we will assume applications of the $M$ rule (scalar multiplication) in the statements of the propositions.

\begin{lemma} \label{lemZQTranslationZ} Translation of the Z spider
	\begin{align}
		ZQ \entails F_Q\left(\spider{gn}{}\right) = \spider{smallZ}{}
	\end{align}
\end{lemma}

\begin{proof} \label{prfLemZQTranslationZ}
	\begin{align}
		LHS & = {\vc{\InputIfFileExists{./figures/ZQ/q_spider_0.tikz}{}{Missing file!}}} \by{I_q} {\vc{\InputIfFileExists{./figures/ZQ/q_spider_1.tikz}{}{Missing file!}}} \by{S} \spider{smallZ}{}
	\end{align}
\end{proof}

\begin{proposition} \label{propZQTranslationS} Translation of the Z spider rule
	\begin{align}
		ZQ     & \entails F_Q\left(ZX_S\right)                        \\
		\ie ZQ & \entails {\vc{\InputIfFileExists{./figures/ZQ/q_ts1.tikz}{}{Missing file!}}} = {\vc{\InputIfFileExists{./figures/ZQ/q_ts5.tikz}{}{Missing file!}}}
	\end{align}
	(The diagonal dots represent at least one wire between the Z spiders)
\end{proposition}

\begin{proof} \label{prfPropZQTranslationS}
	\begin{align}
		{\vc{\InputIfFileExists{./figures/ZQ/q_ts1.tikz}{}{Missing file!}}}
		  & \by{S} {\vc{\InputIfFileExists{./figures/ZQ/q_ts2.tikz}{}{Missing file!}}}    \\
		\by{P, Y} {\vc{\InputIfFileExists{./figures/ZQ/q_ts3.tikz}{}{Missing file!}}}
		  & \by{Q, S} {\vc{\InputIfFileExists{./figures/ZQ/q_ts4.tikz}{}{Missing file!}}} \\
		\by{P, Q, S} {\vc{\InputIfFileExists{./figures/ZQ/q_ts5.tikz}{}{Missing file!}}}
	\end{align}
\end{proof}

\begin{proposition} \label{propZQTranslationIg} Translation of the Z spider identity
	\begin{align}
		ZQ     & \entails  F_Q\left(ZX_{I_g}\right)                                                \\
		\ie ZQ & \entails \node{smallZ}{} \comp \node{qn}{1} \comp \node{smallZ}{} = \node{none}{}
	\end{align}
\end{proposition}

\begin{proof} \label{prfPropZQTranslationIg}
	\begin{align}
		\node{smallZ}{} \comp \node{qn}{1} \comp \node{smallZ}{} \by{I_z} \node{qn}{1}
		  & \by{I_q} \node{none}{}
	\end{align}
\end{proof}

\begin{proposition} \label{propZQTranslationIr} Translation of the X spider identity
	\begin{align}
		ZQ     & \entails F_Q\left(ZX_{I_r}\right)                                                                                                            \\
		\ie ZQ & \entails  \node{qn}{H} \comp \node{smallZ}{} \comp \node{qn}{1} \comp \node{smallZ}{} \comp \node{qn}{H} \lambda_i \lambda_i = \node{none}{}
	\end{align}
\end{proposition}

\begin{proof} \label{prfPropZQTranslationIr}
	\begin{align}
		LHS =    & \node{qn}{H} \comp \node{smallZ}{} \comp \node{qn}{1} \comp \node{smallZ}{} \comp \node{qn}{H} \lambda_i \lambda_i \\
		\by{I_z} & \node{qn}{H} \comp \node{qn}{1} \comp \node{qn}{H} \lambda_i \lambda_i                                             \\
		\by{I_q} & \node{qn}{H} \comp \node{qn}{H} \lambda_i \lambda_i
		\by{Q} \node{qn}{-1} \lambda_i \lambda_i \\
		\by{N}   & \node{qn}{1} \lambda_{-1} \lambda_i \lambda_i
		\by{I_q}  \node{none}{} \lambda_{-1} \lambda_i \lambda_i
		\by{M}  \node{none}{} \lambda_1
		\by{I_1}  \node{none}{}
	\end{align}
\end{proof}

We introduce our first three intermediate lemmas, corresponding to properties of the following three ZX diagrams:

\begin{align}
	\rg{\alpha}{\beta} \label{eqnSmallRG} , \qquad \vc{\InputIfFileExists{./figures/ZX/HH.tikz}{}{Missing file!}}, \qquad \vc{\InputIfFileExists{./figures/ZX/HHH.tikz}{}{Missing file!}}
\end{align}

\begin{lemma} \label{lemZQAlphaHBeta} Interaction of a Z state and Z effect joined by a Hadamard
	\begin{align}
		ZQ & \entails  {\vc{\InputIfFileExists{./figures/ZQ/q_alpha_H_beta.tikz}{}{Missing file!}}} & = \lambda_{-\sqrt{2}\left(\left(\sin\frac{\alpha+\beta}{2}\right) + i \cos\frac{\alpha-\beta}{2}\right)}
	\end{align}
\end{lemma}

\begin{proof} \label{prfLemZQAlphaHBeta}
	\begin{align}
		{\vc{\InputIfFileExists{./figures/ZQ/q_alpha_H_beta2.tikz}{}{Missing file!}}}
		              & \by{I_z}  {\vc{\begin{tikzpicture}
	\begin{pgfonlayer}{nodelayer}
		\node [style=qn] (7) at (0, 0) {$H$};
		\node [style=smallZ] (9) at (0, 1.5) {};
		\node [style=smallZ] (11) at (0, -1.5) {};
		\node [style=qn] (12) at (0, 1) {$\alpha,z$};
		\node [style=qn] (13) at (0, -1) {$\beta,z$};
	\end{pgfonlayer}
	\begin{pgfonlayer}{edgelayer}
		\draw (9.center) to (11.center);
	\end{pgfonlayer}
\end{tikzpicture}
}}
		\by{Q} {\vc{\begin{tikzpicture}
	\begin{pgfonlayer}{nodelayer}
		\node [style=qn] (7) at (0, 0) {$(\alpha,z) \times H \times (\beta, z)$};
		\node [style=smallZ] (9) at (0, 0.75) {};
		\node [style=smallZ] (11) at (0, -0.75) {};
	\end{pgfonlayer}
	\begin{pgfonlayer}{edgelayer}
		\draw (9.center) to (11.center);
	\end{pgfonlayer}
\end{tikzpicture}
}} \by{A} \lambda_{-\sqrt{2}((\sin\frac{\alpha+\beta}{2}) + i \cos\frac{\alpha-\beta}{2})} \\[\rowgap]
		\text{Since } & (\cos \frac{\alpha}{2} + k \sin \frac{\alpha}{2}) \times H \times (\cos \frac{\beta}{2} + k \sin \frac{\beta}{2}) =\nonumber \\
		              & \frac{1}{\sqrt{2}}(
		-\sin\frac{\alpha+\beta}{2} +
		i(\cos\frac{\alpha-\beta}{2}) +
		j(\sin\frac{\alpha-\beta}{2}) +
		k(\cos\frac{\alpha+\beta}{2}))
	\end{align}

\end{proof}

\begin{lemma} \label{propZQHH} Interaction of two Hadamard rotations
	\begin{align}
		ZQ \entails & \node{qn}{H} \comp \node{qn}{H} = \lambda_{-1} \node{none}{}
	\end{align}
\end{lemma}

\begin{proof} \label{prfPropZQHH}
	\begin{align}
		\node{qn}{H} \comp \node{qn}{H} \by{Q} \node{qn}{H \times H} = \node{qn}{-1} \by{N} \node{qn}{1} \lambda_{-1} \by{I_q} \node{none}{} \lambda{-1}
	\end{align}
\end{proof}

\begin{lemma} \label{propZQHHH} The value of the scalar describing three Hadamard rotations in parallel
	\begin{align}
		ZQ \entails {\vc{\InputIfFileExists{./figures/ZQ/q_HHH.tikz}{}{Missing file!}}} & = \lambda_{\frac{i}{\sqrt{2}}}
	\end{align}
\end{lemma}

\begin{proof} \label{prfPropZQHHH}
	\begin{align}
		{\vc{\InputIfFileExists{./figures/ZQ/q_HHH.tikz}{}{Missing file!}}}
		\by{S,\,\ref{propZQHH}} & \lambda_{-1} {\vc{\InputIfFileExists{./figures/ZQ/q_HHH2.tikz}{}{Missing file!}}}          & \by{Y,S} \lambda_{-1}\ {\vc{\InputIfFileExists{./figures/ZQ/q_HHH3.tikz}{}{Missing file!}}}                             \\
		\by{M, B}               & \lambda_{-i/\sqrt{2}}{\vc{\InputIfFileExists{./figures/ZQ/q_HHH4.tikz}{}{Missing file!}}}  & \by{Y,S,I_z} \lambda_{-i/\sqrt{2}}\ {\vc{\InputIfFileExists{./figures/ZQ/q_HHH5.tikz}{}{Missing file!}}}                \\
		\by{CP, Y}              & \lambda_{i/2\sqrt{2}}{\vc{\InputIfFileExists{./figures/ZQ/q_HHH6.tikz}{}{Missing file!}}}  & \by{I_z}  \lambda_{i/2\sqrt{2}}\ {\vc{\InputIfFileExists{./figures/ZQ/q_HHH7.tikz}{}{Missing file!}}}                   \\
		\by{M,\ref{propZQHH}}   & \lambda_{-i/2\sqrt{2}}{\vc{\InputIfFileExists{./figures/ZQ/q_HHH8.tikz}{}{Missing file!}}} & \by{A}  \lambda_{-i/2\sqrt{2}}\ \lambda_{-i\sqrt{2}}\lambda_{-i\sqrt{2}} \\
		\by{M}& \lambda_{\frac{i}{\sqrt{2}}}
	\end{align}
\end{proof}

\begin{proposition} \label{propZQTranslationIV} Translation of the IV rule
	\begin{align}
		ZQ \entails     & F_Q\left(ZX_{IV}\right)          \\
		\ie ZQ \entails & \lambda_{e^{i \frac{\alpha}{2}}}
		{\vc{\InputIfFileExists{./figures/ZQ/q_IV_1.tikz}{}{Missing file!}}}= \epsilon
	\end{align}
\end{proposition}

\begin{proof} \label{prfPropZQTranslationIV}
	\begin{align}
		\lambda_{e^{i \frac{\alpha}{2}}}
		{\vc{\InputIfFileExists{./figures/ZQ/q_IV_1.tikz}{}{Missing file!}}} \by{I_z, I_q, S, M}   & \lambda_{e^{i \frac{\alpha}{2}}}  {\vc{\InputIfFileExists{./figures/ZQ/q_IV_2.tikz}{}{Missing file!}}}                                             \\
		\by{\ref{lemZQAlphaHBeta}, \ref{propZQHHH}} & \lambda_{e^{i \frac{\alpha}{2}}} \lambda_{-\sqrt{2}ie^{-i \frac{\alpha}{2}}} \lambda_{i / \sqrt{2}} \\
		\by{M}                                      & \lambda_{1} \by{I_1} \epsilon
	\end{align}
\end{proof}

\begin{proposition} \label{propZQTranslationCP} Translation of the CP rule
	\begin{align}
		ZQ \entails     & F_Q\left(ZX_{CP}\right)                                               \\
		\ie ZQ \entails & \lambda_{-1} {\vc{\InputIfFileExists{./figures/ZQ/q_CP1.tikz}{}{Missing file!}}} = {\vc{\InputIfFileExists{./figures/ZQ/q_CP3.tikz}{}{Missing file!}}} \lambda_{-1}
	\end{align}
\end{proposition}

\begin{proof} \label{prfPropZQTranslationCP}
	\begin{align}
		LHS \by{1_z, 1_q} & \lambda_{-1} {\vc{\InputIfFileExists{./figures/ZQ/q_CP2.tikz}{}{Missing file!}}} \by{A} \lambda_{-1} \lambda_{-i\sqrt{2}}{\vc{\InputIfFileExists{./figures/ZQ/q_CP2a.tikz}{}{Missing file!}}}                              \\
		\by{CP}           & \lambda_{-1} \lambda_{-i\sqrt{2}} \lambda_{i / \sqrt{2}} {\vc{\InputIfFileExists{./figures/ZQ/q_CP3.tikz}{}{Missing file!}}} \by{M} \lambda_{-1} {\vc{\InputIfFileExists{./figures/ZQ/q_CP3.tikz}{}{Missing file!}}} = RHS
	\end{align}
\end{proof}

\begin{proposition} \label{propZQTranslationB} Translation of the B rule
	\begin{align}
		ZQ \entails     & F_Q\left(ZX_{B}\right)                                               \\
		\ie ZQ \entails & \lambda_{-i}\ {\vc{\InputIfFileExists{./figures/ZQ/q_b1.tikz}{}{Missing file!}}} = {\vc{\InputIfFileExists{./figures/ZQ/q_b4.tikz}{}{Missing file!}}} \lambda_{-i}
	\end{align}
\end{proposition}

\begin{proof} \label{prfPropZQTranslationB}
	\begin{align}
		LHS \by{I_z, I_q} & \lambda_{-i}{\vc{\InputIfFileExists{./figures/ZQ/q_b2.tikz}{}{Missing file!}}} \by{B} \ \lambda_{-i} \lambda_{i / \sqrt{2}} {\vc{\InputIfFileExists{./figures/ZQ/q_b3.tikz}{}{Missing file!}}} \\
		\by{A, M, \ref{propZQHH}}&\lambda_{-i} \vc{\InputIfFileExists{./figures/ZQ/q_b4.tikz}{}{Missing file!}}
	\end{align}
\end{proof}

\begin{proposition} \label{propZQTranslationH} Translation of the H rule
	\begin{align}
		ZQ \entails     & F_Q\left(H\right)                                                                                                      \\
		\ie ZQ \entails &
		\left(\lambda_i \node{qn}{H}\right)^{\tensor n} \comp
		\left(\lambda_i \node{qn}{H}\right)^{\tensor n} \comp
		{\vc{\InputIfFileExists{./figures/ZQ/q_spider_expand.tikz}{}{Missing file!}}}
		\lambda_{e^{i\alpha/2}} \comp\\
		                & \quad \left(\lambda_i \node{qn}{H}\right)^{\tensor m}  \comp \left(\lambda_i \node{qn}{H}\right)^{\tensor m} \nonumber \\
		                & = {\vc{\InputIfFileExists{./figures/ZQ/q_spider_expand.tikz}{}{Missing file!}}} \lambda_{e^{i\alpha/2}}
	\end{align}
\end{proposition}

\begin{proof} \label{prfPropZQTranslationH}
	\begin{align}
		LHS & \by{\ref{propZQHH}} \left(\lambda_i \lambda_i \lambda_{-1}\right)^{\tensor n}  \left(\lambda_i \lambda_i \lambda_{-1}\right)^{\tensor m}
		{\vc{\InputIfFileExists{./figures/ZQ/q_spider_expand.tikz}{}{Missing file!}}} \lambda_{e^{i\alpha/2}}
		\by{M} {\vc{\InputIfFileExists{./figures/ZQ/q_spider_expand.tikz}{}{Missing file!}}} \lambda_{e^{i\alpha/2}}
	\end{align}
\end{proof}

Before proving the translation of the (EU') rule (Proposition~\ref{propZQTranslationEU}) we introduce some helpful lemmas.
For the more trigonometric lemmas we delay their proof until Appendix~\ref{secTrig},
because the details of those proofs are not very informative with respect to proving
ZQ $\entails F_Q\left(EU'\right)$.
We reproduce the side conditions for the rule $\ZX_{EU'}$ for reference here:

\begin{displayquote}[Figure~2, A Near-Minimal Axiomatisation of ZX-Calculus
		for Pure Qubit Quantum Mechanics \cite{VilmartZX}]
	In rule (EU'), $\beta_1$, $\beta_2$, $\beta_3$ and $\gamma$
	can be determined as follows: $x^+ := \frac{\alpha_1 + \alpha_2}{2}$,
	$x^- := x^-\alpha_2$, $z := - \sin (x^+) + i \cos(x^-)$
	and $z' := \cos(x^+) - i \sin (x^-)$,
	then $\beta_1 = \arg z + \arg z'$, $\beta_2 = 2 \arg(i + \frac{\abs{z}}{\abs{z'}})$,
	$\beta_3 = \arg z - \arg z'$, $\gamma = x^+ - \arg(z) + \frac{\pi - \beta_2}{2}$
	where by convention $\arg(0) := 0$ and $z' = 0 \implies \beta_2 = 0$.
\end{displayquote}

\begin{lemma} \label{lemZQQLambdas}
	With the conditions of $\ZX_{EU'}$
	\begin{align}
		\left(e^{i\left(\beta_1 + \beta_2 + \beta_3 + \gamma + 9\pi\right)/2}\right) \times
		\left(\frac{i}{\sqrt{2}}\right) \times
		\left(-\sqrt{2}\left(e^{i\gamma/2}\right)\right) = e^{i\left(\alpha_1 + \alpha_2 + \pi\right)/2}
	\end{align}
\end{lemma}

\begin{proof} \label{prfLemZQQLambdas}
	\begin{align}
		LHS = & \left(e^{i\left(\beta_1 + \beta_2 + \beta_3 + \gamma + 9\pi\right)/2}\right) \times
		\left(\frac{i}{\sqrt{2}}\right) \times
		\left(-\sqrt{2}\left(e^{i\gamma/2}\right)\right) \\
		=     & \left(-i\right)\left(e^{i\left(\beta_1 + \beta_2 + \beta_3 + 2\gamma + \pi\right)/2}\right)                                           \\
		=     & \left(-i\right)\left(e^{i\left(\arg z + \arg z' + \beta_2 + \arg z - \arg z' + 2x^+ - 2 \arg z + \pi - \beta_2 + \pi\right)/2}\right) \\
		=     & \left(-i\right)\left(e^{i\left(2x^++2\pi\right)/2}\right)                                                                             \\
		=     & e^{i\left(\alpha_1 + \alpha_2 + \pi\right)/2}
	\end{align}
\end{proof}

The quaternion $\left(\pi, \frac{x+z}{\sqrt{2}}\right)$, also known as $H$
is the the ZQ equivalent of the Hadamard gate in ZX.
This quaternion changes Z rotations to X rotations and vice versa,
but is not quite self-inverse, with $H \times H = -1$.

\begin{restatable}{lemma}{lemZQHComm}
	\label{lemZQHComm}
	The quaternion $\left(\pi, \frac{x+z}{\sqrt{2}}\right)$ and its interactions with $\left(\alpha, z\right)$ and $\left(\alpha, x\right)$:
	\begin{align}
		\left(\pi, \frac{x+z}{\sqrt{2}}\right) \times \left(\alpha,z\right) = \left(\alpha, x\right) \times \left(\pi, \frac{x+z}{\sqrt{2}}\right) \\
		\left(\alpha,z\right) \times \left(\pi, \frac{x+z}{\sqrt{2}}\right) = \left(\pi, \frac{x+z}{\sqrt{2}}\right) \times \left(\alpha,x\right)  \\
		\left(\pi, \frac{x+z}{\sqrt{2}}\right) \times \left(\pi, \frac{x+z}{\sqrt{2}}\right) = -1
	\end{align}
\end{restatable}

The proof of lemma \ref{lemZQHComm} is in \S\ref{secTrig}.
We now show that the two sequences of rotations
in the rule $\ZX_{EU'}$ correspond to the same two sequences of rotations in ZQ.

\begin{restatable}{lemma}{lemZQEUQuaternion} \label{lemZQEUQuaternion}
	With the conditions of $\ZX_{EU'}$:
	\begin{align}
		\left(\alpha_1, z\right) \times H \times \left(\alpha_2, z\right) = H \times \left(\beta_1, z\right) \times H  \times \left(\beta_2,z\right) \times H \times \left(\beta_3, z\right) \times H
	\end{align}
\end{restatable}

The proof of lemma \ref{lemZQEUQuaternion} is in \S\ref{secTrig}.
The final lemma before Proposition~\ref{propZQTranslationEU}
is the evaluation of the parameterised scalar in $\ZX_{EU'}$.

\begin{lemma} \label{lemZQGamma}
	\begin{align}
		ZQ \entails {\vc{\InputIfFileExists{./figures/ZQ/q_pigamma.tikz}{}{Missing file!}}} & = \lambda_{-\sqrt{2}\left(e^{i\gamma/2}\right)}
	\end{align}
\end{lemma}

\begin{proof} \label{prfLemZQGamma}
	This is a special case of Lemma~\ref{lemZQAlphaHBeta}.
\end{proof}

\begin{proposition} \label{propZQTranslationEU}
	\begin{align}
		ZQ \entails     & F_Q\left(EU'\right)                                                                                                                                                         \\
		\ie ZQ \entails & \lambda_{e^{\frac{i}{2}\left(\alpha_1 + \alpha_2 + \pi\right)}} {\vc{\InputIfFileExists{./figures/ZQ/q_EU1.tikz}{}{Missing file!}}} =\lambda_{e^{i(\beta_1 + \beta_2 + \beta_3 + \gamma + 9\pi)/2}} \ {\vc{\InputIfFileExists{./figures/ZQ/q_eur.tikz}{}{Missing file!}}}
	\end{align}
\end{proposition}

\begin{proof} \label{prfPropZQTranslationEU}
	\begin{align}
		LHS \by{I_q, Q}              & \node{qn}{\left(\alpha_1, z\right) \times H \times \left(\alpha_2, z\right)} \lambda_{e^{i\left(\alpha_1 + \alpha_2 + \pi\right)/2}} \\
		RHS =                        &
		\overset{\lambda_{e^{i\left(\beta_1 + \beta_2 + \beta_3 + \gamma + 9\pi\right)/2}}}{{\vc{\InputIfFileExists{./figures/ZQ/q_HHH.tikz}{}{Missing file!}}}}
		{\vc{\InputIfFileExists{./figures/ZQ/q_pigamma.tikz}{}{Missing file!}}}
		\node{qn}{H \times \left(\beta_1, z\right) \times H  \times \left(\beta_2,z\right) \times H \times \left(\beta_3, z\right) \times H} \\
		\by{\ref{lemZQEUQuaternion}} &
		\lambda_{e^{i\left(\beta_1 + \beta_2 + \beta_3 + \gamma + 9\pi\right)/2}}
		{\vc{\InputIfFileExists{./figures/ZQ/q_HHH.tikz}{}{Missing file!}}}
		{\vc{\InputIfFileExists{./figures/ZQ/q_pigamma.tikz}{}{Missing file!}}}
		\node{qn}{\left(\alpha_1, z\right) \times H \times \left(\alpha_2, z\right)}  \\
		\by{\ref{lemZQGamma}}        &
		\lambda_{e^{i\left(\beta_1 + \beta_2 + \beta_3 + \gamma + 9\pi\right)/2}}
		{\vc{\InputIfFileExists{./figures/ZQ/q_HHH.tikz}{}{Missing file!}}}
		\lambda_{-\sqrt{2}\left(e^{i\gamma/2}\right)}
		\node{qn}{\left(\alpha_1, z\right) \times H \times \left(\alpha_2, z\right)} \\
		\by{\ref{propZQHHH}}         &
		\lambda_{e^{i\left(\beta_1 + \beta_2 + \beta_3 + \gamma + 9\pi\right)/2}}
		\lambda_{\frac{i}{\sqrt{2}}}
		\lambda_{-\sqrt{2}\left(e^{i\gamma/2}\right)}
		\node{qn}{\left(\alpha_1, z\right) \times H \times \left(\alpha_2, z\right)} \\
		\by{\ref{lemZQQLambdas}}     &
		\lambda_{e^{i\left(\alpha_1 + \alpha_2 + \pi\right)/2}}
		\node{qn}{\left(\alpha_1, z\right) \times H \times \left(\alpha_2, z\right)}
	\end{align}
\end{proof}

We have shown that for every rule $L=R$ in ZX, $\ZQ \entails F_Q\left(L\right) = F_Q\left(R\right)$.
We have therefore shown that if $\ZX \entails D_1 = D_2$ then $\ZQ \entails F_Q\left(D_1\right) = F_Q\left(D_2\right)$.

\subsection{From ZQ to ZX and back again} \label{secCompletenessZXZQBack}

It remains to be shown that $\ZQ \entails F_Q\left(F_X\left(D\right)\right) = D$

\begin{proposition} \label{propZQRetranslateZSpider} Re-translating the Z spider
	\begin{align}
		ZQ \entails F_Q\left(F_X\left(\spider{smallZ}{}\right)\right) = \spider{smallZ}{}
	\end{align}
\end{proposition}

\begin{proof} \label{prfPropZQRetranslateZSpider}
	\begin{align}
		LHS = F_Q\left(\spider{gn}{}\right) \by{\ref{lemZQTranslationZ}} \spider{smallZ}{}
	\end{align}
\end{proof}

The following lemmas are necessary for the re-translation of the $Q$ node in Proposition~\ref{propZQRetranslateQ}.

\begin{lemma} \label{lemZQRetranslateTripleBlobs}
	\begin{align}
		ZQ \entails F_Q\left(\tripleblobs\right) = \lambda_{1 / \sqrt{2}}
	\end{align}
\end{lemma}

\begin{proof} \label{prfLemZQRetranslateTripleBlobs}
	\begin{align}
		F_Q\left(\tripleblobs\right) \by{M} \lambda_{-i} \vc{\InputIfFileExists{./figures/ZQ/q_HHH.tikz}{}{Missing file!}} \by{M, \ref{propZQHHH}} \lambda_{1 / \sqrt{2}}
	\end{align}
\end{proof}

\begin{lemma} \label{propZQRetranslateGammaAlphaPi}
	\begin{align}
		ZQ \entails F_Q\left(\galpharpi{\gamma}\right) = \lambda_{-\sqrt{2}e^{i\gamma/2}}
	\end{align}
\end{lemma}

\begin{proof} \label{prfPropZQRetranslateGammaAlphaPi}
	\begin{align}
		F_Q\left(\galpharpi{\gamma}\right) \by{Y} \vc{\InputIfFileExists{./figures/ZQ/q_pigamma.tikz}{}{Missing file!}} \lambda_{e^{i\gamma/2}} \lambda_{i} \lambda_{i} \by{\ref{lemZQGamma}, M} \lambda_{\sqrt{2}e^{i\gamma}}
	\end{align}

\end{proof}

\begin{lemma} \label{lemZQRetranslateZNode}
	\begin{align}
		ZQ \entails F_Q\left(\node{gn}{\alpha}\right) = \node{qn}{\alpha,z} \lambda_{e^{i\frac{\alpha}{2}}}
	\end{align}
\end{lemma}

\begin{proof} \label{prfLemZQRetranslateZNode}
	\begin{align}
		LHS = & \node{smallZ}{} \comp \node{qn}{\alpha,z} \comp \node{smallZ}{} \lambda_{e^{i\frac{\alpha}{2}}} =\node{qn}{\alpha,z}  \lambda_{e^{i\frac{\alpha}{2}}}
	\end{align}
\end{proof}

\begin{lemma} \label{lemZQRetranslateXNode}
	\begin{align}
		ZQ \entails F_Q\left(\node{rn}{\alpha}\right) = \node{qn}{\alpha,x} \lambda_{e^{i\frac{\alpha}{2}}}
	\end{align}
\end{lemma}

\begin{proof} \label{prfLemZQRetranslateXNode}
	\begin{align}
		LHS \by{M} & \node{qn}{H} \comp\node{smallZ}{} \comp  \node{qn}{\alpha,z} \comp\node{smallZ}{} \comp  \node{qn}{H} \lambda_{-e^{i\frac{\alpha}{2}}} \\
		\by{Q}     & \node{qn}{H \times \left(\alpha,z\right) \times H}  \lambda_{-e^{i\frac{\alpha}{2}}}
		\by{\ref{lemZQHComm}}
		\node{qn}{H \times H \times \left(\alpha,x\right)}  \lambda_{-e^{i\frac{\alpha}{2}}} \\
		           & \by{\ref{propZQHH}, N} \node{qn}{\alpha,x}  \lambda_{e^{i\frac{\alpha}{2}}}
	\end{align}
\end{proof}

\begin{proposition} \label{propZQRetranslateQ}
	\begin{align}
		ZQ \entails F_Q\left(F_X\left(\node{qn}{q}\right)\right) = \node{qn}{q}
	\end{align}
	Where $q = \left(\alpha_1, z\right) \times \left(\alpha_2, x\right) \times \left(\alpha_3 , z\right)$, as in Proposition~\ref{propZQQDecomp}
\end{proposition}

\begin{proof} \label{prfPropZQRetranslateQ}
	\begin{align}
		LHS =                                                                       & F_Q\left(
		\node{gn}{\alpha_1} \comp  \node{rn}{\alpha_2} \comp  \node{gn}{\alpha_3} \right) \tensor
		\galpharpi{-\alpha_1/2}
		\galpharpi{-\alpha_2/2}
		\galpharpi{-\alpha_3/2}
		\tripleblobs \tripleblobs \tripleblobs \\
		\by{\ref{propZQRetranslateGammaAlphaPi}, \ref{lemZQRetranslateTripleBlobs}} &
		F_Q\left(
		\node{gn}{\alpha_1} \comp  \node{rn}{\alpha_2} \comp  \node{gn}{\alpha_3}\right) \tensor
		\lambda_{\sqrt{2}e^{-i\alpha_1/2}}
		\lambda_{\sqrt{2}e^{-i\alpha_2/2}}
		\lambda_{\sqrt{2}e^{-i\alpha_3/2}}
		\lambda_{\frac{1}{\sqrt{2}}} \lambda_{\frac{1}{\sqrt{2}}} \lambda_{\frac{1}{\sqrt{2}}} \\
		\by{M}                                                                      & F_Q\left(
		\node{gn}{\alpha_1} \comp  \node{rn}{\alpha_2} \comp  \node{gn}{\alpha_3}\right) \tensor
		\lambda_{e^{-\frac{i}{2}\left(\alpha_1 + \alpha_2 + \alpha_3\right)}} \\
		\by{\ref{lemZQRetranslateXNode}, \ref{lemZQRetranslateZNode}}               &
		\left(\node{qn}{\alpha_1,z} \comp  \node{qn}{\alpha_2,x} \comp  \node{qn}{\alpha_3,z}\right) \tensor
		\lambda_{e^{-\frac{i}{2}\left(\alpha_1 + \alpha_2 + \alpha_3\right)}} \lambda_{e^{i\alpha_1/2}} \lambda_{e^{i\alpha_2/2}} \lambda_{e^{i\alpha_3/2}} \\
		\by{M}                                                                      & \node{qn}{\alpha_1,z} \comp  \node{qn}{\alpha_2,x} \comp  \node{qn}{\alpha_3,z}                      \\
		\by{Q}                                                                      & \node{qn}{\left(\alpha_1, z\right) \times \left(\alpha_2, x\right) \times \left(\alpha_3 , z\right)}
	\end{align}
\end{proof}

Finally we need the following lemma for Proposition~\ref{propZQRetranslateLambda}.

\begin{lemma} \label{propZQRetranslateRG}
	\begin{align}
		ZQ \entails F_Q\left(\rg{\beta}{-\beta}\right) = \lambda_{\sqrt{2}\cos \beta}
	\end{align}
\end{lemma}

\begin{proof} \label{prfPropZQRetranslateRG}
	\begin{align}
		LHS =  & {\vc{\InputIfFileExists{./figures/ZQ/q_rgbb.tikz}{}{Missing file!}}} \lambda_{e^{i\beta_2}} \lambda_{e^{-i\beta_2}} \lambda_i
		\by{M} {\vc{\InputIfFileExists{./figures/ZQ/q_rgbb.tikz}{}{Missing file!}}} \lambda_i
		\by{Q} {\vc{\begin{tikzpicture}
	\begin{pgfonlayer}{nodelayer}
		\node [style=qn] (0) at (0, 0) {$(\beta,z) \times H \times (-\beta , z)$};
		\node [style=Z] (3) at (0, 0.5) {};
		\node [style=Z] (6) at (0, -0.5) {};
	\end{pgfonlayer}
	\begin{pgfonlayer}{edgelayer}
		\draw (3.center) to (6.center);
	\end{pgfonlayer}
\end{tikzpicture}
}} \lambda_i \\
		\by{A} & \lambda_{\sqrt{2}\cos \beta}
	\end{align}
\end{proof}

\begin{proposition} \label{propZQRetranslateLambda}
	\begin{align}
		ZQ \entails F_Q\left(F_X\left(\lambda_{c}\right)\right) = \lambda_{c}
	\end{align}
\end{proposition}

\begin{proof} \label{prfPropZQRetranslateLambda}
	\begin{align}
		c =                            & \left(\sqrt{2}\right)^n e^{i\alpha} \cos \beta \quad \text{ for some } \alpha,\,\beta                \\
		LHS =                          & F_Q\left(\galpharpi{\alpha} \halfblobs \left(\rg{}{\pi}\right)^{\tensor n} \rg{\beta}{-\beta}\right)
		\by{\ref{propZQRetranslateGammaAlphaPi}, \ref{propZQHHH}} \lambda_{\sqrt{2}e^{i\alpha}} \lambda_{\frac{1}{\sqrt{2}}} \lambda_{\frac{1}{\sqrt{2}}} \left(\lambda_{\sqrt{2}}\right)^{\tensor n} F_Q\left(\rg{\beta}{-\beta}\right) \\
		\by{\ref{propZQRetranslateRG}} &
		\lambda_{\sqrt{2}e^{i\alpha/2}}
		\lambda_{\frac{1}{\sqrt{2}}}
		\lambda_{\frac{1}{\sqrt{2}}}
		\left(\lambda_{\sqrt{2}}\right)^{\tensor n}
		\lambda_{\sqrt{2}\cos\beta}
		\by{M}  \lambda_{\left(\sqrt{2}\right)^n e^{i\alpha} \cos \beta} = \lambda_c
	\end{align}
\end{proof}

We have shown that for each of the generators of ZQ, $\ZQ \entails F_Q\left(F_X\left(G\right)\right) = G$,
and since $F_Q$ and $F_X$ are monoidal functors we know that $\ZQ \entails F_Q\left(F_X\left(D\right)\right) = D$ for any diagram $D$.
This concludes our proof of completeness for the rules of ZQ.

\section{Summary}

This chapter introduced ZQ,
a complete graphical calculus that extends the phase group concept of ZX.
ZQ cannot fit into the categorical framework
of $\Sigma$ Graphical Calculi given in the previous chapter,
because the group action of ZQ cannot be expressed as a monoid
(and therefore also cannot be expressed as a spider)
and is the first qubit graphical calculus with a non-commutative phase group.
ZQ allows for the expression of arbitrary single qubit rotations in the form of quaternions,
similar to the compiler TriQ.

This expression of arbitrary rotations in ZQ allows for a clearer presentation
of the Euler Angle structure than ZX exhibits in the rule (EU').
The author feels this clarity,
and the clarity of the ZQ rules in general,
lends ZQ to replacing ZX for pedagogy and research purposes.

In the following chapters we will return to the idea of conjecture synthesis,
exploring how conjecture synthesis interacts with our new calculi,
and how we can use these results to further
improve our inference and verification methods.
Future work based on this chapter would include
implementing ZQ in a graphical proof assistant,
or exploring the category of models of $\groupprop_G$.
The author also hopes that ZQ will act as a bridge between the optimisation
methods of both TriQ and of ZX (such as Refs~\cite{Backens2020circuit}, \cite{Kissinger2019Reducing}).  \thispagestyle{empty} 
\chapter{Conjecture Inference} \label{chapConjectureInference}
\thispagestyle{plain}
\noindent\emph{In this chapter:}
\begin{itemize}
	\item[\chapterbullet] We introduce conjecture inference for graphical calculi
	\item[\chapterbullet] We provide a link between phase algebras and algebraic geometry
	\item[\chapterbullet] We use this geometric link to infer rules from sample instances of those rules
	\item[\chapterbullet] We present a simple method for adding !-boxes and phase variables
	      to diagram boundaries
\end{itemize}

Graphical Conjecture Synthesis\footnote{Outside of the contributions of this thesis. See \S\ref{chapCoSy}.} produces simple equational theorems,
i.e. theorems of the form $L=R$
where $L$ and $R$ are diagrams containing neither !-boxes nor phase variables
(see \S\ref{secSimpleDiagrams}).
In all the calculi we have encountered
their complete rulesets can be expressed using either an infinite number
of simple rules, or a finite number (around 20) of decorated rules.
We therefore would like to know the answer to
`how can we find rules containing !-boxes and phase variables?'
We break this down into two situations:
\begin{itemize}
	\item Finding these rules with certainty, and
	\item Finding these results in a way we can easily verify (e.g. by checking a few more examples).
\end{itemize}
This chapter covers the inference aspect of this process,
with the next chapter (\S\ref{chapConjectureVerification}) devoted to the task of verification.
Our aim in conjecture inference is given in the definition below,
but we note that it is simply a small part of the greater philosophical
program called `The Problem of Induction'.
This Problem began with David Hume \cite[\S1.iii.6]{Hume}
and concerns the construction of new ideas from existing evidence.

\begin{definition}[Conjecture Inference] \label{defConjectureInference}
	Conjecture Inference is the algorithmic act of hypothesising
	families of theorems, given a few examples from that family.
\end{definition}

In this chapter we also demonstrate a link between phase algebras and Algebraic Geometry (\S\ref{secInferringPhaseVariables}).
This link is novel, useful, and neatly generalises existing concepts in the literature.
We will use this link to clarify the relationship between rules with and without side conditions,
explore concepts like rebasing of variables,
and then rely on geometric ideas as a core part of our inference process.
This geometric framework builds on
the earlier ideas of phase homomorphism pairs (\S\ref{chapPhaseRingCalculi}),
where these pairs will be rephrased as symmetries of the geometric space.
This idea will also be carried forward
when we come to discuss verification:
The geometric interpolation discussed in this chapter
is directly related to the interpolation discussed
in the next,
and we combine these ideas with the above symmetries
(and some Galois Theory)
for a pleasantly surprising result in \S\ref{secPhaseVariablesOverQubits}.

\section{Adding parameters at boundaries} \label{secInferenceDeductive}

There is a very simple way to hypothesise adding !-boxes and phase variables to certain boundary spiders in ZX, ZH, ZW, and ZQ rules,
and what's more the method comes with a proof of correctness,
meaning that we don't need to verify any more examples.
For each of our calculi we can use spider laws to infer more general results from simple equations.
These results are included in part for posterity
(any proof assistant that uses these calculi should ideally be capable of automatically inferring variations on Proposition~\ref{propZXAddBang})
but also as a contrast to the methods of \S\ref{secInferringPhaseVariables}.
The hypotheses inferred in \S\ref{secInferringPhaseVariables} do not come with proofs of correctness
and must be verified in some manner (the topic of \S\ref{chapConjectureVerification}).

In order to signify that this process can be applied at the boundary of a diagram
we use the following to signify an arbitrary diagram $D$:

\begin{align}
	\spider{box}{D}
\end{align}

\begin{proposition} \label{propZXAddBang}
	A ZX$_G$ equation of the form
	\begin{align}
		\vc{\InputIfFileExists{./figures/ZX/bbox_add_tl.tikz}{}{Missing file!}}
		=
		\vc{\InputIfFileExists{./figures/ZX/bbox_add_tr.tikz}{}{Missing file!}} &
		\begin{split}
			b, c \text{ fixed phases} \in G
			\\ \text{non-zero, fixed number of inputs}
			\\ D, D' \text{ arbitrary diagrams}
		\end{split}
	\end{align}
	implies and is implied by the family of ZX$_G$ equations parameterised over the phase variable $\alpha$ and a !-box:
	\begin{align}
		\family{
			\vc{\InputIfFileExists{./figures/ZX/bbox_add_bl.tikz}{}{Missing file!}}
			=
			\vc{\InputIfFileExists{./figures/ZX/bbox_add_br.tikz}{}{Missing file!}}
		}_{\alpha, !} &
		\begin{split}
			\alpha \text{ phase variable } \in G
			\\ \text{!-box on inputs}
		\end{split}
	\end{align}
	This holds for Z and for X spiders, providing the spider is the same colour on both sides of the rule.
\end{proposition}

\begin{proof}
	We can apply a spider $s$ with phase $(\alpha - b)$ to both sides of the top rule $r$.
	The spider $s$ has as many outputs as $r$ has inputs, and can have any number of inputs.
	For the converse we specify the !-box to have has many instances as $r$ has inputs,
	and we instantiate $\alpha$ to the value $b$.
\end{proof}

\begin{remark} \label{remAddingBBoxGeneralising}
	Proposition~\ref{propZXAddBang} and its proof both generalise easily to the spiders of ZW.
	For ZQ one can generalise either the phase of a $Q$ node, or add a !-box to a Z-spider,
	depending on which is joined to the boundary.
	In $\ring_R$ one can add a !-box to the inputs or outputs
	of the multiplication spider, or to the inputs of the addition spider.
	In both cases you can also generalise the phase.
	The Z-spider of ZH has a phase-free analogue of Proposition~\ref{propZXAddBang},
	however the H-box result requires a little more work and is presented after Lemma~\ref{lemZHNBox}.
\end{remark}

\begin{restatable}[n-joined ZH spider law]{lemma}{lemZHNBox}
	\label{lemZHNBox}
	In the ZH calculus:
	\begin{align}
		\vc{\InputIfFileExists{./figures/ZH/ZH_S1_gen_l.tikz}{}{Missing file!}}
		=
		\vc{\InputIfFileExists{./figures/ZH/ZH_S1_gen_r.tikz}{}{Missing file!}}
		  &
		\begin{split}a,b \in \mathbb{C} \less \set{1}
			\\ n > 0
		\end{split}
	\end{align}
\end{restatable}

Since the proof isn't relevant to the discussion of conjecture inference
we defer it to \S\ref{secnfoldZH}.

\begin{proposition} \label{propZHbbox}
	A simple ZH equation of the form:
	\begin{align} \label{eqnPropZHbbox}
		\vc{\InputIfFileExists{./figures/ZH/bbox_add_zh_tl.tikz}{}{Missing file!}}
		=
		\vc{\InputIfFileExists{./figures/ZH/bbox_add_zh_tr.tikz}{}{Missing file!}}
		  &
		\begin{split}
			b, c \text{ fixed phases } \in \bb{C} \less \set{1}
			\\n \text{, non-zero number of inputs}
			\\D, D' \text{ arbitrary diagrams}
		\end{split}
	\end{align}
	implies and is implied by the family of equations parameterised over $\alpha$ and a !-box:
	\begin{align}
		\family{
			\vc{\InputIfFileExists{./figures/ZH/bbox_add_zh_bl.tikz}{}{Missing file!}} =
			\vc{\InputIfFileExists{./figures/ZH/bbox_add_zh_br.tikz}{}{Missing file!}}
		}_{\alpha, !}
		  &
		\begin{split}
			\alpha \in \mathbb{C}
			\\\text{!-box on inputs}
		\end{split}
	\end{align}
\end{proposition}

\begin{proof}
	Apply the following diagram to the bottom of each side of Equation~\eqref{eqnPropZHbbox}:

	\begin{equation}
		\vc{\InputIfFileExists{./figures/ZH/bbox_add_proof.tikz}{}{Missing file!}}
	\end{equation}

	Then apply lemma \ref{lemZHNBox}.
	As with the ZX case we justify the !-box by saying that the above is true for any number of inputs and value of $\alpha$.
	For the converse we specify the !-box to have has many instances as needed, and we set $\alpha$ to $b$.
\end{proof}


\section{Linear rules, parameter spaces, and algebraic geometry} \label{secInferringPhaseVariables}

We move now to a method of conjecture inference that does not come with a proof of correctness.
The fundamental idea behind this method is to transform the problem of conjecture inference
into a problem in algebraic geometry.
Rather than seeing a parameterised family of diagrams as a \emph{set} of diagrams
we see them as points defining a surface in some affine space.
The question of conjecture inference then becomes `given a sample of these points, can I reconstruct the surface?'
In order to frame this question correctly we begin with some necessary definitions (\S\ref{secSkeletons}),
show how to view parameterised families as submodules (\S\ref{secEquationsAsSubmodules}),
show how to get back to parameterised equations from submodules (\S\ref{secSubmodulesAsEquations}),
and finally show how to find submodules from sample points (\S\ref{secFindingSubmodules}).

One way of thinking about the content of this section is that in the usual presentation
of our calculi and rules we have as much information as possible
in the algebra written as phases, and as little information as possible in the side conditions.
This section describes how to move from the algebraic paradigm
into a geometric one; describing everything as restrictions on a space (i.e. as side conditions).
In the case of linear rules over $\bbZ$-modules we show how to go in the other direction
as well (from geometry to algebra),
describing the linear geometric constraints using just the algebra in the phases.
This concept neatly generalises the idea of linear rules for ZX already present in the literature.

\subsection{Skeletons and parameter spaces}

We will assume for the rest of this chapter that the phase algebra we are considering is $R$-linear for some commutative ring $R$.
In the case of ZX we would consider the phase group as a $\bbZ$-module, for $\ring_R$, $\ZH_R$ and $\ZW_R$
we would consider the phase ring as an $R$-module.
$\ZQ$ can be considered as an $\bbR$-module,
where a phase is an element of $\bbR^4$,
but only some of the points in the space result in valid phases.

The key to this section is the novel separation of a diagram
into its skeleton 
and parameter space.
The skeleton holds all the graph information except for the phases,
and the parameter space holds the information from the phase algebra.

\label{secSkeletons}
\begin{definition}[Skeleton]
	The \emph{skeleton} of a diagram, written $|D|$, is the phase-free version of that diagram.
	The skeleton of a rule is the pair of skeletons from its two diagrams.
\end{definition}

\begin{example}[ZX diagram skeleton] \label{exaSkeletonOfADiagram}
	An example diagram and skeleton:
	\begin{align}
		D \quad& \spider{gn}{\alpha+\pi} & \abs{D} \quad& \spider{gn}{}
	\end{align}
\end{example}

\begin{definition}[Parameter space] \label{defParameterSpace}
	The \emph{parameter space} of a skeleton is the direct product (in $R$-mod)
	of all the phase algebras of the nodes in the skeleton.

	\begin{align}
		P := \prod_{\text{nodes in the skeleton}} \set{\text{Phase algebra of that node}}
	\end{align}

	The parameter space of a rule is the direct product of the parameter spaces of each diagram.
\end{definition}

\begin{remark} \label{remTrivialParameterSpace}
	Some nodes contribute trivially to the parameter space.
	For example the Z spider in ZH is always blank.
	We treat this as the trivial $R$-module, $\set{0}$,
	which contributes nothing to the overall product $P$.
	As such these cases will be ignored from here on.
\end{remark}

\begin{remark} \label{remZQParameter}
	In the case of calculi such as ZQ, where the phase algebra is a four-dimensional $\bbR$-module,
	the parameter space can have many more dimensions than the rule has vertices.
	In later subsections of this chapter this may require viewing the phase of a $Q$ node as
	four distinct, real numbers rather than as a single quaternion,
	but the alterations that need to be made to the methods are small and will not be mentioned again.
	In essence they are `treat this node labelled by an element from $R^n$ as $n$ nodes each labelled by an element of $R$'.
\end{remark}

\begin{example}[Clifford+T ZX parameter space]
	Consider, as an example, the left hand side of a simplified Clifford+T ZX spider law:
	\begin{center}
		\begin{tabular}{cc}
			$D$                 & $\abs{D}$                    \\
			$\vc{\InputIfFileExists{./figures/ZX/S1_l.tikz}{}{Missing file!}}$ & $\vc{\InputIfFileExists{./figures/ZX/S1_l_skeleton.tikz}{}{Missing file!}}$
		\end{tabular}
	\end{center}
	This skeleton has parameter space:
	\begin{align}
		<\pi / 4> \times <\pi / 4> \quad \simeq \quad \bbZ_8 \times \bbZ_8
	\end{align}
	Which is the product of the two cyclic phase groups, viewed as $\bbZ$-modules, one for each node.
	Assuming we have made (or retained) a choice of which phase group corresponds to which node
	then every point in the parameter space corresponds to a diagram.
	For example the point $p = (\pi/2, 3\pi/4)$ in the parameter space corresponds to the diagram:
	\begin{align}
		\vc{\InputIfFileExists{./figures/ZX/S1_l_pi2_3pi4.tikz}{}{Missing file!}}
	\end{align}
\end{example}

\begin{definition}[Sound points in a parameter space]  \label{defParameterSpaceSound}
	A point $p$ in a parameter space is \emph{sound} if the corresponding rule,
	with phases as described by $p$, is sound (see Definition~\ref{defSimpleSound}).
	Any point in $P$ that does not correspond to a sound rule is \emph{unsound}.
\end{definition}

\subsection{Equation families as submodules} \label{secEquationsAsSubmodules}

In this subsection we view a family of equations parameterised by phase variables
as a collection of simple rules that all have the same skeleton,
and a collection of sound points in the parameter space.
This will allow us to rephrase families parameterised by phase variables as submodules of a parameter space.

\begin{example}[Parameter space of the spider rule] \label{exaParameterSpace}
	Here is a (slightly unusual) way of presenting the simplified spider rule in $\ZX_G$:

	\begin{align} \label{eqnSpiderPedantic}
		\vc{\InputIfFileExists{./figures/ZX/S1_l_skeleton_vars.tikz}{}{Missing file!}} = \vc{\InputIfFileExists{./figures/ZX/S1_r_skeleton_vars.tikz}{}{Missing file!}} &
		\begin{split}(\alpha, \beta,\gamma) \in P \\ \gamma = \alpha + \beta \end{split}
	\end{align}

	We have given the parameterisation as a collection of points in the parameter space $P$,
	subject to the condition $\alpha + \beta = \gamma$.
\end{example}

\begin{remark} \label{remParameterSpacesCommon}
	This presentation of rules is not uncommon; see for example the $EU$ rule of \cite[\S{}2.3]{VilmartZX}.
\end{remark}

We will now introduce the bare minimum of algebraic geometry; the notion of a \emph{variety}
or \emph{zero set}. It is this tool that links the two notions of algebra (in the form of a system of equations)
and geometry (in the form of a topological shape).

\begin{definition}[Affine algebraic variety {\cite[\S1.1]{InvitationAlgGeom}}]  \label{defAlgebraicVariety} 
	A complex \emph{affine algebraic variety} is the common zero set of a collection $\set{F_i}_{i \in I}$
	of complex polynomials on complex $n$-space $\bbC^n$.
\end{definition}

\begin{definition}[Linear rule] \label{defLinearRule}
	A parameterised family of rules is \emph{linear} if
	the sound points in $P$ are the zero set of a collection of linear polynomials on $P$.
\end{definition}

\begin{remark}[Parameter space restrictions as side conditions] \label{remLinearButIllDefined}
For a calculus like ZQ, where only certain points in the parameter space
correspond to valid diagrams, we impose these same
restrictions when talking about rules. I.e. we
say that a parameterised family of rules is linear
if the sound and valid points in $P$
are the intersection of the zero set of a collection of linear polynomials on $P$
with the valid points in $P$.
We could even define ZQ in such a way that
these `valid points in $P$' were phrased as side conditions
on the generators and rules.
\end{remark}

This definition is, on the surface, different to the definition of a linear rule given in \cite{BeyondCliffordT},
which we reproduce (in language consistent with the rest of this thesis) as:

\begin{definition}[Previous definition of linear rule {\cite[Definition~1]{BeyondCliffordT}}] \label{defOldLinearRule} 
	A ZX-diagram is linear in $\alpha_1, \dots, \alpha_k$ with constants in $C \subset \bbR$,
	if it is generated by Z-spiders with phase $E$, X-spiders with phase $E$,
	the Hadamard node, and compact closed wire structure, where E is of the form
	$\sum_i n_i \alpha_i + c$ with $n_i \in \bbZ$ and $c \in C$.
\end{definition}

Thankfully we can show these definitions coincide,
starting with:

\begin{proposition} \label{propLinearRules}
	Every linear ZX rule in the sense of Definition~\ref{defOldLinearRule} is a linear rule in the sense of
	Definition~\ref{defLinearRule}.
\end{proposition}

\begin{proof} \label{prfPropLinearRules}
	Give every spider in the linear rule an index $j \in J$, which as per Definition~\ref{defOldLinearRule}
	has phase $\sum_{i} n_{i,j} \alpha_i + c_j$.
	Assign to each spider the parameter $\beta_j$,
	and construct the system of linear equations (as in Definition~\ref{defLinearRule})
	as the set $\set{\beta_j = \sum_{i} n_{i,j} \alpha_i + c_j}_{j \in J}$
\end{proof}

\begin{remark} \label{remConverseToLemLinearRules}
	The converse to Proposition~\ref{propLinearRules} is given as Corollary~\ref{corConcidingLinearDefinitions}
	later on in this chapter.
	While this shows that the two definitions are equivalent in the case of $\ZX_G$,
	we prefer the geometric definition because it can be generalised to other calculi;
	for example it is unclear how to generalise Definition~\ref{defOldLinearRule} to the case of ZQ.
\end{remark}

\begin{example}[Linear Clifford+T ZX rule] \label{exaLinearRule}
	The following Clifford+T $\ZX$ rule is linear in the sense of Definition~\ref{defLinearRule}:
	\begin{align}
		\vc{\begin{tikzpicture}
	\begin{pgfonlayer}{nodelayer}
		\node [style=none] (l) at (-1, 0) {};
		\node [style=none] (r) at (2, 0) {};
		\node [style=gn] (0) at (0, 0) {$\alpha$};
		\node [style=rn] (1) at (1, 0) {$\beta$};
	\end{pgfonlayer}
	\begin{pgfonlayer}{edgelayer}
		\draw (l.center) to (r.center);
	\end{pgfonlayer}
\end{tikzpicture}
} = \vc{\begin{tikzpicture}
	\begin{pgfonlayer}{nodelayer}
		\node [style=none] (l) at (-1, 0) {};
		\node [style=none] (r) at (2, 0) {};
		\node [style=gn] (0) at (1, 0) {$\delta$};
		\node [style=rn] (1) at (0, 0) {$\gamma$};
	\end{pgfonlayer}
	\begin{pgfonlayer}{edgelayer}
		\draw (l.center) to (r.center);
	\end{pgfonlayer}
\end{tikzpicture}
} &
		\begin{split}
			\set{\alpha = 0, \delta = 0, \beta - \gamma = 0}
		\end{split}
	\end{align}
	The affine, linear $\bbZ$-submodule described by the linear polynomials is $\bbZ(0,\pi/4,\pi/4,0)$.
\end{example}

\begin{example}[Rule with multiple linear subspaces] \label{exaMultipleLinearRules}
	The following Clifford+T $\ZX$ rule (although very similar to Example~\ref{exaLinearRule}) is not linear, but the union of four different sound (up to scalar), linear rules
	on the same skeleton:
	\begin{align}
		\vc{\begin{tikzpicture}
	\begin{pgfonlayer}{nodelayer}
		\node [style=none] (l) at (-1, 0) {};
		\node [style=none] (r) at (2, 0) {};
		\node [style=gn] (0) at (0, 0) {$\alpha$};
		\node [style=rn] (1) at (1, 0) {$\beta$};
	\end{pgfonlayer}
	\begin{pgfonlayer}{edgelayer}
		\draw (l.center) to (r.center);
	\end{pgfonlayer}
\end{tikzpicture}
} = \vc{\begin{tikzpicture}
	\begin{pgfonlayer}{nodelayer}
		\node [style=none] (l) at (-1, 0) {};
		\node [style=none] (r) at (2, 0) {};
		\node [style=gn] (0) at (1, 0) {$\delta$};
		\node [style=rn] (1) at (0, 0) {$\gamma$};
	\end{pgfonlayer}
	\begin{pgfonlayer}{edgelayer}
		\draw (l.center) to (r.center);
	\end{pgfonlayer}
\end{tikzpicture}
} &
		\begin{split}
			\set{\alpha = 0, \delta = 0, \beta - \gamma = 0}  \\
			\set{\beta = 0, \gamma = 0, \alpha - \delta = 0} \\
			\set{\alpha - \pi = 0, \delta - \pi = 0, \beta + \gamma = 0} \\
			\set{\beta - \pi = 0, \gamma - \pi = 0, \alpha + \delta = 0}
		\end{split}
	\end{align}

	This is because it is sound on the following linear varieties of $P$:

	\begin{align}
		\text{defining polynomials}                  & \qquad \text{affine submodule} \nonumber                            \\
		\alpha = 0, \delta = 0, \beta = \gamma       & \qquad (0,0,0,0) + \bbZ(0,\pi/4,\pi/4,0)  \label{eqnEarlierVariety} \\
		\beta = 0, \gamma = 0, \alpha = \delta       & \qquad (0,0,0,0) + \bbZ(\pi/4, 0, 0, \pi/4)                         \\
		\alpha = \pi, \delta = \pi, \beta  = -\gamma & \qquad (\pi,0,0,\pi) + \bbZ(0,\pi/4,-\pi/4,0)                       \\
		\beta = \pi, \gamma = \pi, \alpha = -\delta  & \qquad(0,\pi,\pi,0) + \bbZ(\pi/4,0,0,-\pi/4)
	\end{align}
	For example the rule given by this skeleton and just the variety given by \eqref{eqnEarlierVariety} is Example~\ref{exaLinearRule}.
\end{example}

\begin{remark} \label{remFourLinearSubspaces}
	The rule in Example~\ref{exaMultipleLinearRules} can be seen as applications of the spider and
	$\pi$-commutation rules,
	and their colour-swapped counterparts,
	(see (S2) and (K) of Ref.~\cite{BeyondCliffordT}) representing the four linear subspaces we have exhibited.
\end{remark}

Having shown how to view families of equations as submodules
we now work backwards; from submodules to families of equations.
The benefit of the submodule viewpoint is deferred until
\S\ref{secFindingSubmodules},
where we examine how to hypothesise larger submodules
containing a known collection of points.

\subsection{Submodules as families of equations} \label{secSubmodulesAsEquations}

Assume we have been given a sound, affine submodule $q+Q$ of $P$.
In the case where $Q$ is finitely generated,
and $R$ admits an Euclidean function,
we can find the Hermite Normal Form to simplify relationships between generators.
(In the case where $R$ is a field we can do Gaussian Elimination, and the following is even simpler.)

\begin{proposition} \label{propAlgorithmForSubmodule}
	It is possible to describe any affine, linear $R$-submodule of a parameter space
	as a single diagram with phases that are linear in some set of phase variables.
	Additionally we can ensure the set of phase variables
	is of minimal cardinality.
\end{proposition}

\begin{proof} \label{prfPropAlgorithmForSubmodule} 
	We provide here an algorithm that assigns to each vertex such a phase.
	The input for the algorithm is the rule skeleton, and a description of the finitely generated, linear $R$-submodule $q+Q$.

	We treat $Q$ as a matrix where each row corresponds to a generator,
	and each column corresponds to a vertex in the rule.
	We then find a linearly independent set of generators for the space $Q$ by putting the matrix $Q$ into Hermite Normal Form.
	I.e. we find matrices $U$ and $H$ such that $H = UQ$,
	where $U$ is a product of elementary column and row operations, and $H$ has the following properties:
	\begin{itemize}
		\item The leading coefficient of each row is not in the same column as the leading coefficient of any other row,
		\item The leading coefficient of each row has only the entry $0$ in that column in any of the rows below it,
		\item If the leading coefficient of a row has value $r$ then the values in places above it have smaller Euclidean degree than $r$,
		\item All rows containing only $0$ come after any rows that are not all $0$,
		\item All non-zero rows are linearly independent of the other non-zero rows.
	\end{itemize}
	This is a standard procedure in mathematics software (see, for example, an implementation in Ref.~\cite{Magma}),
	and all of these properties can be shown to be obtainable via Euclid's Algorithm.
	The matrix $H$ corresponds to a linearly independent set of generators of $Q$,
	i.e. if $H$ is split into rows $\set{r_j}_{j=1}^n$ then $Q = R r_1 + R r_2 + \dots + R r_n$.

	To show that all linear independent sets of generators for $Q$ are the same size:
	Given two sets of generators $B := \set{b_j}$ and $B':= \set{b'_k}$, each linearly independent, and with $\abs{B} < \abs{B'}$,
	one can construct the matrix $C$ where the $k^{th}$ row is $\set{r_j}_{j=1}^{\abs{B}}$ and $\sum_j r_j b_j = b'_k$.
	If we find the Hermite Normal Form of $C$ the result is upper triangular of shape $\abs{B}\times\abs{B'}$,
	so the final row must be all $0$s,
	and so there must be a linear dependence inside $B'$, which is a contradiction.
	Therefore any linear independent sets of generators for $Q$ are of the same size.

	We translate the matrix $H$ to a diagram by assigning to each non-zero row of $H$ a phase variable $\alpha_k$,
	and for each vertex $v_j$ in the skeleton (with corresponding column $\set{c_l}_{l=1}^m$)
	that vertex is assigned the phase $q_j + \sum_l c_l \alpha_l$, recalling that $q$ was the affine shift of the original linear subspace.
\end{proof}

\begin{corollary} \label{corConcidingLinearDefinitions}
The two definitions of a linear rule
(Definition~\ref{defLinearRule} and Definition~\ref{defOldLinearRule})
coincide. One direction is Proposition~\ref{propLinearRules},
the other is Proposition~\ref{propAlgorithmForSubmodule}.
\end{corollary}

\begin{example}[Two different subspace representations as rules]

	The skeleton of the rule S1 for $\ZX_{\pi/4}$ has two vertices on the left hand side, and one on the right,
	which we label $v_1$, $v_2$ and $v_3$

	\begin{align}
		\vc{\InputIfFileExists{./figures/ZX/S1_l_skeleton_verts.tikz}{}{Missing file!}} = \spider{gn}{v_3}
	\end{align}

	$P$ is therefore isomorphic to $ \bb{Z}_8 \times \bb{Z}_8 \times  \bb{Z}_8 $,
	and we will assume for this example that we already know that the rule is sound for $p \in q + Q$ where:

	\begin{align}
		q & := (0,0,0)                                 \\
		Q & := \mathbb{Z}(1, 0, 1) + \mathbb{Z}(0,1,1)
	\end{align}

	As in the proof of Proposition~\ref{propAlgorithmForSubmodule} we consider the $Q$ part as a matrix:

	\begin{align}
		\begin{bmatrix}
			1 & 0 & 1 \\
			0 & 1 & 1
		\end{bmatrix}
	\end{align}

	Which is already in Hermite Normal Form.
	We assign to each row a phase variable, and then add together every entry in each column:

	\begin{align}
		\begin{bmatrix}
			\alpha & 0     & \alpha \\
			0      & \beta & \beta
		\end{bmatrix} \qquad   (\alpha, \beta, \alpha + \beta)
	\end{align}

	This gives us the phases we assign to the skeleton of the rule:

	\begin{align}
		v_1                             & \mapsto \alpha                                               &
		v_2                             & \mapsto \beta                                                &
		v_3                             & \mapsto \alpha + \beta                                       &
		\vc{\InputIfFileExists{./figures/ZX/S1_l_skeleton_vars.tikz}{}{Missing file!}} & = \spider{gn}{\alpha+\beta} \label{eqnS1IndependentVariable}
	\end{align}

	Suppose instead we started with the following presentation of the same linear submodule:

	\begin{align}
		q  & := 0                                              \\
		Q' & :=\mathbb{Z}(-1, 1, 0) + \mathbb{Z}(0,1,1) \iso Q
	\end{align}

	Corresponding to the matrix:

	\begin{align}
		\begin{bmatrix}
			-1 & 1 & 0 \\
			0  & 1 & 1
		\end{bmatrix}
	\end{align}

	If we applied the same rewriting to get a family of equations (without finding Hermite Normal Form)
	then we would get:

	\begin{align}
		\begin{bmatrix}
			-\alpha & \alpha & 0     \\
			0       & \beta  & \beta
		\end{bmatrix}     \qquad
		(-\alpha, \alpha+\beta, \beta)
	\end{align}

	This gives us the `same' rule:

	\begin{align}
		v_1                              & \mapsto -\alpha                                     &
		v_2                              & \mapsto \alpha+\beta                                &
		v_3                              & \mapsto \beta                                       &
		\vc{\InputIfFileExists{./figures/ZX/S1_l_skeleton_prime.tikz}{}{Missing file!}} & = \spider{gn}{\beta} \label{eqnS1DependentVariable}
	\end{align}

	Equations \eqref{eqnS1IndependentVariable} and \eqref{eqnS1DependentVariable} are the same rule,
	in that they describe the same set of simple rules,
	but they are presented differently because they reflect the different ways of presenting the underlying affine, linear
	spaces $Q$ and $Q'$.

\end{example}

\begin{remark} \label{remMatcher}
	This use of the Hermite Normal Form is not unexpected, since this process bears a strong resemblance to the way
	the matcher works in the software Quantomatic \cite{Quantomatic}.
	The matcher is the subroutine that matches instances of rules onto diagrams,
	and in doing so needs to match variable names onto concrete values.
\end{remark}

\begin{definition}[Independent phase variables]  \label{defIndependentVariable}
	We call the phase variable $\alpha$ in some rule \emph{independent}
	if there is any vertex
	with the phase $\alpha + c$, where $c$ is a non-variable element of the phase algebra.
\end{definition}

\begin{remark} \label{remRebase}
	This `rebasing' of the affine linear subspace actually allows one to shift the independent phases
	around in a rule; in particular from one side to the other.
	Doing so in a proof assistant would allow for clearer representations of rule inverses,
	where the independent variables are presented on the left hand side of the rule.
	This can be achieved by choosing the order of the vertices in the method of Proposition~\ref{propAlgorithmForSubmodule}
	such that the vertices from the left hand side of the rule appear to the left
	(as columns in the matrix $Q$) of the vertices from the right hand side of the rule.
\end{remark}

\begin{lemma} \label{lemGaussianElim}
	If the phase algebra is a $K$-module for some field $K$ then all phase variables can be re-based as independent phase variables.
\end{lemma}

\begin{proof} \label{prfLemGaussianElim}
	Working over a field $K$ we can put the matrix $Q$ into reduced row echelon form
	via Gaussian Elimination,
	i.e. each leading coefficient is $1$ and is the only entry in that column.
	When translated into phases this results in a phase of $\alpha + c$,
	where $c$ is a non-variable element of the phase algebra.
\end{proof}

\begin{remark} \label{remLinearRulesOfZQ}
We mentioned in \S\ref{secGeneralisingTheorems}
that it seems unlikely
that it is possible to find a ruleset of Universal ZX that can be expressed using
linear rules and the group action.
This was the impetus for ZQ,
and the author notes that the rules of ZQ
presented in \S\ref{secZQRules}
use only linear rules (as an $\bbR$-algebra) and the group action.
\end{remark}

\subsection{Finding submodules} \label{secFindingSubmodules}

Now that we know how to represent linear rules as affine linear submodules,
we want to go about the business of inferring which submodules lead to sound rules.
Given a skeleton rule $|D| = |D'|$, we know it is sound for some (possibly empty) subset of the parameter space $P$.
The question is which linear affine submodules of $P$ to check for soundness.
We suggest here three ways of hypothesising larger linear submodules containing the points we already know to be sound in $P$:

\begin{itemize}
	\item \textbf{Full Linear Interpolation.} Given two points $p$ and $q$ sound in $P$:
	      Extend the points into a line $p + R(q-p)$.

	\item \textbf{Sparse Linear Interpolation.} Given two points $p$ and $q$ sound in $P$,
	      extend the points into a sparse line $p + \bbZ r$, where there is some  $a \in \bbZ$ such that $a \times r = q-p$.
	      Note that this line does not necessarily lie inside the parameter space,
	      and in such cases this technique cannot be applied.

	\item \textbf{Summation.} Given two, affine submodules $s + S$, $t + T$, sound in $P$, that intersect at $q$:
	      The affine submodules can be re-based to $q + S +(s-q) = q + S$ and $q + T + (t-q) = q + T$.
	      Consider $q + S + T$.
\end{itemize}

\begin{remark}[Higher Order Interpolation] \label{remLinearNotOnlyOption}
With the exception of the (EU) rule (and variants) all of the rules
for all of the calculi discussed in this thesis are linear,
and the above interpolation methods will only infer linear subspaces.
We can, however, still apply higher order methods,
see Ref.~\cite{MultivariateSurvey} for a survey of higher-order polynomial
interpolation techniques.
These higher-order techniques would still not be able to infer
the (EU) rule of ZX since it does not express a polynomial relationship.
\end{remark}

\begin{remark} \label{remExtrapolatingQuantifiers}
	This technique is similar to the method of
	altering the existential and universal quantifiers in term-based conjecture synthesis,
	mentioned in Example~\ref{exaAmmonQuantifiers}.
	Our method makes use of both the linear structure
	and the information from multiple theorems (sound points) at once.
\end{remark}

\begin{example}[Finding the Clifford+T ZX spider law's subspace]
	Suppose, using $\ZX_{\pi/4}$, that you know the spider law of \eqref{eqnSpiderPedantic} holds true for $p = (0,0,0)$, $q = (\pi/2,0,\pi/2)$.
	\begin{itemize}
		\item Sparse Linear Interpolation suggests the subspace $Q_1 := (0,0,0) + \bbZ(\pi/2,0,\pi/2)$
		\item Full Linear Interpolation suggests the subspace $Q_2 := (0,0,0) + \bbZ(\pi/4,0,\pi/4)$
		\item Suppose we also knew the subspaces $Q_2$ and $Q_3 := (0,0,0) + \bbZ(0, \pi/4, \pi/4)$
		      were sound,
		      then summation would suggest the subspace $Q_4:= (0,0,0) + \bbZ(\pi/4,0,\pi/4) + \bbZ(0,\pi/4,\pi/4)$,
		      which would give you every sound point in $P$.
	\end{itemize}
\end{example}

What these methods require is some non-exhaustive method
of verifying the inferred subspaces of $P$.
The chapter on verification (\S\ref{chapConjectureVerification}) answers this problem for phase variables,
and is inspired by this chapter's link to algebraic geometry,
in particular by the geometric notion of degree.

\begin{remark} \label{remNonlinearVarieties}
	We have been working with linear varieties to satisfy the constraints of ZX and ZQ.
	When working with phase rings we can consider non-linear varieties
	(i.e. higher degree polynomials in the phase variables)
	although we have not developed heuristics for which non-linear polynomials to conjecture.
\end{remark}

\subsection{Phase homomorphisms and parameter space symmetries} \label{secParameterSymmetries}

The concepts of this chapter can also be combined with that of \S\ref{secPhaseRingHomomorphisms},
where we constructed the phase homomorphism pair $(\phi, \tilde \phi)$,
where $\phi: R \to R$ (lifted to $\hat \phi$ acting on the diagrams of $\ring_R$)
and $\tilde \phi$ acting on matrices in $\bit{R}$.

\begin{proposition} \label{propRingRParameterSymmetries}
	Given a $\ring_R$ equation $D_1 = D_2$ with parameter space $P$,
	then if $(p_1, \dots, p_n) \in P$ is sound then $(\phi(p_1), \dots, \phi(p_n))$ is also sound
	for any ring endomorphism $\phi$.
\end{proposition}

\begin{proof} \label{prfPropRingRParameterSymmetries}
	The action of $(p_1, \dots, p_n) \mapsto (\phi(p_1), \dots, \phi(p_n))$ is precisely the action of $\hat \phi$ in
	Proposition~\ref{propRingHomPreservesSoundness}.
\end{proof}

This shows that any phase homomorphism pair
is a symmetry of the parameter space.
There is one final piece
to the conjecture inference puzzle,
and that is conjectures involving !-boxes.
In order to approach this problem
we first need to know about a property called the
\emph{join} (Definition~\ref{defJoin}).
The join will tell us how far into a series
of !-box expansions we need to check before
the pattern continues indefinitely.
This line of thought is continued
in Remark~\ref{remGeneratingPatternGraphs},
after the results on finite verification.

\subsection{Varying the skeleton} \label{secInferenceVaryingSkeleton}

\S\ref{secInferringPhaseVariables} has so far concerned itself with investigating the parameter space
of an equation with a fixed skeleton.
We will now explore what happens when we vary a skeleton,
initially by adding components via $\comp$ or $\tensor$.

\begin{proposition} \label{propSoundSubspaceEmbedding}
Given an equation $D = D'$  with parameter space $P$,
a sound subspace $Q$ of $P$ embeds as
the sound subspace $Q \oplus R \subset P \oplus R$
for the equation $D \tensor E = D' \tensor E$
and for the equation
$D \comp E = D' \comp E$.
\end{proposition}

\begin{proof} \label{prfPropSoundSubspaceEmbedding}
Write $M_q$ for the interpretation of the diagram $D$
at point $q \in Q$.
Likewise write $M'_q$ for the interpretation of $D'$
at point $q$,
and $N_r$ for the interpretation of $E$ at the
point $r$ in $R$ (the parameter space of the diagram $E$).
Since $q$ is a sound point in $Q$:
\begin{align}
M_q &= M'_q \\
\therefore \quad M_q \comp N_r &= M'_q \comp N_r  \quad \forall r \\
\therefore \quad M_q \tensor N_r &= M'_q \tensor N_r  \quad \forall r
\end{align}
Therefore the sound subspace $Q \subset P$ embeds as $Q \oplus R \subset P \oplus R$.
\end{proof}

What if, rather than compose diagrams,
we want to change the underlying graph of existing ones?
We saw earlier in this chapter
how to generalise from a single boundary wire
to a boundary wire with a !-box,
and here we shall go from no boundary wire to a single boundary wire.
In this next proposition we will assume the diagrams
are ZX$_G$ diagrams 
but the result is easy enough to generalise to diagrams
in other graphical calculi.

\begin{proposition} \label{propLinearSubspaceToBoundary}
Starting with an equation $D = D'$ of $ZX_G$ diagrams
with parameter space $P$:
Let $e_\alpha$ be the vector in $P$ that is 0 everywhere except for a 1 at position $\alpha$,
and let $\alpha$ and $\beta$ be labels for
two vertices of the same colour in $D$ and $D'$ respectively.
If, and only if, $Q + \bbZ ( e_\alpha + e_\beta)$ is sound in $P$ then
the diagram equation formed by adding an additional boundary
connected to $\alpha$ and $\beta$ on each side respectively
is sound at $Q + \bbZ ( e_\alpha + e_\beta)$.
\end{proposition}

\begin{proof} \label{prfPropLinearSubspaceToBoundary}
We begin by picking a point $q \in Q$,
and then rewriting the diagrams $D$ and $D'$
as compositions of the vertices $\alpha$ and $\beta$ with the subdiagrams $\bar D$ and $\bar D'$.
We are assuming both $\alpha$ and $\beta$ label Z vertices,
the proof for X vertices is identical, just with colours swapped.
\begin{align}
\interpret{D_q} &= \interpret{D'_q} &
\boundaryAddNone{\bar D_q}{q_\alpha} &= \boundaryAddNone{\bar D'_q}{q_\beta}
\end{align}
We then use the spider law to extend $\alpha$ and $\beta$ downwards:
\begin{align}
\boundaryAddNone{\bar D_q}{q_\alpha} = \boundaryAddSome{\bar D_q}{q_\alpha}{gn}{0} &
= \boundaryAddSome{\bar D'_q}{q_\beta}{gn}{0} = \boundaryAddNone{\bar D'_q}{q_\beta}
\end{align}
The above equation is sound for all $q \in Q$.
Since $Q + \bbZ ( e_\alpha + e_\beta)$ is sound in $P$
we therefore know that $q + \pi (e_\alpha + e_\beta)$ is sound in $P$.
Accordingly the following equation is also sound:
\begin{align}
\boundaryAddNone{\bar D_q}{q_\alpha + \pi} = \boundaryAddSome{\bar D_q}{q_\alpha}{gn}{\pi} &
= \boundaryAddSome{\bar D'_q}{q_\beta}{gn}{\pi} = \boundaryAddNone{\bar D'_q}{q_\beta + \pi}
\end{align}
We therefore know that the following diagram is sound for all values of $q$,
since it is sound on a basis of inputs:
\begin{align}
\boundaryAddSome{\bar D_q}{q_\alpha}{none}{} &
= \boundaryAddSome{\bar D'_q}{q_\beta}{none}{}
\end{align}
By Proposition~\ref{propZXAddBang} the following diagram is sound for all values of $q$,
$\gamma$ being a phase variable:
\begin{align}
\boundaryAddSome{\bar D_q}{q_\alpha + \gamma}{none}{} &
= \boundaryAddSome{\bar D'_q}{q_\beta + \gamma}{none}{}
\end{align}
Therefore $q + \bbZ(e_\alpha + e_\beta)$ is sound for all values of $Q$.

For the converse, with $\gamma$ a phase variable not appearing in $D$ or $D'$:
\begin{align}
\boundaryAddSome{\bar D_q}{q_\alpha + \gamma}{none}{} &
= \boundaryAddSome{\bar D'_q}{q_\beta + \gamma}{none}{} \qquad \text{sound at $q$} \\
\implies \qquad \boundaryAddSome{\bar D_q}{q_\alpha + \gamma}{gn}{} &
= \boundaryAddSome{\bar D'_q}{q_\beta + \gamma}{gn}{} \qquad \text{sound at $q$} \\
\implies \qquad \boundaryAddNone{\bar D_q}{q_\alpha + \gamma} &
= \boundaryAddNone{\bar D'_q}{q_\beta + \gamma} \qquad \text{sound at $q$} \label{eqnBoundaryAddNone}
\end{align}
\end{proof}

We showed earlier on in this chapter
that parameter subspaces can be reflected in the phase algebra.
In the case of Proposition~\ref{propLinearSubspaceToBoundary}
we saw the existence of a specific sound subspace, $Q + \bbZ(e_\alpha + e_\beta)$,
and deduced the ability to add boundaries to the diagram.
In terms of phases this sound subspace would appear, given the right choice of presentation,
as a variable appearing exactly once in $D$ and once in $D'$, as in Equation~\eqref{eqnBoundaryAddNone}.

\begin{remark}[From arbitrary many to one to none] \label{remWhyConsiderBoundaries}
In Proposition~\ref{propZXAddBang}
we showed that if a !-boxed boundary is
connected to the same colour vertex on each side of the equation
then we only need to check the case where the !-box is expanded once.
In Proposition~\ref{propLinearSubspaceToBoundary}
we showed that if we are already calculating the sound subspaces
then we only need to check the case where the !-box is expanded 0 times;
i.e. without any boundaries at all.
From the soundness of the equation with 0 boundaries,
Proposition~\ref{propLinearSubspaceToBoundary} will let us deduce the soundness
of the equation with 1 boundary,
and then Proposition~\ref{propZXAddBang} will let us generalise immediately to a !-boxed version.
\end{remark}

\begin{remark}[From sparse to full] \label{remFromSparseToFull}
Readers may also have noticed that we did not use the full set
$Q + \bbZ(e_\alpha + e_\beta)$,
simply two points in it, $q$ and $q + \pi(e_\alpha + e_\beta)$.
We will see in the next chapter how the existence of two sound, distinct points
inside a subspace known to be linear implies the entire subspace is sound.
The proposition could have been phrased:
``If $q$ and $q+a(e_\alpha + e_\beta)$ are distinct and sound in $P$
then $q + \bbZ(e_\alpha + e_\beta)$ is sound for
the diagram equation formed by adding an additional boundary
connected to $\alpha$ and $\beta$ on each side respectively.''
The current phrasing emphasises the three-way link
between the linear subspace $\bbZ(e_\alpha + e_\beta)$,
the presentation of some variable $\gamma$ appearing only
on vertices $\alpha$ and $\beta$,
and the ability to add a boundary to the vertices.
\end{remark}

\section{Summary}

This chapter took the existing notion
of linear rules in ZX and generalised it to
apply to all our graphical calculi,
and also to beyond linear rules.
The purpose of this was to find ways
to generalise theorems from just a few examples,
but it also shows how to link the two different
presentations of rules (as phase algebras or as side conditions)
and unites the algebraic presentation
of diagrams with a geometric notion of spaces.
With this framework in place we gave heuristics
for inferring new polynomial hypotheses based on an existing sample of theorems.
We also examined generalisations at boundaries;
ways to deductively add phase variables and !-boxes to theorems,
immediately increasing their applicability.
Even without reference to conjecture synthesis
any proof assistant benefits from increasing the applicability
of lemmas, and in doing so reducing the number of proof steps for the user.

Our examination of the parameter space showed
that sound subspaces were preserved under $\comp$ and $\tensor$,
that examination of the spaces could lead to important generalisations
(such as adding boundaries to diagrams),
and that phase homomorphisms led to symmetries of sound subspaces.
We will combine this knowledge of symmetries with
geometric constraints on the sound subspaces in the next chapter,
but already we have an indication that the study of these subspaces
could be as rich a topic as the study of quantum graphical calculi.
Parameter spaces give a novel and important method of studying
parameterised diagram equations.

Future work includes further examination of parameter spaces;
in particular their symmetries,
other generalisations like Proposition~\ref{propLinearSubspaceToBoundary},
and the effects of common operations on the skeleton such as adding or removing an edge.
Further work should also be done on deductive generalisations,
as any such theorem immediately improves conjecture synthesis runs.
Our focus will now turn to the problem of verifying these inferred hypotheses.
  \thispagestyle{empty} 
\chapter{Conjecture Verification}
\noindent\emph{In this chapter:}
\begin{itemize}
	\item[\chapterbullet] We introduce conjecture verification for phase ring and phase group calculi
	\item[\chapterbullet] We demonstrate the ability to finitely verify certain classes of hypotheses
	\item[\chapterbullet] We discuss the relationship between conjecture verification and conjecture inference
	\item[\chapterbullet] We use phase homomorphisms to simplify the process even further over qubits
\end{itemize}
\label{chapConjectureVerification}
Conjecture Verification is the act of checking whether a hypothesis is true.
In this chapter we will only consider hypotheses of the form $\interpret{D_1} = \interpret{D_2}$,
i.e. semantic equality of two diagrams.
This is in contrast to the syntactic question of whether $\cal{R} \entails D_1 = D_2$ for some ruleset $\cal{R}$.
We only consider the question of semantic equality because our ultimate
aim is Conjecture Synthesis (\S\ref{chapCoSy}),
in which we try to construct a ruleset that is sound and useful, rather than starting with a ruleset and showing it is complete.

This is the final results chapter of this thesis,
providing verification methods for the inference procedures discussed
in \S\ref{chapConjectureInference} (among others).
In particular the verification of phase variables
in \S\ref{secparameter} relies on ideas of polynomial fitting
that mirror those of \S\ref{secInferringPhaseVariables}.
\S\ref{secPhaseVariablesOverQubits} then builds on this foundation
by including results from Galois Theory
and the phase homomorphism pairs of \S\ref{chapPhaseRingCalculi}.

These verification methods rely on the structure inherent in phase algebras,
and on a novel presentation of !-box expansions that allow us to view these
primarily as $\comp$ rather than $\tensor$ products.

\section{Motivation} \label{secVerificationMotivation}

In earlier chapters, with the exception of parts of \S\ref{chapPhaseRingCalculi},
we considered a phase algebra to be a founding property of a calculus;
without specifying a phase algebra we could not specify the generators and so could not construct diagrams.
We present here a different point of view:
Consider a `base' calculus with a phase algebra $\cal{A}$ and without !-boxes.
We call diagrams constructed in this calculus \emph{simple} (Definition~\ref{defSimpleDiagram}).
By contrast we then consider a `higher' calculus,
identical to the base calculus except that we now allow formal variables in $\cal{A}$ (Definition~\ref{defPhaseVariable}),
and also !-boxes (Definition \ref{defBBox}).

Our motivation for this perspective is that this is how many papers present their rulesets.
For example the calculus $\ZH_\bbC$ formally only allows elements of the complex numbers as labels on H-boxes,
but the rules presented in \cite{ZH} use !-boxes and formal variables, and the meaning is clear from the context.
Our ultimate aim is for computers to be able to reason with a similar degree of expressive power;
to be able to hypothesise and then verify equations between decorated diagrams,
initially assuming only that they are capable of reasoning with simple diagrams.

Using $\ZHC$ as an example we list in Figure~\ref{figVerificationCubeDefinitions}
some structures of interest: The sets of diagrams that fit certain criteria,
as well as their interpretations as matrices or sequences of matrices.
The $ev$ (evaluation) and $as$ (assignment) maps generalise in the obvious ways.
From these maps, and ones very similar to them, we can construct the commutative cube in Figure~\ref{figVerificationCube}.
Our reason for doing so is that the map of interest
($\interpret{\cdot} : \ZH_{\bbC[X]}^! \to ( \bbN \to \Mat_\bbC)$)
which interprets the `higher' structure
can then be linked to the map we assume we have some knowledge of ($\interpret{\cdot} : \ZHC \to \Mat_\bbC$)
which interprets the `base' structure.
Although we will use the language of category theory (such as commutative diagrams) in this section
the structures of interest are considered as discrete categories.

The property we are interested in for this chapter is that of finite sampling:
Determining the nature of a function by knowing its value at enough places.
This is directly linked to our ideas of conjecture inference
in \S\ref{chapConjectureInference}
where we hypothesise some surface as sound,
but need to be able to verify this by only checking a finite sample of points on the surface.
The quintessential example of this is interpolation of complex polynomials,
i.e. once you know the values $P(X)$ at more than $\deg P$ points, you can determine the coefficients of $P$.
We extend this notion to our `higher' structure by investigating
the question
`given $D_1$ and $D_2$ in $\ZH_{\bbC[X]}^!$,
and knowing that $as(ev(D_1)) = as(ev(D_2))$ for some number of values of $X$ and $!$,
does $D_1 = D_2$?'
The answer with one !-box and one phase variable is \emph{yes} (Proposition~\ref{propBoth}),
and, crucially, we can work how how many values is \emph{enough} just from the syntax of the diagram.
The phase variable part of the result (\S\ref{secparameter})
always extends to situations with multiple phase variables,
however the !-box part of the result (\S\ref{secBBoxes}) cannot always be extended to multiple !-boxes,
leading us to the notion of \emph{separability} (\S\ref{secSeparability}).
When generating hypotheses this allows us
to verify hypotheses containing !-boxes
where previously we could not (provided the !-boxes are separable).

For equations with phase variables but without !-boxes over the complex numbers
we can actually go a step further.
Theorem~\ref{thmQubitSingleVerifyingPhaseEquation}
allows us to verify all the phase variables in a $\ring_\bbC$, $\ZW_\bbC$, or $\ZH_\bbC$
equation (without !-boxes) using a \emph{single} equation without any phase variables.
(The ZX analogue is Theorem~\ref{thmZXPhasesOverQubits}.)
The converse to this is that whenever a phase appears in an equation
less frequently than its algebraic degree it can be replaced with a variable (Remark~\ref{remNonQImpliesPhaseVariable}).

This work complements the work of \cite{QuickThesis},
which introduced !-induction.
That work allows the verification of an infinite family of equations parameterised by !-boxes,
by showing an inductive step involving those !-boxes.
Our work instead requires more equations to be checked, but where none of the equations involve !-boxes.

\begin{figure}
	\centering
	\renewcommand{\arraystretch}{2}
	\tabcolsep=30pt
	\begin{tabular}{l  l}
		$ \ZH_\bbC  $                                   & $\ZHC$ diagrams, labelled with complex numbers \\
		$ \ZH_{\bbC[X]}  $                              & $\ZHC$ with a phase variable, $X$              \\
		$ \ZH_\bbC^!  $                                 & $\ZHC$ with a !-box, $!$                       \\
		$ \ZH_{\bbC[X]}^!  $                            & $\ZHC$ with $X$ and $!$                        \\
		$ \Mat_\bbC  $                                  & Complex matrices                               \\
		$ \bbN \to \Mat_\bbC  $                         & Sequences of complex matrices                  \\
		$ ev : \ZH_{\bbC[X]} \to \ZH_\bbC  $            & Evaluation of $X$                              \\
		$ as : \ZH_\bbC^! \to \ZH_\bbC  $               & Assignment of a number to $!$                 \\
		$ \interpret{\cdot} : \ZH_\bbC \to \Mat_\bbC  $ & Interpretation as a complex matrix
	\end{tabular}
	\caption{Example sets of diagrams,
		matrices and sequences
		as well as interpretations and some maps between them
		\label{figVerificationCubeDefinitions}}
\end{figure}

\begin{figure}
	\centering
	\begin{align*}
		\begin{tikzcd}[ampersand replacement=\&]
			\ZH^!_{\bbC[X]} \arrow[dd, "ev"] \arrow[rr, "as"]  \arrow[dr, "\interpret{\cdot}"] \& \& \ZH_{\bbC[X]} \arrow[dd]  \arrow[dr, "\interpret{\cdot}"]\& \\
			\& (\bb{N} \to \Mat_{\bb{C}[X]}) \arrow[rr,crossing over] \& \&  \Mat_{\bb{C}[X]} \arrow[dd,"ev"] \\
			\ZHC^!  \arrow[dr, "\interpret{\cdot}"] \arrow[rr] \& \& \ZHC  \arrow[dr, "\interpret{\cdot}"] \&\\
			\& (\bb{N} \to \Mat_{\bb{C}}) \arrow[from=uu,crossing over] \arrow[rr, "as"] \& \&  \Mat_{\bb{C}}
		\end{tikzcd}
	\end{align*}
	\caption{The motivational commutative cube for phase variable and !-box conjecture verification.
		\label{figVerificationCube}}
\end{figure}
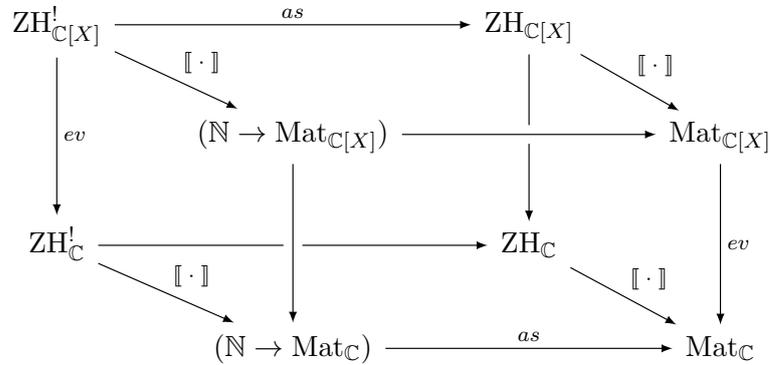

\section{Verifying phase variables}
\label{secparameter}

Our first result will concern diagrams that contain a finite number of phase variables and no !-boxes.
It relies on a certain property of polynomials: If you know the value of the polynomial $P(Y)$
for sufficiently many values of $Y$ then you can determine all the coefficients of $P$.
For example the polynomial $P(Y_1, Y_2) = a + b Y_1 + c Y_2 + d Y_1 Y_2$ can have all of its coefficients
determined by knowing the values $P(0,0)$, $P(0,1)$, $P(1,0)$, and $P(1,1)$.
We use this fact by extending the matrix interpretation of simple diagrams to one for decorated diagrams.
Continuing our ZH example from \S\ref{secVerificationMotivation} we are investigating this face of the commutative cube:

\begin{align}
	\begin{tikzcd}[ampersand replacement=\&]
		\ZH_{\bbC[X]}  \arrow[dr, "\interpret{\cdot}"] \arrow[dd, crossing over, "ev"]\& \\
		\&  \Mat_{\bb{C}[X]} \arrow[dd,"ev"] \\
		\ZHC \arrow[dr, "\interpret{\cdot}"] \&\\
		\&  \Mat_{\bb{C}}
	\end{tikzcd}
\end{align}

\subsection{Matrix interpretations}

The graphical calculi we have considered in this thesis all have matrix interpretations\footnote{The
	exception being SZX in Example~\ref{remSZXAsDiagramsForDiagrams}, and even then only to make a point.}.
Moreover for each of these calculi a diagram with $n$ inputs and $m$ outputs will be mapped to a matrix
with $\dim \Hilbert^{\tensor n}$ columns and $\dim \Hilbert^{\tensor m}$ rows,
where $\Hilbert$ represents $R^2$ (e.g. $\bbC^2$ in the case of quantum computing).
We will use the language of fields and vector spaces in this chapter but note that many of the results
only require the base ring to be an integral domain.
A family of diagrams parameterised over $\alpha$ is a set of simple diagrams,
and we could extend our interpretation
so that a set of simple diagrams is sent to a set of matrices.
This would, however, lose any structure from the phase algebra.
Instead we try to find a polynomial matrix interpretation;
for example one that sends a family of $\ZH_\bbC$ equations $\set{\bbE}_{\alpha}$ to a matrix in $\Mat_{\bb{C}[\alpha]}$.

While this works well for $\ZH_R$, $\ZW_R$ and $\ring_R$, there is a complication with ZX:
A phase variable $\alpha$ in a ZX diagram corresponds to an $e^{i \alpha}$ in
the matrix interpretation.
We can try performing the substitution $Y:= e^{i \alpha}$,
but run into trouble if there is a node containing, for example, $-\alpha$
(and accordingly $Y\inv$ in the matrix),
since polynomials do not normally allow negative powers.
Rather than stick with standard polynomials we instead move to Laurent polynomials;
polynomials that do allow positive and negative powers,
and define all the properties we will need of them below.
We do not consider the case of ZQ in this chapter,
because its phase group is not commutative and commutativity is necessary for our results.

ZX also introduces one more subtlety:
There is an extra relation from the phase group ($2\pi = 0$) that we should take care to reflect in our matrix interpretation.
This does not impact Universal ZX,
but does affect the fragments of ZX with a finite phase group
(Example \ref{excliffordt} demonstrates this for the Clifford+T fragment).
Before we get started we define Laurent polynomials
and show how to interpret diagrams with phases linear in phase variables into
matrices over Laurent polynomials.

\begin{definition}[Laurent polynomials {\cite[p356]{AlgebraArtin}}]
	A Laurent polynomial in the variables $Y_1, \dots, Y_n$ with coefficients in $R$ is an element
	of
	\begin{align}
		R[Y_1, Y_1', \dots, Y_n, Y_n'] / (Y_1Y_1' = 1, \dots, Y_nY_n'=1)
	\end{align}
	We simply write $Y_j\inv$ instead of $Y_j'$.
	Laurent polynomial degrees are given by
	\begin{itemize}
		\item The 0 polynomial has degree $- \infty$ by convention
		\item the non-zero Laurent polynomial $a_n Y^n + a_{n-1} Y^{n-1} + \dots + a_0 + a_{-1} Y^{-1} + \dots + a_{-m} Y^{-m}$
		      with $a_n \neq 0$ and $a_{-m} \neq 0$
		      has \emph{positive degree} $n \geq 0$ and \emph{negative degree} $m \geq 0$.
	\end{itemize}
	Note that we can factorise this Laurent polynomial as $Y^{-m}$ multiplied by a (non-Laurent) polynomial.
\end{definition}

\begin{definition}[Complex Laurent Polynomial Interpretation]
	A \emph{Laurent polynomial interpretation} for a graphical calculus $\bbL$
	is a matrix interpretation:
	\begin{align}
		\interpret{\cdot}
		  & : \bbL[\alpha_1, \dots, \alpha_n] \to \Mat_{R[Y_1, Y_1\inv, \dots, Y_n, Y_n\inv]}
	\end{align}
	The source category here is the PROP of diagrams where generator labels are now taken not from the phase algebra of $\bbL$
	but from this phase algebra adjoin the formal variables $\alpha_1, \dots, \alpha_n$.
\end{definition}

We give the Laurent polynomial interpretations of Universal $\ZX$, $\ZW_R$, $\ZH_R$ and $\ring_R$ here,
showing how to link a phase variable $\alpha$ with the polynomial indeterminate $Y$.
These definitions are, frankly, so close to the original that the difference is easy to miss
so we show the ZX example in more detail in Example~\ref{exaZXLaurent} and then the others in Example~\ref{exaLaurentZWRing}.

\begin{example}[Laurent polynomial interpretations of Universal ZX] \label{exaZXLaurent}
	The Z spider from Universal ZX is parameterised by an $\alpha \in [0, 2\pi)$,
	and the (simple) matrix interpretation of some Z spider (with $\alpha$ instantiated at $a$) is:

	\begin{align}
		\interpret{\family{\spider{gn}{\alpha}}_{\alpha | \alpha = a}} =
		\begin{bmatrix}
			1      & 0     & \dots  &   & 0      \\
			0      & 0     &        &   & \vdots \\
			\vdots &       & \ddots &   &        \\
			       &       &        & 0 & 0      \\
			0      & \dots &        & 0 & e^{ia}
		\end{bmatrix} \quad \in \Mat_{\bbC}
	\end{align}

	Rather than instantiate the value of $\alpha$ before we apply the map,
	we instead make the substitution $Y := e^{i\alpha}$ to find a Laurent polynomial matrix interpretation:

	\begin{align}
		\interpret{\spider{gn}{\alpha}} =
		\begin{bmatrix}
			1      & 0     & \dots  &   & 0      \\
			0      & 0     &        &   & \vdots \\
			\vdots &       & \ddots &   &        \\
			       &       &        & 0 & 0      \\
			0      & \dots &        & 0 & Y
		\end{bmatrix} \quad \in \Mat_{\bbC[Y, Y\inv]}
	\end{align}
\end{example}

\begin{example}[Laurent polynomial interpretations of $\ring$, $\ZW$ and $\ZH$] \label{exaLaurentZWRing}
	For a Laurent polynomial interpretation we simply set $\alpha = Y$ to get:
	\begin{align}
		  & \ring_R & \spider{white}{\alpha} & \interpretedas \ket{0\dots 0}\bra{0\dots 0 } + Y\ket{1\dots 1}\bra{1\dots 1}                                   \\
		  & \ZW_R   & \spider{white}{\alpha} & \interpretedas \ket{0\dots 0}\bra{0\dots 0 } + Y\ket{1\dots 1}\bra{1\dots 1}                                   \\
		  & \ZH_R   & \spider{ZH}{\alpha}    & \interpretedas \sum_{\text{bitstrings}} (Y)^{i_1\dots i_m j_1 \dots j_n} \ket{j_1\dots j_n}\bra{i_1 \dots i_m} \\
	\end{align}
	The rest of the interpretations remain as usual (See Definitions~\ref{defZW}, \ref{defZHR}, and \ref{defRingR})
\end{example}

\begin{remark} \label{remEvaluationCommuting}
	This is an unusually pedantic treatment of the interplay between interpretation and phase variables,
	but it illustrates the commutativity of interpretation and evaluation.
\end{remark}

\begin{remark} \label{remZXLaurent}
	The generators for $\ZH_\bbC$, $\ZW_\bbC$ and $\ring_\bbC$ have Laurent polynomial matrix interpretations
	that are, in fact, just polynomial matrix interpretations.
	Really it is only for $\ZX$ that we are considering the broader Laurent polynomial version.
\end{remark}

\subsection{Degree of a matrix}

Our hope is to find properties of the interpretations of diagrams
(which are matrices of polynomials) just from the diagrams themselves.
An obvious property to investigate is the degree of the polynomials in the matrix interpretation,
and while we cannot easily determine this value precisely,
we can find bounds for the degrees just by looking at the diagrams themselves.

\begin{definition}[Matrix and diagram degree]
	We define the degree of a matrix and diagram by:
	\begin{itemize}
		\item
		      The $Y_j^+$-degree of a matrix in $\Mat_{\bbC[Y_1, Y_1\inv, \dots, Y_n, Y_n\inv]}$
		      is the maximum of the positive $Y_j$-degrees of the entries in that matrix
		\item The $Y_j^-$-degree is the maximum of the negative $Y_j$-degrees of
		      the entries in that matrix
		\item
		      The positive degree of a diagram is the positive degree of the matrix interpretation of that diagram
		      (likewise for negative degrees).
		\item
		      When clear from context we will refer to the degree of a phase variable $\alpha_j$ in the diagram,
		      meaning the degree of $Y_j$ in the interpretation.
	\end{itemize}

\end{definition}

\begin{example}[Laurent polynomial and matrix degrees]
	Here are example positive and negative degrees for first a polynomial, and then a $2\times2$ matrix of polynomials.
	\begin{align}
		\maxdegree{+}{Y}{Y^8+1+Y^{-2}}             & = 8                     \\
		\maxdegree{-}{Y}{Y^8+1+Y^{-2}}             & = 2                     \\
		\nonumber\\
		\maxdegree{+}{Y}{\begin{pmatrix}2 & Y^{-3} \\ Y & Y^2 - 2\end{pmatrix}} & = \max\set{0,0,1,2} = 2 \\
		\maxdegree{-}{Y}{\begin{pmatrix}2 & Y^{-3} \\ Y & Y^2 - 2\end{pmatrix}} & = \max\set{0,3,0,0} = 3
	\end{align}
\end{example}

\begin{proposition} \label{propUpperBound}
	For two diagrams $\bbD$ and $\bbD'$ we can find an upper bound for the
	degrees of the horizontal or vertical compositions of $\bbD$ and $\bbD'$,
	i.e.:
	\begin{align}
		\maxdegree{+}{}{
			(\bbD \comp \bbD')
		}
		  & \leq \maxdegree{+}{Y}{\bbD} + \maxdegree{+}{Y}{\bbD'} \\
		\maxdegree{+}{}{
			(\bbD \tensor \bbD')
		}
		  & \leq \maxdegree{+}{Y}{\bbD} + \maxdegree{+}{Y}{\bbD'} \\
		\nonumber \\
		\maxdegree{-}{}{
			(\bbD \comp \bbD')
		}
		  & \leq \maxdegree{-}{Y}{\bbD} + \maxdegree{-}{Y}{\bbD'} \\
		\maxdegree{-}{}{
			(\bbD \tensor \bbD')
		}
		  & \leq \maxdegree{-}{Y}{\bbD} + \maxdegree{-}{Y}{\bbD'}
	\end{align}
\end{proposition}

\begin{proof}
	We first note that for Laurent polynomials $P$ and $P'$ in $\bbC[Y, Y\inv]$, and for $\lambda \in \bbC$:
	\begin{align}
		\maxdegree{+}{Y}{(\lambda P)}              & \leq \maxdegree{+}{Y}{P}                         \\
		\maxdegree{-}{Y}{(\lambda P)}              & \leq \maxdegree{-}{Y}{P}                         \\
		\nonumber  \\
		\maxdegree{+}{Y}{(P \times P')}            & \leq \maxdegree{+}{Y}{P}  + \maxdegree{+}{Y}{P'} \\
		\maxdegree{-}{Y}{(P \times P')}            & \leq \maxdegree{-}{Y}{P}  + \maxdegree{-}{Y}{P'} \\
		\nonumber  \\
		\maxdegree{+}{Y}{\left(\sum_j{P_j}\right)} & \leq \max_j \maxdegree{+}{Y}{P_j}                \\
		\maxdegree{-}{Y}{\left(\sum_j{P_j}\right)} & \leq \max_j \maxdegree{-}{Y}{P_j}
	\end{align}

	The composition $A \comp B$ or tensor product $A \tensor B$ of matrices produces a new matrix with entries that
	are linear combinations of products of the entries of $A$ and $B$.
	Therefore:

	\begin{align}
		\maxdegree{+}{Y}{(M \comp M')}   & \leq \maxdegree{+}{Y}{M} + \maxdegree{+}{Y}{M'} \\
		\maxdegree{+}{Y}{(M \tensor M')} & \leq \maxdegree{+}{Y}{M} + \maxdegree{+}{Y}{M'} \\
		\nonumber \\
		\maxdegree{-}{Y}{(M \comp M')}   & \leq \maxdegree{-}{Y}{M} + \maxdegree{-}{Y}{M'} \\
		\maxdegree{-}{Y}{(M \tensor M')} & \leq \maxdegree{-}{Y}{M} + \maxdegree{-}{Y}{M'}
	\end{align}

	Recalling that the degree of a diagram is the degree of its polynomial matrix interpretation
	this gives us the result for $\bbD$ and $\bbD'$.
\end{proof}

Now that we have shown how to bound the degree of a diagram
by looking at the degrees of its subdiagrams,
all that is left is to calculate the degrees of the generators of the diagrams.
From there we can then find bounds for the degrees of any of the diagrams
we consider in this chapter.
For simplicity we will consider nodes parameterised by a single variable $\alpha$,
and express their degree with respect to a variable $Y$.
Degrees of wire components are all 0.

\begin{itemize}
	\item[ZX:] Using $Y^n:=e^{n i \alpha}$, the degrees in $Y$ of the generators (for $n \geq 0$) are:

	      \begin{align}
		      \maxdegree{+}{\alpha}{\spider{gn}{n \alpha}}   & = n & \maxdegree{-}{\alpha}{\spider{gn}{n \alpha}}   & = 0 \\
		      \maxdegree{+}{\alpha}{\spider{gn}{- n \alpha}} & = 0 & \maxdegree{-}{\alpha}{\spider{gn}{- n \alpha}} & = n \\
		      \maxdegree{+}{\alpha}{\spider{rn}{n \alpha}}   & = n & \maxdegree{-}{\alpha}{\spider{rn}{n \alpha}}   & = 0 \\
		      \maxdegree{+}{\alpha}{\spider{rn}{- n \alpha}} & = 0 & \maxdegree{-}{\alpha}{\spider{rn}{- n \alpha}} & = n \\
		      \maxdegree{+}{\alpha}{\unary[H]{}} &= 0  &
		      \maxdegree{-}{\alpha}{\unary[H]{}} &= 0
	      \end{align}
	\item[ZH:]
	      We equate $Y:= \alpha$, and for $P$ any Laurent polynomial:

	      \begin{align}
		      \maxdegree{+}{\alpha}{\spider{ZH}{P(\alpha)}} &= \maxdegree{+}{Y} P &
		      \maxdegree{-}{\alpha}{\spider{ZH}{P(\alpha)}} &= \maxdegree{-}{Y} P \\
		      \maxdegree{+}{\alpha}{\spider{white}{}} & = 0 & \maxdegree{-}{\alpha}{\spider{white}{}} & = 0
	      \end{align}
	\item[ZW:]
	      We equate $Y:= \alpha$, and for $P$ any Laurent polynomial:

	      \begin{align}
		      \maxdegree{+}{\alpha}{\spider{black}{}}      & = 0 & \maxdegree{-}{\alpha}{\spider{black}{}}      & = 0 \\
		      \maxdegree{+}{\alpha}{\vc{\InputIfFileExists{./figures/ZW/crossing.tikz}{}{Missing file!}}} & = 0 & \maxdegree{-}{\alpha}{\vc{\InputIfFileExists{./figures/ZW/crossing.tikz}{}{Missing file!}}} & = 0 \\
		      \maxdegree{+}{\alpha}{\spider{white}{P(\alpha)}} &= \maxdegree{+}{Y} P &
		      \maxdegree{-}{\alpha}{\spider{white}{P(\alpha)}} &= \maxdegree{-}{Y} P \\
	      \end{align}

	\item[$\ring$:]
	      We equate $Y:= \alpha$, and for $P$ any Laurent polynomial:

	      \begin{align}
		      \maxdegree{+}{\alpha}{\state{white}{P}} & = \maxdegree{+}{Y} P & \maxdegree{-}{\alpha}{\state{white}{P}} & = \maxdegree{-}{Y} P \\
		      \maxdegree{+}{\alpha}{\binary[white]{}} & = 0                  & \maxdegree{-}{\alpha}{\binary[white]{}} & = 0                  \\
		      \maxdegree{+}{\alpha}{\binary[poly]{}}  & = 0                  & \maxdegree{-}{\alpha}{\binary[poly]{}}  & = 0                  \\
	      \end{align}
\end{itemize}

\begin{remark} \label{remVerifyPhaseAlgebraDirectly}
Using the interpretations above one could
evaluate diagrams containing phase variables
directly into polynomial rings and use this to verify an equation.
This would sidestep much of the need for a verification
result like the one in Theorem~\ref{thmparameter},
but would also not reveal results like
the bound on algebraic complexity
in Remark~\ref{remNonQImpliesPhaseVariable}
nor indicate where one could replace phases with phase variables for conjecture inference.
If the direct evaluation could
deal with quotient rings then it may be able
to avoid pitfalls like those in Example~\ref{excliffordt}.
\end{remark}

\subsection{Finite verification of phase variables}

Phase variables are, as it turns out,
remarkably well behaved when it comes to verification.
As we will see in Theorem~\ref{thmparameter}
we simply need to know a bound for the degree
of the polynomial entries in the matrix the diagram represents.
For all of the graphical calculi we have been considering
this is tantamount to counting occurrences of the variable
in the diagram, taking powers into account.
The only other ingredient for the theorem
is that of interpolation on a grid of points,
a technique that is sadly not well documented as to
its origin, but has existed for well over a century \cite{Multivariate}.

\begin{restatable}[Finite verification of phase variables]{theorem}{thmparameter}
	For a diagrammatic equation with phase variables
	\label{thmparameter}
	\[ \family{\bb{\bbD}_1 = \bb{\bbD}_2}_{\alpha_1, \dots, \alpha_n}\]

	that has a Laurent polynomial matrix interpretation,
	and the equation is sound at all values of $(\alpha_1, \dots, \alpha_n) = (a_1, \dots, a_n) \in A_1 \times \dots \times A_n$,
	where each $\abs{A_j}$ is sufficiently large,
	then the equation is sound for all values of $(\alpha_1, \dots, \alpha_n)$. That is:

	\begin{align}
		         & \interpret{\family{\bbD_1}_{\alpha_1, \dots, \alpha_n|\alpha_1 = a_1, \dots, \alpha_n = a_n}} = \nonumber
		\interpret{\family{\bbD_2}_{\alpha_1, \dots, \alpha_n|\alpha_1 = a_1, \dots, \alpha_n = a_n}} \\
		         & \qquad \qquad \forall a_1 \in A_1, \dots, a_n \in A_n
		\\ \nonumber
		\implies & \interpret{\family{\bbD_1}_{\alpha_1, \dots, \alpha_n|\alpha_1 = a_1, \dots, \alpha_n = a_n}} =
		\interpret{\family{\bbD_2}_{\alpha_1, \dots, \alpha_n|\alpha_1 = a_1, \dots, \alpha_n = a_n}} \\ \nonumber
		         & \qquad \qquad \forall a_1, \dots, a_n
	\end{align}

	The necessary size of $\abs{A_j}$ is given by:
	\begin{align}
		\abs{A_j} = & \max ( \maxdegree{+}{Y_j}{\bbD_1}, \maxdegree{+}{Y_j}{\bbD_2})
		+ \max ( \maxdegree{-}{Y_j}{\bbD_1}, \maxdegree{-}{Y_j}{\bbD_2})
		+ 1
	\end{align}

	`The maximum of the positive degrees, plus the maximum of the negative degrees, plus one.'
\end{restatable}

Before we embark on the proof we sketch it as follows:
\begin{itemize}
	\item Manipulate the equation $\interpret{\bbD} = \interpret{\bbD'}$ into an equation of the form $M = 0$,
	      where $M$ is a matrix of (non-Laurent) polynomials.
	\item Perform multivariate polynomial interpolation element-wise on $M$.
	\item In doing this interpolation we need to know an upper bound for the degrees of the polynomials in $M$,
	      which we calculate using Proposition~\ref{propUpperBound}
\end{itemize}

\begin{proof}

	We are seeking the multivariate complex polynomials that populate the matrices $\interpret{\bb{D}_1}$ and $\interpret{\bb{D}_2}$.
	We begin by combining the two matrices of Laurent polynomials into one matrix of (not-Laurent) polynomials and a scale factor of the form $Y_1^{m_1} \dots Y_n^{m_n}$.

	\begin{itemize}

		\item We define:
		      \begin{align}
			      M_1 :=  \interpret{\bb{D}_1} \quad \quad
			      M_2 :=  \interpret{\bb{D}_2}
		      \end{align}
		      and wish to show $M_1 = M_2$.

		\item First we pull enough copies of $Y_1\inv , \dots, Y_n\inv$ out of each side so that we have an equation of the form:
		      \begin{align}
			      M_1' \prod_j (Y_j\inv)^{\maxdegree{-}{Y_j}{M_1}} =  M_2'  \prod_j (Y_j\inv)^{\maxdegree{-}{Y_j}{M_2}}
		      \end{align}
		      Where $M_1'$ and $M_2'$ are matrices of (not-Laurent) polynomials.

		\item Let $m_j := \max(\maxdegree{-}{Y_j}{M_1}, \maxdegree{-}{Y_j}{M_2})$
		      and multiply both sides by $\prod_j Y_j^{m_j}$ to clear any negative powers of $Y_j$.
		      \begin{align}
			      M_1' \prod_j Y^{m_j - \maxdegree{-}{Y_j}{M_1}}  =  M_2' \prod_j Y^{m_j - \maxdegree{-}{Y_j}{M_2}}
		      \end{align}

		\item Then subtract the right hand side from the left:
		      \begin{align}
			                & M_1' \prod_j Y_j^{m_j - \maxdegree{-}{Y_j}{M_1}}  -  M_2' \prod_j Y^{m_j - \maxdegree{-}{Y_j}{M_2}} = 0 \\
			      \bb{M} := & M_1' \prod_j Y^{m_j - \maxdegree{-}{Y_j}{M_1}}  -  M_2' \prod_j Y^{m_j - \maxdegree{-}{Y_j}{M_2}}
		      \end{align}

		      Note that $\bb{M}$ is a matrix of (not-Laurent) polynomials.
		      The statement $\bb{M} = 0$ can be viewed as a stating that each entry of $\bb{M}$ is equal to the 0 polynomial.

		\item We will use the notation $\maxdegree{}{Y_j}{\cdot}$ for the degree of a (not-Laurent) polynomial,
		      or matrix of polynomials.
		      We could continue to use the term positive degree, the definitions coincide, but want to make it clear
		      whenever we are not in the Laurent polynomial setting.

		\item We wish to find a bound for the maximum degree of any polynomial in $\bb{M}$:
		      \begin{align}
			                             & \ \maxdegree{}{Y_j}{\bb{M}}                                                                                             \\
			      =                      & \ \maxdegree{}{Y_j}{ M_1' \prod_j Y^{m_j - \maxdegree{-}{Y_j}{M_1}}  -  M_2' \prod_j Y^{m_j - \maxdegree{-}{Y_j}{M_2}}} \\
			      \leq                   & \max ( \maxdegree{}{Y_j}{M_1'} + m_j - \maxdegree{-}{Y_j}{M_1},                                                         \\
			                             & \nonumber \qquad \maxdegree{}{Y_j}{M_2'} + m_j - \maxdegree{-}{Y_j}{M_2} )                                              \\
			      =                      & \max ( \maxdegree{+}{Y_j}{M_1} + \maxdegree{-}{Y_j}{M_1} + m_j - \maxdegree{-}{Y_j}{M_1},                               \\
			                             & \nonumber \qquad \maxdegree{+}{Y_j}{M_2} +  \maxdegree{-}{Y_j}{M_2} + m_j - \maxdegree{-}{Y_j}{M_2} )                   \\
			      =                      & \max ( \maxdegree{+}{Y_j}{M_1} + m_j  , \maxdegree{+}{Y_j}{M_2} + m_j )                                                 \\
			      =                      & \max ( \maxdegree{+}{Y_j}{M_1}, \maxdegree{+}{Y_j}{M_2}) + m_j                                                          \\
			      \label{eqnparam-deg} = & \max ( \maxdegree{+}{Y_j}{M_1}, \maxdegree{+}{Y_j}{M_2}) +
			      \max ( \maxdegree{-}{Y_j}{M_1}, \maxdegree{-}{Y_j}{M_2})
		      \end{align}

		\item Suppose we know that our diagram equation
		      is sound for parameter choices in a large enough \emph{regular grid} of values.
		      That is to say we know:
		      \begin{align}
			      &\interpret{\family{\bb{D}_1}_{\alpha_1, \dots, \alpha_n|\alpha_1 = a_1, \dots, \alpha_n=a_n}} =
			      \interpret{\family{\bb{D}_2}_{\alpha_1, \dots, \alpha_n|\alpha_1 = a_1, \dots, \alpha_n=a_n}}  \\
			      \nonumber\text{for}\quad   & (a_1, \dots, a_n) \in A_1 \times \dots \times A_n &   \\
			      \nonumber\text{where}\quad & \abs{A_j} = \deg_{Y_j}(\bb{M}) +1                 &   \\
			      \nonumber A_j \quad & := \set{a_{j,0}\;, \dots,\; a_{j,\deg{Y_j}}}
		      \end{align}

		      By picking a polynomial entry $P$ of $\bb{M}$, expressing $P$ using the multi-index $\beta$ as
		      $P = \sum_{\beta} c_\beta Y^\beta$,
		      and then evaluating $P$ at every point in $A_1 \times \dots \times A_n$
		      we construct the system of equations:
		      \begin{align}
			      \begin{bmatrix}
				      c_{0, \dots,0}a^{0, \dots,0}_{0, \dots, 0}                 & c_{1, \dots,0}a^{1, \dots,0}_{0, \dots, 0}                 & \dots  & c_{\deg{Y_1}, \dots, \deg{Y_n}}a^{\deg{Y_1}, \dots,\deg{Y_n}}_{0, \dots, 0}                 \\
				      \vdots                                                     & \vdots                                                     & \ddots & \vdots                                                                                      \\
				      c_{0, \dots,0}a^{0, \dots,0}_{\abs{A_1}, \dots, \abs{A_n}} & c_{1, \dots,0}a^{1, \dots,0}_{\abs{A_1}, \dots, \abs{A_n}} & \dots  & c_{\deg{Y_1}, \dots, \deg{Y_n}}a^{\deg{Y_1}, \dots,\deg{Y_n}}_{\abs{A_1}, \dots, \abs{A_n}} \\
			      \end{bmatrix}
			      =
			      \begin{bmatrix}
				      0 \\ \vdots \\ 0
			      \end{bmatrix}
		      \end{align}

		      Which we view as:
		      \begin{align}
			      \begin{bmatrix}
				      V
			      \end{bmatrix}
			      \begin{bmatrix}
				      c_{0, \dots, 0} \\ \vdots \\ c_{\deg{Y_1}, \dots, \deg{Y_n}}
			      \end{bmatrix}
			      =
			      \begin{bmatrix}
				      0 \\ \vdots \\ 0
			      \end{bmatrix}
		      \end{align}

		      Where $V$ contains all the products $a_1^{\beta_1} \times \dots \times a_n^{\beta_n}$,
		      $\beta$ ranging from $(0, \dots, 0)$ to $(\deg{Y_1}, \dots, \deg{Y_n})$.
		      Thankfully $V$ decomposes as:

		      \begin{align}
			      V =    & V_1 \tensor \dots \tensor V_n \\
			      V_j := &
			      \begin{bmatrix}
				      a_{j,0}^0         & a_{j,0}^1         & \dots & a_{j,0}^{\deg{Y_j}}         \\
				      a_{j,1}^0         & a_{j,1}^1         & \dots & a_{j,1}^{\deg{Y_j}}         \\
				      \vdots            & \vdots            &       & \vdots                      \\
				      a_{j,\deg{Y_j}}^0 & a_{j,\deg{Y_j}}^1 & \dots & a_{j,\deg{Y_j}}^{\deg{Y_j}} \\
			      \end{bmatrix}
		      \end{align}

		\item Since $\det(A \tensor B) \neq 0$ if and only if $\det(A) \neq 0$ and $\det(B) \neq 0$,
		      and since $\det(V_j) \neq 0$ because each $V_j$ is a Vandermonde matrix,
		      we know that $\det(V) \neq 0$.
		      Since $V$ is therefore invertible we know that all the coefficients $c_\beta$ must be 0,
		      and therefore $P$ is the 0 polynomial.

		\item In the presence of a regular grid on which $\bb{D}_1$ and $\bb{D}_2$ agree we know:

		      \begin{align}
			      &&\interpret{\family{\bb{D}_1}_{\alpha_1, \dots, \alpha_n|\alpha_k \in A_k \forall k}}
			      &= \interpret{\family{\bb{D}_2}_{\alpha_1, \dots, \alpha_n|\alpha_k \in A_k  \forall k}} \\
			        & \implies & \quad \text{ for any entry } P \text{ of } \bb{M} \qquad P & = 0                                                         \\
			        & \implies & \bb{M}                                                     & = 0                                                         \\
			        & \implies & M_1' \prod_j Y^{m_j - \maxdegree{-}{Y_j}{M_1}}             & = M_2' \prod_j Y^{m_j - \maxdegree{-}{Y_j}{M_2}}            \\
			        & \implies & M_1                                                        & = M_2                                                       \\
			        & \implies & \interpret{\family{\bb{D}_1}_{\alpha_1, \dots, \alpha_n}}  & = \interpret{\family{\bb{D}_2}_{\alpha_1, \dots, \alpha_n}}
		      \end{align}

		\item By setting $d_j$ to be strictly greater than $\maxdegree{}{Y_j}{\bb{M}}$
		      (which we can find a bound for using \eqref{eqnparam-deg}), and $\abs{A_j} = d_j$ we attain our result;
		      that if we know that the interpretations of the families of diagrams agree on the regular grid described by the sizes $d_j$
		      then the interpretations agree on all points in the phase algebra.

	\end{itemize}
\end{proof}

\begin{example}[Finding the sizes of the $A_j$]
	The following Universal ZX diagram contains no !-boxes and two phase variables.
	\label{exaPauliSpiderVerification}
	\begin{align}
		\vc{\InputIfFileExists{./figures/ZX/param_example_zx.tikz}{}{Missing file!}} & \qquad \label{eqnVerSpider}
		\begin{matrix}
			\deg^+_{\alpha_1}\bb{D}_1 = 1 &
			\deg^+_{\alpha_2}\bb{D}_1 = 1 \\
			\deg^-_{\alpha_1}\bb{D}_1 = 0 &
			\deg^-_{\alpha_2}\bb{D}_1 = 0 \\
		\end{matrix}
		\qquad \begin{matrix}
			\deg^+_{\alpha_1}\bb{D}_2 = 1 &
			\deg^+_{\alpha_2}\bb{D}_2 = 1 \\
			\deg^-_{\alpha_1}\bb{D}_2 = 0 &
			\deg^-_{\alpha_2}\bb{D}_2 = 0 \\
		\end{matrix}
	\end{align}

	In order to apply Theorem~\ref{thmparameter} we should therefore construct $A_1$ and $A_2$ such that
	$\abs{A_1} =\max\set{1,1} + \max\set{0,0}+1$ and  $\abs{A_2} = \max\set{1,1} + \max\set{0,0}+1$.
	By picking $A_1 = A_2 = \set{0 , \pi}$
	we therefore know that we can verify this parameterised family of diagram equations for all values
	of $\alpha_1$ and $\alpha_2$ by verifying this equation on the following grid of values:

	\begin{align}
		\begin{matrix}
			               & \alpha_1 = 0 & \alpha_1 = \pi \\
			\alpha_2 = 0   & (0,0)        & (0,\pi)        \\
			\alpha_2 = \pi & (\pi, 0)     & (\pi, \pi)
		\end{matrix}
	\end{align}

	I.e. by verifying the four equations:
	\begin{align}
		\zxtwospider{0}{0}   & \qquad
		\zxtwospider{0}{\pi}   \\
		\zxtwospider{\pi}{0} & \qquad
		\zxtwospider{\pi}{\pi}
	\end{align}
	we can assert that the diagram equation in \eqref{eqnVerSpider}
	is sound for all values of $\alpha_1$ and $\alpha_2$ in $[0, 2\pi)$.
\end{example}

\begin{remark} \label{remVerificationOutsidePauli}
	Example~\ref{exaPauliSpiderVerification} has verified a rule that applies
	to all phases in $[0, 2\pi)$, but the verification used only Clifford phases.
	Quantum circuits using just Clifford phases are efficiently simulable on classical computers,
	but at time of writing Universal ZX diagrams are not \cite{Gottesman98}.
\end{remark}

\begin{corollary}[{\cite[Theorem 3]{BeyondCliffordT}}] \label{corZXPhases}
	In the Universal ZX calculus it suffices to check $(\alpha_1, \dots, \alpha_n) \in A_1 \times \dots \times A_n$
	to prove an equation parameterised by $\alpha_j$, where the $A_j$ are sets of distinct angles
	with
	\begin{align}
		\abs{A_j} = & \max
		\set{ \maxdegree{+}{\alpha_j}{\bbD_1},\maxdegree{+}{\alpha_j}{\bbD_2} }
		+ \max \set{ \maxdegree{-}{\alpha_j}{\bbD_1}, \maxdegree{-}{\alpha_j}{\bbD_2} }
		+1
	\end{align}

\end{corollary}

\begin{remark} \label{remVilmart}
	Corollary \ref{corZXPhases} was first proved in \cite{BeyondCliffordT}.
	The authors of that paper use the symbol $\mu$ to count appearances of $\alpha_j$ (with coefficient),
	and $T_j$ to denote a large enough set of values.
	Their method does not use Laurent polynomials, instead examining ranks of certain matrices,
	but this also means their method does not extend neatly to other graphical calculi.

	The result of \cite{BeyondCliffordT} is in fact stronger than that of Theorem \ref{thmparameter};
	as they show that under the conditions given here there is a \emph{universal} proof of the parameterised equation,
	something that the method given in Theorem \ref{thmparameter} does not show.
	For the purposes of verification, however, it is not important how an equation is derived,
	only whether the equation is sound.
\end{remark}

\begin{remark}[Distances between points] \label{remAnyPointsWillDo}
It is worth noting that there is no restriction
on the points in each $A_j$ used in Theorem~\ref{thmparameter},
beyond being distinct.
If one were to work over arbitrary rings rather than fields we would
require that the differences between any pair of points were not zero-divisors.
\end{remark}

\begin{example}[ZH phase variable verification]
	Note that the ZX result required the variables to be linear,
	but for ZH and ZW this result applies to diagrams whose phases are polynomial
	(or even Laurent polynomial) in the $\alpha_j$.
	For example we can verify the following ZH equation by checking 3 distinct values of $\alpha$:
	\begin{align}
		\vc{\InputIfFileExists{./figures/ZH/ZH4.tikz}{}{Missing file!}}
	\end{align}

\end{example}

\begin{example}[Phase variables in a quotient ring]
	\label{excliffordt}
	Consider a Clifford+T ZX diagram that contains at least 8 nodes labelled by a positive $\alpha$.
	Our theorem says that for any equation containing this diagram it suffices to try at least 9 distinct values of $\alpha$,
	but this is impossible since there are only 8 distinct values of $\alpha$ available in Clifford+T.

	This is because our Laurent polynomial matrix interpretation
	needs to be viewed not in $\bb{C}[Y, Y\inv]$ but in $\bb{C}[Y] / (Y^8-1)$,
	reflecting the property $8 \times \alpha = 0$ in our phase group.
	All polynomials in $\bb{C}[Y] / (Y^8-1)$ have degree at most 7,
	and so it is never necessary to check more than 8 points.
\end{example}

\section{Phase variables over qubits} \label{secPhaseVariablesOverQubits}

For $\ring_\bbC$, $\ZW_\bbC$, $\ZH_\bbC$ and Universal ZX we can combine Theorem~\ref{thmparameter}
with Theorem~\ref{thmLiftRingHomSoundness} to achieve phase variable verification with a single equation.
The implication of Theorem~\ref{thmparameter} is that if we find `enough' sound instantiations
we can then infer that every instantiation is sound.
The result of Theorem~\ref{thmLiftRingHomSoundness} is that we can use phase homomorphism pairs
to discover more sound equations from an existing sound equation.
All we therefore need to do is find the right phase homomorphisms,
and the right starting equation, and we will have sufficient data to conclude soundness
for a family of equations parameterised by phase variables.

In order to show this we shall require some Galois theory,
which we include below, leading up to Theorem~\ref{thmQubitSingleVerifyingPhaseEquation}.
Galois Theory concerns itself with the study of two fields,
$K \subset L$,
and the field automorphisms of the larger that do not affect the smaller.
A classic example is the pair of fields $\bbR \subset \bbC$,
where the only automorphisms of $\bbC$ that preserves $\bbR$ are
the identity and complex conjugation.
A second example is the pair of fields (usually just called a `field extension')
$\bbQ \subset \bbQ(\sqrt{2})$,
where the only non-trivial automorphism of $\bbQ(\sqrt{2})$ that preserves $\bbQ$ is
the map that sends $a + b \sqrt{2} \mapsto a - b \sqrt{2}$.

Our reason for using Galois Theory is that we will want to fix almost
all of the phases in a diagram, while affecting some others.
For example in a ZX diagram with phases that are multiples of $\pi/8$
we can define the fields $K := \bbQ(e^{i\pi/4})$
and $L := \bbQ(e^{i\pi/8})$.
The automorphisms of $L$ that fix $K$
will lift to phase automorphisms that affect only those phases
that are odd multiples of $\pi/8$.
Much of the work below is in proving
the existence of suitable field extensions,
and then showing that suitable automorphisms exist.
Definitions, lemmas, and theorems
leading up to Theorem~\ref{thmQubitSingleVerifyingPhaseEquation}
are either directly from the literature
(and cited)
or are expected to be well know (and a reference could not be found).
Theorems~\ref{thmQubitSingleVerifyingPhaseEquation}
and \ref{thmZXPhasesOverQubits} are the crux of this section.
This section concerns itself with qubit graphical calculi,
and as such all the fields considered will be subfields of $\bbC$.

\begin{definition}[Degree of a field extension {\cite[Definition~6.2]{GaloisTheory}}]  \label{defGaloisTowerDegree} 
	Where $K$ is a subfield of $L$, written $L / K$,
	we define the degree of the extension $[L : K]$ as
	the dimension of $L$ considered as a vector space over $K$.
\end{definition}

\begin{lemma}[{\cite[Proposition~6.7]{GaloisTheory}}] \label{lemDegreeOfMinimalPolynomial} 
	For $K(\alpha) / K$ a field extension
	where $\alpha$ satisfies some minimal polynomial $m \in K[X]$,
	\begin{align}
		[K(\alpha) : K] & = \deg m
	\end{align}
\end{lemma}

\begin{definition}[Normal extension {\cite[Definition 9.8]{GaloisTheory}}]  \label{defGaloisNormal}
	The field extension $L / K$ is called \emph{normal}
	if every irreducible polynomial over $K$ with a root in $L$ splits in $L$
	as a product of linear factors.
\end{definition}

\begin{theorem}[Normal closure {\cite[Theorem 11.6]{GaloisTheory}}] \label{thmGaloisNormalClosure}
	
	If $L / K$ is a finite extension of subfields of $\bbC$, then there exists
	a unique smallest normal extension $N / L / K$ which is a finite extension of $K$.
\end{theorem}

\begin{theorem}[Tower Law {\cite[Corollary~6.6]{GaloisTheory}}] \label{thmTowerLaw}
	If $K_0 \subseteq K_1 \subseteq \dots \subseteq K_n$ are subfields of $\bbC$, then
	\begin{align}
		[K_n : K_0] & = [K_n : K_{n-1}] [K_{n-1} : K_{n-2}] \dots [K_1 : K_0]
	\end{align}
\end{theorem}

\begin{theorem}[Part of the fundamental theorem of Galois theory {\cite[Theorem~12.2]{GaloisTheory}}]
	\label{thmFundamentalTheoremGaloisTheory} 
	If $L / K$ is a finite, normal field extension inside $\bbC$
	then $\Gal(L / K)$, the group of field automorphisms of $L$ that fix $K$,
	has order $[L : K]$
\end{theorem}

We will need to be able to place an upper bound on the degree of an algebraic extension,
in a similar way to how we needed to place an upper bound on the degree of a matrix.
Assuming we know for each $\alpha$ a polynomial $p$ where $p(\alpha) = 0$,
then we know that the degree of the minimal polynomial $m_\alpha$
has lower degree than $p$.
We can use this, in conjunction with the following corollary,
to get an upper bound for the degree of an algebraic extension.

\begin{corollary} \label{corGaloisBigEnough}
	For $\beta_1, \dots, \beta_r$ algebraic over $K$,
	and $m_{\beta_j}$ the minimal polynomial of $\beta_j$ over $K$
	\begin{align}
		[K(\beta_1, \dots, \beta_r) : K] & \leq \Product_j \deg m_{\beta_j}
	\end{align}
\end{corollary}

Now that we know how to find certain bounds
we shall also need to find extensions that are guaranteed to exceed those bounds.
For this we shall use primitive roots of unity
and cyclotomic polynomials to introduce into a field $K$
elements that guarantee $[K(\alpha) : K] \geq d$ for a given $d$.

\begin{definition}[Primitive Roots of Unity {\cite[Definition 1.5]{GaloisTheory}}]  \label{defPrimitiveRoots} 
	For $n \in \bbN$ the primitive $n$-th roots of unity are
	\begin{align}
		\set{e^{2 i j \pi / n} | (j, n) = 1} &
	\end{align}
	and we shall use $\w_n$ to indicate a generic primitive $n$-th root of unity.
\end{definition}

\begin{definition}[Cyclotomic polynomials {\cite[Definition~21.5]{GaloisTheory}}]  \label{defCyclotimicPolynomial} 
	The $n$-th cyclotomic polynomial $\Phi_n$ is defined as:
	\begin{align}
		\Phi_n(X) := \Product_{1 \leq a \leq n,\ (a,n)=1} (X - \w_n^a)
	\end{align}
	for $\w_n$ a primitive $n$-th root of unity
\end{definition}

\begin{lemma}[Galois group of a cyclotomic extension {\cite[Theorem 21.9]{GaloisTheory}}] \label{lemDegreeOfACyclotomic}
	The Galois group of $\bbQ(\w_n) / \bbQ$ is isomorphic to $(\bbZ / n \bbZ)^*$,
	and the degree of $\Phi_n$ is $\abs{(\bbZ / n \bbZ)^*}$
\end{lemma}

\begin{lemma}[$\Phi_m$ irreducible over transcendental extension] \label{lemCyclotomicIrreducible}
	$\Phi_m$ is irreducible over $\bbQ(\pi_1, \dots, \pi_n)$,
	where $n \geq 0$
	and the $\pi_j$ are transcendental and algebraically independent over $\bbQ$.
\end{lemma}

\begin{proof} \label{prfLemCyclotomicIrreducible}
	If $n=0$ then $\Phi_m$ is irreducible (\cite[Corollary~21.6]{GaloisTheory}).
	If $n > 0$ then let us assume that
	we can factorise $\Phi_m$, i.e. for $\w_m$ a primitive $m$-th root of unity
	there is a polynomial $f$ such that
	\begin{align}
		f(\w_m) & = 0 & \deg f & < \deg \Phi_m & f(X) & \in \bbQ(\pi_1, \dots, \pi_n)
	\end{align}
	and without loss of generality we can assume that $\pi_n$ is used non-trivially
	as a coefficient of $f$.
	We can modify $f$ by clearing denominators and negative powers of $\pi_n$
	to get
	\begin{align}
		g(\pi_n) & = 0         &
		\deg g   & < \infinity &
		g(X) & \in \bbQ(\pi_1, \dots, \pi_{n-1}, \w_m)
	\end{align}
	We therefore can place upper bounds on the degrees of the relevant field extensions:
	\begin{align}
		[\bbQ(\pi_1, \dots, \pi_{n}, \w_m) : \bbQ(\pi_1, \dots, \pi_{n-1}, \w_m)] & \leq \deg g      \\
		[\bbQ(\pi_1, \dots, \pi_{n-1}, \w_m) : \bbQ(\pi_1, \dots, \pi_{n-1})]     & \leq \deg \Phi_m
	\end{align}
	Therefore $[\bbQ(\pi_1, \dots, \pi_{n}, \w_m) : \bbQ(\pi_1, \dots, \pi_{n-1})]$
	is finite,
	and therefore $[\bbQ(\pi_1, \dots, \pi_{n}) : \bbQ(\pi_1, \dots, \pi_{n-1})]$
	is finite, contradicting the algebraic independence of $\set{\pi_j}$ over $\bbQ$.
\end{proof}

\begin{lemma} \label{lemCyclotomicLargeEnough}
	For $K$ a finitely generated field extension of $\bbQ$,
	and any choice of natural number $d$,
	there exist infinitely many primes $p$
	such that $[K(\w_p) : K] \geq d$
\end{lemma}

\begin{proof} \label{prfLemCyclotomicLargeEnough}
	We decompose $K / \bbQ$ as $K / L / \bbQ$
	where $L / \bbQ$ is transcendental
	(and finitely generated by algebraically independent additions to $\bbQ$)
	and $K / L$ is algebraic
	(and finitely generated by algebraically dependent additions to $L$).

	By Lemma~\ref{lemCyclotomicIrreducible}
	$\Phi_m$ is irreducible in $L$,
	and therefore $[L(\w_m) : L] = \deg \Phi_m$.
	Let $r$ be $[K : L]$,
	and $s$ be $[K(\w_m) : K]$.
	Note that $r$ is fixed and we are trying to choose $m$
	such that $s$ is larger than $d$.
	We illustrate this with the diagram of extensions:

	\begin{align}
		\vc{\InputIfFileExists{./figures/wire/bigEnoughFieldExtension.tikz}{}{Missing file!}}
	\end{align}

	We know that
	\begin{align}
		[K(\w_m) : L] & = [K(\w_m) : L(\w_m)][L(\w_m) : L] \geq \deg \Phi_m \\
		[K(\w_m) : L] & = [K(\w_m) :K][K : L] = rs
	\end{align}

	The degree of $\Phi_p$ for $p$ a prime is $p-1$,
	therefore by choosing $m = p$ for $p > rd$
	we know that $[K(\w_m) : K] \geq \deg \Phi_p / r \geq d$
\end{proof}

\begin{lemma} \label{lemCyclotomicsDistinct}
	For the finite set of distinct primes $\set{p_j}$
	there is an isomorphism of groups
	\begin{align}
		\Gal(K(\w_{p_1}, \dots, \w_{p_n}) / K) \iso \Product_j \Gal(K(\w_{p_j}) / K)
	\end{align}
\end{lemma}

\begin{proof} \label{prfLemCyclotomicsDistinct}
	Because the roots of $\Phi_p$ are all powers of any other root,
	either all the roots of $\Phi_p$ are in K or none of them are.
	So either $\w_p$ is in $K$ and $\Gal(K(\w_p) / K) \iso \set{e}$ or
	$\w_p$ is not in $K$ and $\Gal(K(\w_p) / K) \iso (\bbZ / p\bbZ)^*$.
	We can therefore ignore any $\w_p \in K$
	as it contributes trivially to both $\Gal(K(\w_{p_1}, \dots, \w_{p_n}) / K)$
	and $\Product_j \Gal(K(\w_{p_j}) / K)$.

	Only considering those $\w_p \notin K$, Lemma~\ref{lemDegreeOfACyclotomic}
	gives us the isomorphism
	\begin{align}
		\Gal(K(\w_p) / K)              & \iso (\bbZ / p \bbZ)^*     \\
		(\sigma : \w_p \mapsto \w_p^k) & \leftrightarrow (k \mod p)
	\end{align}

	By writing $N := \Product_j p_j$ and $\Omega: = \Product_j \w_{p_j}$, and noting that $K(\w_{p_1}, \dots, \w_{p_n}) = K(\Omega)$,
	we construct the isomorphism we need via the Chinese Remainder Theorem:

	\begin{align}
		\Gal(K(\Omega) / K)                    & \iso \Product_j \Gal(K(\w_{p_j}) / K)                                                                                             \\
		(\sigma : w_{N} \mapsto w_{N}^k ) & \leftrightarrow (\sigma_1 : \w_{p_1} \mapsto \w_{p_1}^{(k \mod p_1)}, \dots ,\sigma_1 : \w_{p_n} \mapsto \w_{p_n}^{(k \mod p_n)}) \\
		(k \mod N)                        & \leftrightarrow (k \mod p_1, \dots, k \mod p_n)
	\end{align}
\end{proof}

\begin{lemma} \label{lemComplexArbitrarilyLargeDegree}
	For any list of natural numbers $d_1, \dots, d_n$ and
	finitely generated field extension $K$ of $\bbQ$
	we can construct elements $a_1, \dots, a_n$ of $\bbC$
	(which happen to be roots of unity)
	such that
	\begin{align}
		\Gal(K(a_1, \dots, a_n) / K) & \iso \Gal(K(a_1) / K) \times \dots \times \Gal(K(a_n) / K) \\
		\abs{\Gal(K(a_j) / K)}       & \geq d_j \qquad \forall j
	\end{align}
\end{lemma}

\begin{proof} \label{proofLemComplexArbitrarilyLargeDegree}
	For any natural number $d$ there exists a prime $p$ such that
	$[K(\w_p) : K] \geq d$ by Lemma~\ref{lemCyclotomicLargeEnough}.
	Therefore by Lemma~\ref{lemCyclotomicsDistinct}
	\begin{align}
		\Gal(K(\w_{p_1}, \dots, w_{p_n}) / K) \iso \Gal(K(\w_{p_1}) / K)
		\times \dots \times \Gal(K(\w_{p_n}) / K)
	\end{align}
	and for each subgroup $\abs{\Gal(K(\w_{p_j}) / K)} \geq d$.
	Set $a_j$ to a choice of $\w_{p_j}$, such as $e^{2i\pi / {p_j}}$.
\end{proof}

\begin{theorem}[Single equation verification of qubit phase variables, ring version] \label{thmQubitSingleVerifyingPhaseEquation}
	For a $\ring_\bbC$, $\ZW_\bbC$, or $\ZH_\bbC$ equation
	with phases that are polynomial in some phase variables,
	there exists a single verifying equation without any phase variables.
	If we can put an upper bound on the degrees of the minimal polynomials
	of the phase constants then we can construct such an equation.
\end{theorem}

\begin{proof} \label{prfThmQubitSingleVerifyingPhaseEquation}
	Let $D_1 = D_2$ be the equation of interest,
	with phase variables $\alpha_1, \dots, \alpha_n$,
	and phase constants $\set{c_1, \dots, c_m}$.
	Let $K$ be the normal closure of $\bbQ(c_1, \dots, c_m)$ in $\bbC$;
	this exists by Theorem~\ref{thmGaloisNormalClosure},
	is finite over $\bbQ(c_1, \dots, c_m)$,
	and is therefore finitely generated over $\bbQ$
	by the $\set{c_1, \dots, c_m}$ and finitely many more elements (any missing Galois conjugates).
	Note that $K$ satisfies the conditions of Lemma~\ref{lemCyclotomicLargeEnough}.
	We will construct the assignments
	\begin{align}
		\alpha_1 = a_1, \alpha_2 = a_2, \dots, \alpha_n = a_n
	\end{align}
	and then show that this instantiation leads to the desired outcome.

	\begin{align}
		\vc{\InputIfFileExists{./figures/wire/galoisTheorem.tikz}{}{Missing file!}}
	\end{align}

	By Lemma~\ref{lemComplexArbitrarilyLargeDegree}
	there exist elements $a_1 , \dots, a_n \in \bbC$ such that
	\begin{align}
		\Gal(K(a_1, \dots, a_n) / K) \iso \Gal(K(a_1) / K) \times \dots \times \Gal(K(a_n) / K)
	\end{align}
	and $\abs{\Gal(K(a_j) / K)} \geq d_j$ for all $j$;
	note that $d_j$ is the required size of the sets $A_j$, as in Theorem~\ref{thmparameter}.

	In order to construct the phases we shall need to know some
	algebraic information about all the phases $c_j$ already present in the diagram.
	That is for a given $c_j$ we will need to know either
	that $c_j$ is transcendental over $\bbQ$, or some polynomial that $c_j$ satisfies.
	We can then apply Corollary~\ref{corGaloisBigEnough} to get an upper bound on the degree
	of the algebraic extension.
	Then using Lemma~\ref{lemComplexArbitrarilyLargeDegree} we can construct suitable values $a_1, \dots, a_n$.

	Suppose the instantiation $\alpha_1 = a_1, \dots, \alpha_n = a_n$
	of $D_1 = D_2$ is sound.
	By Theorem~\ref{thmLiftRingHomSoundness}
	we know that applying $\hat \phi$ to this instantiation is also sound,
	for $\phi$ an automorphism that fixes $K$.

	Therefore for any element $\phi$ of $\Gal(K(a_1, \dots, a_n) / K)$
	we know that the instantiation $\alpha_1 = \phi(a_1), \dots, \alpha_n = \phi(a_n)$
	is sound.
	Since the Galois group splits as a product of groups each affecting a single $a_j$
	(Lemma~\ref{lemComplexArbitrarilyLargeDegree})
	we can apply Theorem~\ref{thmparameter}
	with the sets:
	\begin{align}
		A_j := \set{\phi(a_j) | \phi \in \Gal(K(a_j) / K)}
	\end{align}
	knowing that each of these choices of instantiations yields a sound equation.
	Therefore by verifying the phase-variable-parameterised equation
	$D_1 = D_2$ with the instantiation $\alpha_1 = a_1, \dots, \alpha_n = a_n$
	we have verified $D_1 = D_2$ for all values of the variables.
\end{proof}

\begin{example}[Verifying a $\ring_\bbC$ family of equations with a single equation] \label{exaGalois}
	Consider the $\ring_{\bbC[x,y]}$ equation
	\begin{align}
		\bap{white}{}{x}{y} & = \state{white}{xy}
	\end{align}
	By Theorem~\ref{thmparameter} we require sets $A_x$ and $A_y$,
	with sizes $d_x \geq 2$ and $d_y \geq 2$.
	We can therefore instantiate $x$ to $e^{2i\pi/3}$ and $y$ to $e^{2i\pi/5}$,
	knowing that
	\begin{align}
		\bapwide{white}{}{e^{2i\pi/3}}{e^{2i\pi/5}}                       & = \state{white}{e^{16i\pi/15}} \quad \text{sound}                                                       \\
		\implies \bapwwide{white}{}{\phi(e^{2i\pi/3})}{\phi(e^{2i\pi/5})} & = \state{white}{\phi(e^{16i\pi/15})} \quad \text{sound } \forall \phi \in \Gal(\bbQ(\w_3, \w_5) / \bbQ) \\
		\implies \bap{white}{}{x}{y}                                      & = \state{white}{xy} \quad \text{sound $\forall x, y \in \bbC$ by Theorem~\ref{thmparameter}}
	\end{align}
	In the language of \S\ref{chapConjectureInference}: The equation
	\begin{align}
		\bapwide{white}{}{e^{2i\pi/3}}{e^{2i\pi/5}} = \state{white}{e^{16i\pi/15}} &
	\end{align}
	implies and is implied by the equation
	\begin{align}
		\set{\bap{white}{}{x}{y} = \state{white}{xy}}_{x,y} & x,y \text{ phase variables}
	\end{align}

\end{example}

\begin{theorem}[Single equation verification of qubit phase variables, ZX version] \label{thmZXPhasesOverQubits}
	For a Universal ZX equation
	with phases that are linear in some phase variables,
	there exists a single verifying equation without any phase variables.
	If all the phase constants are rational multiples of $\pi$
	then we can construct such an equation.
\end{theorem}

\begin{proof} \label{prfThmZXPhasesOverQubits}
	The idea of the proof of Theorem~\ref{thmZXPhasesOverQubits}
	is nearly identical to that of Theorem~\ref{thmQubitSingleVerifyingPhaseEquation}.
	The difference is that we cannot rely directly on the results
	about phase homomorphisms and so must reconstruct the necessary properties directly.

	We first regard the ZX equation as generated by Z spiders and Hadamard gates.
	Each concrete phase $c$ on a green node
	corresponds to $e^{ic}$ in the matrix representation.
	We therefore construct
	$K$ as the normal closure of $\bbQ(e^{ic_1}, \dots, e^{ic_m}, e^{i\pi/4})$.
	The $e^{i\pi/4}$ ensures that $\sqrt{2}$ is in $K$ and so the
	automorphisms created will not unintentionally interact
	with the interpretation of the Hadamard gate.
	By Lemma~\ref{lemComplexArbitrarilyLargeDegree}
	there exist roots of unity $a_1 , \dots, a_n \in \bbC$ such that
	\begin{align}
		\Gal(K(a_1, \dots, a_n) / K) \iso \Gal(K(a_1) / K) \times \dots \times \Gal(K(a_n) / K)
	\end{align}
	and $\abs{\Gal(K(a_j) / K)} \geq d_j$ for all $j$;
	note that $d_j$ is the required size of the sets $A_j$, as in Theorem~\ref{thmparameter}.

	Any given $\phi \in \Gal(K(a_1, \dots, a_n) / K)$
	sends $\set{a_1, \dots, a_n}$ to $\set{b_1, \dots, b_n}$
	where $b_j$ is a Galois conjugate of $a_j$
	and can be viewed as conjugating the phases that each root of unity represents
	(via $a \leftrightarrow e^{i\alpha}$).
	We lift this permutation to a map $\hat \phi$ for the ZX diagrams,
	noting that this is a very similar construction to that of phase group homomorphism pairs.
	Since $K$ is fixed by $\phi$ we know that $\hat \phi$ does not change any of the concrete
	phases from the initial diagram.
	Since $e^{i\pi/4}$ is in $K$ we can
	simply check that for Z spiders and for Hadamard gates that $\phi$ and $\hat\phi$ respect $\interpret{\cdot}$
	in the sense that:
	\begin{align}
		\phi(\interpret{\cdot}) = \interpret{\hat\phi({\cdot})}
	\end{align}
	This gives us the property that if our original equation was sound then
	so is the image under $\hat \phi$.
	We can therefore again apply Theorem~\ref{thmparameter}
	to show that the constructed equation (with phases $\set{q_j}$ instead of variables)
	and all the images of this equation under $\hat \phi$
	verify the original family of parameterised equations.

	The argument about construction is also similar to that
	in the proof of Theorem~\ref{thmQubitSingleVerifyingPhaseEquation},
	noting that the Universal ZX phase $2a\pi/b$ with $a/b$ rational
	represents the complex number $e^{i2a\pi/b}$,
	which is a solution to the equation $X^b-1 = 0$,
	and so the degree of the corresponding minimal polynomial is bounded above by $b$.
\end{proof}

\begin{remark} \label{remNonQImpliesPhaseVariable}
	One does not have to work from parameterised equation to verifying, instantiated equation.
	The presence of a phase $\tau$ in $\bb{A}\setminus\bbQ$
	for a $\ring_\bbC$ equation
	already suggests at least two sound equations
	(since there is a non-identity Galois automorphism $\phi$ for $\tau$ over $\bbQ$,
	and we know that $\hat \phi$ preserves soundness).
	Such a phase $\tau$ suggests, in the conjecture inference sense, the family of equations
	where $\tau$ has been replaced by a phase variable.
	Again this is similar to the replacement of existential quantifiers with universal quantifiers
	as mentioned in \S\ref{secGeneralisingTheorems}.

	One could also view this result as a bound on the `algebraic complexity'
	of phases in an equation of diagrams:
	If any field element appears less frequently
	than the degree of its minimal polynomial then we should seek to replace it with a phase variable.
	One does, however, need to be careful when trying to calculate how often an element
	appears in a diagram.
	For example $\sqrt{6}$ and $\sqrt{2}$ are not algebraically independent,
	so naively counting $\sqrt{2}$ as not occurring in $\sqrt{6}$ will lead to errors.
\end{remark}

Having brought the number of verifying equations needed down to one
(in the qubit case) we now turn our attention to !-boxes,
which are only marginally less well behaved.

\section{Verifying !-boxes}	 \label{secBBoxes}

!-boxes allow the succinct representation
of arbitrary numbers of copies of subdiagrams.
In this section we will show that
unless two !-boxes are connected (and not nested)
then this family of diagrams will behave in
a very constrained manner.
The main result is that
if an equation includes !-boxes that are \emph{separable}
(i.e. not connected as described above)
then we need only check the first $N$ instances
of the !-box expansions in order to verify the entire family.
As with the result on phase variable verification
we can read this value $N$ off the diagram directly.
Part of this result is the definition of \emph{series !-box form},
a manipulation of the diagram (requiring spider laws)
that presents the !-box expansions
in a novel manner.
This manipulation allows us to view
each further !-box expansion in terms of the $\comp$ product
rather than the more traditional $\tensor$ presentation.

In terms of our running ZH example from \S\ref{secVerificationMotivation}
we are now investigating this face of the commutative cube:
\begin{align}
	\begin{tikzcd}[ampersand replacement=\&]
		\ZHC^! \arrow[dr, "\interpret{\cdot}"] \arrow[rr, "as"] \& \& \ZHC  \arrow[dr, "\interpret{\cdot}"] \&\\
		\& (\bb{N} \to \Mat_{\bb{C}}) \arrow[rr, "as"] \& \&  \Mat_{\bb{C}}
	\end{tikzcd}
\end{align}
With a third face given in \eqref{eqnThirdFace}.

\subsection{Copies, nesting, and separability}
\label{secnesting}
There is a choice to be made when expanding a !-box that contains a phase variable.
The approach taken in \cite{MerryThesis}, which we adopt here,
is to copy the variable name exactly.
This is in contrast to creating a `fresh' name for each new instance of the variable
(Quantomatic, Ref.~\cite{Quantomatic}, implements a version of this for users that want it).
The original paper that this work appeared in, \cite{millerbakewell2019finite},
contains versions of these proofs that encompass both options,
but since only minor alterations are required at any point we now present just the proof for the following definition
of copying:

\begin{definition}[Copying variables inside !-boxes] \label{defCopy}
	When a !-box expansion creates new instances of a phase variable
	we \emph{copy} that variable name,
	so that all instances are linked by the same name (the approach taken in \cite[\S 4.4.2]{MerryThesis})
\end{definition}

This definition results in the commutativity of $ev$ and $as$,
giving us another face in the commutative cube from our running example (see \S\ref{secVerificationMotivation}).

\begin{align} \label{eqnThirdFace}
	\begin{tikzcd}[ampersand replacement=\&]
		\ZH_{\bbC[X]}^! \arrow[d, "ev"] \arrow[r, "as"] \& \ZH_{\bbC[X]} \arrow[d, "ev"] \\
		\ZHC^! \arrow[r, "as"] \& \ZHC
	\end{tikzcd}
\end{align}

\begin{definition}[nesting order]

	We define a partial order, called the nesting order, on !-boxes in a diagram:
	\begin{align}
		\delta_1 < \delta_2 & \quad \text{ if $\delta_1$ is inside $\delta_2$}
	\end{align}

	And use this partial order to draw a nesting diagram.
	For example this Universal ZX diagram:

	\begin{align}
		\vc{\InputIfFileExists{./figures/wire/net-example-l.tikz}{}{Missing file!}} \quad \text{ has nesting diagram } \quad \vc{\InputIfFileExists{./figures/wire/net-example-r.tikz}{}{Missing file!}}
	\end{align}

\end{definition}

\begin{remark} \label{remPatternNesting}
	The nesting diagram is the same as just taking the !-nodes in a pattern graph.
\end{remark}

\begin{definition}[Well nested]
	We say an equation is \emph{well nested} if the nesting diagrams corresponding to
	the left and right hand sides of the equation are identical.
\end{definition}

\begin{definition}[Join]  \label{defJoin}
	The \emph{join} of a !-box is the collection of wires that leave that !-box,
	i.e. the edges linking a vertex inside $\delta$ to one outside $\delta$.
	The \emph{size} of a join is the number of wires, and
	the \emph{dimension} of a join is the dimension of the diagram formed by just the wires of that join,
	i.e. $\dim \interpret{ \;\vc{\begin{tikzpicture}
	\begin{pgfonlayer}{nodelayer}
		\node [style=none] (6) at (0, 0.25) {};
		\node [style=none] (7) at (0, -0.25) {};
	\end{pgfonlayer}
	\begin{pgfonlayer}{edgelayer}
		\draw (6.center) to (7.center);
	\end{pgfonlayer}
\end{tikzpicture}
}\;^{\tensor n} }$ where $n$ is the size of the join.
	This is equivalently $\left(\dim \interpret{ \;\vc{\begin{tikzpicture}
	\begin{pgfonlayer}{nodelayer}
		\node [style=none] (6) at (0, 0.25) {};
		\node [style=none] (7) at (0, -0.25) {};
	\end{pgfonlayer}
	\begin{pgfonlayer}{edgelayer}
		\draw (6.center) to (7.center);
	\end{pgfonlayer}
\end{tikzpicture}
}\; }\right)^n$.
\end{definition}

\begin{definition}  \label{defSeparated}

	We describe a pair of !-boxes as \textbf{separated} if either:
	\begin{itemize}
		\item They are nested, or
		\item There is no edge joining a vertex in one to a vertex in the other
	\end{itemize}
\end{definition}

\label{secSeparability}

Theorem~\ref{thmbbox2} will require that the !-boxes all be separated,
but before we look at that theorem we will investigate how stringent a requirement that is;
or rather the question `if two !-boxes are not separated, then how hard is it to separate them?'
This question is directly relevant for conjecture synthesis,
as it allows us to convert an equation that is not amenable to
finite verification into an equation that is amenable to finite verification.
This determines which sorts of conjectures we should generate in our conjecture synthesis.

\begin{definition}[Separable]  \label{defSeparable}
	We describe a non-separated pair of !-boxes as separable if we can perform the following operation:
	\begin{align}
		\vc{\InputIfFileExists{./figures/wire/sep-l.tikz}{}{Missing file!}} \quad = \vc{\InputIfFileExists{./figures/wire/sep-r.tikz}{}{Missing file!}}
	\end{align}
	We define pairs of nodes as separable if we can always separate !-boxes that are joined by edges between these pairs of nodes.
	Note that we only need to consider nodes that have arbitrary arity,
	since only they can be connected to !-boxes.
\end{definition}

\begin{lemma} \label{lemSeparability}
	The following pairs of nodes are separable, by calculus:
	\begin{align}
		\text{ZX:} \quad &
		\left(\spider{gn}{}, \spider{rn}{}\right) \quad
		\left(\spider{gn}{}, \spider{gn}{}\right) \quad
		\left(\spider{rn}{}, \spider{rn}{}\right) \\
		\text{ZH:} \quad &
		\left(\spider{smallZ}{}, \spider{smallZ}{}\right) \quad
		\left(\spider{smallZ}{}, \spider{smallgrey}{}\right) \quad
		\left(\spider{smallgrey}{}, \spider{ZH}{}\right) \\
		\text{ZW:} \quad &
		\left(\spider{white}{}, \spider{white}{}\right) \quad
		\left(\spider{white}{}, \spider{smallblack}{}\right) \\
		\ring: \quad     &
		\left(\spider{white}{}, \spider{white}{}\right)
	\end{align}
\end{lemma}

\begin{proof} \label{prfLemSeparability}
	These proofs follow immediately from the spider and bialgebra laws in the respective calculi.
	Note that it is enough to specify the phase-free versions of these interactions,
	because phases can always be moved away from the critical nodes.
\end{proof}

Having shown how to separate !-boxes from each other, we will now talk about
how to spread apart the instances of a !-box once it has been instantiated.
The form we are aiming for is quite particular;
one that allows us to use properties of the vertical composition $\comp$.

\subsection{Series !-box form} \label{secSeriesBBoxForm}

Instantiation of !-boxes is an operation of pattern graphs
and it is only because the calculi we are considering are built from spiders
that the instantiated outcomes still form valid diagrams.
After instantiating !-boxes, and in the presence of suitable spider laws,
we can change the presentation of the repeated elements into more amenable forms.
In particular we can represent the $d$ copies created by a !-box
as a $d$-fold composition of diagrammatic elements,
within the context of a larger diagram.
We call this `series !-box form', providing a definition, example and method below.

\begin{definition}[Series !-box form] \label{defSeriesBBoxForm}
	\emph{Series !-box form} for a given non-nested, separated !-box $\delta_1$
	is a presentation $C \comp (B^{\comp d}) \comp G$ of each the $(\delta_1 = d)$-instantiated diagrams, where
	\begin{itemize}
		\item $C  := $ An unparameterised end diagram                                                                                    \\
		\item $B  := $ Repeated element, which may contain $\alpha_j \forall j,\,\delta_k \forall k \geq 2$ and some boundary nodes
		\item $G  :=$ The rest of the diagram outside of $\delta_1$, which may contain
		      $\alpha_j \forall j,\,\delta_k \forall k \geq 2$ and some boundary nodes
	\end{itemize}
\end{definition}

\begin{example}[Series !-box form] \label{exaSeriesBBoxForm}
	To give a simple ZX example of series !-box form, with just one !-box and just one wire between $G$ and $B'$:

	\begin{align} \label{eqnSingleSeriesBBox}
		\family{\vc{\InputIfFileExists{./figures/ZX/internal_spider_join.tikz}{}{Missing file!}}}_{\delta|\delta=d}\quad = \quad \vc{\InputIfFileExists{./figures/ZX/bbox_sideways.tikz}{}{Missing file!}}
	\end{align}
	Note that $B$ is not just the subdiagram $B'$
	but also everything directly below it as depicted in \eqref{eqnSingleSeriesBBox}.
	Although we have only used one example node and joining wire we can perform this action
	on all nodes and joining wires.
	Where two wires travel from the !-box to the same spider inside $G$
	we first spread out that spider so each wire from the !-box connects to a different spider in $G$.
	Here is an example for $n$ wires between $B'$ and $G$, and $p$ boundaries inside $\delta$.

	\begin{align}\label{eqnseries-bbox}
		\vc{\InputIfFileExists{./figures/ZX/bbox_sideways_compact_external.tikz}{}{Missing file!}}
	\end{align}

	From here it is easy to see that we have the diagram $G : m \to n$,
	beside $d$ copies of a diagram we call $B : p + n \to n$ (containing the $p$ boundary nodes and the new connecting spiders),
	and finally an ending diagram $C: n \to 0$.

\end{example}

\begin{proposition} \label{propSeriesBBoxForm}
	In the calculi ZX, ZW and ZH,
	for any diagram $\bbD$ and any value of $d$ we can put $\bbD_{\alphadelta | \delta_1=d}$ into series !-box form.
\end{proposition}

\begin{proof} \label{prfPropSeriesBBoxForm}
	We will only indicate the !-box we are interested in ($\delta$) and not explicitly write out the other parameters.
	First we manipulate the diagram (by Choi-Jamiolkowski we can transform inputs into outputs,
	provided we turn them back again later) until it is in the following form:

	\begin{align} \bb{D} \quad = \quad  \vc{\InputIfFileExists{./figures/wire/d_and_b.tikz}{}{Missing file!}} \end{align}

	We note that the nodes inside $G$ that join with $B'$ must be spiders;
	since there can be arbitrary many instances of $B'$ the connecting node in $G$ must be able to have arbitrary arity.
	There may also be $p$ boundaries that are contained in $\delta$.
	We will now rely on the existence of a spider law such that we may do the following:
	\begin{align}
		\vc{\InputIfFileExists{./figures/wire/gen-spider-l.tikz}{}{Missing file!}} \quad = \quad \vc{\InputIfFileExists{./figures/wire/gen-spider-r.tikz}{}{Missing file!}}
	\end{align}
	Which is the ability to `spread' a spider with $q$ outputs into $q$
	repeated copies of a spider with 1 output, with suitable initial, terminal, and joining subdiagrams.
	Such spider laws exists for $\ZW$, $\ZX$, $\ZH$ and $\ring$.
	We apply this to all the spiders in $G$ that are connected to the instantiated subdiagrams
	formed by $\delta$, so that the spiders form a chain as in \eqref{eqnseries-bbox}.
\end{proof}

\subsection{The verification process for !-boxes}

Using this notion of series !-box form (Definition~\ref{defSeriesBBoxForm})
we will now show that checking just an initial segment of a !-box instantiation
suffices for showing soundness for all possible instantiations.
Before stating and proving our theorem we will first need two simple lemmas
concerning tensor products of bases,
and increasing chains of vector subspaces.

\begin{lemma} \label{lemSeriesBasis}
	The set of all vectors of the form
	$\set{v_d \tensor \dots \tensor v_1 \tensor x}$
	where
	$v_j \in \Hilbert^{\tensor p}$
	and
	$x \in \Hilbert^{\tensor m}$ contains a basis for
	$\Hilbert^{\tensor(m + dp)}$
\end{lemma}

\begin{proof} \label{prfLemSeriesBasis}
	Note that we may form a basis for $V \tensor V'$ by taking the tensor products of the bases of $V$ and $V'$,
	and therefore the above set contains all the basis elements of $\Hilbert^{\tensor p} \tensor \dots \Hilbert^{\tensor p} \tensor
		\Hilbert^{\tensor m} \iso \Hilbert^{\tensor (m +dp)}$.
\end{proof}

\begin{lemma} \label{lemIncreasingSubspace}
	A chain of subspaces $V_j$ of an $N$-dimensional vector space $V'$
	has the following property:
	\begin{align}
		\text{If}   & \nonumber                                               \\
		            & V_j \leq V_{j+1}\ \forall j                             \\
		            & (V_j = V_{j+1}) \implies (V_{j+1} = V_{j+2})            \\
		\text{Then} & \nonumber                                               \\
		            & \exists n \leq N \st V_{n+k} = V_n\ \forall  k \in \bbN
	\end{align}
\end{lemma}

\begin{proof} \label{prfLemIncreasingSubspace}
	By the assumptions in the statement $V_j$ is non-decreasing, $\dim V_j \leq N$,
	and if $V_j = V_{j+1}$ then $V_{j} = V_{k} \forall k \geq j$;
	therefore the sequence $\set{V_j}$ has an initial increasing section, followed by a stable section.
	By looking at the sequence $\set{\dim V_j}$ we know that this initial increasing section has length at most $N$.
\end{proof}

\begin{definition}[Verifying subset]
	Given a parameterised equation $\bbE$ we say that $\set{\bbE_1, \dots, \bbE_n}$ {verifies} $\bbE$ if:
	\[
		(\forall j \;\bbE_j \text{ is sound}) \implies \bbE \issound
	\]
\end{definition}

\begin{restatable}[Finite verification of !-boxes]{theorem}{thmbboxx}\label{thmbbox2}
	Given a family 
	\begin{align*}
	\bbE = \family{\bb{D}_1 = \bb{D}_2}_{\alphadelta}
	\end{align*} of diagrammatic equations,
	in ZX, ZH, or ZW,
	parameterised by a !-box $\delta_1$
	where $\delta_1$ is separated from all other !-boxes
	and is nested in no other !-box;
	then $\bbE$ is verified by the finite family
	\begin{align*}
	\set{\bbE|_{\delta_1 = 0}, \dots, \bbE|_{\delta_1 = N}}
	\end{align*}
	where $N$ is the dimension of the join of $\delta_1$ in $\bb{D}_1$
	plus the dimension of the join of $\delta_1$ in $\bb{D}_2$.
	That is:
	\begin{align}
		n_1 := & \text{ join of $\delta_1$ in $\bbD_1$} \nonumber                                   \\
		n_2 := & \text{ join of $\delta_1$ in $\bbD_2$} \nonumber                                   \\
		N :=   & \dim \left( \Hilbert^{\tensor n_1} \oplus \Hilbert^{\tensor n_2} \right) \nonumber \\[2em]
		       & \interpret{\family{\bb{D}_1}_{\alphadelta | \delta_1 = d_1}} \nonumber             \\
		=      &
		\interpret{\family{\bb{D}_2}_{\alphadelta | \delta_1 = d_1}}
		       & \forall d_1 \leq N                                                                 \\[2em]
		\implies  \label{eqnbbox2}
		       & \interpret{\family{\bb{D}_1}_{\alphadelta}}  \nonumber                             \\
		=      &
		\interpret{\family{\bb{D}_2}_{\alphadelta}}
		       &
	\end{align}
\end{restatable}

The sketch of the proof is:
\begin{enumerate}
	\item Manipulate the diagrams into series !-box form (Definition~\ref{defSeriesBBoxForm})
	\item Move to the matrix interpretation
	\item Manipulate the equation between two matrices into an expression on a single vector space of dimension $N$
	\item Demonstrate the required property as a condition on subspaces
\end{enumerate}

\begin{proof}
	Considering either $\bbD_1$ or $\bbD_2$ for now,
	manipulate ${\bbD}_{\delta_1 = d}$ into series !-box form;
	$C \comp (B)^{\comp d} \comp G$.
	The $p$ boundary vertices in every copy of $B$ pose an issue for our intended method,
	and so we instead consider $B$ as being parameterised not just by $\alpha_j \; \forall j$
	and $\delta_k \; \forall k > 1$, but also by
	input vectors $v \in \Hilbert^{\tensor p}$ that `plug' the inputs inside $B$.

	Since we can show equivalence of complex matrices by showing that they perform the same operation on any input,
	it will be enough to show that for any choice of $\alpha_j$, $\delta_{k>1}$
	and $v$ that \eqref{eqnbbox2} holds.
	Assuming we have instantiated the $\alpha_j$ and $\delta_{k>1}$,
	Lemma~\ref{lemSeriesBasis} justifies that we can choose the input vector for the $p$ inputs of
	each of the $d$ copies of $B$ independently.

	Given a choice of values for the $\alpha_j$, $\delta_{k > 1}$ and $v$
	we denote this choice by $q$ and use $B_q$ to mean `the sub-diagram $B$ from the series !-box form with this choice of variables'.
	We also define $G_q$ as the subdiagram $G$ instantiated with the values
	of $\alpha_j$ and $\delta_{k > 1}$
	described by $q$, same as for $B_q$
	(we do not plug any values into the inputs of $G$).
	Once we have chosen values for $q$ we may consider the matrix interpretation of the diagram:
	\begin{align}
		    & \interpret{C} \interpret{B_{q_d}} \dots \interpret{B_{q_1}} \interpret{G_{q_o}} \\
		G_q & : \Hilbert^{\tensor m} \to  \Hilbert^{\tensor n}                                \\
		B_q & :  \Hilbert^{\tensor n} \to  \Hilbert^{\tensor n}                               \\
		C   & :  \Hilbert^{\tensor n} \to \bbC
	\end{align}

	Going back to our original question:
	Given an equation $\bbD_1 = \bbD_2$ of two families of diagrams, both parameterised by a (non-nested, separated) !-box $\delta_1$
	(among other parameters)
	we wish to remove our dependence on $\delta_1$ by
	instead verifying a finite set of equations, each of which has a different value for $\delta_1$.
	Note that for this to be the case we require the number of inputs to be equal;
	i.e. $m := m_1 = m_2$ and $p := p_1 = p_2$, but we do not require $n_1 = n_2$ in \eqref{eqnseries-bbox}.
	We instantiate $\delta_1 = d$ and consider $D_1$ and $D_2$ to be in series !-box form,
	with matrix interpretations:
	\begin{align}
		  & \interpret{C_1} \interpret{B_{1,q_d}} \dots \interpret{B_{1,q_1}} \interpret{G_{1,q_0}} \\
		  & \interpret{C_2} \interpret{B_{2,q_d}} \dots \interpret{B_{2,q_1}} \interpret{G_{2,q_0}}
	\end{align}
	And we wish to know when these two interpretations are equal.
	Rather than consider the matrices acting on two independent spaces
	we view them as acting on the direct sum of those two spaces
	and represent these maps as block matrices.
	(We drop the $\interpret{\cdot}$ notation when it would appear inside a matrix.)
	\begin{align}
		     & \interpret{C_1} \interpret{B_{1,q_d}} \dots \interpret{B_{1,q_1}} \interpret{G_{1,q_0}}
		= \interpret{C_2} \interpret{B_{2,q_d}} \dots \interpret{B_{2,q_1}} \interpret{G_{2,q_0}}\\
		\iff &
		\begin{bmatrix}	1 & -1\end{bmatrix}
		\begin{bmatrix}C_1 & 0 \\ 0 & C_2\end{bmatrix}
		\begin{bmatrix}B_{1, q_d} & 0 \\ 0 & B_{2, q_d}\end{bmatrix}
		\dots
		\begin{bmatrix}B_{1, q_1} & 0 \\ 0 & B_{2, q_1}\end{bmatrix}
		\begin{bmatrix}G_{1,q_0} & 0 \\ 0 & G_{2,q_0}\end{bmatrix}
		\begin{bmatrix}	\id_m \\ \id_m\end{bmatrix}
		= 0
	\end{align}
	One can think of the above as copying an input vector $x \in \Hilbert^m$ as $x :: x$ in $\Hilbert^m \oplus \Hilbert^m$,
	then applying
	\begin{align}
		  & \interpret{C_1} \interpret{B_{1,q_d}} \dots \interpret{B_{1,q_1}} \interpret{G_{1,q_0}} \\
		\text{and} \nonumber \\
		  & \interpret{C_2} \interpret{B_{2,q_d}} \dots \interpret{B_{2,q_1}} \interpret{G_{2,q_0}}
	\end{align}
	to the left and right copies respectively.
	After that we apply a minus sign to the right hand result and add that to the left hand result,
	effectively comparing them and demanding the difference to be 0.
	We seek to prove:
	\begin{align}
		  & \nonumber
		\begin{bmatrix}	1 & -1\end{bmatrix}
		\begin{bmatrix}C_1 & 0 \\ 0 & C_2\end{bmatrix}
		\begin{bmatrix}B_{1, q_d} & 0 \\ 0 & B_{2, q_d}\end{bmatrix}
		\dots
		\begin{bmatrix}B_{1, q_1} & 0 \\ 0 & B_{2, q_1}\end{bmatrix}
		\begin{bmatrix}G_{1,q_0} & 0 \\ 0 & G_{2,q_0}\end{bmatrix}
		\begin{bmatrix}	\id_m \\ \id_m\end{bmatrix}
		= 0 \nonumber \\
		  & \qquad \forall d \leq N, q  \nonumber \\
		\implies
		  &
		\begin{bmatrix}	1 & -1\end{bmatrix}
		\begin{bmatrix}C_1 & 0 \\ 0 & C_2\end{bmatrix}
		\begin{bmatrix}B_{1, q_d} & 0 \\ 0 & B_{2, q_d}\end{bmatrix}
		\dots
		\begin{bmatrix}B_{1, q_1} & 0 \\ 0 & B_{2, q_1}\end{bmatrix}
		\begin{bmatrix}G_{1,q_0} & 0 \\ 0 & G_{2,q_0}\end{bmatrix}
		\begin{bmatrix}	\id_m \\ \id_m\end{bmatrix}  = 0   \nonumber \\
		  & \qquad \forall d, q
	\end{align}
	Recalling that $q$ is the choice of values for $\alpha_j$, $\delta_{k>1}$ and $v \in \Hilbert^{\tensor p}$,
	we use $Q$ to denote the set of all possible choices.
	We use $\hat B_q$ for the matrix that acts as the direct sum of $B_{1,q}$ and $B_{2,q}$:
	\begin{align}
		\hat B_q := &
		\begin{bmatrix}B_{1, q} & 0 \\ 0 & B_{2, q}\end{bmatrix}
	\end{align}
	and inductively define the spaces:
	\begin{align}
		V_0 & := \text{span} \set{ \bigcup_{q \in Q} \text{Im} \left(\begin{bmatrix}G_{1,q} & 0 \\ 0 & G_{2,q}\end{bmatrix}\begin{bmatrix}\id_m \\ \id_m\end{bmatrix}\right)} \\
		V_j & := \text{span} \set{ V_{j-1} \cup \; \bigcup_{q \in Q} \hat B_q V_{j-1}}
	\end{align}
	The $V_j$ form an ascending sequence of subspaces, each containing the potential images of up to
	$j$ applications of $\hat B_q$:
	\begin{align}
		V_j \geq \text{Im}(\;\hat B_{q_k} \dots \hat B_{q_1}\, V_0\;)  \quad \forall k \leq j\; \forall q_k, \dots, q_1 \in Q
	\end{align}
	By Lemma~\ref{lemIncreasingSubspace} there is a $V_b$ with the following properties:
	\begin{itemize}
		\item if $j <b$ then $V_j > V_{j-1}$
		\item if $j \geq b$ then $V_j = V_{j-1}$
		\item $b \leq \dim ( \; \Hilbert^{n_1} \oplus \Hilbert^{n_2} \; )$
	\end{itemize}
	Consider the subspace $W$ defined as
	\begin{align}
		W & := \begin{bmatrix}	1 & -1\end{bmatrix}
		\begin{bmatrix}C_1 & 0 \\ 0 & C_2\end{bmatrix}
	\end{align}
	We wish to show that:
	\begin{align}
		         & V_j \leq W & \forall j \leq N \\ \nonumber
		\implies & V_j \leq W & \forall j
	\end{align}
	Since $V_c = V_N $ when $ c \geq N \geq b$ it is enough to show that this is the case for all $V_j$ when $j \leq N$.
	This is, however, the assumption in our theorem;
	that for $d \leq N$ our diagrammatic equation holds,
	and so for any choice of $d$ and $q_1, \dots, q_d$ this matrix equation holds:
	\begin{align}
		\begin{bmatrix}	1 & -1\end{bmatrix}
		\begin{bmatrix}C_1 & 0 \\ 0 & C_2\end{bmatrix}
		\begin{bmatrix}B_{1, q_d} & 0 \\ 0 & B_{2, q_d}\end{bmatrix}
		\dots
		\begin{bmatrix}B_{1, q_1} & 0 \\ 0 & B_{2, q_1}\end{bmatrix}
		\begin{bmatrix}G_{1,q_0} & 0 \\ 0 & G_{2,q_0}\end{bmatrix}
		\begin{bmatrix}	\id_m \\ \id_m\end{bmatrix}
		= 0
	\end{align}

	We have shown that for any choice of inputs $v \in \Hilbert^p \tensor \dots \tensor \Hilbert^p \tensor \Hilbert^m$,
	and parameters $\alpha_j$, $\delta_{k>1}$ our matrix equations hold,
	and by extension they hold on all elements of the space $\Hilbert^{\tensor (m + dp)}$,
	i.e. that:

	\begin{align}
		  & \forall d \leq N \; \forall q_d \dots q_1 \nonumber &   \\
		&	\begin{bmatrix}	1 & -1\end{bmatrix}
		\begin{bmatrix}C_1 & 0 \\ 0 & C_2\end{bmatrix}
		\begin{bmatrix}B_{1, q_d} & 0 \\ 0 & B_{2, q_d}\end{bmatrix}
		\dots
		\begin{bmatrix}B_{1, q_1} & 0 \\ 0 & B_{2, q_1}\end{bmatrix}
		\begin{bmatrix}G_{1,q_0} & 0 \\ 0 & G_{2,q_0}\end{bmatrix}
		\begin{bmatrix}	\id_m \\ \id_m\end{bmatrix}
		&= 0 \\
		\implies
		  & \forall d\  \forall q_d \dots q_1 \nonumber         &   \\
		&	\begin{bmatrix}	1 & -1\end{bmatrix}
		\begin{bmatrix}C_1 & 0 \\ 0 & C_2\end{bmatrix}
		\begin{bmatrix}B_{1, q_d} & 0 \\ 0 & B_{2, q_d}\end{bmatrix}
		\dots
		\begin{bmatrix}B_{1, q_1} & 0 \\ 0 & B_{2, q_1}\end{bmatrix}
		\begin{bmatrix}G_{1,q_0} & 0 \\ 0 & G_{2,q_0}\end{bmatrix}
		\begin{bmatrix}	\id_m \\ \id_m\end{bmatrix}
		&= 0
	\end{align}
	And therefore:
	\begin{align}
		         & \interpret{C_1} \interpret{B_{1,q_d}} \dots \interpret{B_{1,q_1}} \interpret{G_{1,q_0}} =
		\interpret{C_2} \interpret{B_{2,q_d}} \dots \interpret{B_{2,q_1}} \interpret{G_{2,q_0}} \nonumber \\
		         & \qquad \forall d \leq N \; \forall q_d \dots q_1                \nonumber                 \\
		\implies &
		\interpret{C_1} \interpret{B_{1,q_d}} \dots \interpret{B_{1,q_1}} \interpret{G_{1,q_0}} =
		\interpret{C_2} \interpret{B_{2,q_d}} \dots \interpret{B_{2,q_1}} \interpret{G_{2,q_0}} \\
		         & \qquad \forall d \; \forall q_d \dots q_1 \nonumber
	\end{align}

	And therefore for any choice of value for $\alpha_j$ and $\delta_{k>1}$:

	\begin{align}
		         & \interpret{\family{\bb{D}_1 }_{\alphadelta|\delta_1=d}} = \interpret{\family{ \bb{D}_2 }_{\alphadelta|\delta_1=d}} & \forall d \leq N \\
		\implies & \interpret{\family{\bb{D}_1 }_{\alphadelta|\delta_1=d}} = \interpret{\family{ \bb{D}_2 }_{\alphadelta|\delta_1=d}} & \forall d
	\end{align}
\end{proof}

\begin{example}[!-box verification]
	Consider the following Universal ZX family of equations, parameterised by a single \mbox{!-box}:
	\begin{align}
		\bbE : =\family{\vc{\InputIfFileExists{./figures/ZX/internal_rule.tikz}{}{Missing file!}}}_{\delta}
	\end{align}
	The join between the !-box and the rest of the diagram on the left is two wires (dimension $2^2$),
	and on the right is one wire (dimension $2^1$).
	These sum to have dimension 6, and we therefore need only to check the !-box instances
	($0, \dots, 6$)
	to be sure that the equation $\bbE$ is sound for any value of $\delta$.
\end{example}

A longer example is provided in \S\ref{secVerificationExample}.

\section{Verifying !-boxes and phase variables together}
\label{sectogether}

Theorem \ref{thmparameter} and theorem \ref{thmbbox2} deal with equations containing
multiple phase variables and nested !-boxes respectively.
This section will put together the necessary results such that we can combine these approaches to deal with
equations containing multiple !-boxes and phase variables, any of which could potentially be nested inside other !-boxes.
For our running ZH example
the previous sections covered the faces of the commutative cube,
and this is the section where we put them all together.

\begin{proposition} \label{propBoth}
	Given a parameterised equation\ $\bbE$\ and ways of finding
	\begin{itemize}
		\item finite verifying sets $A_j$ for phase variables in diagrams without !-boxes (e.g. Theorem~\ref{thmparameter})
		\item finite verifying sets $D_k$ for the !-boxes (e.g. Theorem~\ref{thmbbox2})
	\end{itemize}
	we may verify the entire family $\bbE$ by checking the (finite)
	set given by the cartesian product of all the $A_j$ and $D_k$.
\end{proposition}

\begin{proof}
	We define:
	\begin{align}
		\bar D           & := D_1 \times D_2  \times \dots \times D_m \\
		\bar \delta      & := (\delta_1, \dots, \delta_m)             \\
		A_j^{\bar \delta} & := \text{The verifying set for } \alpha_j
		\text { once } E \text{ has had !-boxes instantiated at } \bar \delta
	\end{align}
	Construct $A_j$ by choosing as many points as there are in $\max_{\bar \delta}\set{ \abs{A_j^{\bar \delta}}}$.
	This is finite because the $D_k$ and the $A_j^{\bar \delta}$ are finite.
	$A_j$ therefore contains enough points to be a verifying set for $A_j^{\bar \delta} \quad \forall \bar \delta \in \bar D$.
	We now show that $A_1 \times \dots \times A_n \times D_1 \times \dots \times D_m$
	describes a verifying set for the parameterised equation $\bbE$:
	\begin{align}
		&\bbE_{\alphadelta | \forall j, k \; \alpha_j \in A_j ,\, \delta_k \in D_k} & \issound\\
		= \quad & \bbE_{\alphadelta | \forall i \; \alpha_j \in A_j ,\, \bar \delta \in \bar D} & \issound
		& \quad \text{(rewrite using }\bar \delta \text{ notation})\\
		\implies & \bbE_{\alphadelta | \forall i \; \alpha_j \in A_j^{\bar \delta}, \, \bar \delta \in \bar D} & \issound & \quad  \text{(construction of $A_j$)}     \\
		\implies & \bbE_{\alphadelta | \bar \delta \in \bar D}                                                & \issound & \quad \text{(theorem \ref{thmparameter})} \\
		\implies & \bbE_{\alphadelta}                                                                         & \issound & \quad \text{(theorem \ref{thmbbox2})}
	\end{align}

\end{proof}

We now proceed towards the main theorem of this chapter,
giving the definition that we will need in the proof.
We shall need algorithms to determine the sets $A_j$ and $D_k$,
as required by Proposition~\ref{propBoth},
and for that we shall need a way of ordering our !-boxes,
and a notion of `the largest diagram'.
Proof of existence, or correctness, of these definitions are given in the proof
of Theorem~\ref{thmfinite}.

\begin{definition}[!-box ordering]
	Given a parameterised equation $\bbE$,
	we construct the ordered list $\delta_{k_1} \succ \delta_{k_2} \succ \dots \succ \delta_{k_n}$
	of !-boxes
	by recursively picking a !-box that is nested in no other !-boxes,
	then removing that !-box from the nesting diagram of $\bbE$. Repeat on the new nesting diagram.
\end{definition}

\begin{definition}[The algorithm !-Removal] \label{defBRemoval}
	Given a verifying set of equations $\set{\bbE_\kappa}_{\kappa \in K}$,
	and a !-box $\delta_k$ nested in no other !-boxes present in the ${\bbE_\kappa}$:

	We define !-Remove($\delta_k$) as the process described in Theorem~\ref{thmbbox2}.
	It acts on the set $\set{\bbE_\kappa}_{\kappa \in K}$
	by acting on each of the $\bbE_\kappa$ in turn,
	finding the value $N_\kappa$, and
	creating the new verification set:

	\begin{align}\set{\bbE_\kappa}_{\kappa \in K'} :=
		\bigcup_{\kappa \in K} \set{ {\bbE_{\kappa|\delta_k = 1}}, \dots, {\bbE_{\kappa|\delta_k = N_\kappa} }}
	\end{align}
\end{definition}

\begin{definition}[Largest diagram] \label{defLargestDiagram}
	Given a verifying set of equations $\set{\bbE_\kappa}_{\kappa \in K}$,
	we use $\bar \bbE$ to denote the !-box free equation that is the result of instantiating
	each !-box $\delta_k$ at its largest required amount $N_k$ given by Definition~\ref{defBRemoval}
\end{definition}

\begin{definition}[The algorithm $\alpha$-Removal]
	Given a verifying set $\set{\bbE_\kappa}_{\kappa \in K}$ we construct the set $A_j$ for the variable $\alpha_j$
	by considering the degree of $\alpha_j$ in $\bar \bbE$,
	and then choosing enough valid values of $\alpha_j$ to reach the amount dictated by Theorem~\ref{thmparameter}.
	We then form:

	\begin{align}
		\set{\bbE_\kappa}_{\kappa \in K'} :=
		\bigcup_{\kappa \in K, a \in A_j} \set{ {\bbE_\kappa} \at{\alpha_j = a}}
	\end{align}

	If we are working with $\ring_\bbC$, $\ZW_\bbC$, or $\ZH_\bbC$ we may  use
	Theorem~\ref{thmQubitSingleVerifyingPhaseEquation} to generate the sets
	$A_1 := \set{a_1}, \dots A_n := \set{a_n}$ instead.
\end{definition}

\begin{restatable}[Finite verification of !-boxes and phase variables]{theorem}{thmFinite}\label{thmfinite}
	Given a parameterised family of equations $\bbE$,
	where the !-boxes are separated and well nested,
	we can construct a finite set of simple equations $\set{\bbE_\kappa}_{\kappa \in K}$
	that verifies $\bbE$.
\end{restatable}

The idea of the proof is to iteratively remove dependencies on !-boxes via Theorem~\ref{thmbbox2},
each time generating a larger set of verifying equations.
Once we have removed all !-box dependence we then use the method of Proposition~\ref{propBoth}
to remove phase variable dependence;
using the `largest' equation in $\set{\bbE_\kappa}$ to determine the sizes of the $A_j$.
Since every step removes either a !-box or a phase variable (and introduces neither !-boxes nor phase variables) this process terminates.

\begin{proof}

	Throughout this proof we iterate on the set $\set{ \bbE_\kappa }$.
	While at all times $\set{ \bbE_\kappa }$ is a finite verifying set for $\bbE$,
	it is only at the end that $\set{ \bbE_\kappa }$ is a finite verifying set of simple equations.
	We first show the existence of an ordered list of the !-boxes present in $E$,
	compatible with the nesting order on both of the nesting diagrams of $E$.

	\textbf{Claims:} \;
	\begin{enumerate}
		\item !-Remove($\delta_k$) removes any dependency on $\delta_k$ in the verifying set $\set{\bbE_\kappa}_{\kappa \in K'}$
		\item $\set{\bbE_\kappa}_{\kappa \in K'}$ verifies $\set{\bbE_\kappa}_{\kappa \in K}$
		\item !-Remove($\delta_k$) does not alter the nesting ordering of any remaining !-boxes and phase variables
		      in the verification pair
		\item The ordered list $\delta_{k_1} \succ \delta_{k_2} \succ \dots$ provides us with a sequence of !-boxes
		      such that we can apply !-Remove($\delta_{k_{n+1}}$) to the output of !-Remove($\delta_{k_n}$).
		\item Applying !-Remove to this ordered sequence results in a finite verifying set that has no dependence on any !-box.
	\end{enumerate}
	Proof of claims: The first, second and fifth claims follow from Theorem~\ref{thmbbox2}.
	The third and fourth claims are clear from the definitions.
	\vskip 2em
	Starting with $\set{\bbE_\kappa}_{\kappa \in K} = \bbE$
	we iteratively apply !-Remove according to the nesting order to obtain $\set{\bbE_\kappa}_{\kappa \in K'}$
	which is a set of equations with phase variables but no !-boxes,
	and also gives us the value of $\bar \bbE$ (see Definition ~\ref{defLargestDiagram}).
	We then iteratively apply $\alpha$-Remove to construct
	a finite set $\set{\bbE_\kappa}_{\kappa \in K''}$ of simple equations that verifies $\bbE$.

\end{proof}

\begin{remark} \label{remGeneratingPatternGraphs}
When we introduced conjecture synthesis we said
in Remark~\ref{remGeneratingPatterns}
that because we can now verify conjectures that contain !-boxes
it is now worth generating conjectures that contain !-boxes.
This can be done in one of two ways:
We can infer !-boxes onto an existing theorem,
by perhaps spotting repeating subdiagrams on both sides of the equation.
Alternatively we could synthesise
separated pattern graphs directly
(i.e. generate diagrams with !-boxes in)
and then store the first $0, \dots, N$
instances of each !-box expansion where $N$ is the join.
This would not give us all the details needed for verifying
whether two such diagrams were equal,
as additional instances will need to be calculated,
but would allow some initial checks to be performed quickly.
\end{remark}

\section{Summary}

This chapter combines two very different verification results;
one for phase variables
(reliant on the structure of the phase algebra),
and the other for !-boxes
(reliant on the new series !-box form).
We showed how to unify these two results
into one that works in all cases where the !-boxes
are separated,
and also how to separate out !-boxes in most, but not all, cases.
When considering qubits we used some Galois theory
to reduce the number of verifying equations needed
from `finite' to `one'.
Further work includes finding similar
results for the calculus ZQ,
and potentially reducing the number
of equations needed to verify !-boxes over qubits.
For conjecture synthesis this chapter
shows how to verify the sorts of hypotheses created
in \S\ref{chapConjectureInference},
but also shows that we can compare
diagrams containing !-boxes directly if we were to
generate them as part of conjecture synthesis.
This chapter also gave the final piece of our conjecture inference;
hypothesising theorems involving !-boxes.
It also brings the results chapters of this thesis to a close.  \thispagestyle{empty}

\chapter{Conclusion} 
\thispagestyle{plain}
\label{chapConclusion}

This thesis demonstrates that
phase algebras conceal a wealth of algebraic complexity
that can and should be used in both quantum computing and conjecture synthesis.
This wealth comes from the interplay between
the phase algebra of the diagrams
and the underlying category of the interpretation.
In our exploration of this theme
we have introduced new graphical calculi,
demonstrated a deep relationship between diagrammatic rules and algebraic geometry,
and introduced phase homomorphisms to link
the algebra of labels with the algebra of interpretations.
While doing this we improved on the design of conjecture synthesis,
exhibited new methods for conjecture verification,
and combined all these, along with some Galois Theory,
to show how a single equation can imply an abundance of generalisations.

The impetus for this work was the task of Conjecture Synthesis
for Quantum Graphical Calculi.
We have extended the core process found in the literature
by including a generalisation step at the point of successful synthesis.
This generalisation step is performed (although not exclusively) via the discovery,
presented in this thesis, of a link
between our rule presentations and algebraic geometry.
Through the related work on conjecture verification, also presented in this thesis,
it was shown that the evidence needed for such a generalisation could be as low
as a single equation;
this is equivalent to inferring
all the points of a surface based on just a single point and some symmetries.

These symmetries of the parameter space,
also discovered in this thesis,
are the phase homomorphism pairs
exhibited in \S\ref{chapPhaseRingCalculi}.
We introduced them,
demonstrated when they preserved soundness and derivations,
and then classified them for ZW, ZH, $\ring$, and certain fragments of ZX.
The author hopes that this examination,
and the contents of \S\ref{chapPhaseRingCalculi}
overall, prompt further research into graphical calculi with a
greater variety of phase algebras.
These parameter spaces do not arise by chance:
They are a direct consequence of our choices of phase algebras,
and have the capacity to reveal profound insights into our graphical calculi.

The conjecture verification results of this thesis
rely on the (relatively simple) properties
of polynomial interpolation
and chains of vector subspaces.
The hard work in these results was in
extracting the necessary information
from the diagrams.
Again we emphasise that this algebraic machinery is present,
but because we are dealing with an unusual presentation of the information,
and in a context where these results are less commonly applied,
these properties are far easier to miss.

Graphical calculus design is therefore clearly important
but we have also shown ways in which it is (usefully) limited.
Any future work to design a quantum graphical calculus
with a phase ring will
necessarily construct something similar to $\ring_\bbC$ (Remark~\ref{remRingKGeneric}).
This is not to say that $\ring_\bbC$ is the only option
(ZH and ZW have good reasons for their choices of generators)
but that $\ring_\bbC$ will act as common ground between these calculi,
and will make the construction of equivalences easier.
The language ZQ,
on the other hand,
will be useful for those creating optimisers,
or reasoning about certain hardware implementations.
In fact the author would argue that the rules of ZQ
are more intuitive than those for Universal ZX, and so there
is a pedagogical argument for explaining quantum circuitry using ZQ.
The unwieldy nature of the Euler Decomposition rule of ZX,
and its comparative elegance in ZQ, was a major influence on the direction of this research.
Both ZQ and $\ring$ were introduced
and then showed to be sound and complete in this thesis.

The central point of this thesis was that phase algebras are important,
and that they enable algebraic machinery that was previously underused.
This is a continuation of the program of Categorical Quantum Mechanics,
a program that aims to make reasoning about these systems easier
through selective abstraction:

\begin{displayquote}[Categorical Quantum Mechanics \cite{CQM}]
The arguments for the benefits of a high-level,
conceptual approach to designing and reasoning about quantum computational
systems are just as compelling as for classical computation.
In particular, we have in mind the hard-learned lessons from Computer Science of
the importance of compositionality, types, abstraction, and the use of tools from
algebra and logic in the design and analysis of complex informatic processes.
\end{displayquote}

As a final demonstration of this point let us examine the generalisation step
of an updated simple $\ring_\bbC$ conjecture synthesis run:
\begin{itemize}
    \item First an unparameterised, sound equation in $\ring_\bbC$ is generated.
    \item Let us assume that this equation shares a skeleton with other equations,
    already shown to be sound.
    We may then use the link between the algebra of equations
    and the geometry of parameter spaces
    to infer a potential generalisation.
    \item We can then use the results of \S\ref{secPhaseVariablesOverQubits}
    to pick a single, unparameterised, verifying equation for this generalisation.
    This equation will already be in our search space,
    we simply move it forward in the queue.
    \item With our generalisation verified
    we can then reduce the search space by a far greater amount
    than if we had not performed the verification.
\end{itemize}
This generalisation process uses Algebraic Geometry, Galois Theory,
Laurent Polynomials,
Polynomial Interpolation, and Category Theory.
Each of their particular uses are novel to this thesis.
All of these uses rely on the interplay between the phase algebra
and the underlying matrices.

The author hopes that this research prompts further work
in the program of Categorical Quantum Mechanics,
in particular further collaboration between
quantum computing and traditionally algebraic fields.
The link to algebraic geometry given here
deserves further exploration,
both for quantum computing and for the inference step in conjecture synthesis.
Even outside the geometric inference framework of this thesis our novel generalisation step
should increase efficacy in future conjecture synthesis projects.
The calculus ZQ would be ideal for formal verification of optimisers such as TriQ. \thispagestyle{empty} 

\chapter*{Acknowledgements}
\thispagestyle{plain}
This thesis was supervised by Bob Coecke,
and the author would like to thank Aleks Kissinger and Miriam Backens
for their academic support.
Tessa Darbyshire, Maaike Zwart and Maryam Ahmed
were invaluable.
Thank you.
The research was funded by the EPSRC. \thispagestyle{empty}

\backmatter \pagestyle{plain}

\cleardoublepage 
\addcontentsline{toc}{chapter}{Bibliography}
\small 
\printbibliography

@phdthesis{Backens16Thesis,
  Author         = {Miriam Backens},
  Title          = {Completeness and the ZX-calculus},
  eprint         = {arXiv:1602.08954},
  year           = 2016,
  school         = {University of Oxford},
}

@InBook{Coecke08,
  Author         = "Coecke, Bob and Duncan, Ross",
  Editor         = "Aceto, Luca and Damg{\aa}rd, Ivan and Goldberg, Leslie
                   Ann and Halld{\'o}rsson, Magn{\'u}s M. and
                   Ing{\'o}lfsd{\'o}ttir, Anna and Walukiewicz, Igor",
  Title          = "Automata, Languages and Programming: 35th
                   International Colloquium, ICALP 2008, Reykjavik,
                   Iceland, July 7-11, 2008, Proceedings, Part II",
  Chapter        = "Interacting Quantum Observables",
  Pages          = "298--310",
  Publisher      = "Springer Berlin Heidelberg",
  Address        = "Berlin, Heidelberg",
  doi            = "10.1007/978-3-540-70583-3_25",
  isbn           = "978-3-540-70583-3",
  year           = 2008
}

@article{Coecke11Entanglement,
   title={Three qubit entanglement within graphical Z/X-calculus},
   volume={52},
   ISSN={2075-2180},
   DOI={10.4204/eptcs.52.3},
   journal={Electronic Proceedings in Theoretical Computer Science},
   publisher={Open Publishing Association},
   author={Coecke, Bob and Edwards, Bill},
   year={2011},
   month={Mar},
   pages={22–33}
}

@InCollection{Duncan09,
  Author         = {Duncan, Ross and Perdrix, Simon},
  Title          = {{G}raph {S}tates and the {N}ecessity of {E}uler
                   {D}ecomposition},
  BookTitle      = {Mathematical Theory and Computational Practice},
  Publisher      = {Springer Berlin Heidelberg},
  Editor         = {Ambos-Spies, Klaus and L{\"o}we, Benedikt and Merkle,
                   Wolfgang},
  Volume         = {5635},
  Series         = {Lecture Notes in Computer Science},
  Pages          = {167-177},
  doi            = {10.1007/978-3-642-03073-4_18},
  isbn           = {978-3-642-03072-7},
  year           = 2009
}

@article{VBE,
  title = {Quantum networks for elementary arithmetic operations},
  author = {Vedral, Vlatko and Barenco, Adriano and Ekert, Artur},
  journal = {Phys. Rev. A},
  volume = {54},
  issue = {1},
  pages = {147--153},
  numpages = {0},
  year = {1996},
  month = {Jul},
  publisher = {American Physical Society},
  doi = {10.1103/PhysRevA.54.147}
}

@INPROCEEDINGS {Shor,
author = {P. Shor},
booktitle = {2013 IEEE 54th Annual Symposium on Foundations of Computer Science},
title = {Algorithms for quantum computation: discrete logarithms and factoring},
year = {1994},
volume = {},
issn = {},
pages = {124-134},
doi = {10.1109/SFCS.1994.365700},
publisher = {IEEE Computer Society},
address = {Los Alamitos, CA, USA},
month = {nov}
}

@misc{Gottesman98,
    title={The Heisenberg Representation of Quantum Computers},
    author={Daniel Gottesman},
    year={1998},
    primaryClass={quant-ph},
    eprint={quant-ph/9807006},
    archivePrefix={arXiv},
}

@Article{Perdomo-Ortiz2012,
author={Perdomo-Ortiz, Alejandro
and Dickson, Neil
and Drew-Brook, Marshall
and Rose, Geordie
and Aspuru-Guzik, Al{\'a}n},
title={Finding low-energy conformations of lattice protein models by quantum annealing},
journal={Scientific Reports},
year={2012},
month={Aug},
day={13},
volume={2},
number={1},
pages={571},
issn={2045-2322},
doi={10.1038/srep00571}
}

@Book{NielsonChuang,
  Author         = {Nielsen, Michael A and Chuang, Isaac L},
  Title          = {Quantum computation and quantum information},
  Publisher      = {Cambridge university press},
  doi            = {10.1017/CBO9780511976667},
  year           = 2010
}

@Book{Hume,
  author = {Hume, David},
  title = {A treatise of human nature},
  Publisher      = {Dover Publications},
  isbn            = {0486432505},
  year           = {2003}
}

@article{Quantomatic,
   title={Quantomatic: A Proof Assistant for Diagrammatic Reasoning},
   ISBN={9783319214016},
   ISSN={1611-3349},
   DOI={10.1007/978-3-319-21401-6_22},
   journal={Lecture Notes in Computer Science},
   publisher={Springer International Publishing},
   author={Kissinger, Aleks and Zamdzhiev, Vladimir},
   year={2015},
   pages={326–336}
}

@article{Spekkens,
   title={A Complete Graphical Calculus for Spekkens’ Toy Bit Theory},
   volume={46},
   ISSN={1572-9516},
   DOI={10.1007/s10701-015-9957-7},
   number={1},
   journal={Foundations of Physics},
   publisher={Springer Science and Business Media LLC},
   author={Backens, Miriam and Duman, Ali Nabi},
   year={2015},
   month={Oct},
   pages={70–103}
}

@article{Ranchin14,
   title={Depicting qudit quantum mechanics and mutually unbiased qudit theories},
   volume={172},
   ISSN={2075-2180},
   DOI={10.4204/eptcs.172.6},
   journal={Electronic Proceedings in Theoretical Computer Science},
   publisher={Open Publishing Association},
   author={Ranchin, André},
   year={2014},
   month={Dec},
   pages={68–91}
}

@book{PQP,
  title={Picturing Quantum Processes},
  author={Coecke, B. and Kissinger, A.},
  isbn={9781107104228},
  lccn={2016035537},
  year={2017},
  publisher={Cambridge University Press}
}

@phdthesis{AmarThesis,
   author = {{Hadzihasanovic}, A.},
    title = "{The algebra of entanglement and the geometry of composition}",
  journal = {ArXiv e-prints},
  school  = {University of Oxford},
archivePrefix = "arXiv",
   eprint = {1709.08086},
 primaryClass = "math.CT",
     year = 2017,
    month = sep
}

@misc{UniversalComplete,
    title={A universal completion of the ZX-calculus},
    author={Kang Feng Ng and Quanlong Wang},
    year={2017},
    eprint={1706.09877},
    archivePrefix={arXiv},
    primaryClass={quant-ph}
}

@ARTICLE{BeyondCliffordT,
   title={Diagrammatic Reasoning beyond Clifford+T Quantum Mechanics},
   ISBN={9781450355834},
   DOI={10.1145/3209108.3209139},
   journal={Proceedings of the 33rd Annual ACM/IEEE Symposium on Logic in Computer Science - LICS  ’18},
   publisher={ACM Press},
   author={Jeandel, Emmanuel and Perdrix, Simon and Vilmart, Renaud},
   year={2018}
}

@article{VilmartZX,
   title={A Near-Minimal Axiomatisation of ZX-Calculus for Pure Qubit Quantum Mechanics},
   ISBN={9781728136080},
   DOI={10.1109/lics.2019.8785765},
   journal={2019 34th Annual ACM/IEEE Symposium on Logic in Computer Science (LICS)},
   publisher={IEEE},
   author={Vilmart, Renaud},
   year={2019},
   month={Jun}
}

@phdthesis{MerryThesis,
    title={Reasoning with !-Graphs},
    author={Alexander Merry},
    school ={University of Oxford},
    year={2014},
    eprint={1403.7828},
    archivePrefix={arXiv},
    primaryClass={cs.LO}
}

@article{Multivariate,
  author = "Mariano Gasca, Thomas Sauer",
  journal = "Journal of Computational and Applied Mathematics",
  year = 2000,
  volume = 122,
  pages= {23-35},
  title = "On the history of multivariate polynomial interpolation.",
  doi = {https://doi.org/10.1016/S0377-0427(00)00353-8}
}

@article{ZW,
   title={A Diagrammatic Axiomatisation for Qubit Entanglement},
   ISBN={9781479988754},
   DOI={10.1109/lics.2015.59},
   journal={2015 30th Annual ACM/IEEE Symposium on Logic in Computer Science},
   publisher={IEEE},
   author={Hadzihasanovic, Amar},
   year={2015},
   month={Jul}
}

@ARTICLE{ZH,
   title={ZH: A Complete Graphical Calculus for Quantum Computations Involving Classical Non-linearity},
   volume={287},
   ISSN={2075-2180},
   DOI={10.4204/eptcs.287.2},
   journal={Electronic Proceedings in Theoretical Computer Science},
   publisher={Open Publishing Association},
   author={Backens, Miriam and Kissinger, Aleks},
   year={2019},
   month={Jan},
   pages={23–42}
}

@inproceedings{Kissinger2012PatternGR,
   title={Pattern graph rewrite systems},
   volume={143},
   ISSN={2075-2180},
   DOI={10.4204/eptcs.143.5},
   journal={Electronic Proceedings in Theoretical Computer Science},
   publisher={Open Publishing Association},
   author={Kissinger, Aleks and Merry, Alex and Soloviev, Matvey},
   year={2014},
   month={Mar},
   pages={54–66}
}

@book{Terms,
 author = {Baader, Franz and Nipkow, Tobias},
 title = {Term Rewriting and All That},
 year = {1998},
 isbn = {0-521-45520-0},
 publisher = {Cambridge University Press},
 address = {New York, NY, USA},
}

@article{Maclane1965,
author = "MacLane, Saunders",
journal = "Bulletin of the American Mathematical Society",
month = "01",
number = "1",
pages = "40--106",
publisher = "American Mathematical Society",
title = "Categorical algebra",
doi = "10.1090/S0002-9904-1965-11234-4",
volume = "71",
year = "1965"
}

@book{Maclane2013,
  title={Categories for the working mathematician},
  author={Mac Lane, Saunders},
  volume={5},
  year={2013},
  publisher={Springer Science \& Business Media},
  doi={10.1007/978-1-4757-4721-8}
}

@phdthesis{AleksThesis,
  author    = {Aleks Kissinger},
  title     = {Synthesising Graphical Theories},
  journal   = {CoRR},
  volume    = {abs/1202.6079},
  school    = {University of Oxford},
  year      = {2012},
  archivePrefix = {arXiv},
  eprint    = {1202.6079}
}

@article{Selinger07,
  title = "Dagger Compact Closed Categories and Completely Positive Maps: (Extended Abstract)",
  journal = "Electronic Notes in Theoretical Computer Science",
  volume = "170",
  pages = "139 - 163",
  year = "2007",
  note = "Proceedings of the 3rd International Workshop on Quantum Programming Languages (QPL 2005)",
  issn = "1571-0661",
  doi = "https://doi.org/10.1016/j.entcs.2006.12.018",
  author = "Peter Selinger"
}

@article{Shoemake,
 author = {Shoemake, Ken},
 title = {Animating Rotation with Quaternion Curves},
 journal = {SIGGRAPH Comput. Graph.},
 issue_date = {Jul. 1985},
 volume = {19},
 number = {3},
 month = jul,
 year = {1985},
 issn = {0097-8930},
 pages = {245--254},
 numpages = {10},
 doi = {10.1145/325165.325242},
 acmid = {325242},
 publisher = {ACM},
 address = {New York, NY, USA},
}

@book{InvitationAlgGeom,
  title={An Invitation to Algebraic Geometry},
  author={Smith, Karen and Kahanp{\"a}{\"a}, Lauri and Kek{\"a}l{\"a}inen, Pekka and Traves, William},
  year={2004},
  publisher={Springer Science \& Business Media},
  doi={10.1007/978-1-4757-4497-2},
  isbn={978-1-4419-3195-5}
}

@article {Magma,
    AUTHOR = {Bosma, Wieb and Cannon, John and Playoust, Catherine},
     TITLE = {The {M}agma algebra system. {I}. {T}he user language},
      NOTE = {Computational algebra and number theory (London, 1993)},
   JOURNAL = {J. Symbolic Comput.},
  FJOURNAL = {Journal of Symbolic Computation},
    VOLUME = {24},
      YEAR = {1997},
    NUMBER = {3-4},
     PAGES = {235--265},
      ISSN = {0747-7171},
   MRCLASS = {68Q40},
  MRNUMBER = {MR1484478},
       DOI = {10.1006/jsco.1996.0125}
}

@ARTICLE{ZWArithmetic,
   title={The GHZ/W-calculus contains rational arithmetic},
   volume={52},
   ISSN={2075-2180},
   DOI={10.4204/eptcs.52.4},
   journal={Electronic Proceedings in Theoretical Computer Science},
   publisher={Open Publishing Association},
   author={Coecke, Bob and Kissinger, Aleks and Merry, Alex and Roy, Shibdas},
   year={2011},
   month={Mar},
   pages={34–48}
}

@misc{AZX,
    title={An algebraic axiomatisation of ZX-calculus},
    author={Quanlong Wang},
    year={2019},
    eprint={1911.06752},
    archivePrefix={arXiv},
    primaryClass={quant-ph}
}

@book{Hazewinkel2004algebras,
  title={Algebras, rings and modules},
  author={Hazewinkel, Michiel and Gubareni, Nadiya and Kirichenko, Vladimir V},
  volume={1},
  year={2004},
  publisher={Springer Science \& Business Media},
  doi={10.1007/1-4020-2691-9}
}

@article{Duncan2020,
   title={Graph-theoretic Simplification of Quantum Circuits with the ZX-calculus},
   volume={4},
   ISSN={2521-327X},
   DOI={10.22331/q-2020-06-04-279},
   journal={Quantum},
   publisher={Verein zur Forderung des Open Access Publizierens in den Quantenwissenschaften},
   author={Duncan, Ross and Kissinger, Aleks and Perdrix, Simon and van de Wetering, John},
   year={2020},
   month={Jun},
   pages={279}
}

@phdthesis{QuickThesis,
author={David Quick},
title={!-Logic : first order reasoning for families of non-commutative string diagrams},
school={University of Oxford},
year={2015},
url={https://ora.ox.ac.uk/objects/uuid:baf2d50d-8c5f-419d-9b3d-f2f700f8acbd}
}

@book{GaloisTheory,
Author = {Stewart, Ian},
ISBN = {9781138401709},
Publisher = {Chapman and Hall/CRC},
Title = {Galois Theory.},
Volume = {Fourth edition},
Year = {2015},
}

@Book{AlgebraArtin,
  Author         = {Artin, Michael},
  Title          = {Algebra: Pearson New International Edition},
  Publisher      = {Pearson Education Limited},
  isbn           = {1292027665},
  year           = 2013
}

@article{Penrose,
  title={Applications of negative dimensional tensors},
  author={Penrose, Roger},
  journal={Combinatorial mathematics and its applications},
  volume={1},
  pages={221--244},
  year={1971}
}

@InProceedings{SZX,
  author ={Titouan Carette and Dominic Horsman and Simon Perdrix},
  title ={{SZX-Calculus: Scalable Graphical Quantum Reasoning}},
  booktitle ={44th International Symposium on Mathematical Foundations of Computer Science (MFCS 2019)},
  pages ={55:1--55:15},
  series ={Leibniz International Proceedings in Informatics (LIPIcs)},
  ISBN ={978-3-95977-117-7},
  ISSN ={1868-8969},
  year ={2019},
  volume ={138},
  editor ={Peter Rossmanith and Pinar Heggernes and Joost-Pieter Katoen},
  publisher ={Schloss Dagstuhl--Leibniz-Zentrum fuer Informatik},
  address ={Dagstuhl, Germany},
  URN ={urn:nbn:de:0030-drops-109999},
  doi ={10.4230/LIPIcs.MFCS.2019.55}
}

@Article{MultivariateSurvey,
author={Gasca, Mariano
and Sauer, Thomas},
title={Polynomial interpolation in several variables},
journal={Advances in Computational Mathematics},
year={2000},
month={Mar},
day={01},
volume={12},
number={4},
pages={377},
issn={1572-9044},
doi={10.1023/A:1018981505752}
}

@article{CQM,
   title={Categorical Quantum Mechanics},
   ISBN={9780444528698},
   DOI={10.1016/b978-0-444-52869-8.50010-4},
   journal={Handbook of Quantum Logic and Quantum Structures},
   publisher={Elsevier},
   author={Abramsky, Samson and Coecke, Bob},
   year={2009},
   pages={261–323}
}

@article{JPVNormalForm,
   title={A Generic Normal Form for ZX-Diagrams and Application to the Rational Angle Completeness},
   ISBN={9781728136080},
   DOI={10.1109/lics.2019.8785754},
   journal={2019 34th Annual ACM/IEEE Symposium on Logic in Computer Science (LICS)},
   publisher={IEEE},
   author={Jeandel, Emmanuel and Perdrix, Simon and Vilmart, Renaud},
   year={2019},
   month={Jun}
}

@ARTICLE{Duncan13,
   title={Pivoting makes the ZX-calculus complete for real stabilizers},
   volume={171},
   ISSN={2075-2180},
   DOI={10.4204/eptcs.171.5},
   journal={Electronic Proceedings in Theoretical Computer Science},
   publisher={Open Publishing Association},
   author={Duncan, Ross and Perdrix, Simon},
   year={2014},
   month={Dec},
   pages={50–62}
}

@article{SdW14,
   title={The ZX-calculus is incomplete for quantum mechanics},
   volume={172},
   ISSN={2075-2180},
   DOI={10.4204/eptcs.172.20},
   journal={Electronic Proceedings in Theoretical Computer Science},
   publisher={Open Publishing Association},
   author={Schröder de Witt, Christian and Zamdzhiev, Vladimir},
   year={2014},
   month={Dec},
   pages={285–292}
}

@Article{QuantumSupremacy,
author={Arute, Frank
and Arya, Kunal
and Babbush, Ryan
and Bacon, Dave
and Bardin, Joseph C.
and Barends, Rami
and Biswas, Rupak
and Boixo, Sergio
and Brandao, Fernando G. S. L.
and Buell, David A.
and Burkett, Brian
and Chen, Yu
and Chen, Zijun
and Chiaro, Ben
and Collins, Roberto
and Courtney, William
and Dunsworth, Andrew
and Farhi, Edward
and Foxen, Brooks
and Fowler, Austin
and Gidney, Craig
and Giustina, Marissa
and Graff, Rob
and Guerin, Keith
and Habegger, Steve
and Harrigan, Matthew P.
and Hartmann, Michael J.
and Ho, Alan
and Hoffmann, Markus
and Huang, Trent
and Humble, Travis S.
and Isakov, Sergei V.
and Jeffrey, Evan
and Jiang, Zhang
and Kafri, Dvir
and Kechedzhi, Kostyantyn
and Kelly, Julian
and Klimov, Paul V.
and Knysh, Sergey
and Korotkov, Alexander
and Kostritsa, Fedor
and Landhuis, David
and Lindmark, Mike
and Lucero, Erik
and Lyakh, Dmitry
and Mandr{\`a}, Salvatore
and McClean, Jarrod R.
and McEwen, Matthew
and Megrant, Anthony
and Mi, Xiao
and Michielsen, Kristel
and Mohseni, Masoud
and Mutus, Josh
and Naaman, Ofer
and Neeley, Matthew
and Neill, Charles
and Niu, Murphy Yuezhen
and Ostby, Eric
and Petukhov, Andre
and Platt, John C.
and Quintana, Chris
and Rieffel, Eleanor G.
and Roushan, Pedram
and Rubin, Nicholas C.
and Sank, Daniel
and Satzinger, Kevin J.
and Smelyanskiy, Vadim
and Sung, Kevin J.
and Trevithick, Matthew D.
and Vainsencher, Amit
and Villalonga, Benjamin
and White, Theodore
and Yao, Z. Jamie
and Yeh, Ping
and Zalcman, Adam
and Neven, Hartmut
and Martinis, John M.},
title={Quantum supremacy using a programmable superconducting processor},
journal={Nature},
year={2019},
month={Oct},
day={01},
volume={574},
number={7779},
pages={505-510},
issn={1476-4687},
doi={10.1038/s41586-019-1666-5}
}

@misc{IBMSupremacy,
  title = {On “Quantum Supremacy”},
  howpublished = {\url{https://www.ibm.com/blogs/research/2019/10/on-quantum-supremacy/}},
  note = {Accessed: 2020-01-21},
  year={2019},
  shortauthor={IBM},
  author={International Business Machines Corporation}
}

@article{NISQ,
  doi = {10.22331/q-2018-08-06-79},
  title = {Quantum {C}omputing in the {NISQ} era and beyond},
  author = {Preskill, John},
  journal = {{Quantum}},
  issn = {2521-327X},
  publisher = {{Verein zur F{\"{o}}rderung des Open Access Publizierens in den Quantenwissenschaften}},
  volume = {2},
  pages = {79},
  month = aug,
  year = {2018}
}

@article{MagicStateFactory,
   title={Efficient magic state factories with a catalyzed $|CCZ⟩$ to $2|T⟩$ transformation},
   volume={3},
   ISSN={2521-327X},
   DOI={10.22331/q-2019-04-30-135},
   journal={Quantum},
   publisher={Verein zur Forderung des Open Access Publizierens in den Quantenwissenschaften},
   author={Gidney, Craig and Fowler, Austin G.},
   year={2019},
   month={Apr},
   pages={135}
}

@misc{ZXAccessibility,
  title = {ZX-calculus and accessibility},
  howpublished = {\url{http://zxcalculus.com/accessibility.html}},
  note = {Accessed: 2020-01-23},
  year={2019},
  author={Miller-Bakewell, Hector}
}

@unpublished{ChristianMSC,
    author = {Schr{\"o}der de Witt, Christian },
    title = {The ZX calculus is incomplete for non-stabilizer quantum mechanics},
    year = {2013},
    url = {http://www.cs.ox.ac.uk/people/bob.coecke/Christian_thesis.pdf},
    note = {University of Oxford Master thesis}
}

@article{BackensStabilizerComplete,
   title={The ZX-calculus is complete for stabilizer quantum mechanics},
   volume={16},
   ISSN={1367-2630},
   DOI={10.1088/1367-2630/16/9/093021},
   number={9},
   journal={New Journal of Physics},
   publisher={IOP Publishing},
   author={Backens, Miriam},
   year={2014},
   month={Sep},
   pages={093021}
}

@article{BackensSingleQubit,
   title={The ZX-calculus is complete for the single-qubit Clifford+T group},
   volume={172},
   ISSN={2075-2180},
   DOI={10.4204/eptcs.172.21},
   journal={Electronic Proceedings in Theoretical Computer Science},
   publisher={Open Publishing Association},
   author={Backens, Miriam},
   year={2014},
   month={Dec},
   pages={293–303}
}

@inproceedings{Backens15,
    author = {Miriam Backens},
    title = {Making the stabilizer ZX-calculus complete for scalars},
    year = {2015},
    booktitle = {Proceedings of the 12th International Workshop on Quantum Physics and Logic (QPL 2015)},
    editor = {Chris Heunen and Peter Selinger and Jamie Vicary},
    series = {Electronic Proceedings in Theoretical Computer Science},
    volume = {195},
    pages = {17--32},
    doi = {10.4204/EPTCS.195.2},
    link = {https://arxiv.org/abs/1507.03854}
}

@inproceedings{BackensSimplified,
    author = {Backens, Miriam and Perdrix, Simon and Wang, Quanlong},
    title = {A Simplified Stabilizer ZX-calculus},
    year = {2017},
    booktitle = {Proceedings 13th International Conference on Quantum Physics and Logic, Glasgow, Scotland, 6-10 June 2016},
    editor = {Duncan, Ross and Heunen, Chris},
    series = {Electronic Proceedings in Theoretical Computer Science},
    volume = {236},
    pages = {1-20},
    publisher = {Open Publishing Association},
    doi = {10.4204/EPTCS.236.1},
    link = {https://arxiv.org/abs/1602.04744}
}

@inproceedings{cyclo,
    author = {Emmanuel Jeandel and Simon Perdrix and Renaud Vilmart and Quanlong Wang},
    title = {{ZX-Calculus: Cyclotomic Supplementarity and Incompleteness for Clifford+T Quantum Mechanics}},
    year = {2017},
    booktitle = {42nd International Symposium on Mathematical Foundations of Computer Science (MFCS 2017)},
    editor = {Kim G. Larsen and Hans L. Bodlaender and Jean-Francois Raskin},
    series = {Leibniz International Proceedings in Informatics (LIPIcs)},
    volume = {83},
    pages = {11:1--11:13},
    publisher = {Schloss Dagstuhl--Leibniz-Zentrum fuer Informatik},
    doi = {10.4230/LIPIcs.MFCS.2017.11},
    link = {https://arxiv.org/abs/1702.01945},
    address = {Dagstuhl, Germany},
    isbn = {978-3-95977-046-0},
    issn = {1868-8969}
}

@inproceedings{SimonCompleteness,
    author = {Jeandel, Emmanuel and Perdrix, Simon and Vilmart, Renaud},
    title = {{A Complete Axiomatisation of the ZX-Calculus for Clifford+T Quantum Mechanics}},
    year = {2018},
    booktitle = {Proceedings of the 33rd Annual ACM/IEEE Symposium on Logic in Computer Science},
    series = {LICS '18},
    pages = {559--568},
    publisher = {ACM},
    doi = {10.1145/3209108.3209131},
    link = {https://arxiv.org/abs/1705.11151},
    acmid = {3209131},
    address = {New York, NY, USA},
    isbn = {978-1-4503-5583-4},
    location = {Oxford, United Kingdom},
    numpages = {10}
}

@article{CQCZX,
   title={Techniques to Reduce $\pi$/4-Parity-Phase Circuits, Motivated by the ZX Calculus},
   volume={318},
   ISSN={2075-2180},
   DOI={10.4204/eptcs.318.9},
   journal={Electronic Proceedings in Theoretical Computer Science},
   publisher={Open Publishing Association},
   author={de Beaudrap, Niel and Bian, Xiaoning and Wang, Quanlong},
   year={2020},
   month={May},
   pages={131–149}
}

@misc{CQCSite,
    title={Cambridge Quantum Computing publications webpage},
    year={2019},
    note = {Accessed: 2020-01-23},
    shortauthor={CQC},
    howpublished = {\url{https://cambridgequantum.com/publications/}},
}

@article{CQCShallow,
   title={Phase Gadget Synthesis for Shallow Circuits},
   volume={318},
   ISSN={2075-2180},
   DOI={10.4204/eptcs.318.13},
   journal={Electronic Proceedings in Theoretical Computer Science},
   publisher={Open Publishing Association},
   author={Cowtan, Alexander and Dilkes, Silas and Duncan, Ross and Simmons, Will and Sivarajah, Seyon},
   year={2020},
   month={May},
   pages={213–228}
}

@article{Horsman2017Surgery,
    author = {de Beaudrap, Niel and Horsman, Dominic},
    title = {The ZX calculus is a language for surface code lattice surgery},
    year = {2020},
    journal = {Quantum},
    volume = {4},
    doi = {10.22331/q-2020-01-09-218}
}

@inproceedings{ISACOSY,
  title={Isa{C}o{S}y: Synthesis of inductive theorems},
  author={Johansson, Moa and Dixon, Lucas and Bundy, Alan},
  booktitle={Workshop on Automated Mathematical Theory Exploration (Automatheo)},
  url={http://hdl.handle.net/1842/4751},
  year={2009}
}

@Article{ISACOSY2,
author={Johansson, Moa
and Dixon, Lucas
and Bundy, Alan},
title={Conjecture Synthesis for Inductive Theories},
journal={Journal of Automated Reasoning},
year={2011},
month={Oct},
day={01},
volume={47},
number={3},
pages={251-289},
issn={1573-0670},
doi={10.1007/s10817-010-9193-y}
}

@inproceedings{AleksCoSy,
  author    = {Aleks Kissinger},
  title     = {Synthesising Graphical Theories},
  booktitle = {ATx'12/WInG'12: Joint Proceedings of the Workshops on Automated Theory eXploration and on Invariant Generation},
  editor    = {Jacques Fleuriot and Peter H\textbackslash{}"ofner and Annabelle McIver and Alan Smaill},
  series    = {EPiC Series in Computing},
  volume    = {17},
  pages     = {26--35},
  year      = {2013},
  publisher = {EasyChair},
  bibsource = {EasyChair, https://easychair.org},
  issn      = {2398-7340},
  doi       = {10.29007/5dkd}
}

@article{PyZX,
   title={PyZX: Large Scale Automated Diagrammatic Reasoning},
   volume={318},
   ISSN={2075-2180},
   DOI={10.4204/eptcs.318.14},
   journal={Electronic Proceedings in Theoretical Computer Science},
   publisher={Open Publishing Association},
   author={Kissinger, Aleks and van de Wetering, John},
   year={2020},
   month={May},
   pages={229–241}
}

@article{Kissinger2019Reducing,
   title={Reducing the number of non-Clifford gates in quantum circuits},
   volume={102},
   ISSN={2469-9934},
   DOI={10.1103/physreva.102.022406},
   number={2},
   journal={Physical Review A},
   publisher={American Physical Society (APS)},
   author={Kissinger, Aleks and van de Wetering, John},
   year={2020},
   month={Aug}
}

@inproceedings{TRIQ,
 author = {Murali, Prakash and Linke, Norbert Matthias and Martonosi, Margaret and Abhari, Ali Javadi and Nguyen, Nhung Hong and Alderete, Cinthia Huerta},
 title = {Full-Stack, Real-System Quantum Computer Studies: Architectural Comparisons and Design Insights},
 year = {2019},
 isbn = {9781450366694},
 publisher = {Association for Computing Machinery},
 address = {New York, NY, USA},
 doi = {10.1145/3307650.3322273},
 booktitle = {Proceedings of the 46th International Symposium on Computer Architecture},
 pages = {527–540},
 numpages = {14},
 location = {Phoenix, Arizona},
 series = {ISCA ’19}
}

@Article{Ammon1993,
author="Ammon, Kurt",
title="A learning procedure for mathematics",
journal="Annals of Mathematics and Artificial Intelligence",
year="1993",
month="Sep",
day="01",
volume="8",
number="3",
pages="407--423",
issn="1573-7470",
doi="10.1007/BF01530800"
}

@article{ColbournRead, 
  author = {Colbourn, Charles J and Read, Ronald C}, 
  title = {{Orderly algorithms for generating restricted classes of graphs}}, 
  issn = {1097-0118}, 
  doi = {10.1002/jgt.3190030210}, 
  pages = {187--195}, 
  number = {2}, 
  volume = {3}, 
  journal = {Journal of Graph Theory}, 
  year = {1979}
}

@Inbook{SelingerSurvey,
author="Selinger, P.",
editor="Coecke, Bob",
title="A Survey of Graphical Languages for Monoidal Categories",
bookTitle="New Structures for Physics",
year="2011",
publisher="Springer Berlin Heidelberg",
address="Berlin, Heidelberg",
pages="289--355",
isbn="978-3-642-12821-9",
doi="10.1007/978-3-642-12821-9_4"
}

@article{JoyalStreet,
  title={The geometry of tensor calculus, I},
  author={Joyal, Andr{\'e} and Street, Ross},
  journal={Advances in mathematics},
  volume={88},
  number={1},
  pages={55--112},
  year={1991},
  publisher={Academic Press},
  doi = "10.1016/0001-8708(91)90003-P"
}

@article{DixonKissingerOpenGraphs,
   title={Open-graphs and monoidal theories},
   volume={23},
   ISSN={1469-8072},
   DOI={10.1017/s0960129512000138},
   number={02},
   journal={Mathematical Structures in Computer Science},
   publisher={Cambridge University Press (CUP)},
   author={DIXON, LUCAS and KISSINGER, ALEKS},
   year={2013},
   month={Feb},
   pages={308–359}
}

@inproceedings{DPO,
  title={Graph-grammars: An algebraic approach},
  author={Ehrig, Hartmut and Pfender, Michael and Schneider, Hans J{\"u}rgen},
  booktitle={14th Annual Symposium on Switching and Automata Theory (swat 1973)},
  pages={167--180},
  year={1973},
  organization={IEEE},
  doi={10.1109/SWAT.1973.11}
}

@article{millerbakewell2019finite,
   title={Finite Verification of Infinite Families of Diagram Equations},
   volume={318},
   ISSN={2075-2180},
   DOI={10.4204/eptcs.318.3},
   journal={Electronic Proceedings in Theoretical Computer Science},
   publisher={Open Publishing Association},
   author={Miller-Bakewell, Hector},
   year={2020},
   month={May},
   pages={27–52}
}

@misc{Backens2020circuit,
    title={There and back again: A circuit extraction tale},
    author={Miriam Backens and Hector Miller-Bakewell and Giovanni de Felice and Leo Lobski and John van de Wetering},
    year={2020},
    eprint={2003.01664},
    archivePrefix={arXiv},
    primaryClass={quant-ph}
}
\clearpage \thispagestyle{empty}
\cleardoublepage
\appendix
\renewcommand\thesection{\Alph{section}}
\chapter{Appendices}

\section{Extending the calculus ZH}
\label{chapZHR}
\thispagestyle{plain}
\noindent\emph{In this appendix:}
\begin{itemize}
	\item[\chapterbullet] We extend the calculus $\ZH_\bbC$ to the calculus $\ZH_R$,
	      where $R$ is any commutative ring with a half.
	\item[\chapterbullet] The proof of this result is very similar to the original proof
	      for $\ZH_\bbC$ in Ref.~\cite{ZH}, and we imitate, quote, and make use of that paper where possible.
\end{itemize}

\newcommand\greyphase[1]{\smallnode{smallGrey}{#1}}
\newcommand\whitemult{\smallbinary[smallZ]{}}
\newcommand\whiteunit{\smallstate{smallZ}{}}
\newcommand\whitecounit{\smalleffect{smallZ}{}}
\newcommand\whitedot{\scalar{smallZ}{}}
\newcommand\intf[1]{\interpret{#1}}

For the definition of the calculus $\ZHR$ see Section~\ref{secZHR}.
The majority of the text below is from Ref.~\cite{ZH},
it has simply been altered to reflect the more general case of $\ZHR$ over $\ZH_\bbC$.
It can be seen from the proofs in Ref.~\cite{ZH}
that the only property used of the ring $\bbC$ is the existence of the element $\half{}$.

We will show that $\ZHR$ is complete by demonstrating the existence of a unique normal form for $\ZHR$ diagrams.
It is first worth noting that, because we can turn inputs into outputs arbitrarily (since all the generators are spiders),
it suffices to consider diagrams which have only outputs.

For states $\psi,\phi$,
let $\psi * \phi$ be the \textit{Schur product} of $\psi$ and $\phi$
obtained by plugging the $i$-th output of $\psi$ and $\phi$ into \whitemult, for each $i$:
\begin{align}
	\vc{\InputIfFileExists{./figures/ZH/schur.tikz}{}{Missing file!}}
\end{align}
It follows from (ZS1) that $*$ is associative and commutative,
so we can write $k$-fold Schur products $\psi_1 * \psi_2 * \ldots * \psi_k$ without ambiguity.
For any finite set $J$ with $|J| = k$, let $\prod_{j\in J} \psi_j$ be the $k$-fold Schur product.

Let $\mathbb B^n$ be the set of all $n$-bitstrings.
For any $\vec{b} := b_1\ldots b_n \in \mathbb B^n$,
define the \textit{indexing map} $\iota_{\vec{b}}$ as follows:
\begin{equation}\label{eq:iota-dfn}
	\iota_{\vec{b}} \; = \;
	\vc{\InputIfFileExists{./figures/ZH/indexing_box.tikz}{}{Missing file!}} \; = \; \left(\greyphase{\neg}\right)^{1 - b_1} \ldots \left(\greyphase{\neg}\right)^{1 - b_n}.
\end{equation}
Then normal forms are given by the following $2^n$-fold Schur products:
\begin{equation}\label{eq:nf-formula}
	\prod_{\vec{b} \in \mathbb B^n} \big( \iota_{\vec{b}} \circ H_n(a_{\vec{b}}) \big)
\end{equation}
where $H_n(a_{\vec{b}})$ is the arity-$n$ H-box (considered as a state) labelled by an arbitrary elements $a_{\vec{b}}$ of $R$.

A normal form diagram can be seen as a collection of $n$ spiders,
fanning out to $2^n$ H-boxes,
each with a distinct configuration of NOT's corresponding to the $2^n$ bitstrings in $\mathbb B^n$.
Diagrammatically, normal forms are:

\begin{align}
	\vc{\InputIfFileExists{./figures/ZH/nf_bbox.tikz}{}{Missing file!}} \quad := \quad
	\vc{\InputIfFileExists{./figures/ZH/nf_picture.tikz}{}{Missing file!}}
\end{align}

\begin{theorem}\label{thm:nf-unique} \cite[Theorem~3.1]{ZH}
	Normal forms are unique. In particular:
	\begin{equation}\label{eq:nf-concrete}
		\intf{ \, \prod_{\vec{b} \in \mathbb B^n} \big( \iota_{\vec{b}} \circ H_n(a_{\vec{b}}) \big) } =
		\sum_{\vec{b} \in \mathbb B^n} a_{\vec{b}} \ket{\vec{b}}.
	\end{equation}
\end{theorem}

\begin{proof}
	The map $\iota_{\vec b}$ is a permutation that acts on computational basis elements as
	$\ket{\vec c} \mapsto \ket{\vec c \oplus \vec b \oplus \vec 1}$.
	In particular, it sends the basis element $\ket{\vec 1}$ to $\ket{\vec b}$.
	Hence $\iota_{\vec b} \circ H_n(a_{\vec b})$ is a vector with
	$a_{\vec b}$ in the $\vec b$-th component and $1$ everywhere else.
	The Schur product of all such vectors gives the RHS of \eqref{eq:nf-concrete}.
\end{proof}

Since equation~\eqref{eq:nf-concrete} gives us a means of constructing any vector in
$R^{2^n}$,
Theorem~\ref{thm:nf-unique} can also be seen as a proof of universality of $\ZHR$.

\begin{lemma}[{\cite[Lemma~3.2]{ZH}}]\label{lem:X-copy}
	The NOT operator copies through white spiders:
	\begin{align}
		\vc{\InputIfFileExists{./figures/ZH/X_copy.tikz}{}{Missing file!}}
	\end{align}
\end{lemma}

\begin{lemma}[{\cite[Lemma~3.3]{ZH}}]\label{lem:iota-copy}
	The $\iota_{\vec{b}}$ operator copies through white spiders, i.e.\ for any $\vec{b}\in\mathbb B^n$:
	\begin{align}
		\vc{\InputIfFileExists{./figures/ZH/iota_copy.tikz}{}{Missing file!}}
	\end{align}
\end{lemma}
\begin{lemma}[{\cite[Lemma~3.4]{ZH}}]\label{lem:convolution-iota}
	$\ZHR$ enables the computation of the Schur product of two maps of the form
	$\iota_{\vec{b}}\circ H_n(x)$ and $\iota_{\vec{b}}\circ H_n(y)$ for any $\vec{b}\in\mathbb B^n$ and $x,y\in R$:
	\begin{align}
		\vc{\InputIfFileExists{./figures/ZH/convolution_iota.tikz}{}{Missing file!}}
	\end{align}
\end{lemma}

We will now show that normal form diagrams, when combined in various ways,
can also be put into normal form.
Let
\begin{align}
	\vc{\InputIfFileExists{./figures/ZH/nf.tikz}{}{Missing file!}}
\end{align} denote an arbitrary normal-form diagram.
It is straightforward to see that permuting the outputs of a normal-form diagram merely interchanges the bits
in the coefficients $a_{\vec b}$.
Hence, normal forms are preserved under permutations of outputs.
Furthermore:

\begin{proposition}[{\cite[Proposition~3.5]{ZH}}]\label{prop:extension} 
	A diagram consisting of a normal form diagram juxtaposed
	with \whiteunit\ can be brought into normal form using the rules of $\ZHR$:
	\begin{align}
		\vc{\InputIfFileExists{./figures/ZH/extension.tikz}{}{Missing file!}}
	\end{align}
\end{proposition}

\begin{proposition}[{\cite[Proposition~3.6]{ZH}}]\label{prop:convolution}
	The Schur product of two normal form diagrams can be brought into normal form using the rules of $\ZHR$.
	\begin{align}
		\vc{\InputIfFileExists{./figures/ZH/convolution_nf.tikz}{}{Missing file!}}
	\end{align}
\end{proposition}

\begin{corollary}[{\cite[Corollary~3.7]{ZH}}]\label{cor:tensor-product} 
	The tensor product of two normal form diagrams can be brought into normal form using the rules of $\ZHR$.
\end{corollary}

\begin{remark}\label{rem:scalar-juxtaposition}
	Note that a single scalar H-box is a normal form diagram.
	Corollary~\ref{cor:tensor-product} thus implies that a diagram consisting of a normal form diagram juxtaposed with a scalar H-box can be brought into normal form.
	In the following proofs, we will therefore ignore scalars for simplicity: they can be added back in and then incorporated to the normal form without problems.
\end{remark}

\begin{proposition}[{\cite[Proposition~3.9]{ZH}}]\label{prop:contraction}
	The diagram resulting from applying \whitecounit\ to an output of a normal form diagram can be brought into normal form:
	\begin{align}
		\vc{\InputIfFileExists{./figures/ZH/whitecounit_nf.tikz}{}{Missing file!}}
	\end{align}
\end{proposition}

Our strategy will now be to show that any diagram can be decomposed into H-boxes, combined via the operations of extension, convolution, and contraction. This will give us a completeness proof, thanks to the following proposition.

\begin{lemma}{[\cite[Lemma~3.10]{ZH}}]\label{lem:H-box-nf}
	Any H-box can be brought into normal form using the rules of $\ZHR$.
\end{lemma}

\begin{corollary}[{\cite[Corollary~3.11]{ZH}}]\label{cor:cup-nf}
	The diagram of a single cup can be brought into normal form:
	\begin{align}
		\vc{\InputIfFileExists{./figures/ZH/cup_nf.tikz}{}{Missing file!}}
	\end{align}
\end{corollary}

\begin{corollary}[{\cite[Corollary~3.12]{ZH}}]\label{cor:whitemult-nf}
	The diagram resulting from applying \whitemult\ to a pair of outputs of a normal form diagram can be brought into normal form.
	\begin{align}
		\vc{\InputIfFileExists{./figures/ZH/whitemult_nf.tikz}{}{Missing file!}}\label{eq:whitemult-nf}
	\end{align}
\end{corollary}

\begin{corollary}[{\cite[Corollary~3.13]{ZH}}]\label{cor:cap-nf}
	Applying a cap to a normal form diagram results in another normal form diagram:
	\begin{align}
		\vc{\InputIfFileExists{./figures/ZH/cap_nf.tikz}{}{Missing file!}}
	\end{align}
\end{corollary}

Thanks to Corollaries~\ref{cor:tensor-product} and \ref{cor:cap-nf},
we are able to turn any diagram of normal forms into a normal form.
It only remains to show that the generators of $\ZHR$ can themselves be made into normal forms.
We have already shown the result for H-boxes, so we only need the following.

\begin{lemma}[{\cite[Lemma~3.14]{ZH}}]\label{lem:Z-spider-nf}
	Any Z-spider can be brought into normal form using the rules of $\ZHR$.
\end{lemma}

\begin{theorem}[\ZHR\ is complete]
	For any $\ZHR$ diagrams $D_1$ and $D_2$,
	if $\interpret{D_1} = \interpret{D_2}$ then $D_1$ is convertible into $D_2$ using the rules of $\ZHR$.
\end{theorem}

\begin{proof}
	By Theorem~\ref{thm:nf-unique},
	it suffices to show that any $\ZHR$ diagram can be brought into normal form.
	Lemmas~\ref{lem:H-box-nf} and \ref{lem:Z-spider-nf} suffice to turn any generator into normal form.
	Corollary~\ref{cor:tensor-product} lets us turn any tensor product of generators into a normal form and Corollary~\ref{cor:cap-nf} lets us normalise any arbitrary wiring.
\end{proof}

\section{\texorpdfstring{\ZHR}{ZHR}'s relationship to RING} \label{appZHRing}

In this appendix we prove completeness for the set of rules for $\ring$
presented in \S\ref{secRingRZHRules},
using a translation of the complete set of rules for ZH$_R$ given in Figure~\ref{figZHRRules}.
The translation was provided in Figure~\ref{figRingRZHTranslation}.

\subsection{Translations of generators of RING to ZH and back} 
\label{secTranslationRulesZHQR}

\begin{itemize}
	\item Translation of $\spider{white}{}$ from $\ring_R$ to ZH and back:
	      \begin{align} \label{eqnZH-QR-times}
		      \interpret{\interpret{\spider{white}{}}_\ZH}_{\ring} = \interpret{\spider{white}{}}_{\ring} = \spider{white}{}
	      \end{align}
	\item Translation of $\binary[poly]{}$ from $\ring_R$ to ZH and back:
	      \begin{align} \label{eqnZH-QR-plus}
		      \interpret{\interpret{\binary[poly]{}}_\ZH}_{\ring} = \interpret{
			      \vc{\InputIfFileExists{./figures/ZH/ZH_plus.tikz}{}{Missing file!}}
		      }_{\ring}  =
		      \vc{\InputIfFileExists{./figures/RingR/QR_of_ZH_of_plus.tikz}{}{Missing file!}}
	      \end{align}
	\item Translation of $\state{white}{a}$ from $\ring_R$ to ZH and back:
	      \begin{align} \label{eqnZH-QR-state}
		      \interpret{\interpret{\state{white}{a}}_\ZH}_{\ring} =
		      \interpret{\state{ZH}{a}}_{\ring} =
		      \vc{\InputIfFileExists{./figures/QR/QR_of_state.tikz}{}{Missing file!}}
	      \end{align}

\end{itemize}

\subsection{Translation of the derived generators of ZH}

	      \begin{align} \label{eqnZH-QR-grey}
		      \vc{\InputIfFileExists{./figures/ZH/ZH_greyspider.tikz}{}{Missing file!}}
		      \mapsto
		      \vc{\InputIfFileExists{./figures/RingR/QR_of_greyspider.tikz}{}{Missing file!}}&\qquad \qquad
		      \unary[grey]{\neg}
		      \mapsto
		      \vc{\InputIfFileExists{./figures/RingR/QR_of_negate.tikz}{}{Missing file!}}
	      \end{align}

\subsection{Proof of completeness}

In this section we prove the following:

\thmRingRCompleteZH*

We will show that the rules of \S\ref{secRingRZHRules}, here called $\ring^2_R$,
can derive all of the equations $\interpret{L = R}_{\ring}$ where $L=R$ is a rule of ZH,
and all the equations $\interpret{\interpret{g}_\ZH}_{\ring} = g$ where $g$ is a generator of $\ring$.
First we will show some useful intermediate results.

\begin{proposition} \label{propQRHCycle}
	In $\ring^2_R$ we can `cycle' the gates of two adjacent Hadamard maps,
	and call this rule $H'$
	\begin{align}
		\vc{\InputIfFileExists{./figures/QR/h2_1.tikz}{}{Missing file!}} \by{H'} \vc{\InputIfFileExists{./figures/QR/h2_2.tikz}{}{Missing file!}} \by{H'} \vc{\InputIfFileExists{./figures/QR/h2_3.tikz}{}{Missing file!}} \tag{H'}
	\end{align}
\end{proposition}

\begin{proof} \label{prfPropQRHCycle}
	We simultaneously pre- and post-compose both sides of rule $H$ with the elements $\smallunary[polyT]{1}$ and $\smallunary[polyT]{-1}$
	for the first equality,
	then  $\smallunary[white]{-2}$ and $\smallunary[white]{-\half}$
	for the second.
\end{proof}

\begin{proposition} \label{propQRHInv}
	We can also find another useful representation of the Hadamard gate, and call this rule $H''$
	\begin{align}
		\unary[h2]{} = \vc{\InputIfFileExists{./figures/QR/h3_1.tikz}{}{Missing file!}} \by{H''} \vc{\InputIfFileExists{./figures/QR/h3_2.tikz}{}{Missing file!}} \tag{H''}
	\end{align}
\end{proposition}

\begin{proof} \label{prfPropQRHInv}
	We apply the right hand side of this rule (minus the scalar) to both sides of rule $H$,
	then cancel terms using $+_a$ and $S$.
\end{proof}
\begin{proposition} \label{propQRHNeg}
	The following derivation is helpful for the rules translated from ZH that contain $\smallunary[grey]{\neg}$:
	\begin{align}
		\vc{\InputIfFileExists{./figures/QR/h4_5.tikz}{}{Missing file!}} \by{\neg} \vc{\InputIfFileExists{./figures/QR/h4_6.tikz}{}{Missing file!}} \tag{$\neg$}
	\end{align}
\end{proposition}

\begin{proof} \label{prfPropQRHNeg}
	\begin{align}
		\vc{\InputIfFileExists{./figures/QR/h4_6.tikz}{}{Missing file!}}
		\by{H, sc} \vc{\InputIfFileExists{./figures/QR/h4_7.tikz}{}{Missing file!}}
		\by{+, S} \vc{\InputIfFileExists{./figures/QR/h4_8.tikz}{}{Missing file!}}   \\
		\by{H''} \vc{\InputIfFileExists{./figures/QR/h4_1.tikz}{}{Missing file!}}
		\by{D} \vc{\InputIfFileExists{./figures/QR/h4_2.tikz}{}{Missing file!}}
		\by{S, +_a} \vc{\InputIfFileExists{./figures/QR/h4_3.tikz}{}{Missing file!}} \\
		\by{D} \vc{\InputIfFileExists{./figures/QR/h4_4.tikz}{}{Missing file!}}
		\by{S} \vc{\InputIfFileExists{./figures/QR/h4_5.tikz}{}{Missing file!}}
	\end{align}
\end{proof}

\begin{corollary} \label{corDerivedPlus}
	The derived generator $\smallunary[poly]{a}$ obeys the following rules that follow immediately
	from $+$, $+_a$ and $D$, so we will use these labels when referring to these rules:
	\begin{align}
		\vc{\begin{tikzpicture}
	\begin{pgfonlayer}{nodelayer}
		\node [style=polynomial] (0) at (0, 0) {a};
		\node [style=white] (1) at (0, -0.5) {b};
		\node [style=none] (2) at (0, 0.5) {};
	\end{pgfonlayer}
	\begin{pgfonlayer}{edgelayer}
		\draw (0.center) to (1.center);
		\draw (2.center) to (0.center);
	\end{pgfonlayer}
\end{tikzpicture}
} &\by{+} \state{white}{a+b}        &
		\vc{\InputIfFileExists{./figures/QR/der_assoc_l.tikz}{}{Missing file!}} &\by{+_a} \unary[poly]{a+b}     &
		\vc{\InputIfFileExists{./figures/QR/der_distr_l.tikz}{}{Missing file!}} &\by{D} \vc{\InputIfFileExists{./figures/QR/der_distr_r.tikz}{}{Missing file!}}
	\end{align}
\end{corollary}

\begin{lemma} \label{lemQRImpliesZH-QR-ZS1} $\ring^2_R \entails F(\ZH_{ZS1})$, i.e.
	\begin{align}
		\ring^2_R \entails
		      \vc{\InputIfFileExists{./figures/ZH/Z_spider_rule.tikz}{}{Missing file!}}
	\end{align}
\end{lemma}

\begin{proof} \label{prfLemQRImpliesZH-QR-ZS1}
	This follows immediately from the rule S1
\end{proof}

\begin{lemma} \label{lemQRImpliesZH-QR-HS1}  $\ring^2_R \entails F(\ZH_{HS1})$, i.e.
	\begin{align}
		\ring^2_R \entails
		      \vc{\InputIfFileExists{./figures/QR/QR_of_HS1.tikz}{}{Missing file!}}
	\end{align}
\end{lemma}

\begin{proof} \label{prfPropQRImpliesZH-QR-HS1}
	\begin{align}
		       \vc{\InputIfFileExists{./figures/QR/pf_hs1_1.tikz}{}{Missing file!}}
		\by{S}
		\vc{\InputIfFileExists{./figures/QR/pf_hs1_2.tikz}{}{Missing file!}}
		\by{H'}
		\vc{\InputIfFileExists{./figures/QR/pf_hs1_2a.tikz}{}{Missing file!}} 
		\by{H} 
		\vc{\InputIfFileExists{./figures/QR/pf_hs1_3.tikz}{}{Missing file!}}
		\by{S}
		\vc{\InputIfFileExists{./figures/QR/pf_hs1_4.tikz}{}{Missing file!}}
	\end{align}
\end{proof}

\begin{proposition} \label{propQRImplesZH-QRZS2}   $\ring^2_R \entails F(\ZH_{ZS2})$, i.e.
	\begin{align}
		\ring^2_R \entails
		\vc{\InputIfFileExists{./figures/QR/QR_of_ZS2.tikz}{}{Missing file!}}
	\end{align}
\end{proposition}

\begin{proof} \label{prfPropQRImplesZH-QRZS2}
	\begin{align}
		\vc{\InputIfFileExists{./figures/QR/pf_zs2_1.tikz}{}{Missing file!}}
		\by{S}
		\vc{\begin{tikzpicture}
	\begin{pgfonlayer}{nodelayer}
		\node [style=none] (6) at (-3.75, 0.75) {};
		\node [style=white] (7) at (-3.75, 0) {};
		\node [style=none] (9) at (-3.75, -0.75) {};
	\end{pgfonlayer}
	\begin{pgfonlayer}{edgelayer}
		\draw (9.center) to (7.center);
		\draw (6.center) to (7.center);
	\end{pgfonlayer}
\end{tikzpicture}
}
		\by{\times_1}
		\vc{\begin{tikzpicture}
	\begin{pgfonlayer}{nodelayer}
		\node [style=none] (6) at (-3.75, 0.75) {};
		\node [style=none] (9) at (-3.75, -0.75) {};
	\end{pgfonlayer}
	\begin{pgfonlayer}{edgelayer}
		\draw (6.center) to (9.center);
	\end{pgfonlayer}
\end{tikzpicture}
}
	\end{align}
\end{proof}

\begin{proposition} \label{propQRImpliesZH-QR-HS2}  $\ring^2_R \entails F(\ZH_{HS2})$, i.e.
	\begin{align}
		\ring^2_R  \entails
		\vc{\InputIfFileExists{./figures/QR/QR_of_HS2.tikz}{}{Missing file!}}
	\end{align}
\end{proposition}

\begin{proof} \label{prfPropQRImpliesZH-QR-HS2}
	This is the same as rule H
\end{proof}

\begin{proposition} \label{propQRImpliesZH-QR-BA1}   $\ring^2_R \entails F(\ZH_{BA1})$, i.e.
	\begin{align}
		\ring^2_R  \entails
		\vc{\InputIfFileExists{./figures/QR/QR_of_BA1.tikz}{}{Missing file!}}
	\end{align}
\end{proposition}

\begin{proof} \label{prfPropQRImpliesZH-QR-BA1}
	This is the same as rule BA1
\end{proof}

\begin{proposition} \label{propQRImpliesZH-QR-BA2}   $\ring^2_R \entails F(\ZH_{BA2})$, i.e.
	\begin{align}
		\ring^2_R  \entails
		\vc{\InputIfFileExists{./figures/QR/QR_of_BA2.tikz}{}{Missing file!}}
	\end{align}
\end{proposition}

\begin{proof} \label{prfPropQRImpliesZH-QR-BA2}
	\begin{align}
		  &
		\vc{\InputIfFileExists{./figures/QR/pf_ba2_1.tikz}{}{Missing file!}}
		\by{S}
		\vc{\InputIfFileExists{./figures/QR/pf_ba2_2.tikz}{}{Missing file!}}
		\by{+_a}
		\vc{\InputIfFileExists{./figures/QR/pf_ba2_2a.tikz}{}{Missing file!}}
		\\
		\by{H}
		  &
		\vc{\InputIfFileExists{./figures/QR/pf_ba2_3.tikz}{}{Missing file!}}
		\by{sc, BA2}
		\vc{\InputIfFileExists{./figures/QR/pf_ba2_4.tikz}{}{Missing file!}}
		\by{+_a, S}
		\vc{\InputIfFileExists{./figures/QR/pf_ba2_5.tikz}{}{Missing file!}}
	\end{align}
\end{proof}

\begin{lemma} \label{lemLemmaQRImpliesZH-QR-M}   $\ring^2_R \entails F(\ZH_{M})$, i.e.
	\begin{align}
		\ring^2_R  \entails
		\vc{\InputIfFileExists{./figures/QR/QR_of_M.tikz}{}{Missing file!}}
	\end{align}
\end{lemma}

\begin{proof} \label{prfLemLemmaQRImpliesZH-QR-M}
	\begin{align}
		  &
		\vc{\InputIfFileExists{./figures/QR/pf_m_1.tikz}{}{Missing file!}}
		\by{\times}
		\state{white}{a \times b}
		\by{+}
		\vc{\begin{tikzpicture}
	\begin{pgfonlayer}{nodelayer}
		\node [style=white] (3) at (0, -0.75) {$(a \times b) - 1$};
		\node [style=none] (4) at (0, 0.5) {};
		\node [style=polynomial] (5) at (0, 0) {1};
	\end{pgfonlayer}
	\begin{pgfonlayer}{edgelayer}
		\draw (5.center) to (3.center);
		\draw (4.center) to (5.center);
	\end{pgfonlayer}
\end{tikzpicture}
}
	\end{align}
\end{proof}

\begin{lemma} \label{lemQREntailsQR-ZH-U}   $\ring^2_R \entails F(\ZH_{U})$, i.e.
	\begin{align}
		\ring^2_R  \entails
		\vc{\InputIfFileExists{./figures/QR/QR_of_U.tikz}{}{Missing file!}}
	\end{align}
\end{lemma}

\begin{proof} \label{prfLemQREntailsQR-ZH-U}
	This is an instance of rule $+$
\end{proof}

\begin{lemma} \label{lemQREntailsQR-ZH-QR-plus}  $\ring^2_R \entails$ equation \eqref{eqnZH-QR-plus}
(the re-translation of the $+$ generator), i.e.
	\begin{align}
		\ring^2_R  \entails
		\binary[poly]{} =
		\vc{\InputIfFileExists{./figures/RingR/QR_of_ZH_of_plus.tikz}{}{Missing file!}}
	\end{align}
\end{lemma}

\begin{proof} \label{prfLem}
	\begin{align}
		             & \vc{\InputIfFileExists{./figures/QR/pf_plus_1.tikz}{}{Missing file!}}
		\by{S, +_a, +_0}
		\vc{\InputIfFileExists{./figures/QR/pf_plus_2.tikz}{}{Missing file!}} \\
		\by{H, sc} &
		\vc{\InputIfFileExists{./figures/QR/pf_plus_3.tikz}{}{Missing file!}} 
		\by{+, S} 
		\vc{\InputIfFileExists{./figures/QR/pf_plus_4.tikz}{}{Missing file!}} 
		\by{S} 
		\vc{\InputIfFileExists{./figures/QR/pf_plus_6.tikz}{}{Missing file!}} \\
		\by{H'', sc} &
		\vc{\InputIfFileExists{./figures/QR/pf_plus_7.tikz}{}{Missing file!}} 
		\by{a}
		\vc{\InputIfFileExists{./figures/QR/pf_plus_8.tikz}{}{Missing file!}}
		\by{S}
		\vc{\InputIfFileExists{./figures/QR/pf_plus_9.tikz}{}{Missing file!}}
		\by{+'}
		\binary[poly]{}
	\end{align}
\end{proof}

\begin{lemma} \label{lemQREntailsQR-ZH-A}  $\ring^2_R \entails F(\ZH_{A})$, i.e.
	\begin{align}
		\ring^2_R  \entails
		\vc{\InputIfFileExists{./figures/QR/QR_of_A.tikz}{}{Missing file!}}
	\end{align}
\end{lemma}

\begin{proof} \label{prfLemQREntailsQR-ZH-A}
	\begin{align}
		         & \vc{\InputIfFileExists{./figures/QR/pf_a_1.tikz}{}{Missing file!}}
		\by{H, sc}
		\vc{\InputIfFileExists{./figures/QR/pf_a_2.tikz}{}{Missing file!}}
		\by{+_a, S}
		\vc{\InputIfFileExists{./figures/QR/pf_a_3.tikz}{}{Missing file!}} \\
		\by{S}   &
		\vc{\InputIfFileExists{./figures/QR/pf_a_4.tikz}{}{Missing file!}}
		\by{+}
		\vc{\InputIfFileExists{./figures/QR/pf_a_5.tikz}{}{Missing file!}} \\
		\by{H''} &
		\vc{\InputIfFileExists{./figures/QR/pf_a_7.tikz}{}{Missing file!}}
		\by{+_a, +_0}
		\vc{\InputIfFileExists{./figures/QR/pf_a_8.tikz}{}{Missing file!}}
		\by{S}
		\vc{\InputIfFileExists{./figures/QR/pf_a_9.tikz}{}{Missing file!}} \\
		\by{+'}  &
		\vc{\InputIfFileExists{./figures/QR/pf_a_10.tikz}{}{Missing file!}}
		\by{+, S}
		\vc{\begin{tikzpicture}
	\begin{pgfonlayer}{nodelayer}
		\node [style=none] (1) at (-2, 2.75) {};
		\node [style=white] (30) at (-1.25, 2.25) {};
		\node [style=white] (32) at (-2, 1.5) {$\frac{a+b}{2}$};
	\end{pgfonlayer}
	\begin{pgfonlayer}{edgelayer}
		\draw (1.center) to (32.center);
	\end{pgfonlayer}
\end{tikzpicture}
} \by{+_a} RHS
	\end{align}
\end{proof}

\begin{lemma} \label{lemQREntailsQR-ZH-I}  $\ring^2_R \entails F(\ZH_{I})$, i.e.
	\begin{align}
		\ring^2_R  \entails
		\vc{\InputIfFileExists{./figures/QR/QR_of_I.tikz}{}{Missing file!}}
	\end{align}
\end{lemma}

\begin{proof} \label{prfLemQREntailsQR-ZH-I}
	\begin{align}
		       & \vc{\InputIfFileExists{./figures/QR/pf_i_1.tikz}{}{Missing file!}}
		\by{\neg}
		\vc{\InputIfFileExists{./figures/QR/pf_i_2.tikz}{}{Missing file!}}
		\by{+_a, +_0}
		\vc{\InputIfFileExists{./figures/QR/pf_i_3.tikz}{}{Missing file!}} \\
		\by{D} &
		\vc{\InputIfFileExists{./figures/QR/pf_i_4.tikz}{}{Missing file!}}
		\by{I}
		\vc{\begin{tikzpicture}
	\begin{pgfonlayer}{nodelayer}
		\node [style=white] (9) at (-1.5, 0.5) {$a$};
		\node [style=white] (10) at (-1.5, -0.75) {};
		\node [style=none] (11) at (-1.5, 1.5) {};
		\node [style=none] (12) at (-1.5, -1.5) {};
	\end{pgfonlayer}
	\begin{pgfonlayer}{edgelayer}
		\draw (11.center) to (9.center);
		\draw (10.center) to (12.center);
	\end{pgfonlayer}
\end{tikzpicture}
}
	\end{align}
\end{proof}

\begin{lemma} \label{lemQREntailsQR-ZH-O}  $\ring^2_R \entails F(\ZH_{O})$, i.e.
	\begin{align}
		\ring^2_R  \entails
		\vc{\InputIfFileExists{./figures/QR/QR_of_O.tikz}{}{Missing file!}}
	\end{align}
\end{lemma}

\begin{proof} \label{prfLemQREntailsQR-ZH-O}
	\begin{align}
		\vc{\InputIfFileExists{./figures/QR/pf_o_1.tikz}{}{Missing file!}}
		\by{H, sc}
		\vc{\InputIfFileExists{./figures/QR/pf_o_2.tikz}{}{Missing file!}}
		\by{+, S}
		\vc{\InputIfFileExists{./figures/QR/pf_o_3.tikz}{}{Missing file!}} \\
		\by{H''}
		\vc{\InputIfFileExists{./figures/QR/pf_o_4.tikz}{}{Missing file!}}
		\by{D}
		\vc{\InputIfFileExists{./figures/QR/pf_o_5.tikz}{}{Missing file!}}
		\by{D}
		\vc{\InputIfFileExists{./figures/QR/pf_o_6.tikz}{}{Missing file!}} \\
		\by{H''}
		\vc{\InputIfFileExists{./figures/QR/pf_o_7.tikz}{}{Missing file!}}
		\by{D}
		\vc{\InputIfFileExists{./figures/QR/pf_o_8.tikz}{}{Missing file!}}
		\by{D}
		\vc{\InputIfFileExists{./figures/QR/pf_o_9.tikz}{}{Missing file!}}
	\end{align}
\end{proof}

\begin{lemma} \label{lemQREntailsQR-ZH-times}  $\ring^2_R \entails$ equation \eqref{eqnZH-QR-times}
(the re-translation of the white spider) i.e.
	\begin{align}
		\ring^2_R  \entails
		\spider{white}{} = \spider{white}{}
	\end{align}
\end{lemma}

\begin{proof} Nothing to prove
\end{proof}
\begin{lemma} \label{lemQREntailsQR-ZH-state}  $\ring^2_R \entails$ equation \eqref{eqnZH-QR-state}
(the re-translation of the state $a$)  i.e.
	\begin{align}
		\ring^2_R  \entails
		\state{white}{a} = \vc{\InputIfFileExists{./figures/QR/QR_of_state.tikz}{}{Missing file!}}
	\end{align}
\end{lemma}

\begin{proof} Instance of the $+$ rule
\end{proof}

We have shown that the ruleset $\ring^2_R$ can prove all the translated rules of $\ZH_R$,
and can prove that generators equal their images under the translation into ZH and back.
The ruleset $\ring^2_R$ is therefore complete for $R$ any commutative ring with a half,
because $\ZH_R$ is complete for any commutative ring with a half.

\section{Trigonometry proofs} \label{secTrig}

In this appendix we prove the trigonometric proofs required in \S\ref{chapZQ}.
We reproduce the side conditions for the rule $\ZX_{EU'}$ for reference here:

\begin{displayquote}[Figure~2, A Near-Minimal Axiomatisation of ZX-Calculus
		for Pure Qubit Quantum Mechanics \cite{VilmartZX}]
	In rule (EU'), $\beta_1$, $\beta_2$, $\beta_3$ and $\gamma$
	can be determined as follows: $x^+ := \frac{\alpha_1 + \alpha_2}{2}$,
	$x^- := x^-\alpha_2$, $z := - \sin (x^+) + i \cos(x^-)$
	and $z' := \cos(x^+) - i \sin (x^-)$,
	then $\beta_1 = \arg z + \arg z'$, $\beta_2 = 2 \arg(i + \frac{\abs{z}}{\abs{z'}})$,
	$\beta_3 = \arg z - \arg z'$, $\gamma = x^+ - \arg(z) + \frac{\pi - \beta_2}{2}$
	where by convention $\arg(0) := 0$ and $z' = 0 \implies \beta_2 = 0$.
\end{displayquote}

\lemZQHComm*

\begin{proof} \label{prfLemZQHComm}
	\begin{align}
		(\pi, \frac{x+z}{\sqrt{2}}) \times (\alpha, z) =                 & \frac{1}{\sqrt{2}}(i+k)(\cos \frac{\alpha}{2} + k \sin \frac{\alpha}{2})                                              \\
		=                                                                & \frac{1}{\sqrt{2}}(-\sin\frac{\alpha}{2} + i \cos \frac{\alpha}{2} - j \sin \frac{\alpha}{2} +k \cos \frac{\alpha}{2} \\
		=                                                                & \frac{1}{\sqrt{2}}(\cos\frac{\alpha}{2}+i\sin\frac{\alpha}{2})(i+k)                                                   \\
		=                                                                & (\alpha, x) \times (\pi, \frac{x+z}{\sqrt{2}})                                                                        \\
		\nonumber \\
		(\pi, \frac{x+z}{\sqrt{2}})\times(\alpha, x) =                   & \frac{1}{\sqrt{2}}(i+k)(\cos \frac{\alpha}{2} + i \sin \frac{\alpha}{2})                                              \\
		=                                                                & \frac{1}{\sqrt{2}}(-\sin\frac{\alpha}{2} + i \cos \frac{\alpha}{2} + j \sin \frac{\alpha}{2} +k \cos \frac{\alpha}{2} \\
		=                                                                & \frac{1}{\sqrt{2}}(\cos\frac{\alpha}{2}+k\sin\frac{\alpha}{2})(i+k)                                                   \\
		=                                                                & (\alpha, z) \times (\pi, \frac{x+z}{\sqrt{2}})                                                                        \\
		\nonumber \\
		(\pi, \frac{x+z}{\sqrt{2}}) \times (\pi, \frac{x+z}{\sqrt{2}}) = & \frac{1}{\sqrt{2}}(i+k)\frac{1}{\sqrt{2}}(i+k)                                                                        \\
		=                                                                & \frac{1}{2}(-1-j-1+j) = -1
	\end{align}
\end{proof}

\lemZQEUQuaternion*

In the hope of easing legibility we separate out the real, $i$, $j$, and $k$ components
of quaternions onto separate lines where suitable.

\begin{proof} \label{prfLemZQEUQuaternion}
	\begin{align}
		RHS =& H \times (\beta_1, z) \times H  \times (\beta_2,z) \times H \times (\beta_3, z) \times H \\
		=&H \times H \times H \times H \times (\beta_1, x) \times  (\beta_2,z) \times (\beta_3, x) & (\ref{lemZQHComm}) \\
		=& (\beta_1, x) \times  (\beta_2,z) \times (\beta_3, x) & (\ref{lemZQHComm}) \\
		=& (\cos \beta_1 / 2 + i \sin \beta_1 / 2) \times \\\nonumber& (\cos \beta_2 / 2 + k \sin \beta_2 / 2) \times \\\nonumber& (\cos \beta_3 / 2 + i \sin \beta_3 / 2) \\
		=& 1(\cos \beta_1 /2 \cos \beta_2/2 \cos \beta_3/2 - \sin \beta_1 / 2 \cos \beta_2 /2 \sin \beta_3 / 2) +\\
		\nonumber & i (\cos \beta_1 /2 \cos \beta_2/2 \sin \beta_3/2 + \sin \beta_1 / 2 \cos \beta_2 /2 \cos \beta_3 / 2) + \\
		\nonumber&j (\cos \beta_1 /2 \sin \beta_2/2 \sin \beta_3/2 - \sin \beta_1 / 2 \sin \beta_2 /2 \cos \beta_3 / 2) + \\
		\nonumber&k (\cos \beta_1 /2 \sin \beta_2/2 \cos \beta_3/2 + \sin \beta_1 / 2 \sin \beta_2 /2 \sin \beta_3 / 2) \\
		= & 1 (\cos \beta_2/2)(\cos \beta_1 /2 \cos \beta_3/2 - \sin \beta_1 / 2 \sin \beta_3 / 2) + \\
		\nonumber & i (\cos \beta_2 /2)(\cos \beta_1 /2 \sin \beta_3/2 + \sin \beta_1 / 2  \cos \beta_3 / 2) + \\
		\nonumber & j (\sin \beta_2/2)(\cos \beta_1 /2  \sin \beta_3/2 - \sin \beta_1 / 2  \cos \beta_3 / 2) + \\
		\nonumber & k (\sin \beta_2/2)(\cos \beta_1 /2 \cos \beta_3/2 + \sin \beta_1 / 2\sin \beta_3 / 2) \\
		= & 1 (\cos \beta_2/2)(\cos \frac{\beta_1 + \beta_3}{2}) + \\
		\nonumber & i (\cos \beta_2 /2)(\sin \frac{\beta_1 + \beta_3}{2}) + \\
		\nonumber & j (\sin \beta_2/2)(\sin \frac{\beta_3 - \beta_1}{2}) + \\
		\nonumber & k (\sin \beta_2/2)(\cos \frac{\beta_1 - \beta_3}{2}) \\
		= & 1 (\cos \beta_2/2)(\cos \arg z) + \\
		\nonumber & i (\cos \beta_2 /2)(\sin \arg z) + \\
		\nonumber & j (\sin \beta_2/2)(- \sin \arg z') + \\
		\nonumber & k (\sin \beta_2/2)(\cos \arg z') \\
	\end{align}
	Using properties of arguments and moduli we then show the following:
	\begin{align}
		\cos \arg (a + ib) =     & a / \abs{a+ib}                                                                            \\
		\sin \arg (a + ib) =     & b / \abs{a+ib}                                                                            \\
		\abs{z}^2 =              & \sin(x^+)^2 + \cos(x^-)^2                                                                 \\
		\abs{z'}^2 =             & \sin(x^-)^2 + \cos(x^+)^2                                                                 \\
		\abs{z}^2 + \abs{z'}^2 = & \cos^2 x^+ + \sin^2 x^+ + \cos^2 x^- + \sin^2 x^- = 2                                     \\                                            \\
		\cos \arg z =            & \Re(z) / \abs{z} =  \frac{-\sin(\alpha_1 + \alpha_2)/2}{\abs{z}}                          \\
		\sin \arg z =            & \Im(z) / \abs{z} = \frac{\cos(\alpha_1 - \alpha_2)/2}{\abs{z}}                            \\
		\cos \arg z' =           & \Re(z') / \abs{z'}  = \frac{\cos(\alpha_1 + \alpha_2)/2}{\abs{z'}}                        \\
		\sin \arg z' =           & \Im(z') / \abs{z'} =\frac{-\sin(\alpha_1 - \alpha_2)/2}{\abs{z'}}                         \\
		\nonumber \\
		\cos(\beta_2 / 2) =      & \cos \arg (i + \abs{z}/\abs{z'})  = \cos \arg (\abs{z'}i + \abs{z})  = \abs{z} / \sqrt{2} \\
		\sin(\beta_2 / 2) =      & \sin \arg (i + \abs{z}/\abs{z'}) = \sin \arg (\abs{z'}i + \abs{z}) =\abs{z'} / \sqrt{2}
	\end{align}
	And now substitute these values into our expression for the right hand side:
	\begin{align}
		RHS =     & 1 (\cos \beta_2/2)(\cos \arg z) +                                           \\
		\nonumber & i (\cos \beta_2 /2)(\sin \arg z) +                                          \\
		\nonumber & j (\sin \beta_2/2)(- \sin \arg z') +                                        \\
		\nonumber & k (\sin \beta_2/2)(\cos \arg z')                                            \\
		=         & 1 (\abs{z} / \sqrt{2})(\frac{-\sin(\alpha_1 + \alpha_2)/2}{\abs{z}}) +      \\
		\nonumber & i (\abs{z} / \sqrt{2})(\frac{\cos(\alpha_1 - \alpha_2)/2}{\abs{z}}) +       \\
		\nonumber & j (\abs{z'} / \sqrt{2} )(- \frac{-\sin(\alpha_1 - \alpha_2)/2}{\abs{z'}}) + \\
		\nonumber & k (\abs{z'} / \sqrt{2} )(\frac{\cos(\alpha_1 + \alpha_2)/2}{\abs{z'}})      \\
		=         & (\frac{1}{\sqrt{2}}) \times                                                 \\
		\nonumber & (-1 (\sin(\alpha_1 + \alpha_2)/2) +                                         \\
		\nonumber & i (\cos(\alpha_1 - \alpha_2)/2) +                                           \\
		\nonumber & j (\sin(\alpha_1 - \alpha_2)/2) +                                           \\
		\nonumber & k (\cos(\alpha_1 + \alpha_2)/2))                                            \\
	\end{align}
	And now for the left hand side:
	\begin{align}
		LHS =     & (\alpha_1, z) \times H \times (\alpha_2, z)                                                                                             \\
		=         & \frac{1}{\sqrt{2}} (\cos\alpha_1 + k\sin\alpha_1) (i+k) (\cos\alpha_2 + k\sin\alpha_2)                                                  \\
		=         & \frac{1}{\sqrt{2}} (i \cos\alpha_1 \cos\alpha_2 - j \cos\alpha_1\sin\alpha_2 + k \cos\alpha_1 \cos\alpha_2 - \cos\alpha_1\sin\alpha_2 + \\
		\nonumber & j \sin\alpha_1\cos\alpha_2 + i \sin\alpha_1\sin\alpha_2  - \sin\alpha_1\cos\alpha_2 - k \sin\alpha_1\sin\alpha_2)                       \\
		=         & \frac{1}{\sqrt{2}} (-(\cos\alpha_1\sin\alpha_2 + \sin\alpha_1\cos\alpha_2)                                                              \\
		\nonumber & i(\sin\alpha_1\sin\alpha_2 + \cos\alpha_1\alpha_2) +                                                                                    \\
		\nonumber & j(\sin\alpha_1\cos\alpha_2 - \cos\alpha_1\sin\alpha_2) +                                                                                \\
		\nonumber & k(\cos\alpha_1\cos\alpha_2 - \sin\alpha_1\sin\alpha_2))                                                                                 \\
		=         & (\frac{1}{\sqrt{2}}) \times                                                                                                             \\
		\nonumber & (-1 (\sin(\alpha_1 + \alpha_2)/2) +                                                                                                     \\
		\nonumber & i (\cos(\alpha_1 - \alpha_2)/2) +                                                                                                       \\
		\nonumber & j (\sin(\alpha_1 - \alpha_2)/2) +                                                                                                       \\
		\nonumber & k (\cos(\alpha_1 + \alpha_2)/2))
	\end{align}
\end{proof}

\begin{proposition} \label{propZQPhi}
	The map $\phi$, given by
	\begin{align}
		  & \phi : & (\text{Unit Quaternions}, \times) & \to (2\times 2\text{ complex matrices}, \comp) \\
		  & \phi:  & q_w + iq_x + jq_y + kq_z          & \mapsto \begin{pmatrix}
			q_w - iq_z & q_y - iq_x \\
			-q_y-iq_x  & q_w + iq_z
		\end{pmatrix}
	\end{align}
	is a group homomorphism with trivial kernel.
\end{proposition}

\begin{proof} \label{prfPropZQPhi}

	Write $q_1$ as $w + x + y + z$ and $q_2$ as $w' + x' + y' + z'$:

	\begin{itemize}
		\item Show that $\phi(1) = \begin{pmatrix}
				      1 & 0 \\ 0 & 1
			      \end{pmatrix}$:
		      \begin{align}
			      \phi(1) = & \begin{pmatrix}
				      1 - i0 & 0 - i0 \\
				      -0 -i0 & 1 + i0
			      \end{pmatrix}  \\
			      =         & \begin{pmatrix}
				      1 & 0 \\ 0 & 1
			      \end{pmatrix}
		      \end{align}
		\item Show that $\phi(q_1)\phi(q_2) = \phi(q_1 \times q_2)$:
		      \begin{align}
			      LHS =         &
			      \begin{pmatrix}
				      w - iz & y - ix \\
				      -y-ix  & w + iz
			      \end{pmatrix} \comp
			      \begin{pmatrix}
				      w' - iz' & y' - ix' \\
				      -y'-ix'  & w' + iz'
			      \end{pmatrix} \\
			      LHS_{(1,1)} = & ( (ww' - zz'-yy' - xx') - i(wz' - xy' + yx' + w'z ) ) \\
			      LHS_{(1,2)} = & ( ((w -iz)(y' - i x') + (y-ix)(w'+iz'))               \\
			      =             & ( wy'- zx' + yw'+xz') - i  (wx' + xw' - yz' + zy')    \\
			      LHS_{(2,1)} = & ( (-y - ix)(w'-iz') + (w+iz)(-y'-ix') )               \\
			      =             & - (yw' + xz' + wy' - zx') - i (wx' + xw' - yz' + zy') \\
			      LHS_{(2,2)} = & ( (-y -ix)(y'-ix') + (w + iz)(w' + i z') )            \\
			      =             & ( ww'-xx'-yy'-zz') +i( wz' -xy' + yx' +zw')           \\
			      \nonumber\\
			      RHS_{(1,1)} = & (  ww' - xx'-yy'-zz') - i( wz' + xy'- yx' + zw')      \\
			      RHS_{(1,2)} = & ( wy' - xz' + yw') + zx' - i( wx' + xw' +yz' - zy')   \\
			      RHS_{(2,1)} = & -( wy' - xz' + yw') + zx')-i( wx' + xw' +yz' - zy')   \\
			      RHS_{(2,2)} = & ( ww' - xx'-yy'-zz') + i( wz' + xy'- yx' + zw')
		      \end{align}
		\item Show that $\phi(q) = \begin{pmatrix}
				      1 & 0 \\ 0 & 1
			      \end{pmatrix} \implies q = 1$:
		      Looking at the matrix entries individually:
		      \begin{align}
			      1 =            & \ q_w - iq_z     \\
			                     & \implies q_w = 1 \\
			                     & \implies q_z = 0 \\
			      0 =            & -q_y - iq_x      \\
			                     & \implies q_y = 0 \\
			                     & \implies q_x = 0 \\
			      \therefore q = & 1
		      \end{align}
	\end{itemize}

\end{proof}

\section{Proof of the $n$-joined ZH spider law}
\label{secnfoldZH}

\lemZHNBox*

\begin{proof} \label{prfZHNBox}
	Aiming to proceed by induction we first note the base case where $n=1$:
	
	\begin{align}
		\vc{\InputIfFileExists{./figures/ZH/nspider_base.tikz}{}{Missing file!}}
	\end{align}

	This is shown by checking soundness by hand.
	For the induction step we will need the following lemma:

	\begin{align}
		\vc{\InputIfFileExists{./figures/ZH/nspider_l_1.tikz}{}{Missing file!}}
	\end{align}

	Which again can be shown to be sound by direct computation.
	For the induction step we assume the original statement is true for $n-1$, $n > 1$, we then add another joining Hadamard.

	\begin{align}
		\vc{\InputIfFileExists{./figures/ZH/nspider_l_2.tikz}{}{Missing file!}} & = \vc{\InputIfFileExists{./figures/ZH/nspider_l_3.tikz}{}{Missing file!}}  \\
		\vc{\InputIfFileExists{./figures/ZH/nspider_l_4.tikz}{}{Missing file!}} & = \vc{\InputIfFileExists{./figures/ZH/nspider_l_5.tikz}{}{Missing file!}}
	\end{align}
\end{proof}

\section{Worked example of !-box removal}
\label{secVerificationExample}

In this appendix we give a longer example of !-box removal.
Consider the following $\ZH_\bbC$ equation,
which is not sound, but is suitable for demonstrating the principle of !-box removal.
\begin{align}
	\ZHExample{X}
\end{align}

We first need to work out the \emph{join},
which we do by counting the wires entering the !-boxes.
\begin{align}
	2^1 + 2^0 = 3
\end{align}

By Theorem~\ref{thmbbox2} we need to check $d = \set{0, 1, 2, 3}$ instances of the !-box.
So in order to verify:
\begin{align}
	\ZHExample{X}
\end{align}

We just need to verify these equations:

\begin{align}
\ZHExampleZero{X}\quad &,\quad \ZHExampleOne{X} \\
\ZHExampleTwo{X}\quad &,\quad \ZHExampleThree{X}
\end{align}

Note that these equations contain no !-boxes.
 \thispagestyle{empty} 
\cleardoublepage \addcontentsline{toc}{section}{Index}
\normalsize \printindex[main]

\end{document}